\documentclass[12pt,a4paper]{article}
\usepackage[utf8]{inputenc}
\usepackage[T1]{fontenc}
\usepackage{lmodern}
\usepackage{amsmath,amssymb,amsthm}
\usepackage{graphicx}
\usepackage{caption}
\usepackage{subcaption}
\usepackage{booktabs}
\usepackage{tabularx}
\usepackage{cite}
\usepackage{float}
\usepackage{tikz}
\usepackage{braket}
\usepackage{enumitem}
\usepackage{mathtools}
\usepackage{mathrsfs}
\usetikzlibrary{shapes,arrows,positioning}
\usepackage{hyperref}
\hypersetup{
	colorlinks=true,
	linkcolor=blue,
	citecolor=blue,
	urlcolor=blue,
}

\usepackage{geometry}
\newtheorem{theorem}{Theorem}
\newtheorem{lemma}{Lemma}
\newtheorem{property}{Property}
\newtheorem{definition}{Definition}
\newtheorem{remark}{Remark}
\newtheorem{example}{Example}
\newtheorem{proposition}{Proposition}

\newtheorem{corollary}{Corollary}

\title{
	Packaged Quantum States for Gauge-Invariant Quantum Computation and Communication
}

\author{
	Rongchao Ma \\
	\textit{Department of Physics, University of Alberta, Edmonton, Canada}\\
}

\date{\today}

\begin{document}
	
\maketitle

\begin{abstract}
	Packaged quantum states are gauge-invariant states in which all internal quantum numbers (IQNs) form an inseparable block.
	This feature gives rise to novel packaged entanglements that encompass all IQNs, which is important both for fundamental physics and for quantum technology. 
	Here we develop a framework for gauge-invariant quantum information processing based on packaged quantum states.
	We propose the necessary and sufficient conditions for a valid packaged superposition state of a single particle and multi-particle.
	We then present the details of constructing gauge-invariant packaged qubits (or qudits), packaged gates, and packaged circuits (which commute with the total charge operator).
	These serve as alternative foundation for gauge-invariant quantum information science.
	We then adapt conventional quantum error-correction codes, quantum algorithms, and quantum communication protocols to the ($d \times D$)-dimensional hybrid-packaged subspace.
	This high-dimensional hybrid-packaged subspace is flexible for pruning and scaling to match available physics systems.
	Thus, packaged quantum information processing becomes feasible and testable.	
	Our results show that the gauge-invariant packaged quantum states may provide a possible route toward robust, fault-tolerant, and secure quantum technologies.
\end{abstract}

\tableofcontents

\section{Introduction}
\label{SEC:Introduction}

Traditional quantum information science usually uses external \textbf{degrees of freedom (DOFs)} (such as spin, momentum, polarization, or magnetic flux) and assumes that arbitrary superpositions and local operations are available.
This formalism has two drawbacks:
(1) \textbf{Internal quantum numbers (IQNs)} (electric charge, flavor, color, etc.) \cite{WeinbergBook,PeskinSchroeder} are not fully exploited, because they are constrained by the fundamental symmetries of realistic physical systems.
(2) Quantum systems are vulnerable to gauge‑violating errors \cite{Poulin2005,Bacon2006,Zohar2016}.

Recently, we demonstrated the validity of packaged entangled states \cite{Ma2017} using quantum field theory \cite{Ma2025} and group theory \cite{MaGroup2025}.
Due to the packaging principle, all IQNs in the packaged entangled states are inseparably entangled.
In other words, the packaging principle prohibits entanglements in which only part of the IQNs are entangled (e.g., the charge $Q$), while leaving the other quantum numbers unentangled.
These packaged entangled states are restricted to a single superselection sector \cite{WWW1952,DHR1971,DHR1974,StreaterWightman2001} and transform covariantly within that fixed net‐charge sector.
Furthermore, external DOFs (e.g., spin or momentum) can combine with the packaged IQNs to yield hybrid gauge-invariant \cite{Feynman1949,Yang1954,Utiyama1956,Weinberg1967} entanglement.
These unique features have been applied to quantum simulation of lattice gauge theories \cite{MaLGT2025}.

Evidence for packaged entanglement include early works in high-energy physics about electric charge entanglement \cite{WuShaknov1950}, flavor entanglement \cite{Abe2001,Aubert2002}, and color entanglement \cite{Brandelik1979,Abazov2012}.
Recent experimental results for packaged entanglement include the observation of entanglements in quarks, leptons \cite{TheATLASCollaboration2024,FabbrichesiFloreaniniGabrielli2023,Afik2024,TheCMSCollaboration2024,HayrapetyanCMSCollaboration2024, Blasone2009,Go2007}, massive gauge bosons \cite{Barr2022,AguilarSaavedra2023}, quantum correlations in quark-pairs \cite{TheATLASCollaboration2024,TheCMSCollaboration2024}, heavy mesons \cite{Go2007}, leptons \cite{FabbrichesiFloreaniniGabrielli2023}, neutrinos \cite{Blasone2009}, and tau leptons and bottom quarks \cite{Afik2024}.
These results indicate that it is both necessary and feasible to incorporate high-energy physics into quantum information science \cite{AfikNova2022,Afik2025}.

In this paper, we develop a comprehensive framework for gauge-invariant quantum computation and communication based on packaged quantum states.
We encode logical information within a fixed net-charge subspace and design all operations to commute with the total charge operator.
This ensures that the entire dynamics remain within a single superselection sector.
The intrinsic constraint aligns with the fundamental principles of quantum field theory.
Compared with conventional quantum information science, this packaged framework offers at least two advantages:
(1) It exploits the full structure of IQNs. For instance, simulations of lattice gauge theory become native \cite{MaLGT2025}.
(2) It leads to significant gauge-violating error suppression and fault tolerance \cite{Steane1996,Kitaev2003,Bombin2006}.

Our manuscript is organized as follows.
Section \ref{SEC:GaugeInvariantPackagedStates} develops the foundations of gauge‑invariant quantum information processing, including the necessary and sufficient conditions for single‑particle packaged superposition states.
Section \ref{SEC:PackagedQubitsGatesAndCircuits} introduces packaged qubits, gates, and circuits.
Section \ref{SEC:PackagedStatesForQuantumCommunication} presents packaged resource states for quantum communication.
Section \ref{SEC:dDDimensionalHybridPackagedSpace} lifts these constructions to high‑dimensional hybrid‑packaged sub‑spaces.
Section \ref{SEC:ErrorAnalysisAndFaultToleranceHyb} analyzes errors and fault‑tolerance in the gauge‑invariant setting.
Section \ref{SEC:PackagedQuantumErrorCorrectionCodes} is devoted to packaged quantum error‑correction codes (Shor‑like, Steane‑like, and surface codes).
Section \ref{SEC:QuantumComputationInPackagedSpace} reconstructs conventional quantum algorithms in the hybrid‑packaged sub‑space.
Section \ref{SEC:QuantumCommunicationAndCryptography} explores enhanced quantum communication and cryptography protocols.
Section \ref{SEC:MetrologyAndSensingInPackagedSpace} investigates metrology and sensing applications.
Section \ref{SEC:ExperimentalProspects} discusses experimental prospects and implementations.

\begin{table}[H]
	\centering
	\caption{Symbols used in this manuscript}
	\begin{tabular}[hbt!]{|p{2cm}|p{5cm}|p{5cm}|}
		\hline\hline
		Symbols &Meaning &Location \\
		\hline
		IQN &Internal quantum number &Sec.~\ref{SEC:Introduction} \\
		\hline
		DOF &Degrees of freedom &Sec.~\ref{SEC:Introduction} \\
		\hline
		$\hat{Q}$ &Charge operator &Sec.~\ref{SEC:ValidPackagedSuperpositionStatesOfSingleParticle} \\
		\hline
		$d$ &Dimension of internal space &Sec.~\ref{SEC:ConstructingdDimensionalPackagedHilbertSpace} \\
		\hline
		$D$ &Dimension of external space &Sec.~\ref{SEC:ConstructingDDimensionalExternalHilbertSpace} \\
		\hline
		$N$ &Dimension of single-index hybrid-packaged subspace &Sec.~\ref{SEC:dDDimensionalHybridPackagedQudits} \\
		\hline
	\end{tabular}
\end{table}

\section{Gauge-Invariant Packaged States}
\label{SEC:GaugeInvariantPackagedStates}

In earlier works \cite{Ma2025,MaGroup2025}, we introduced the packaging principle and packaged entangled states in gauge-invariant quantum systems.
The packaging principle (see Appendix \ref{SEC:PackagingPrinciple}) states that, whenever a local gauge symmetry and superselection rules are present, the \textbf{internal quantum numbers (IQNs)} must appear in indivisible packaged blocks.
One cannot entangle only a subset of these IQNs while leaving the rest unentangled.
All the IQNs associated with a creation operator are locked together into a single \textbf{irreducible representation (irrep)} \cite{Weyl1925,Wigner1939} of the gauge group.

The packaging principle influences how quantum superpositions and entanglement can form in quantum field theories (QFTs) due to gauge constraints.
It then influences quantum computation (e.g., how to build packaged qubits, gates, and circuits that respect gauge-invariance) and quantum communication (e.g., how to send charge‐neutral or global‐only states through a channel without violating superselection).

In this section, we
propose the necessary and sufficient conditions for a valid packaged superposition state of a single particle,
show how to construct complete (maximal) orthonormal bases of multi‐particle packaged entangled states,
and discuss the advantages of packaged states in quantum information.

\subsection{Single-Particle Packaged Superposition States}
\label{SEC:SuperpositionOfSingleParticlePackagedStates}

From an application point of view, the particle packaged state $\lvert P\rangle$ and antiparticle packaged state $\lvert \bar{P}\rangle$ are trivial.
However, the packaged superposition state is nontrivial.
Due to superselection rules, not all packaged states can be coherently superposed.
Only states within the same net-charge sector are allowed to form superpositions.
We now discuss the conditions for a valid superposition of the packaged states of a single-particle.

\subsubsection{Gauged Charge and Global Quantum Numbers}

If we only consider the spin of a single-particle, then we can freely superpose single‐particle spin states $\lvert \uparrow\rangle$ and $\lvert \downarrow\rangle$ to obtain a superposition state like $\psi = \alpha\,\lvert \uparrow\rangle + \beta\,\lvert \downarrow\rangle$.
This is because spin is an external degree of freedom (DOF) and is not subject to superselection.

However, if we consider a gauge quantum number such as electric charge $\pm e$, then we cannot superpose states $\lvert +e\rangle$ and $\lvert -e\rangle$ to form a coherent superposition state like $\psi = \alpha\,\lvert +e\rangle + \beta\,\lvert -e\rangle$.
This is because states with different $\pm e$ charges belong to distinct charge sectors $\mathcal{H}_{+e}$ and $\mathcal{H}_{-e}$ and superselection rules forbid the formation of superpositions that span distinct charge sectors under local gauge group $U(1)$.

This indicates that superselection rules prevent charged single particles (or excitations differing by a gauge charge) from serving as quantum information carriers, such as
qubits in quantum computation and messenger particles in quantum communication.
To solve this problem, let us now compare the difference between a gauged and global IQN:
\begin{enumerate}
	\item Gauge Charge (e.g., electric charge $Q$, color in QCD):
	Local gauge-invariance and superselection protects gauge charge.
	The cross‐sector superpositions are disallowed.
	For example, an electron $(Q=-e)$ and positron $(Q=+e)$ lie in different charge sectors.
	
	\item Global Quantum Number (e.g., strangeness, flavor, isospin): 
	Local gauge symmetries do not enforce global quantum number.
	Therefore, superselection rules do not apply to global quantum number.
	The coherent superpositions are allowed if the net gauge charge is still the same.
	For example, flavor states such as $\lvert K^0\rangle$ and $\lvert \bar{K}^0\rangle$ each have net electric charge $0$, so we can form $\alpha\,\lvert K^0\rangle + \beta\,\lvert \bar{K}^0\rangle$.
\end{enumerate}

This distinction between gauged IQNs (which are fixed by local gauge-invariance) and global quantum numbers (which can be superposed) is important.
It uncovers the fact that only those differences that arise from global (non-gauged) quantum numbers can be used to encode quantum information.

\subsubsection{Valid Packaged Superposition States of Single‐Particle}
\label{SEC:ValidPackagedSuperpositionStatesOfSingleParticle}

Now we see that the obstacle for superposition of single‐particle packaged states comes from gauged charge. More specifically, if a particle $\lvert P\rangle$ and its antiparticle $\lvert \bar{P}\rangle$ differ by a nonzero gauge charge, then superselection forbids the superposition between state $\lvert P\rangle$ and $\lvert \bar{P}\rangle$.

However, if we select a particle $\lvert P\rangle$ and its antiparticle $\lvert \bar{P}\rangle$ that reside in the same net‐charge sector (often $Q=0$), then they differ only by a global quantum number (e.g., flavor, isospin) instead of a gauged one.
In this case, we can successfully bypass the obstacles of superselection and form physically coherent superposition states \cite{Ma2025}.
Therefore, we can choose neutral particles as quantum information carriers.
This is because all relevant states of a neutral particle (single or composite) always remain in one gauge sector and differ only in global quantum numbers, which permit superposition.

For rigor, let us now summarize the conditions for forming a valid packaged superposition states of single-particle as follows:

\begin{proposition}[Necessary and Sufficient Conditions for Packaged Superposition of a Single Particle]\label{PROP:NecessarySufficientConditionsSingleParticle}
	Let $\lvert P \rangle$ be the packaged state of a particle and $\lvert \bar{P} \rangle$ be the packaged state of its antiparticle.
	The basis of single-particle space can be written as $\{\lvert P \rangle, \lvert \bar{P} \rangle\}$.
	Then the packaged superposition state of a single particle
	\begin{equation}\label{EQ:ValidPackagedSuperpositionStatesOfSingleParticle}
		\Psi = \alpha\,\lvert P\rangle + \beta\,\lvert \bar{P}\rangle, \quad |\alpha|^2 + |\beta|^2 = 1,
	\end{equation}
	is valid (physically allowed and nontrivial) iff the following two conditions hold:
	
	\begin{enumerate}
		\item \textit{Zero Net Gauge Charge:} 
		The particle $\lvert P\rangle$ and its antiparticle $\lvert \bar{P}\rangle$ both carry zero gauge charge.
		This ensures they occupy the same superselection sector $\mathcal{H}_{Q=0}$, so $\lvert P\rangle$ and $\lvert \bar{P}\rangle$ are not separated by a gauge-induced superselection boundary.
		Let $\hat{Q}$ be the net-charge operator, we have
		\begin{equation}\label{EQ:ZeroNetGaugeChargeCond}
			\hat{Q}\lvert P\rangle = 0,\quad \hat{Q}\lvert \bar{P}\rangle = 0.
		\end{equation}
				
		\item \textit{Difference is a Global Quantum Number:}
		The particle and its antiparticle differ only by a global quantum number (flavor, strangeness, lepton number, etc.), which we denote by the operator $\hat{F}$.
		Then, the basis states must satisfy
		\begin{equation}\label{EQ:DifferenceGlobalQuantumNumber}
			\hat{F}\lvert P\rangle = f\,\lvert P\rangle,\quad \hat{F}\lvert \bar{P}\rangle = -f\,\lvert \bar{P}\rangle,
		\end{equation}
		where $f$ is a nonzero real number.
		Since $\hat{F}$ is not protected by a local gauge symmetry, coherent superpositions of the form in Eq.~(\ref{EQ:ValidPackagedSuperpositionStatesOfSingleParticle}) are physically allowed.
		The difference in a global quantum number ensures that the two states carry opposite values of the global quantum number and therefore can serve as distinct logical states for encoding quantum information.
	\end{enumerate}
	Thus, to construct physically allowed superposition states for quantum computation, one must choose particles (or composite systems) that are neutral (or that differ only by global quantum numbers).
\end{proposition}

\begin{proof}
	We now show that the two conditions are necessary and sufficient for a single‐particle packaged superposition state to be physically allowed.
	
	\begin{enumerate}
		\item \textbf{Necessary:}
		Suppose Eq.~(\ref{EQ:ValidPackagedSuperpositionStatesOfSingleParticle}) is valid (physically allowed), then we have
		\begin{itemize}
			\item \textit{No Net Gauge Charge:}
			The superselection rules prevent coherent superpositions of states that carry different gauge charges.
			If either $|P\rangle$ or $|\bar{P}\rangle$ carried a nonzero gauge charge, then they would lie in different superselection sectors and could not be coherently superposed.
			Thus, requiring 
			\[
			\hat{Q}\lvert P\rangle = 0 \quad \text{and} \quad \hat{Q}\lvert \bar{P}\rangle = 0
			\]
			is necessary.
			
			\item \textit{Difference is Global, Not Gauged:} 
			If the states differed by a gauged quantum number (e.g., electric charge), then (even if they were both overall neutral) their internal structures would belong to different irreducible representations of the local gauge group (by the packaging principle) and coherent superpositions between different irreps are forbidden by superselection.
			So the difference must be in a quantity (like flavor or strangeness) that is a global quantum number, which is not subject to the superselection constraint.
		\end{itemize}

		\item \textbf{Sufficient:}
		Conversely, suppose the two conditions hold. Then
		\begin{itemize}
			\item \textit{No Net Gauge Charge:}
			If both states have zero gauge charge, then they lie in the same superselection sector (say, $\mathcal{H}_{Q=0}$).
			From the perspective of gauge-invariance, a coherent superposition is allowed.
			
			\item \textit{Difference is Global, Not Gauged:} 
			If the only distinction between the particle and its antiparticle is a global quantum number (i.e., the operator $\hat{F}$ acting nontrivially while $\hat{Q} = 0$ holds), then the states lie in the same superselection sector and any superposition 
			\[
			\Psi = \alpha\,\lvert P\rangle + \beta\,\lvert \bar{P}\rangle
			\]
			is allowed by the rules of quantum mechanics.
		\end{itemize}
		
		In other words, Eq.~(\ref{EQ:ValidPackagedSuperpositionStatesOfSingleParticle}) is valid.
	\end{enumerate}
	
	Thus, the two conditions are necessary and sufficient.
\end{proof}

\begin{remark}
	Remarks on the necessary and sufficient conditions:
	
	\begin{enumerate}
		\item The first condition $\hat{Q}\lvert P\rangle = 0,~ \hat{Q}\lvert \bar{P}\rangle = 0$ is to guarantee that the superposition is physical allowed. Only neutral particles and their antiparticles can be in the same charge sector (superselection sector) and therefore the superposition $\Psi = \alpha\,\lvert P\rangle + \beta\,\lvert \bar{P}\rangle$ is permitted.
		
		\item The second condition $\hat{F}\lvert P\rangle = f\,\lvert P\rangle,\quad \hat{F}\lvert \bar{P}\rangle = -f\,\lvert \bar{P}\rangle$ is to guarantee that the superposition is non-trivial. Otherwise, we have $\lvert P \rangle = \lvert \bar{P} \rangle$ and the superposition $\Psi = \alpha\,\lvert P\rangle + \beta\,\lvert \bar{P}\rangle$ is trivial. For example, photons are their own antiparticle, $\lvert \gamma\rangle = \lvert \bar{\gamma}\rangle$. We cannot use photons for packaged qubits.
	\end{enumerate}	
\end{remark}

\begin{corollary}\label{COR:SuperpositionStateGaugeInvariant}
	The superposition state in Eq.~(\ref{EQ:ValidPackagedSuperpositionStatesOfSingleParticle}) is gauge-invariant.
\end{corollary}

\begin{proof}
	By definition, the particle state $|P\rangle$ and antiparticle state $|\bar{P}\rangle$ are packaged states.
	Therefore, they are gauge-invariant.	
	Let $G$ be a gauge group. Under any local gauge transformation $U_g ~ (g \in G)$, the single-particle packaged states transform as
	\[
	U_g\,\lvert P\rangle = e^{i\phi(g)}\,\lvert P\rangle,\quad U_g\,\lvert \bar{P}\rangle = e^{i\phi(g)}\,\lvert \bar{P}\rangle,
	\]
	where the phase $e^{i\phi(g)}$ is the same for both states because they lie in the same sector.
	
	Thus, the superposition state in Eq.~(\ref{EQ:ValidPackagedSuperpositionStatesOfSingleParticle}),
	\[
	\lvert \Psi\rangle = \alpha\,\lvert P\rangle + \beta\,\lvert \bar{P}\rangle,
	\]
	transforms as
	\[
		U_g\,\lvert \Psi\rangle 
		= \alpha\,U_g\,\lvert P\rangle + \beta\,U_g\,\lvert \bar{P}\rangle
		= e^{i\phi(g)} \left(\alpha\,\lvert P\rangle + \beta\,\lvert \bar{P}\rangle\right)
		= e^{i\phi(g)}\,\lvert \Psi\rangle.
	\]
	This shows that $\lvert \Psi\rangle$ is gauge-invariant and therefore remains in the same gauge sector (i.e. $\hat{Q}=0$).	 
\end{proof}

\begin{example}	\label{EXA:NeutralQuantumInformationCarrier}
	
	Examples of neutral quantum information carriers:
	
	\begin{enumerate}
		\item Photons:
		
		Since photons are their own antiparticles ($\gamma = \bar{\gamma}$), a state like $\lvert \gamma\rangle + \lvert \bar{\gamma}\rangle$ is trivial and does not lead to a distinct basis for quantum information.
				
		\item Neutral Mesons $(K^0, \bar{K}^0), (B^0, \bar{B}^0)$, etc.:
		
		Each pair has $\textit{net electric charge} = 0$.
		Their difference (e.g., strangeness for $K^0$ vs. $\bar{K}^0$) is a global quantum number, not a gauge charge.		
		Experiments show that $\alpha\,\lvert K^0\rangle + \beta\,\lvert \bar{K}^0\rangle$ states do exist physically (flavor oscillations, CP violation, etc.). 		
		Hence these can form coherent superpositions, making them candidate for quantum information carrier in principle. 
		In practical, the short lifetimes and detection complexities hamper their use in real quantum communication, but conceptually they circumvent the superselection constraints for electromagnetism.
		
		\item Neutrinos (Majorana or Dirac):
		
		If neutrinos are Majorana, each neutrino is its own antiparticle (similar to photons), then superposition of Eq.~(\ref{EQ:ValidPackagedSuperpositionStatesOfSingleParticle}) is trivial.
		No net gauge charge difference.		
		If neutrinos are Dirac but lepton number is effectively global (not a local gauge symmetry), partial superpositions of neutrino vs. antineutrino might appear in exotic contexts. 
		Practical communication is again limited by detection challenges.
	\end{enumerate}
\end{example}

\begin{remark}
	Although theoretically it is feasible to use neutral particles $(\lvert K^0\rangle, \lvert \bar{K}^0\rangle)$ as quantum information carriers, there are still challenges in practice at present:
	
	(1) Short Lifetimes and Oscillations:
	Many neutral mesons are short-lived or experience rapid flavor oscillations.
	$K^0$-$\bar K^0$ mixing lives only $ \sim 10^{-10}$ s \cite{GellMannPais1955,Good1961,Zohar2016}.
	This makes them less convenient as quantum information carrier.
	
	(2) Measurement Complexity: 
	Although experimental flavor factories do measure correlated oscillations,
	distinguishing $\lvert K^0\rangle$ from $\lvert \bar{K}^0\rangle$ or performing partial manipulations is not easy.
\end{remark}

\subsection{Pure-packaged Subspaces}
\label{SEC:PurePackagedSubspace}

We begin with the most elementary dynamical object in a gauge theory:
a Hilbert subspace that is completely insensitive to local gauge transformations yet still large enough to carry quantum information.
Such subspaces are the building blocks for all packaged qubits, qudits, codes and resource states introduced later.

In a lattice gauge theory with physical Hilbert space $\mathcal H$, we have
(i) a global conserved charge $\hat Q \;=\; \sum_{x\in\Lambda}\hat q_x$.
(ii) a continuous family of local gauge transformations $\{U_g\}_{g\in G}$ that act on $\mathcal H$.
Super‑selection forbids superpositions of different $\hat Q$ sectors.
Any dynamics that fails to commute with the gauge action will instantaneously leak a state out of the physical manifold.
Thus, we would like to single out subspaces that are automatically confined to a fixed charge sector and on which every $U_g$ acts only by a harmless global phase.

Intuitively, a pure-packaged subspace is a collection of packaged quantum states in a fixed charge sector.
We now give a formal definition for pure-packaged subspace:

\begin{definition}[Pure-packaged subspace]
	\label{DEF:PurePackagedSubspace}
	Let $G$ be a gauge group acting on the physical Hilbert space
	$\mathcal H$ via local unitaries $U_g$ and let $\hat Q$ be the conserved
	(global) charge.
	
	If a sub‑space
	$
	\mathcal H_{\mathrm{pack}}\subset\mathcal H
	$
	satisfies the following conditions:	
	
	\begin{enumerate}[label=(\alph*)]
		\item \textbf{Fixed charge.}\;
		There exists $Q_0$ such that
		$\hat Q\,\ket{\psi}=Q_0\,\ket{\psi}$
		for all $\ket{\psi}\in\mathcal H_{\mathrm{pack}}$;
		
		\item \textbf{Gauge‑invariance (singlet property).}\;
		For every $g\in G$ there is a one-dimensional unitary character
		$\chi(g)\in\mathrm U(1)$, independent of $\ket{\psi}$, such that
		\[
		U_g\,\ket{\psi}= \chi(g)\,\ket{\psi},
		\qquad\forall\,\ket{\psi}\in\mathcal H_{\mathrm{pack}}.
		\]
		(Equivalently, $U_g$ acts as $\chi(g)\,\mathbf 1$
		on~$\mathcal H_{\mathrm{pack}}$).
		\footnote{%
			If $G$ is non‑Abelian, then every one‑dimensional unitary
			representation $\chi$ factors through the Abelianisation
			$G/[G,G]$.
			The special case $\chi(g)\equiv1$
			gives the usual notion of a {\em gauge singlet}.}
	\end{enumerate}
	then we say that $\mathcal H_{\rm pack}$ is a \textbf{pure-packaged subspace} of $\mathcal H$.
\end{definition}

\paragraph{Remarks.}

\begin{itemize}
	\item[(i)]
	The definition is manifestly gauge‑group agnostic:
	it applies to any compact Lie group or discrete group, Abelian or non‑Abelian. 
	In the non‑Abelian case condition (b) simply says that $\mathcal H_{\rm pack}$ carries the trivial	1‑dimensional irrep (or any other 1‑dimensional character, if one exists).
	
	\item[(ii)]
	No dimensionality is imposed. 
	A 2‑dimensional packaged subspace will later be called a pure-packaged qubit. $d$‑dimensional ones give packaged qudits.
	
	\item[(iii)]
	Let $P_{\rm pack}$ be the projector onto $\mathcal H_{\rm pack}$.
	Then condition (b) implies $[U_g,P_{\rm pack}]=0$ and $[\hat Q,P_{\rm pack}]=0$.
\end{itemize}

A pure-packaged subspace has the following basic properties:

\begin{proposition}\label{prop:PackagedProperties}
	Let $\mathcal H_{\rm pack}$ satisfy
	Definition~\ref{DEF:PurePackagedSubspace}.
	Then
	\begin{enumerate}[label=\roman*.]
		\item \textbf{Stability.}\;
		Every bounded operator $V$ that commutes with $\hat Q$
		preserves the packaged subspace:
		$ V\,\mathcal H_{\rm pack}\subseteq\mathcal H_{\rm pack}$.
		
		\item \textbf{Existence (constructive).}\;
		Fix any charge sector $\mathcal H_{Q_0}$ with $\dim\mathcal H_{Q_0}\ge 1$,
		choose an orthonormal set $\{\ket{\psi_j}\}_{j=1}^k\subset\mathcal H_{Q_0}$
		and set $\mathcal H_{\rm pack}:=\mathrm{span}\{\ket{\psi_j}\}$.
		Because all $\ket{\psi_j}$ transform under the same 1‑dimensional representation $\chi$, the span is pure‑packaged.
	\end{enumerate}
\end{proposition}

\begin{proof}
	\emph{(i)}\;
	If $[\hat Q,V]=0$ and $\hat Q\ket{\psi}=Q_0\ket{\psi}$ then
	$\hat Q(V\ket{\psi})=V\hat Q\ket{\psi}=Q_0(V\ket{\psi})$,
	so $V\ket{\psi}\in\mathcal H_{Q_0}$.
	Moreover
	$U_g V\ket{\psi}=V U_g\ket{\psi}
	=\chi(g)\,V\ket{\psi}$,
	hence $V\ket{\psi}\in\mathcal H_{\rm pack}$. 
	Linearity completes the argument.
	
	\emph{(ii)}\;
	The trivial character $\chi(g)\!\equiv\!1$ always exists, so the set
	chosen clearly fulfils conditions (a) and (b).
\end{proof}

The pure-packaged subspaces defined here are the foundation for
multi-particle packaged states (Sec.~\ref{SEC:MultiParticlePackagedStates}),
packaged qubits/qudits (Sec.~\ref{SEC:PackagedQubits}),
packaged gates (Sec.~\ref{SEC:PurePackagedQubitGates}),
packaged circuits (Sec.~\ref{SEC:PackagedQuantumCircuits}),
packaged messenger (Sec.~\ref{SEC:PackagedMessengers}),
and packaged resource states (Sec.~\ref{SEC:PackagedResourceStates}).

\subsection{Multi-Particle Packaged States}
\label{SEC:MultiParticlePackagedStates}

In Ref. \cite{Ma2025}, we showed that for each subspace $\mathcal{H}_Q$ with net charge $Q$, we can construct a complete orthonormal basis of packaged entangled states using a variant of the Gram-Schmidt procedure.
This property is crucial for applications in quantum information because it guarantees the existence of a maximal orthonormal set $\{\lvert \Psi_i\rangle\}$ of packaged entangled states spanning $\mathcal{H}_Q$. 
Consequently, Bell-like measurements become possible \cite{NielsenChuang2010}.

Before constructing entangled resources, we prove the closure property of pure‑packaged sectors.

\subsubsection{Tensor Product of Pure-Packaged Subspaces}
\label{sec:MultiParticlePackagedStates}

The building blocks of every many‑body protocol are the single‑particle
pure-packaged subspaces introduced in Definition~\ref{DEF:PurePackagedSubspace}.
We construct the following lemma to show how they combine.

\begin{lemma}[Tensor Product Closure of Pure-packaged Subspaces]
	\label{LEM:TensorProductPackaged}
	Let $\mathcal H^{(1)}_{\rm pack}\subset\mathcal H^{(1)}$ and
	$\mathcal H^{(2)}_{\rm pack}\subset\mathcal H^{(2)}$ be pure-packaged
	subspaces with the same fixed charge $Q_0$ (possibly on different physical
	systems).
	Then their tensor product
	\[
	\mathcal H^{(1)}_{\rm pack}\;\otimes\;
	\mathcal H^{(2)}_{\rm pack}
	\;\subset\;
	\bigl(\mathcal H^{(1)}\!\otimes\!\mathcal H^{(2)}\bigr)
	\]
	is also a pure-packaged subspace.
\end{lemma}

\begin{proof}
	Both factors lie in charge sector $Q_0$, so their tensor product lies
	in charge $2Q_0$ (or still $Q_0$ if the charges are additive mod
	something). 
	$\forall~ g \in G$, the action is
	$U_g\!\otimes\!U_g
	=\chi(g)\chi(g)\,{\bf1}\otimes{\bf1}$,
	which is a global phase.
	Thus, the tensor product is pure‑packaged.
\end{proof}

For a multi-particle system, it is much simple to stay in a fixed net-charge sector.
For simplicity, one may choose the neutral sector $Q_0=0$.
This is exactly the setting we will use in our packaged Shor, QKD and error‑correction constructions.

\subsubsection{Valid Packaged Superposition States of Multi‑Particle}
\label{SEC:ValidPackagedSuperpositionStatesOfMultiParticle}

In Sec.~\ref{SEC:ValidPackagedSuperpositionStatesOfSingleParticle}, we discussed the necessary and sufficient conditions for a single‑particle packaged superposition.
Here we do the same thing to multi‑particle.

For an $n$‑particle system, the physical Hilbert space is the tensor product
\[
\mathcal H_{\text{phys}}
\;=\;
\bigl(\mathcal H_{\text{single}}\bigr)^{\!\otimes n},
\qquad
\hat Q_{\text{tot}}
=\sum_{i=1}^{n}\hat Q^{(i)},
\]
where each factor carries the local action of the gauge group~$G$.
A packaged $n$‑particle state is any vector that lies in the simultaneous kernel of $\hat Q_{\text{tot}}$ and of all local Gauss-law generators
$\{\hat G_x\}_{x\in\Lambda}$.
Because $\ker\hat G_x$ is a pure‑packaged sub‑space
(Lemma~\ref{LEM:TensorProductPackaged}),
the set of $n$‑particle packaged states is itself pure‑packaged.

By Lemma \ref{LEM:TensorProductPackaged}, the tensor product of
packaged single‑particle states is still packaged.
The only remaining question is when superpositions of different such $n$‑particle vectors are physical?
We answer this question by building the next proposition.

\begin{proposition}[Necessary and Sufficient Conditions for Packaged Superposition of Multi-Particle]\label{PROP:NecessarySufficientConditionsMultiParticle}
	Let	$G$ be a compact gauge group that acts on the physical Hilbert space $\mathcal H$ with a global conserved charge operator $\hat Q$ and a family of local Gauss-law generators $\{\hat G_x\}_{x}$.	
	Let $\ket{\Psi_{1}},\ket{\Psi_{2}}\in\mathcal H$ be two $n$-particle packaged states.
	Consider the normalized superposition
	\begin{equation}\label{EQ:ValidPackagedSuperpositionStatesOfMultieParticle}
		\ket{\Psi}\;=\;\alpha\,\ket{\Psi_{1}}\;+\;\beta\,\ket{\Psi_{2}},
		\quad |\alpha|^{2}+|\beta|^{2}=1 .
	\end{equation}
	The state $\ket{\Psi}$ is valid (physically admissible and non‑trivial) if and only if the following conditions hold:
	\begin{description}
		\item[C1.] Fixed total charge:	
		$$
		\hat Q\,\ket{\Psi_{k}} \;=\; Q\,\ket{\Psi_{k}}, 
		\quad k=1,2 ,
		$$
		for one and the same eigenvalue $Q$.
		
		\item[C2.] Gauss law at every site:		
		$$
		\hat G_{x}\,\ket{\Psi_{k}} \;=\; 0,
		\quad \forall x,\;k=1,2 .
		$$
		
		\item[C3.] Identical local gauge character:
		Every local matter transformation $U_{g}^{(i)}$ acts on both states with the same one‑dimensional representation (character) $\chi(g)$:		
		$$
		U_{g}^{(i)}\ket{\Psi_{1}} = \chi(g)\,\ket{\Psi_{1}}, 
		\quad
		U_{g}^{(i)}\ket{\Psi_{2}} = \chi(g)\,\ket{\Psi_{2}},
		\quad \forall g\in G,\;\forall i.
		$$
		
		\item[C4.] Linear independence:
		$\ket{\Psi_1}$ and $\ket{\Psi_2}$ are not proportional:
		$$
		\bigl|\langle\Psi_1\bigl|\Psi_2\rangle\bigr|<1.
		$$
	\end{description}
\end{proposition}

\begin{proof}
	C1, C2, and C3 guarantee physical admissibility, i.e., not forbidden by any superselection rule.
	C4 guarantees non‑triviality.
	
	\smallskip\noindent
	\emph{I. Necessity.}	
	Assume $\ket{\Psi}$ is valid $\Longrightarrow$ C1, C2, C3, and C4.
	
	\begin{enumerate}
		\item Physically admissible
		
		Because $\ket{\Psi}$ is physical, it must lie in a single superselection sector.
		That forces each component $\ket{\Psi_k}$ to share the same eigenvalue of $\hat Q$ and to satisfy $\hat G_x=0$.
		Otherwise the superposition would lie across orthogonal sectors and be forbidden.
		Thus, we obtain C1 and C2.
		
		Moreover gauge‐invariance of the superposition requires that all local gauge transformations $U_{g}^{(i)}$ act by one and the same phase on $\ket{\Psi_1}$ and $\ket{\Psi_2}$.
		Otherwise, $U_{g}^{(i)}\ket{\Psi}$ would contain relative phases depending on $g$ and the superposition would not transform by a single overall phase.
		This will violate superselection.
		Thus, we must have C3.
		
		\item Non‑triviality.
		If $\ket{\Psi_{1}}\propto\ket{\Psi_{2}}$ the superposition collapses to a single vector and carries no logical degree of freedom, contradicting the premise.
	\end{enumerate}	
	
	\smallskip\noindent
	\emph{II. Sufficiency.}
	
	Assume C1, C2, C3, and C4 $\Longrightarrow$ $\ket{\Psi}$ is physically admissible and non‑trivial.
	\begin{enumerate}
		\item Physically admissible.
		
		By C1 and C2, any linear combination of two common eigenvectors of $\hat Q$ (with eigenvalue $Q$) and of all $\hat G_x$ (with eigenvalue $0$) is again an eigenvector with the same set of eigenvalues.
		Thus, $\ket{\Psi}$ lies in the same super‑selection block as $\ket{\Psi_{1,2}}$.
		
		By C3, every local $U_{g}^{(i)}$ acts as $\chi(g)\mathbf 1$ on each component, then it also acts as $\chi(g)\mathbf 1$ on their superposition:	
		$$
		U_{g}^{(i)}\ket{\Psi}
		= \chi(g)\,\ket{\Psi},
		\qquad \forall g,i .
		$$	
		Thus $\ket{\Psi}$ is gauge‑invariant up to a global phase.
		
		Put these two together, we conclude that $\ket{\Psi}$ is physically admissible (not forbidden by any superselection rule).
		
		\item Non‑triviality.
		By C4, the two basis states are linearly independent, so $\ket{\Psi}$ genuinely spans a two‑dimensional logical subspace when $(\alpha,\beta)$ vary.
	\end{enumerate}
	
	Hence $\ket{\Psi}$ is a valid multi-particle packaged superposition state.
\end{proof}

\begin{remark}[Relation to single‑particle case]
	For $n=1$ the total‑charge operator reduces to $\hat Q$ and the
	Gauss‑law constraints act trivially, so
	Proposition~\ref{PROP:NecessarySufficientConditionsMultiParticle}
	reduces to
	Proposition~\ref{PROP:NecessarySufficientConditionsSingleParticle}.
\end{remark}

\begin{corollary}[Gauge‑invariance]
	\label{COR:MultiSuperpositionGaugeInvariant}
	Every state~$\ket{\Phi}$ satisfying the conditions of
	Proposition~\ref{PROP:NecessarySufficientConditionsMultiParticle}
	obeys
	$
	U_g^{\otimes n}\ket{\Phi}=e^{i\phi(g)}\ket{\Phi},
	\; \forall g\in G,
	$
	and thus remains inside the packaged sector.
\end{corollary}

\paragraph{Examples.}

\begin{enumerate}
	\item \textbf{$K^0\bar K^0$-pair Bell states.}\;
	Each neutral kaon is individually a packaged single‑particle
	state.
	The two‑particle Bell state
	$
	\tfrac1{\sqrt2}\bigl(\ket{K^0}\ket{\bar K^0}
	\pm\ket{\bar K^0}\ket{K^0}\bigr)
	$
	satisfies total electric charge $0$ (C1), Gauss-law $\hat G_{x}\,\ket{\Psi_{k}} = 0$ (C2), identical gauge
	character (C3), and are linearly independent (orthonormal).
	Thus, the Bell superposition is valid.
	
	\item \textbf{Electron-positron pair vs.\ muon-antimuon pair.}\;
	The states
	$
	\ket{e^{+}e^{-}}
	$
	and
	$
	\ket{\mu^{+}\mu^{-}}
	$
	carry the same gauge charge profile (each is a QED singlet) but
	differ by global flavour number.
	Coherent superpositions are allowed in principle, although
	unstable in practice because of electroweak interactions.
	
	\item \textbf{Photon pairs.}\;
	Each photon is its own antiparticle, so condition C4 fails
	(no global observable distinguishes the two basis vectors).
	The corresponding superposition is therefore trivial and cannot offer
	logical degree of freedom for packaged encoding.
\end{enumerate}

\subsubsection{Two-Particle Packaged Product States}

Consider a two-particle system in which $\lvert P \rangle$ and $\lvert \bar{P} \rangle$ are the packaged states of a particle and its antiparticle, respectively.
Under the charge conjunction $\hat{C}$, which swaps $\lvert P \rangle$ and $\lvert \bar{P} \rangle$, the product basis states are
\begin{equation}\label{EQ:ProductBasisStatesOfTwoParticleSystem}
	\{\lvert P\,P \rangle, \lvert P\,\bar{P} \rangle, \lvert \bar{P}\,P \rangle, \lvert \bar{P}\,\bar{P} \rangle\},
\end{equation}
where each term is defined as a tensor product (e.g., $\lvert P\,\bar{P}\rangle \equiv \lvert P \rangle \otimes \lvert \bar{P} \rangle$). 

If $P$ has charge $+q$ and $\bar{P}$ has charge $-q$,
then the states $\lvert P\,P \rangle$ and $\lvert \bar{P}\,\bar{P} \rangle$ lie in charge sectors with net charge $+2q$ and $-2q$, respectively, while only the neutral sector ($Q=0$) supports nontrivial entanglement. 

These product basis states then fall into three charge sectors:
\begin{itemize}
	\item Charge $Q=2q$: $\lvert P\,P \rangle$.
	\item Charge $Q=0$: $\lvert P\,\bar{P} \rangle, \lvert \bar{P}\,P \rangle$.
	\item Charge $Q=-2q$: $\lvert \bar{P}\,\bar{P} \rangle$.
\end{itemize}

\subsubsection{Two-Particle Packaged Entangled States}
\label{SEC:TwoParticlePackagedEntangledStates}

Because a particle and its antiparticle are distinguishable, the composite states are not subject to exchange-symmetry constraints.

\paragraph{(1) Pure Packaged Entangled States.}
To discuss the packaged entangled states, we assume the superselection rule that forbids superpositions of states with different net charge (here, the net difference between the number of $P$’s and $\bar{P}$’s).

Because entanglement requires a Hilbert space of at least two dimensions, 
only states from charge $Q=0$ sector satisfies this requirement.
Therefore, we select the physical Hilbert space: 
\[
\mathcal{H}_{Q=0} = \operatorname{span}\Bigl\{\,\lvert P\,\bar{P} \rangle,\; \lvert \bar{P}\,P \rangle\,\Bigr\}.
\]

This is smaller Hilbert space, which make the packaged quantum communication protocols simpler.
A general state in this two-dimensional subspace is given by
\[
\lvert \Psi \rangle = \alpha\,\lvert P\,\bar{P} \rangle + \beta\,\lvert \bar{P}\,P \rangle,\quad \alpha,\beta\in\mathbb{C}, \quad | \alpha |^2 + | \beta |^2 = 1.
\]

Applying the Gram-Schmidt procedure, we obtain a complete maximal orthonormal basis for the entire sector $\mathcal{H}_{Q=0}$:
\begin{equation}\label{EQ:2DOrthonormalBasis}
	\begin{aligned}
		\lvert \Psi_P^+ \rangle &= \frac{1}{\sqrt{2}}\Bigl(\lvert P\,\bar{P} \rangle + \lvert \bar{P}\,P \rangle\Bigr),\\[1mm]
		\lvert \Psi_P^- \rangle &= \frac{1}{\sqrt{2}}\Bigl(\lvert P\,\bar{P} \rangle - \lvert \bar{P}\,P \rangle\Bigr),
	\end{aligned}
\end{equation}
Here the states $\lvert \Psi^+ \rangle$, $\lvert \Psi^- \rangle$ satisfy the orthonormal conditions: $\langle \Psi^+ \lvert \Psi^+ \rangle = \langle \Psi^- \lvert \Psi^- \rangle = 1$ and $\langle \Psi^+ \lvert \Psi^- \rangle = 0$.

This complete basis enables Bell measurements and underpins the design of two-particle qubits, quantum communication, and quantum error correction schemes.

Although the above discussions fit to our intuition, it looks manually and lacks rigor.
Let us now apply the Clebsch-Gordan decomposition.
The total Hilbert space $\mathcal{H}$ splits into three subspace (charge sectors):
\[
\mathcal{H} = \mathcal{H}_{2q}\;\oplus\;\mathcal{H}_{0}\;\oplus\;\mathcal{H}_{-2q},
\]
where:
\begin{itemize}
	\item $\mathcal{H}_{2q} = \operatorname{span} \{\lvert P\,P \rangle\}$,
	
	\item $\mathcal{H}_{0} = \operatorname{span} \left\{\frac{1}{\sqrt{2}}\left(\lvert P\,\bar{P} \rangle + \lvert \bar{P}\,P \rangle\right), \frac{1}{\sqrt{2}}\Bigl(\lvert P\,\bar{P} \rangle - \lvert \bar{P}\,P \rangle\Bigr)\right\}$,
	
	\item $\mathcal{H}_{-2q} = \operatorname{span} \{\lvert \bar{P}\,\bar{P} \rangle\}$.
\end{itemize}
Then we obtain the same results.

\paragraph{(2) Hybrid-Packaged Entangled States.}

We now consider the hybrid packaging of external spin-$\frac{1}{2}$ and IQNs.

We first construct eight \textbf{fully hybrid-packaged entangled states}, where the external DOFs (spin) and IQNs are unfactorizable:
\begin{align}\label{EQ:HybridPES}
	\begin{aligned}
		&\lvert \Psi_{h,1}^{+} \rangle_{AB}
		\;=\;
		\frac{1}{\sqrt{2}}
		\Bigl(
		\lvert P, \uparrow\rangle_A \,\lvert \bar{P}, \downarrow\rangle_B
		\;+\;
		\lvert \bar{P}, \downarrow\rangle_A \,\lvert P, \uparrow\rangle_B
		\Bigr), \\
		&\lvert \Psi_{h,2}^{+} \rangle_{AB}
		\;=\;
		\frac{1}{\sqrt{2}}
		\Bigl(
		\lvert P, \downarrow\rangle_A \,\lvert \bar{P}, \uparrow\rangle_B
		\;+\;
		\lvert \bar{P}, \uparrow\rangle_A \,\lvert P, \downarrow\rangle_B
		\Bigr), \\
		&\lvert \Psi_{h,1}^{-} \rangle_{AB}
		\;=\;
		\frac{1}{\sqrt{2}}
		\Bigl(
		\lvert P, \uparrow\rangle_A \,\lvert \bar{P}, \downarrow\rangle_B
		\;-\;
		\lvert \bar{P}, \downarrow\rangle_A \,\lvert P, \uparrow\rangle_B
		\Bigr), \\
		&\lvert \Psi_{h,2}^{-} \rangle_{AB}
		\;=\;
		\frac{1}{\sqrt{2}}
		\Bigl(
		\lvert P, \downarrow\rangle_A \,\lvert \bar{P}, \uparrow\rangle_B
		\;-\;
		\lvert \bar{P}, \uparrow\rangle_A \,\lvert P, \downarrow\rangle_B
		\Bigr), \\
		&\lvert \Phi_{h,1}^{+} \rangle_{AB}
		\;=\;
		\frac{1}{\sqrt{2}}
		\Bigl(
		\lvert P, \uparrow\rangle_A \,\lvert \bar{P}, \uparrow\rangle_B
		\;+\;
		\lvert \bar{P}, \downarrow\rangle_A \,\lvert P, \downarrow\rangle_B
		\Bigr), \\
		&\lvert \Phi_{h,2}^{+} \rangle_{AB}
		\;=\;
		\frac{1}{\sqrt{2}}
		\Bigl(
		\lvert P, \downarrow\rangle_A \,\lvert \bar{P}, \downarrow\rangle_B
		\;+\;
		\lvert \bar{P}, \uparrow\rangle_A \,\lvert P, \uparrow\rangle_B
		\Bigr), \\
		&\lvert \Phi_{h,1}^{-} \rangle_{AB}
		\;=\;
		\frac{1}{\sqrt{2}}
		\Bigl(
		\lvert P, \uparrow\rangle_A \,\lvert \bar{P}, \uparrow\rangle_B
		\;-\;
		\lvert \bar{P}, \downarrow\rangle_A \,\lvert P, \downarrow\rangle_B
		\Bigr), \\
		&\lvert \Phi_{h,2}^{-} \rangle_{AB}
		\;=\;
		\frac{1}{\sqrt{2}}
		\Bigl(
		\lvert P, \downarrow\rangle_A \,\lvert \bar{P}, \downarrow\rangle_B
		\;-\;
		\lvert \bar{P}, \uparrow\rangle_A \,\lvert P, \uparrow\rangle_B
		\Bigr). \\
	\end{aligned}
\end{align}

By direct inner‐product calculation, one checks that all eight states in Eqs. (\ref{EQ:HybridPES}) are mutually orthonormal.

Thereafter, let us construct eight \textbf{semi-hybrid-packaged entangled states}, where the spin and IQNs are factorizable.

\begin{align}\label{EQ:SemiHybridPES}
	\begin{aligned}
		&\lvert \Psi_{P,s}^{+,+} \rangle_{AB}
		\;=\;
		\lvert \Psi_P^+ \rangle_{AB} ~ \lvert \Psi_s^+ \rangle_{AB}, \\
		&\lvert \Psi_{P,s}^{+,-} \rangle_{AB}
		\;=\;
		\lvert \Psi_P^+ \rangle_{AB} ~ \lvert \Psi_s^- \rangle_{AB}, \\
		&\lvert \Psi_{P,s}^{-,+} \rangle_{AB}
		\;=\;
		\lvert \Psi_P^- \rangle_{AB} ~ \lvert \Psi_s^+ \rangle_{AB}, \\
		&\lvert \Psi_{P,s}^{-,-} \rangle_{AB}
		\;=\;
		\lvert \Psi_P^- \rangle_{AB} ~ \lvert \Psi_s^- \rangle_{AB}, \\
		&\lvert \Phi_{P,s}^{+,+} \rangle_{AB}
		\;=\;
		\lvert \Psi_P^+ \rangle_{AB} ~ \lvert \Phi_s^+ \rangle_{AB}, \\
		&\lvert \Phi_{P,s}^{+,-} \rangle_{AB}
		\;=\;
		\lvert \Psi_P^+ \rangle_{AB} ~ \lvert \Phi_s^- \rangle_{AB}, \\
		&\lvert \Phi_{P,s}^{-,+} \rangle_{AB}
		\;=\;
		\lvert \Psi_P^- \rangle_{AB} ~ \lvert \Phi_s^+ \rangle_{AB}, \\
		&\lvert \Phi_{P,s}^{-,-} \rangle_{AB}
		\;=\;
		\lvert \Psi_P^- \rangle_{AB} ~ \lvert \Phi_s^- \rangle_{AB}, \\
	\end{aligned}
\end{align}
where
\begin{align}\label{SpinBellStates}
	\begin{aligned}
		&\lvert \Psi_s^{\pm} \rangle_{AB}
		\;=\;
		\tfrac{1}{\sqrt{2}}
		\Bigl(
		\lvert \uparrow\rangle_A \,\lvert \downarrow\rangle_B
		\;\pm\;
		\lvert \downarrow\rangle_A \,\lvert \uparrow\rangle_B
		\Bigr), \\
		&\lvert \Phi_s^{\pm} \rangle_{AB}
		\;=\;
		\tfrac{1}{\sqrt{2}}
		\Bigl(
		\lvert \uparrow \rangle_A \,\lvert \uparrow \rangle_B
		\;\pm\;
		\lvert \downarrow \rangle_A \,\lvert \downarrow \rangle_B
		\Bigr),
	\end{aligned}
\end{align}
are the four spin Bell states.
One checks by direct inner‐product calculation that all eight states in Eqs. (\ref{EQ:SemiHybridPES}) are mutually orthonormal.

It should be emphasized that in the fully \text{pack}-packaged entangled states Eqs.~(\ref{EQ:HybridPES}), the internal DOFs or external DOFs are un-factorizable.
Measuring either the internal DOF or external DOF, the entire state will collapse.
But in the semi-\text{pack}-packaged entangled states Eqs.~(\ref{EQ:SemiHybridPES}), the internal DOFs or external DOFs are factorizable, when measuring the internal DOF, the external DOF is not affectd, and vice versa.

\subsubsection{$N$-Particle Packaged States}

Let us now consider an $N$ particle system.
In many physical systems, conservation laws or superselection rules require that a state must have a fixed charge.
For the particle-antiparticle case, one may require that the net charge vanishes.
For example, if we associate the eigenvalue $+1$ to $P$ and $-1$ to $\bar{P}$, then a zero-charge state has an equal number of $P$’s and $\bar{P}$’s.

Consider an even number of particles, say
\[
N=2n,\quad n\in\mathbb{N}.
\]
Let $w(x)$ denote the number of occurrences of $P$ in the configuration $x$.
Then the balanced subspace (allowed Hilbert space) is
\[
\mathcal{H}_{\rm balanced} = \operatorname{span}\Bigl\{\,\lvert x \rangle \in \{\,\lvert P\rangle,\lvert \bar{P}\rangle\}^{\otimes 2n} :\; w(x)=n\,\Bigr\}.
\]
Here the function $w(x)$ is the Hamming weight of $P$ in the configuration $x$.
The dimension of $\mathcal{H}_{\rm balanced}$ is
\[
\dim \mathcal{H}_{\rm balanced} = \binom{2n}{n}.
\]

A general multi-particle packaged entangled state is an arbitrary normalized superposition over the allowed (balanced) configurations:
\[
\lvert \Psi \rangle = \sum_{x\in \mathcal{B}} c_{x}\,\lvert x \rangle,\quad \sum_{x \in \mathcal{B}} | c_x |^2 = 1,
\]
where
\[
\mathcal{B} = \Bigl\{\,x\in\{P,\bar{P}\}^{2n}: w(x)=n \Bigr\}.
\]

For $2n$ particles with $n>1$, therefore, one may choose a superposition with any number of terms - from two up to $\binom{2n}{n}$ terms. In particular, two natural examples are:

\paragraph{(1) Two-Term Packaged States.} 
One may select two complementary configurations. For instance, for $N=4$ (i.e. $n=2$) one may define
\begin{equation}\label{eq:N4TwoTermPES}
	\begin{aligned}
		\lvert \Psi^{(4),+} \rangle &= \frac{1}{\sqrt{2}}\Bigl(\lvert P\,\bar{P}\,P\,\bar{P} \rangle + \lvert \bar{P}\,P\,\bar{P}\,P \rangle\Bigr),\\[1mm]
		\lvert \Psi^{(4),-} \rangle &= \frac{1}{\sqrt{2}}\Bigl(\lvert P\,\bar{P}\,P\,\bar{P} \rangle - \lvert \bar{P}\,P\,\bar{P}\,P \rangle\Bigr).
	\end{aligned}
\end{equation}
Eq.~(\ref{eq:N4TwoTermPES}) are the direct analogues of the two-particle $\Psi^\pm$ states, but now on four particles.
These states will be used in packaged quantum error correction codes (see Sec.~\ref{SEC:PackagedQuantumErrorCorrectionCodes})

\paragraph{(2) Packaged Dicke State.} 
Another natural choice is the packaged balanced Dicke state \cite{Dicke1954}
\begin{equation}\label{DickeState}
	\lvert D_{n}^{2n} \rangle = \frac{1}{\sqrt{\binom{2n}{n}}} \sum_{x\in\mathcal{B}} \lvert x \rangle,
\end{equation}
where $\mathcal{B}$ is the subspace of all packaged product states with exactly $n$ particles in state $P$ and $n$ in state $\bar{P}$.

The packaged Dicke state $\lvert D_{n}^{2n} \rangle$ is a uniform superposition over all configurations with exactly $n$ particles in state $P$ and $n$ in state $\bar{P}$.
This packaged entangled state is highly symmetric (invariant under particle permutations) and involves $\binom{2n}{n}$ terms.

Thus, depending on the physical or informational requirements, one can choose a packaged entangled state with two terms, three terms, four terms, or in general, any superposition spanning a subset of the $\binom{2n}{n}$-dimensional balanced subspace.
The structure of these states can be exploited in quantum error correction or multi-qubit encoding.

\subsection{Advantages of Packaged States}

Although gauge‐invariance constraints look restrictive, packaged states offer many advantages for quantum‐information applications:

First, packaged states are robust against gauge‐violating errors. 
Since superselection forbids transitions out of a given gauge‐charge sector, some types of environmental noise are automatically disallowed.
This can act as a partial protection mechanism that prevents unphysical errors, which would jump across superselection boundaries.

Packaged states are also more secure in quantum communication. 
An eavesdropper attempting to measure partial or fractional charge DOFs would violate gauge-invariance, thereby revealing their presence.
Certain protocols can exploit superselection to limit an adversary’s manipulations, adding a layer of security.

In the packaged state formalism, gauge-invariance is strictly enforced and the operations like flipping a charge are forbidden.
This restriction simplifies theoretical work in analyzing errors or designing gates because fewer types of operations are needed.

Packaged states possess rich packaged IQNs. 
However, the external DOFs like spin, momentum, position, etc. remain unconstrained by the gauge symmetry.
Thus, one can still form high‐dimensional or multi‐particle entangled states.
The external DOFs remain free while internal charges stay packaged, resulting in a broad space of packaged entangled configurations.

By restricting operations to those that obey the superselection rule, packaged states naturally protect against certain types of errors and therefore offer a promising route toward robust quantum information processing.

In the following sections,
we shall show how to construct packaged qubits, packaged gates, and packaged circuits that respect gauge constraints.
We will also lift the conventional quantum algorithms and protocols into high-dimensional, hybrid-packaged subspace for enhanced robustness and security.

\section{Packaged Qubits, Gates, and Circuits}
\label{SEC:PackagedQubitsGatesAndCircuits}

In conventional quantum information science, we have
qubits \cite{Benioff1980,Feynman1982,Deutsch1985,Schumacher1995}, 
quantum gates \cite{Deutsch1989,Barenco1995,DiVincenzo1995,CiracZoller1995,DiCarlo2009}, 
and quantum circuits \cite{Yao1993,Shor1994,Monroe1995,Chuang1998,Coppersmith2002}.
In this section, we explore how to construct gauge-invariant packaged qubits, packaged gates, and packaged circuits using single-particle packaged states and multi-particle packaged states.
These form the foundation for later studies in gauge-invariant quantum algorithms and protocols based on packaged quantum states.

\subsection{Packaged Qubits}
\label{SEC:PackagedQubits}

Physically, a packaged qubit is a gauge-invariant, two-level system that resides within a fixed net-charge sector.
It is characterized not only by its two-dimensional Hilbert-space structure but also by the locking of internal quantum numbers (IQNs) due to gauge-invariance.

From a pure mathematical point of view, a packaged qubit is an abstract concept.
We require it to be a gauge-invariant two level system.
We also require the two levels to permits superposition so that we can construct a general packaged superposition state.
However, the inner structure of the packaged qubit is not limited.

\subsubsection{Definition of Pure Packaged Qubits}
\label{SEC:DefPackagedQubitsHybridQudits}

Let us first give a formal definition for pure-packaged qubits, which only utilize the internal quantum numbers, but ignore the external quantum numbers.

\begin{definition}[Pure-packaged qubit]
	\label{DEF:PurePackagedQubit}
	Let $\mathcal H_{\rm pack}$ be a pure-packaged subspace of dimension two.
	Let
	\[
	\mathcal{H}_{\rm qubit} = \operatorname{span}\{\lvert 0_P\rangle,\lvert 1_P\rangle\} ~ \subset \mathcal{H}_{\rm pack}
	\]
	and any packaged state $\lvert \Psi\rangle \in \mathcal{H}_{\rm pack}$ can be written as
	\begin{equation}\label{EQ:ExpandArbitraryPackagedState}
		\lvert \Psi\rangle = \alpha\,\lvert 0_P\rangle + \beta\,\lvert 1_P\rangle, \quad | \alpha |^2 + | \beta |^2 = 1,
	\end{equation}
	then we say that the pure-packaged subspace $\mathcal{H}_{\rm qubit}$ is a \textbf{pure-packaged qubit}.
\end{definition}

Since the packaged qubit basis states $\lvert 1_P\rangle$ and $\lvert 0_P\rangle$ are abstract symbols, the inner structure of the packaged qubit is not limited.
The packaged qubit can be either a single-particle or a two-level multi-particle system.

\begin{definition}[Pure packaged qudit of dimension $d$]
	Choose an orthonormal set
	$\{|0_P\rangle,\dots ,|d{-}1_P\rangle\}$
	inside a pure-packaged subspace of dimension $d$.
	Their span is a pure-packaged qudit. 
	For $d=2$ this reduces to Definition~\ref{DEF:PurePackagedQubit}.
\end{definition}

A pure-packaged qubit/qudit is a pure-packaged subspace, which is gauge-invariant according to Definition \ref{DEF:PurePackagedSubspace}.
Thus, a pure-packaged qubit/qudit is gauge-invariant.
This means that the gauge-invariance of a pure-packaged qubit/qudit is given by definition or construction, not a derived property.

In this section, we focus on pure-packaged qubit, while deferring the detailed discussion of qudit to Section Sec.~\ref{SEC:dDDimensionalHybridPackagedSpace}.

\subsubsection{Constructing Single-Particle Packaged Qubits}

A single-particle packaged qubit must satisfy the conditions given in Proposition \ref{PROP:NecessarySufficientConditionsSingleParticle}.
To construct a single-particle packaged qubit, let us consider a neutral particle $\lvert P\rangle$ and its antiparticle $\lvert \bar{P}\rangle$.
Both states are eigenstates of the net-charge operator $\hat{Q}$ with eigenvalue zero (Eq.~(\ref{EQ:ZeroNetGaugeChargeCond})), i.e.,
\[
\hat{Q}\lvert P\rangle = 0,\quad \hat{Q}\lvert \bar{P}\rangle = 0
\]
We assume that $\lvert P\rangle$ and $\lvert \bar{P}\rangle$ differ by a global quantum number.
Let $\hat{F}$ be a global quantum number operator (e.g., the flavor operator).
Then we have (Eq.~(\ref{EQ:DifferenceGlobalQuantumNumber})):
\[
\hat{F}\lvert P\rangle = f\,\lvert P\rangle, \quad \hat{F}\lvert \bar{P}\rangle = -f\,\lvert \bar{P}\rangle, \quad f \ne 0.
\]

According to Proposition \ref{PROP:NecessarySufficientConditionsSingleParticle}, the particle can have valid packaged superposition states (Eq.~(\ref{EQ:ValidPackagedSuperpositionStatesOfSingleParticle}))
\[
\lvert \Psi\rangle = \alpha\,\lvert P\rangle + \beta\,\lvert \bar{P}\rangle, \quad | \alpha |^2 + | \beta |^2 = 1.
\]

According to Corollary \ref{COR:SuperpositionStateGaugeInvariant}, $\lvert \Psi\rangle$ is gauge-invariant.
Specifically, for any local gauge transformation $U_g$ (with $g \in G$), $\lvert \Psi\rangle$ transforms as
\[
U_g \lvert \Psi\rangle = e^{i\phi(g)} \lvert \Psi\rangle.
\]

\paragraph{(1) Single-Particle Qubits Based on Packaged Rectilinear Basis.}
Let us first define a packaged qubit based on single-particle packaged rectilinear (or computational) basis as:
\begin{equation}\label{EQ:SingleParticlePackagedRectilinearBasis}
	\begin{aligned}
		\lvert 0_P\rangle &\;=\; \lvert P\rangle \\[1mm]
		\lvert 1_P\rangle &\;=\; \lvert \bar{P}\rangle,
	\end{aligned}
\end{equation}
where the quantum information is directly encoded in the global quantum number (or equivalently, in the labels that distinguish the two states).

The logical basis in Eq.~(\ref{EQ:SingleParticlePackagedRectilinearBasis}) is very direct:
one basis state is given by the particle $ \lvert P \rangle $ and the other by the antiparticle $ \lvert \bar{P} \rangle $.
This is conceptually simple and easy for state preparation.

If $ P $ and $ \bar{P} $ are chosen such that they are both neutral (i.e., have net charge $ Q=0 $) and differ only in global quantum numbers (for instance, flavor or isospin), then the states naturally lie in the same superselection sector, making them physically admissible without further processing.
However, it has limited flexibility in basis choice. 
Since the basis is fixed by the physical particle content, it may not be the optimal choice for implementing certain logical operations (e.g., Hadamard gates) directly on the computational degrees of freedom.

\paragraph{(2) Single-Particle Qubits Based on Packaged Diagonal Basis.}
Applying a Hadamard transformation to Eq.~(\ref{EQ:SingleParticlePackagedRectilinearBasis}), we obtain the single-particle packaged diagonal basis: $\lvert +_P\rangle = \frac{1}{\sqrt{2}}\left(\lvert 0_P\rangle + \lvert 1_P\rangle\right)$, $\lvert -_P\rangle = \frac{1}{\sqrt{2}}\left(\lvert 0_P\rangle - \lvert 1_P\rangle\right)$.
In this case, we can re-define a packaged qubit based on the packaged diagonal basis, i.e.,
\begin{equation}\label{EQ:SingleParticlePackagedDiagonalBasis}
	\begin{aligned}
		\lvert 0_P\rangle &= \frac{1}{\sqrt{2}} \left(\lvert P\rangle + \lvert \bar{P}\rangle\right) \\[1mm]
		\lvert 1_P\rangle &= \frac{1}{\sqrt{2}} \left(\lvert P\rangle - \lvert \bar{P}\rangle\right),
	\end{aligned}
\end{equation}
where the logical information is encoded in the relative phase between the two components.

These basis states are symmetric superpositions of $ \lvert P\rangle $ and $ \lvert \bar{P}\rangle $.
They are analogous to the diagonal basis states in conventional qubit theory and are well-suited for operations such as the Hadamard gate. 
The superposed form can be more convenient in certain quantum algorithms and error-correcting schemes because the logical information is encoded in the relative phase between the two components.
The disadvantages are that generating and maintaining coherent superpositions may be more challenging than directly using distinct particle states, particularly if the natural system tends to distinguish $ P $ from $ \bar{P} $ due to charge.

\begin{example}[Neutral mesons]
	Neutral mesons such as $K^0$ and $\bar{K}^0$ provide exact examples of packaged qubits.
	Although both have zero electric charge, they differ by strangeness, a global quantum number.
	Let us denote:
	\[
	\lvert 0_P\rangle \;=\; \lvert K^0\rangle, \quad
	\lvert 1_P\rangle \;=\; \lvert \bar{K}^0\rangle.
	\]	
	Hence, the state
	\[
	\lvert \Psi\rangle = \alpha\,\lvert 0_P\rangle + \beta\,\lvert 1_P\rangle
	\]
	is physically realizable and can act as a single-particle packaged qubit.

	Similarly, in systems where neutrinos (or other neutral particles) exhibit oscillatory behavior due to a global quantum number difference, one can imagine encoding information in a single-particle packaged qubit formed by the neutrino and its distinct antiparticle counterpart (if they differ globally).
	
	 These states are natural candidates for encoding qubits precisely because they bypass the restrictions imposed by superselection rules.
\end{example}

This single-particle packaged qubit encodes quantum information in the global quantum number difference rather than in the gauge charge itself.
It is particularly useful in quantum computation and quantum teleportation where one wishes to transfer an unknown state.

\subsubsection{Constructing Multi-Particle Packaged Qubits}

We have now constructed single-particle packaged qubits.
However, these packaged qubits require neutral particle excitations.
This limits the available physical systems or requires careful state engineering.

In fact, we can use multi-particle systems to construct packaged qubits because a multi-particle system is much easier to achieve a total zero charge.
For example, in a particle-antiparticle system $\{P, \bar{P}\}$, even if $ P $ and $ \bar{P} $ individually carry opposite gauge charges, the two-particle packaged product states $ \lvert P\,\bar{P}\rangle $ and $ \lvert \bar{P}\,P\rangle $ have a total net charge of zero.
This automatically satisfies the superselection rule.

In Sec.~\ref{SEC:TwoParticlePackagedEntangledStates}, we already studied the properties of a particle-antiparticle system, in which the charge zero sector has two levels and also allows superposition.
These results will enable us to construct packaged qubits using particle-antiparticle pairs.
We can use either the packaged product basis or packaged entangled basis.

\paragraph{(1) Packaged Qubits Based on Packaged Product Basis: $\{\lvert P\,\bar{P}\rangle, \lvert \bar{P}\,P\rangle\}$.}

These two basis states have the following properties:
\begin{enumerate}
	\item Zero Net Gauge Charge: 
	Both states are eigenstates of the net-charge operator $\hat{Q} \otimes \hat{Q}$ with eigenvalue zero (Eq.~(\ref{EQ:ZeroNetGaugeChargeCond})):
	\[
	\hat{Q} \otimes \hat{Q} ~ \lvert P\,\bar{P}\rangle = 0, \quad \hat{Q} \otimes \hat{Q} ~ \lvert \bar{P}\,P\rangle = 0.
	\]
	
	\item Zero Global Quantum Number: 
	Because the charge conjunction operator $\hat{C}: \lvert P\rangle \leftrightarrow \lvert\bar{P}\rangle$ also flip global quantum numbers, the total global quantum number is zero, i.e.,
	\[
	\hat{F} \otimes \hat{F} ~ \lvert P\,\bar{P}\rangle = 0, \quad \hat{F} \otimes \hat{F} ~ \lvert \bar{P}\,P\rangle = 0.
	\]
	This shows that both packaged product basis states are eigenstates of the total global quantum number operator $\hat{F} \otimes \hat{F}$, each with eigenvalue zero.
\end{enumerate}

Under any local gauge transformation $U_g ~ (g \in G)$, the particle-antiparticle packaged states transform as
\[
U_g\,\lvert P \bar{P}\rangle = e^{i\phi(g)}\,\lvert P \bar{P}\rangle,\quad U_g\,\lvert \bar{P} P\rangle = e^{i\phi(g)}\,\lvert \bar{P} P\rangle.
\]
Thus, any superposition state
\[
\lvert \Psi\rangle = \alpha\,\lvert P\bar{P}\rangle + \beta\,\lvert \bar{P}P\rangle, \quad | \alpha |^2 + | \beta |^2 = 1,
\]
remains in the same gauge sector (i.e. $\hat{Q}=0$) and transforms as
\[
U_g\,\lvert \Psi\rangle = e^{i\phi(g)}\,\lvert \Psi\rangle.
\]

According to Proposition \ref{PROP:NecessarySufficientConditionsMultiParticle}, $\lvert \Psi\rangle$ is a valid multi-particle superposition state.
Thus, we can define a packaged qubit based on \textbf{particle-antiparticle packaged rectilinear basis} (or the other way dependent on the convention):
\begin{equation}\label{EQ:ParticleAntiparticlePackagedRectilinearBasis}
	\begin{aligned}
		\lvert 0_P\rangle &\;=\; \lvert P \bar{P}\rangle \\[1mm]
		\lvert 1_P\rangle &\;=\; \lvert \bar{P} P\rangle,
	\end{aligned}
\end{equation}
where the logical information is encoded in the ordering correlation between the two constituents.

These packaged product states satisfy the relations:
\[
\langle 0_P\vert 0_P\rangle=\langle 1_P\vert 1_P\rangle=1,\quad \langle 0_P\vert 1_P\rangle=0,
\]

Packaged product states are easier to prepare experimentally since they do not require the delicate control to create coherent superpositions.
However, they do not exhibit quantum entanglement between the two particles.
This can be a disadvantage for protocols (like teleportation or superdense coding) that rely on entangled resources for enhanced performance or security.
Furthermore, errors in the packaged product states may act independently on each particle and therefore degrade the encoded information.
This is because the protection provided by entanglement tends to delocalize errors and is now absent.
Let us now start to construct packaged entangled basis.

\paragraph{(2) Packaged Qubits Based on Packaged Entangled Basis $\{ \lvert \Psi^+\rangle, \lvert \Psi^-\rangle\}$.}

In Sec.~\ref{SEC:MultiParticlePackagedStates} we have shown that a particle-antiparticle system has a complete orthonormal basis of packaged entangled states $\{ \lvert \Psi^+\rangle, \lvert \Psi^-\rangle\}$ (see Eq.~(\ref{EQ:2DOrthonormalBasis})).
This basis spans a 2D subspace 
$
\mathcal{H}_{Q=0} = \operatorname{span}\{ \lvert \Psi^+\rangle, \lvert \Psi^-\rangle\}.
$
It is indeed the Hadamard transformation of Eq.~(\ref{EQ:ParticleAntiparticlePackagedRectilinearBasis}), i.e., $\lvert +_P\rangle 
= \frac{1}{\sqrt{2}}\left(\lvert 0_P\rangle + \lvert 1_P\rangle\right)$, $\lvert -_P\rangle 
= \frac{1}{\sqrt{2}}\left(\lvert 0_P\rangle - \lvert 1_P\rangle\right)$.
This can be regarded as a particle-antiparticle packaged diagonal basis.

For the convenience of applications, we can re-define a packaged qubit based on particle-antiparticle packaged entangled basis:
\begin{equation}\label{EQ:ParticleAntiparticlePackagedEntangledBasis}
	\begin{aligned}
		\lvert 0_P\rangle 
		&= \frac{1}{\sqrt{2}}\Bigl(\lvert P\,\bar{P}\rangle + \lvert \bar{P}\,P\rangle\Bigr),\\[1mm]
		\lvert 1_P\rangle 
		&= \frac{1}{\sqrt{2}}\Bigl(\lvert P\,\bar{P}\rangle - \lvert \bar{P}\,P\rangle\Bigr).
	\end{aligned}
\end{equation}
where the logical information is encoded in the entanglement structure (entangled phase between the logical basis states).

These packaged entangled states satisfy the relations:
\[
\langle 0_P\vert 0_P\rangle=\langle 1_P\vert 1_P\rangle=1,\quad \langle 0_P\vert 1_P\rangle=0,
\]
and both $\lvert 0_P\rangle$ and $\lvert 1_P\rangle$ are non-factorizable superpositions over the constituent multi-particle excitations carrying inseparable IQNs.

Thus, any packaged states
$\lvert \Psi\rangle
\in
\mathcal{H}_{Q=0}=\operatorname{span}\{\lvert 0_P\rangle,\lvert 1_P\rangle\}
$
can be written as:
\[
\lvert \Psi\rangle = \alpha\,\lvert 0_P\rangle + \beta\,\lvert 1_P\rangle,\quad | \alpha |^2 + | \beta |^2 = 1,
\]
and $\lvert \Psi\rangle$ remains in the same gauge sector (i.e. $\hat{Q}=0$) and transforms as
\[
U_g\,\lvert \Psi\rangle = e^{i\phi(g)}\,\lvert \Psi\rangle.
\]
This guarantees that $\lvert \Psi\rangle$ is a valid physical state and superpositions are not forbidden by superselection rules.
Therefore, we can use the maximal orthonormal packaged entangled states $ \lvert \Psi^{\pm}\rangle$ as the packaged logical $\lvert 1_P\rangle$ and $\lvert 0_P\rangle$ for our multi-particle packaged qubit.

The maximally packaged entangled states are useful in quantum communication protocols (such as teleportation and superdense coding) because the entanglement can provide both enhanced security and robustness against certain types of errors.
Distributed entanglement often leads to error suppression because local noise is less likely to disrupt a globally entangled state.
In addition, even if one part of the pair suffers an error, the overall correlation might still be detected via Bell measurements.
Since the net charge of each Bell state is zero, the packaged entangled states automatically lie in the correct superselection sector.

The disadvantages are that creating and maintaining maximally entangled states is usually more challenging than preparing product states. The requirement for coherence over two particles means that experimental imperfections (such as decoherence \cite{Zurek1982} or imperfect control) can more easily destroy the entanglement.
Bell measurements, which are needed to fully exploit the entangled resource, are generally more complex to implement than measurements in a product basis. This might increase the overhead in any protocol that uses such states.

At this point we would like to emphasize that, for single-particle packaged qubits, the two conditions (neutrality under the gauge charge (ensuring both states lie in $\mathcal{H}_{Q=0}$) and a nontrivial (global) quantum number difference) are both necessary and sufficient to allow a physically valid coherent superposition
\[
\lvert \Psi\rangle = \alpha\,\lvert 0_P\rangle + \beta\,\lvert 1_P\rangle.
\]
For particle-antiparticle qubits, these conditions translate into that the two-particle basis states (such as $\lvert P\,\bar{P}\rangle$ and $\lvert \bar{P}\,P\rangle$) have both zero net gauge charge and zero net global quantum number.
This is naturally accomplished if the particle and antiparticle have opposite gauge charges and opposite values of the global quantum number.

Finally, let us summarize all the advantages and disadvantages of the four packaged qubit basis in Table \ref{TAB:AdvDisadvOfPackagedQubit}:

\begin{table}[H]
	\centering
	\caption{Comparison of Packaged Qubit Basis}
	\label{TAB:AdvDisadvOfPackagedQubit}
	\begin{tabular}[hbt!]{|p{2.9cm}|p{3.7cm}|p{4cm}|p{4cm}|}
		\hline\hline
		Type &Basis &Advantages &Disadvantages \\
		\hline
		Single-Particle, Rectilinear Basis &$\lvert 0_P\rangle = \lvert P\rangle$, $\lvert 1_P\rangle = \lvert \bar{P}\rangle$ &Simple and direct; minimal state preparation &Valid only if $P$ and $\bar{P}$ are in the same $ \mathcal{H}_Q $ (i.e. both neutral); may not support coherent superpositions if they have different gauge charges \\
		\hline
		Single-Particle, Diagonal Basis &$\lvert 0_P\rangle = \frac{1}{\sqrt{2}}(|P\rangle + |\bar{P}\rangle)$, $\lvert 1_P\rangle = \frac{1}{\sqrt{2}}(|P\rangle - |\bar{P}\rangle)$ &Provides a symmetric basis; ideal for Hadamard operation; optimized for coherent superposition and interference &Requires that $P$ and $\bar{P}$ can be coherently superposed (they differ only in global quantum numbers) \\
		\hline
		Two-Particle, Product Basis & $\lvert 0_P\rangle = |P\,\bar{P}\rangle$,\; $\lvert 1_P\rangle = |\bar{P}\,P\rangle$ & Automatically neutral; allowed even if individual particles have nonzero gauge charges; simpler state preparation & States are not entangled; more vulnerable to local errors since each particle is independent \\
		\hline
		Two-Particle, Entangled Basis & $\lvert 0_P\rangle = \frac{1}{\sqrt{2}}(|P\,\bar{P}\rangle + |\bar{P}\,P\rangle)$,\; $\lvert 1_P\rangle = \frac{1}{\sqrt{2}}(|P\,\bar{P}\rangle - |\bar{P}\,P\rangle)$ & Maximally entangled; provides enhanced security and error resilience; naturally resides in $ \mathcal{H}_{Q=0} $ & More challenging to generate and measure; requires high coherence and more complex control techniques \\
		\hline
	\end{tabular}
\end{table}

\subsection{Packaged Qubit Gates}
\label{SEC:PurePackagedQubitGates}

Recall that a unitary operator $V$ acting on a Hilbert space $\mathcal{H}$ is defined by
\[
V^\dagger V = V V^\dagger = \mathbb{I},
\]
which means that $V$ preserves the inner product.
In our case, we consider the packaged qubit subspace $\mathcal{H}_{\rm qubit}\subset\mathcal{H}_Q$, where all basis states (e.g., $\lvert 0_P\rangle$ and $\lvert 1_P\rangle$) reside in the same net-charge sector.
$\forall ~ U_g\in G$, we have
\[
U_g\,\lvert 0_P\rangle = e^{i\phi(g)}\,\lvert 0_P\rangle,\quad U_g\,\lvert 1_P\rangle = e^{i\phi(g)}\,\lvert 1_P\rangle.
\]
This property ensures that any state in $\mathcal{H}_{\rm qubit}$ is gauge-invariant (up to an overall phase).

\subsubsection{Definition of Packaged Quantum Gates}

A packaged quantum gate \cite{NielsenChuang2010} is a unitary operator that acts on the packaged qubit $\mathcal{H}_{\rm qubit}$.
Since both $\lvert 0_P\rangle$, $\lvert 1_P\rangle \in \mathcal{H}_{\rm qubit}$ belong to the same net-charge sector $\mathcal{H}_Q$, any operators defined on $\mathcal{H}_{\rm qubit}$ must preserve gauge-invariance and remain confined to a single charge sector $\mathcal{H}_Q$.

In particular, a unitary operator $V$ acting on $\mathcal{H}_{\rm qubit}$ must also satisfy the above conditions: preserve gauge-invariance and remain confined to a single charge sector $\mathcal{H}_Q$.
In other words, any physically admissible gate $V$ acting on $\mathcal{H}_{\rm qubit}$ must be gauge-invariant,
i.e., it must commute with the net-charge operator $\hat{Q}$:
\[
[V, \hat{Q}]=0.
\]
This condition implies that only those unitary operations that preserve the charge sector are allowed.
It is automatically true if for all local gauge transformations $U_g ~ (g \in G)$, we have
\[
U_g\,\lvert 0_P\rangle = e^{i\phi(g)}\lvert 0_P\rangle,\quad U_g\,\lvert 1_P\rangle = e^{i\phi(g)}\lvert 1_P\rangle.
\]

Because all packaged states acquire the same gauge phase under local gauge transformations, any packaged gate defined on the logical subspace (constructed from such states) is automatically gauge-invariant.
On the other hand, the packaged states reside entirely in $\mathcal{H}_Q$, every unitary operator $V$ that acts on $\mathcal{H}_{\rm qubit}$ is automatically compatible with the superselection rule.

For abelian gauge groups, we assume every physical gate that preserves net charge also preserves each irrep block.
For non-abelian ones, however, different irreps can mix within the same total $Q$.
We further restrict to block-diagonal operators in each irrep.

We now give a formal definition for the packaged qubit gates:

\begin{definition}[Packaged Qubit Gate]
	Let $\mathcal{H}_{\rm qubit}$ be a two-dimensional subspace of a fixed net-charge (or superselection) sector $\mathcal{H}_Q$ of the full Hilbert space $\mathcal{H}$. A \textbf{packaged qubit gate} is a unitary operator defined on $\mathcal{H}_{\rm qubit}$:
	\[
	V: \mathcal{H}_{\rm qubit} \to \mathcal{H}_{\rm qubit}
	\]
	that satisfies the gauge-invariance condition
	\[
	[V, \hat{Q}] = 0,
	\]
	i.e., for every $\lvert \psi \rangle\in \mathcal{H}_{\rm qubit}$ and every local gauge transformation $U_g$, we have
	\[
	U_g\, V \lvert \psi \rangle = V \,U_g \lvert \psi \rangle.
	\]
\end{definition}

From the definition, we see that any physically admissible gate on packaged qubits is automatically compatible with the superselection rule, since the gauge-invariance of the basis implies that
\[
[V, \hat{Q}] = 0.
\]

Furthermore, one can show that the composition of two packaged quantum gates is again a packaged quantum gate.

\subsubsection{Closure of Packaged Qubit Gates}

Indeed, if $V_1$ and $V_2$ are packaged qubit gates satisfying
\[
[V_1, \hat{Q}] = 0 \quad \text{and} \quad [V_2, \hat{Q}] = 0,
\]
then using the product rule of commutators, one can prove that the composed gate $V_1 V_2$ is also gauge-invariant.
We now state this fact as a lemma:

\begin{lemma}[Closure of Packaged Qubit Gates]\label{LEMMA:ClosureOfPackagedOperations}
	Let $ V_1 $ and $ V_2 $ be packaged qubit gates (unitary operators) that act on the packaged Hilbert space such that 
	\[
	[V_1, \hat{Q}] = 0 \quad \text{and} \quad [V_2, \hat{Q}] = 0.
	\]
	Then, the composition $ V = V_1 V_2 $ also satisfies 
	\[
	[V, \hat{Q}] = 0.
	\]
\end{lemma}

\begin{proof}
	Using the Leibniz product rule for commutators, we have 
	\[
	[V_1 V_2, \hat{Q}] = V_1 [V_2, \hat{Q}] + [V_1, \hat{Q}] V_2.
	\]
	Using the assumptions 
	\[
	[V_2, \hat{Q}] = 0 \quad \text{and} \quad [V_1, \hat{Q}] = 0,
	\]
	we obtain 
	\[
	[V_1 V_2, \hat{Q}] = V_1 \cdot 0 + 0 \cdot V_2 = 0.
	\]
	This shows that $ V_1 V_2 $ commutes with $ \hat{Q} $, which proves the lemma. 
	$\Box$
\end{proof}

This lemma can be easily extended to $n$ unitary operators.
It has important applications in packaged qubit circuit, quantum algorithms, and quantum communication protocols.

\subsubsection{Packaged Single-Qubit Gates}

Consider a packaged qubit subspace
$
\mathcal H_{\text{qubit}}
=\operatorname{span}\{\ket{0_P},\ket{1_P}\}\subset\mathcal H_{Q=0},
$
where the net-charge operator $\hat{Q}$ satisfies
$
\hat{Q}\ket{0_P}=0,
$
$
\hat{Q}\ket{1_P}=0.
$
We can define the following packaged single-qubit gates on $\mathcal H_{\text{qubit}}$:

\begin{enumerate}
	\item Pauli-$X_P$ (Bit-Flip):
	\begin{equation}\label{EQ:PackagedPauliX}
		X_P = \lvert 0_P\rangle\langle 1_P\rvert + \lvert 1_P\rangle\langle 0_P\rvert.
	\end{equation}
	Represent by a $2\times 2$ unitary matrix:
	\begin{equation}\label{EQ:PackagedPauliXMatrix}
		X_P = \begin{pmatrix} 0 & 1 \\ 1 & 0 \end{pmatrix}
	\end{equation}
	Acting on the logical packaged basis states:
	\[
	X_P\lvert 0_P\rangle=\lvert 1_P\rangle,\quad X_P\lvert 1_P\rangle=\lvert 0_P\rangle,
	\]
	
	\item Pauli-$Y_P$:
	\begin{equation}\label{EQ:PackagedPauliY}
		Y_P = -i \lvert 0_P\rangle\langle 1_P\rvert + i \lvert 1_P\rangle\langle 0_P\rvert.
	\end{equation}
	Represent by a $2\times 2$ unitary matrix:
	\begin{equation}\label{EQ:PackagedPauliYMatrix}
		Y_P = \begin{pmatrix} 0 & -i \\ i & 0 \end{pmatrix}
	\end{equation}
	Acting on the logical packaged basis states:
	\[
	Y_P\lvert 0_P\rangle=i\lvert 1_P\rangle,
	\quad 
	Y_P\lvert 1_P\rangle=-i\lvert 0_P\rangle,
	\]
	
	\item Pauli-$Z_P$ (Phase-Flip):
	\begin{equation}\label{EQ:PackagedPauliZ}
		Z_P = \lvert 0_P\rangle\langle 0_P\rvert - \lvert 1_P\rangle\langle 1_P\rvert.
	\end{equation}
	Represent by a $2\times 2$ unitary matrix:
	\begin{equation}\label{EQ:PackagedPauliZMatrix}
		Z_P = \begin{pmatrix} 1 & 0 \\ 0 & -1 \end{pmatrix}
	\end{equation}
	Acting on the logical packaged basis states:
	\[
	Z_P\lvert 0_P\rangle=\lvert 0_P\rangle,\quad Z_P\lvert 1_P\rangle=-\lvert 1_P\rangle.
	\]
		
	\item Hadamard Gate $H_P$:
	\begin{equation}\label{EQ:PackagedHadamard}
		H_P = \frac{1}{\sqrt{2}} 
		\left(
		\lvert 0_P\rangle\langle 0_P\rvert +
		\lvert 0_P\rangle\langle 1_P\rvert + 
		\lvert 1_P\rangle\langle 0_P\rvert - 
		\lvert 1_P\rangle\langle 1_P\rvert
		\right).
	\end{equation}
	Represent by a $2\times 2$ unitary matrix:
	\begin{equation}\label{EQ:PackagedHadamardMatrix}
		H_P = \frac{1}{\sqrt{2}} \begin{pmatrix} 1 & 1 \\ 1 & -1 \end{pmatrix}.
	\end{equation}
	
	Acting on the logical packaged basis states:
	\[
	H_P\lvert 0_P\rangle = \frac{1}{\sqrt{2}}\left(\lvert 0_P\rangle + \lvert 1_P\rangle\right),\quad
	H_P\lvert 1_P\rangle = \frac{1}{\sqrt{2}}\left(\lvert 0_P\rangle - \lvert 1_P\rangle\right).
	\]	
	This operation preserves the net gauge charge while creating nontrivial superpositions of the packaged single-qubit basis states.
	
	\item $\tfrac\pi4$ phase $T_P$:
	\begin{equation}\label{EQ:PackagedT}
		T_P = \ket{0_P}\!\bra{0_P}+e^{i\pi/4}\ket{1_P}\!\bra{1_P}
	\end{equation}
	Represent by a $2\times 2$ unitary matrix:
	\begin{equation}\label{EQ:PackagedTmatrix}
		T_P = \begin{pmatrix} 1 &0 \\ 0 &e^{i\pi/4} \end{pmatrix}
	\end{equation}	
	 This gate is non-Clifford, which turns $X_P$ into $\tfrac{X_P+Y_P}{\sqrt2}$
\end{enumerate}

We now summarize the above defined packaged single-qubit gates in following table:

\begin{table}[H]
	\centering
	\caption{Packaged Single-Qubit Gates}
	\begin{tabular}[hbt!]{|p{2cm}|p{4.5cm}|p{2.1cm}|p{5cm}|}
		\hline\hline
		Name &Abstract definition &$2\times2$ matrix &Action on basis \\
		\hline
		Bit-flip $X_P$ & $\displaystyle X_P=\ket{0_P}\!\bra{1_P}+\ket{1_P}\!\bra{0_P}$ & $\begin{pmatrix}0&1\\[1pt]1&0\end{pmatrix}$ & $X_P\ket{0_P}=\ket{1_P},\;X_P\ket{1_P}=\ket{0_P}$ \\
		\hline
		$Y_P$ & $\displaystyle Y_P=-i\ket{0_P}\!\bra{1_P}+i\ket{1_P}\!\bra{0_P}$ & $\begin{pmatrix}0&-i\\[1pt]i&0\end{pmatrix}$ & $Y_P\ket{0_P}=i\ket{1_P},\;Y_P\ket{1_P}=-i\ket{0_P}$ \\
		\hline
		Phase-flip $Z_P$ & $\displaystyle Z_P=\ket{0_P}\!\bra{0_P}-\ket{1_P}\!\bra{1_P}$ & $\begin{pmatrix}1&0\\[1pt]0&-1\end{pmatrix}$ & $Z_P\ket{0_P}= \ket{0_P},\;Z_P\ket{1_P}=-\ket{1_P}$ \\
		\hline
		Hadamard $H_P$ & $\displaystyle H_P=\frac{1}{\sqrt2}\bigl(X_P+Z_P\bigr)$ & $\tfrac1{\sqrt2}\!\begin{pmatrix}1&1\\[2pt]1&-1\end{pmatrix}$ & Swaps $X_P\!\leftrightarrow\!Z_P$: $H_PX_PH_P^\dagger=Z_P$ \\
		\hline
		$\tfrac\pi4$ phase $T_P$ & $\displaystyle T_P=\ket{0_P}\!\bra{0_P}+e^{i\pi/4}\ket{1_P}\!\bra{1_P}$ & $\begin{pmatrix}1&0\\[2pt]0&e^{i\pi/4}\end{pmatrix}$ & Non-Clifford, turns $X_P$ into $\tfrac{X_P+Y_P}{\sqrt2}$ \\
		\hline
	\end{tabular}
\end{table}

Since states $\{\lvert 0_P\rangle,\lvert 1_P\rangle\}$ lie in $\mathcal{H}_{Q=0}$,
we have $\hat{Q}\ket{0_P}=0$ and $\hat{Q}\ket{1_P}=0$.
One can easily verify that all packaged single-qubit gates commute with $\hat{Q}$, i.e.,
\[
[\,X_P,\hat{Q}]=[\,Y_P,\hat{Q}]=[\,Z_P,\hat{Q}]=[\,H_P,\hat{Q}]=[\,T_P,\hat{Q}]=0.
\]

\subsubsection{Packaged Entangled Two-qubit Gate}

Consider two packaged qubits:
$\mathcal{H}^{(1)}_{\rm qubit} = \operatorname{span}\{\ket{0_P},\ket{1_P}\}$
and $\mathcal{H}^{(2)}_{\rm qubit} =\operatorname{span}\{\ket{0_P},\ket{1_P}\}$.
Their joint space is the tensor product space
$
\mathcal{H}^{(1)}_{\rm qubit} \otimes \mathcal{H}^{(2)}_{\rm qubit}\subset \mathcal{H}_{Q=0}\otimes \mathcal{H}_{Q=0},
$
whose basis is $\{\ket{0_P},\ket{1_P}\}^{\!\otimes2}$.

We define the \textbf{packaged controlled-NOT (CNOT) gate} as:
\begin{equation}\label{EQ:PackagedCNOT}
	\text{CNOT}_P
	=\ket{0_P}\!\bra{0_P}\!\otimes\!\mathbf 1
	+\ket{1_P}\!\bra{1_P}\!\otimes\!X_P,
\end{equation}
where the first qubit acts as control and the second as target,
the projector operators and $X_P$ are defined on packaged single-qubits.

The matrix form (computational ordering $\ket{00},\ket{01},\ket{10},\ket{11}$) of the packaged $\text{CNOT}_P$ gate is:
\begin{equation}\label{EQ:PackagedCNOTMatrix}
	\text{CNOT}_P=
	\begin{pmatrix}
		1&0&0&0\\
		0&1&0&0\\
		0&0&0&1\\
		0&0&1&0
	\end{pmatrix}.
\end{equation}

Since each packaged single-qubit has net charge $\hat{Q}=0$, the action of $\text{CNOT}_P$ remains in the net-zero sector, i.e.,
\[
[\text{CNOT}_P,\hat{Q}_1+\hat{Q}_2]=0.
\]
This ensures gauge-invariance of $\text{CNOT}_P$.

\subsubsection{Universality of Packaged Qubit Gate Set}

The universality of packaged qubit gate set is important in both theory and applications.
But here we focus on pure-packaged qubit gate to keep the discussion concise.
Later in Sec.~\ref{SEC:HybridPackagedQuditGates}, we provide rigorous proofs for ($d \times D$)-dimensional hybrid-packaged gate sets.

Similarly, we consider the ``Clifford + $T$'' library of packaged qubits:
\begin{equation}\label{EQ:CliffordPlusTPackagedQubit}
	\mathcal G_{(2)} \;=\;\{\,X_P,\;Z_P,\;H_P,\;\mathrm{CNOT}_P,\;T\,\}
	\;\subset\;\mathcal C_{\hat{Q}},
\end{equation}
where
\[
\mathcal C_{\hat{Q}} \;=\;\bigl\{\,V \in \mathrm U(2)\;\big|\;[V,\hat{Q}]=0\bigr\}.
\]
is the allowed commutant.
These five gates can be split into three groups:
\begin{enumerate}
	\item Packaged Single-qubit Clifford generators $\{X_P,Z_P,H_P\}$ (see Eqs.~(\ref{EQ:PackagedPauliX}), (\ref{EQ:PackagedPauliZ}), and (\ref{EQ:PackagedHadamard})):
	
	These obey exactly the same algebra as their conventional qubit counterparts:
	\[
	H_P\,X_P\,H_P = Z_P,
	\quad
	H_P\,Z_P\,H_P = X_P^\dagger,
	\quad
	X_P^2 = Z_P^2 = I.
	\]
	Thus, they generate the full single-qubit Clifford group on $\mathcal H_{Q=0}$.
	
	\item Packaged Entangled Clifford $\mathrm{CNOT}_P$ (see Eq.~(\ref{EQ:PackagedCNOT})):
	
	This gate commutes with the total charge $\hat{Q}_1+\hat{Q}_2$ and provides two-qubit Clifford entanglement.
	
	\item Packaged Non-Clifford diagonal $T$ (see Eq.~(\ref{EQ:PackagedT})):
	
	This gate commutes with $\hat{Q}$ and breaks out of the Clifford subgroup.
\end{enumerate}

The ``Clifford + $T$'' library, $\mathcal G_{(2)}$ is dense in $\mathrm{SU}(2^n)$.\cite{Deutsch1989,Lloyd1995,Barenco1995,Kitaev1997,NielsenChuang2010}
Because every generator lies in the commutant $\mathcal C_{\hat{Q}}:=\{V\mid[V,\hat{Q}]=0\}$, this remains true inside the neutral sector $\mathcal H_{Q=0}^{\!\otimes n}$,
which is restricted by superselection.

\begin{theorem}[Universality of a Gauge-Respecting Packaged Gate Set]
	Let $\hat{Q}$ be the conserved total charge on $\mathcal H_{Q=0}\cong\mathbb C^{2}$. 
	Then the ``Clifford + $T$'' library $\mathcal G_{(2)}$ of packaged qubits defined Eq.~(\ref{EQ:CliffordPlusTPackagedQubit}) is dense in $\mathrm{SU}(2^n)$, i.e., satisfies
	\[
	\overline{\langle\mathcal G_{(2)}\rangle}
	\;=\;
	\mathrm{SU}\bigl(\mathcal H_{Q=0}\bigr).
	\]
\end{theorem}

\begin{corollary}[Solovay-Kitaev Theorem \cite{Kitaev1997,DawsonNielsen2005} in Packaged Gates]
	For any $n$-qubit unitary $V \in \mathrm{SU}(2^{n})$ that preserves the neutral sector and any accuracy $\varepsilon\in(0,1]$, there exists a word	
	\[
	\widetilde V=G_{i_L}\cdots G_{i_1},\qquad 
	G_{i_\ell}\in\mathcal G_{\mathrm{(2)}},
	\]	
	such that	
	\[
	\|V-\widetilde V\|_{\mathrm{op}}\le\varepsilon,
	\quad
	L=O\!\bigl(\log^{\kappa}\varepsilon^{-1}\bigr),\;\kappa\le 3.97.
	\]
	Every prefix $G_{i_\ell}\cdots G_{i_1}$ commutes with the total charge $\hat{Q}_{\mathrm{tot}}$.
\end{corollary}

In other words, any gauge-respecting $n$-qubit unitary can be approximated to precision $\varepsilon$ by a word over $\mathcal G_{(2)}$ of length $L=O(\log^{\kappa}\!\varepsilon^{-1})$ with $\kappa\le3.97$ \cite{DawsonNielsen2005}.
Every intermediate circuit element $V$ keeps $[V,\hat{Q}]=0$.

This shows that using packaged qubits does not degrade computational power.
Every familiar single- and two-qubit gate has a direct packaged analogue and satisfies the same algebraic relations.
Adding a non-Clifford $T$, they form a finite and fault-tolerant \cite{KnillLaflamme1886,Aharonov2008,Boykin1999} universal set within the superselection sector.

\subsection{Packaged Qubit Circuits}
\label{SEC:PackagedQuantumCircuits}

A packaged qubit circuit is obtained by sequencing a set of packaged gates: 
\[
V_{\rm circuit} = V_k\,V_{k-1}\,\cdots\,V_2\,V_1,
\]
where each packaged gate $V_j$ satisfies
\[
[V_j,\hat{Q}]=0,
\]
i.e., each $V_j$ is gauge-invariant and acts on one or more single-particle packaged qubits.

\subsubsection{Definition of Packaged Qubit Circuits}

We now give the formal definition of a packaged qubit circuit.

\begin{definition}[Packaged Qubit Circuit]\label{DEF:PackagedQubitCircuit}
	Let $V_{\rm circuit}$ be a unitary operator 
	\[
	V_{\rm circuit} : \mathcal{H}_{\rm qubit}^{\otimes n} \to \mathcal{H}_{\rm qubit}^{\otimes n}
	\]
	that is constructed as a sequential composition of unitary operators $V_j$, i.e.,
	\[
	V_{\rm circuit} = V_k\,V_{k-1}\,\cdots\,V_2\,V_1,
	\]
	If each $ V_j $ is a packaged quantum gate that acts on the logical packaged subspace $\mathcal{H}_{\rm qubit}$ and satisfies
	\[
	[V_j, \hat{Q}] = 0,
	\]
	then we say that $ V_{\rm circuit} $ is a \textbf{packaged qubit circuit}.
\end{definition}

In other words, a packaged qubit circuit is a circuit composed entirely of packaged qubit gates that preserve gauge-invariance and remain confined to the same superselection sector $\mathcal{H}_Q$.

Given Definition \ref{DEF:PackagedQubitCircuit} and Lemma \ref{LEMMA:ClosureOfPackagedOperations}, we intuitively conclude that $\bigl[ V_{\rm circuit}, \hat{Q} \bigr] = 0$.

\subsubsection{Properties of Packaged Qubit Circuits}

In Definition \ref{DEF:PackagedQubitCircuit}, we do not explicitly require that a packaged qubit circuit to be gauge-invariant.
Thus, we need to prove that any packaged qubit circuit is gauge-invariant.

\begin{theorem}[Gauge-Invariance of Packaged Qubit Circuits]
	\label{THM:GaugeInvariantOfPackagedQuantumCircuits}
	A packaged qubit circuit is gauge-invariant.
\end{theorem}

\begin{proof}
	Consider a packaged qubit circuit
	\[
	V_{\rm circuit} = V_k\,V_{k-1}\,\cdots\,V_2\,V_1,
	\]
	where each gate $ V_j $ is a packaged qubit gate that satisfies
	$
	[V_j, \hat{Q}] = 0.
	$
	Then we need to prove
	$
	\bigl[ V_{\rm circuit}, \hat{Q} \bigr] = 0.
	$
	Let us now prove the theorem by induction on the number $ k $ of gates.
	
	\begin{enumerate}
		\item \textbf{Base case:} For $ k=1 $, the circuit is simply $ V_{\rm circuit} = V_1 $, and by assumption,
		\[
		[V_1, \hat{Q}] = 0.
		\]
		
		\item \textbf{Inductive step:} Assume that for a product of $ n $ gates,
		\[
		W_{n} = V_n\,V_{n-1}\,\cdots\,V_1,
		\]
		we have
		\[
		\left[ W_{n}, \hat{Q} \right] = 0.
		\]
		Now consider the product of $ n+1 $ gates,
		\[
		W_{n+1} = V_{n+1}\,W_{n}.
		\]
		Using the product rule for commutators, we have
		\[
		\left[ W_{n+1}, \hat{Q} \right] = V_{n+1} \left[ W_{n}, \hat{Q} \right] + \left[ V_{n+1}, \hat{Q} \right] W_{n}.
		\]
		By the induction hypothesis, $\left[ W_{n}, \hat{Q} \right] = 0$, and by assumption, $\left[ V_{n+1}, \hat{Q} \right] = 0$. Therefore,
		\[
		\left[ W_{n+1}, \hat{Q} \right] = V_{n+1} \cdot 0 + 0 \cdot W_{n} = 0.
		\]
		Thus, by induction, we conclude that for any packaged qubit circuit $ V_{\rm circuit} $,
		\[
		\left[ V_{\rm circuit}, \hat{Q} \right] = 0.
		\]
	\end{enumerate}
\end{proof}

Because the entire packaged qubit circuit is gauge-invariant (confined to the net-zero gauge sector), it respect the superselection rules and maintain the packaged structure.
Thus, it guarantees that the circuit operations are physically realizable.
This constraint also acts like a built-in stabiliser:
errors that would change the net charge anticommute with $\hat{Q}$ are forbidden, while gauge-conserving errors remain correctable by ordinary techniques. 
Hence packaged qubits, gates and circuits provide a native layer of error filtering that we will use in the new quantum-communication protocols.

\begin{example}
	Let $\lvert \Psi_{\rm in}\rangle \in \mathcal{H}_{\rm qubit}^{\otimes n}$ be an input state with net charge zero (that is, each packaged qubit lies in the subspace $\mathcal{H}_{Q=0}$ where the net‐charge operator $\hat{Q}$ has eigenvalue zero).
	Because every gate in a packaged circuit commutes with the total-charge operator $\hat{Q}$, the entire evolution
	\[
	|\Psi_{\text{out}}\rangle
	=\;V_kV_{k-1}\!\cdots V_1\,|\Psi_{\text{in}}\rangle , 
	\qquad [V_j,\hat{Q}]=0 ,
	\]
	remains inside the neutral sector $\mathcal H_{Q=0}^{\otimes n}$. 
	Equivalently, for any local gauge transformation $U_g$,
	\[
	U_g |\Psi_{\text{out}}\rangle
	=\;e^{i\phi(g)} |\Psi_{\text{out}}\rangle ,
	\]
	so packaged circuits are automatically gauge-respecting and physically realizable.
\end{example}

\section{Packaged States for Quantum Communication}
\label{SEC:PackagedStatesForQuantumCommunication}

In Secs.~\ref{SEC:GaugeInvariantPackagedStates} and \ref{SEC:PackagedQubitsGatesAndCircuits}, we have shown that every packaged state resides in a fixed superselection sector (e.g., with net charge $Q$) and is therefore immune to operations that would mix distinct gauge sectors.
Here we develop a mathematical foundation for employing such packaged states as robust resource states in quantum communication protocols.

\subsection{Packaged Messengers}
\label{SEC:PackagedMessengers}

In quantum communication, we need messengers (quantum states) to carry and transfer quantum information. 
Because we are now developing the quantum communication protocols using packaged states, the state superposition is subject to superselection rules \cite{WWW1952,DHR1971,DHR1974,StreaterWightman2001}.
Therefore, we meet the similar problems as those in constructing packaged qubits (see Sec.~\ref{SEC:PackagedQubits}),
so we now address it in detail.

\subsubsection{Definition of Packaged Messengers}

We prefer to define a packaged messenger as the packaged Hilbert space of a particle rather than as the particle itself.
In this way, the logical information is explicitly encoded in the packaged space and its IQNs remain locked by gauge-invariance.
Defining a packaged messenger as the packaged space explicitly expresses that all IQNs are locked together within a fixed superselection sector.
This ensures that the packaged messenger is invariant under local gauge transformations.
It thus naturally protects against errors that could otherwise mix different gauge sectors.
Furthermore, by focusing on the packaged space, we abstract away from the specific nature of the particle and focus on the properties required for robust quantum information processing.
Finally, in communication protocols, what is transmitted, measured, and manipulated is the quantum state in the packaged space, not the particle itself.
Hence, defining a packaged messenger as the packaged Hilbert space clearly identifies the carrier of quantum information.

A packaged messenger at least satisfies the conditions of a pure-packaged subspace given in Definition \ref{DEF:PurePackagedSubspace}.
In addition to this, a packaged messenger must be transmittable.
Similarly to the definition of a pure-packaged qubit (see Definition \ref{DEF:PurePackagedQubit}), we formally define a packaged messenger as:

\begin{definition}[Packaged Messenger]\label{DEF:PackagedMessenger}
	A \textbf{packaged messenger} is a two-dimensional pure-packaged subspace 
	$\mathcal H_{\rm mess} = \operatorname{span}\{\lvert 0_P\rangle,\lvert 1_P\rangle\} ~ \subset \mathcal H_{Q=0}$ that is used to carry and transfer quantum information.
	$\forall~ |\Psi_M\rangle \in \mathcal H_{\rm mess}$, we can write it as 
	\begin{equation}\label{EQ:PackagedMessengerState}
		\boxed{|\Psi_M\rangle \;=\; \alpha\,|0_P\rangle \;+\;\beta\,|1_P\rangle,}
	\end{equation}
	and called it a \textbf{packaged messenger state}.
	The packaged messenger thus inherits the intrinsic error-protection provided by the superselection rules. 
\end{definition}

\subsubsection{Properties of a Packaged Messenger}

According to Definition \ref{DEF:PackagedMessenger}, a packaged messenger has the following properties:

\paragraph{(1) A Packaged Messenger is Gauge-Invariant.}

Let $G$ be a local gauge group.
Because a packaged messenger $\mathcal H_{\rm mess} \subset \mathcal H_{Q=0}$, 
for all $\lvert \Psi_M\rangle \in \mathcal H_{\rm mess}$ and $g \in G$, we have
\[
U_g\,\lvert \Psi_M\rangle = e^{i\phi(g)}\,\lvert \Psi_M\rangle,
\]
This ensures that the messenger state remains in the desired superselection sector and therefore is robust against gauge-violating noise.
It is necessary for a packaged messenger state to be valid in a packaged quantum communication protocol.

\paragraph{(2) Global IQNs differ in $\{|0_M\rangle,|1_M\rangle\}$ of a Single-Particle Messenger.}

If the packaged messenger is built from single-particle packaged qubit,
then according to Proposition \ref{PROP:NecessarySufficientConditionsSingleParticle}, we require that
\[
\hat{Q}\lvert 0_M\rangle = 0,\quad \hat{Q}\lvert 1_M\rangle = 0,
\]
and for a global operator $\hat{F}$ (which is not subject to local gauge constraints),
\[
\hat{F}\lvert 0_M\rangle = f\,\lvert 0_M\rangle,\quad \hat{F}\lvert 1_M\rangle = -f\,\lvert 1_M\rangle,
\]
with $f\neq 0$.
This ensures that a superposition
\[
\lvert \Psi_M\rangle = \alpha\,\lvert 0_M\rangle + \beta\,\lvert 1_M\rangle.
\]
is physically allowed and can serve as a messenger state in protocols such as quantum teleportation.

In real implementations, one would employ neutral mesons $(K^0,\bar K^0)$ or entangled atomic pairs rather than photons ($|\gamma\rangle, |\bar\gamma\rangle$), since the latter satisfy $|\gamma\rangle=|\bar\gamma\rangle$ and the superposition $\lvert \Psi_M\rangle$ is trivial.

\paragraph{(3) Global IQNs are Identical in $\{|0_M\rangle,|1_M\rangle\}$ of a Particle-Antiparticle Messenger.}

If the packaged messenger is built from a particle-antiparticle qubit, then $\lvert 0_M\rangle$ and $\lvert 1_M\rangle$ must have both zero net gauge charge
\[
\hat{Q}\lvert 0_M\rangle = 0,\quad \hat{Q}\lvert 1_M\rangle = 0,
\]
and zero net global quantum number
\[
\hat{F}\lvert 0_M\rangle = 0, \quad \hat{F}\lvert 1_M\rangle = 0.
\]
This is naturally accomplished if the particle and antiparticle are distinct.

\subsection{Packaged Resource States}
\label{SEC:PackagedResourceStates}

In quantum communication (such as those used in teleportation, superdense coding, or QKD), a resource state usually refers to an entangled state that is shared among parties.
It may serve as a quantum channel.
In this subsection, we extend the resource states to packaged entangled states.

\subsubsection{Definition of Packaged Resource States}

In a packaged subspace, a resource state's components must belong to the same net-charge sector and obey the superselection rules.
It must be a packaged entangled state and is gauge-invariant.

\begin{definition}[Packaged Resource States]
	Let $\lvert \Psi\rangle$ be a bipartite packaged state shared between Alice and Bob.
	If $\lvert \Psi\rangle$ satisfies the following conditions:
	\begin{enumerate}
		\item Single-Sector Condition:
		No term in $\lvert \Psi\rangle$ has a different net charge, i.e.,
		$
		\lvert \Psi\rangle \in \mathcal{H}_{Q}
		$
		for a fixed $Q$.
		This ensure that $\lvert \Psi\rangle$ obeys superselection and is therefore physical,
		
		\item Entanglement: 
		$\lvert \Psi\rangle$ is non-factorizable with respect to the bipartite partition, i.e., $\lvert \Psi\rangle$ cannot be written as $\lvert \phi_A\rangle\otimes\lvert \phi_B\rangle$),
		
		\item Gauge-Invariance: $\lvert \Psi\rangle$ is gauge-invariant, i.e.,
		$\forall ~ g\in G$, we have
		$
		U_g\,\lvert \Psi\rangle = e^{i\phi(g)}\,\lvert \Psi\rangle,
		$
	\end{enumerate}
	then we say that $\lvert \Psi\rangle$ is a \textbf{packaged resource state}.
\end{definition}

For application purpose, we construct the packaged resource state on top of abstract packaged qubit states as developed in Sec.~\ref{SEC:PackagedQubits}.

\subsubsection{Constructing Full Packaged Bell Basis as Resource States}

Let us now consider the packaged resource states of a particle-antiparticle pair.
In our framework, if one directly uses a bare particle state $|P\rangle$ and its antiparticle state $|\bar{P}\rangle$, only the two maximal packaged entangled states
\[
\lvert \Psi_P^{\pm} \rangle = \frac{1}{\sqrt{2}}\Bigl(\lvert P\,\bar{P} \rangle \pm \lvert \bar{P}\,P \rangle\Bigr)
\]
are available (see Sec.~\ref{SEC:TwoParticlePackagedEntangledStates}).
The reason is that in many cases the bare states may not belong to the same superselection sector if they carry different gauge charges.

To overcome the above limitations, we instead define logical packaged qubit basis states $\{|0_P\rangle,\,|1_P\rangle\}$ that both lie within the zero net-charge subspace $\mathcal{H}_{Q=0}$.
These packaged logical states are constructed, for example, by combining the bare states in such a way that the overall gauge charge vanishes.
Specifically, the four packaged product states $|0_P\rangle_A\otimes|1_P\rangle_B$, $|1_P\rangle_A\otimes|0_P\rangle_B$, $|0_P\rangle_A\otimes|0_P\rangle_B$, and $|1_P\rangle_A\otimes|1_P\rangle_B$ all remain within the same zero charge sector.
Thus, the superposition between $|0_P\rangle_A\otimes|1_P\rangle_B$ and $|1_P\rangle_A\otimes|0_P\rangle_B$, $|0_P\rangle_A\otimes|0_P\rangle_B$ and $|1_P\rangle_A\otimes|1_P\rangle_B$ are all allowed because each term has zero charge.

With such a packaged 2-qubit basis, we can now construct a full packaged Bell basis with four maximally packaged entangled states:
\begin{align}\label{EQ:FullPackagedBellBasis}
	\boxed{\begin{aligned}
		|\Psi_P^+\rangle_{AB} &= \frac {1}{\sqrt {2}} \Bigl(|0_P\rangle_A\otimes|1_P\rangle_B + |1_P\rangle_A\otimes|0_P\rangle_B\Bigr), \\
		|\Psi_P^-\rangle_{AB} &= \frac {1}{\sqrt {2}} \Bigl(|0_P\rangle_A\otimes|1_P\rangle_B - |1_P\rangle_A\otimes|0_P\rangle_B\Bigr), \\
		|\Phi_P^+\rangle_{AB} &= \frac {1}{\sqrt {2}} \Bigl(|0_P\rangle_A\otimes|0_P\rangle_B + |1_P\rangle_A\otimes|1_P\rangle_B\Bigr), \\
		|\Phi_P^-\rangle_{AB} &= \frac {1}{\sqrt {2}} \Bigl(|0_P\rangle_A\otimes|0_P\rangle_B - |1_P\rangle_A\otimes|1_P\rangle_B\Bigr).
	\end{aligned}}
\end{align}

As a result, by working with the packaged logical qubit basis rather than with the bare particle states, one gains access to the full four-dimensional Bell basis.
This is exactly what we need to adapt all quantum algorithms and communication protocols from conventional quantum information theory into the packaged framework.

One can easily verify that, under the charge operator $Q$,
\[
\hat{Q}\,\lvert \Psi^{\pm}_P\rangle_{A,B} = 0, \quad \hat{Q}\,\lvert \Phi^{\pm}_P\rangle_{A,B} = 0.
\]
and under a local gauge transformation $U_g$,
\[
U_g\,\lvert \Psi^{\pm}_P\rangle_{A,B} = e^{i\psi(g)}\,\lvert \Psi^{\pm}_P\rangle_{A,B},
\quad
U_g\,\lvert \Phi^{\pm}_P\rangle_{A,B} = e^{i\phi(g)}\,\lvert \Phi^{\pm}_P\rangle_{A,B}.
\]
Thus, $\lvert \Psi^{\pm}_P\rangle_{A,B}$ and $\lvert \Phi^{\pm}_P\rangle_{A,B}$ qualify as packaged resource states.

Using Eq.~(\ref{EQ:PackagedMessengerState}) and Eq.~(\ref{EQ:FullPackagedBellBasis}), we can bypass the superselection rules and translate the conventional quantum communication protocols into gauge-invariant packaged protocols.

\subsubsection{Multi-Party and Multi-Particle Resource States}

For quantum communication protocols that involves more than two parties (or multi-particle entangled states), the same principle applies.
For example, if we consider a multipartite state of the form
\[
\lvert \Psi\rangle = \sum_{n} \alpha_n \, \lvert \Theta_n\rangle_A\,\lvert \Theta'_n\rangle_B,
\]
where each term $\lvert \Theta_n\rangle_A\,\lvert \Theta'_n\rangle_B$ lies in the same net-charge sector, i.e.,
\[
\hat{Q}\,\lvert \Theta_n\rangle_A\,\lvert \Theta'_n\rangle_B = Q\, \lvert \Theta_n\rangle_A\,\lvert \Theta'_n\rangle_B,
\]
where $\hat Q = \hat Q_A+\hat Q_B$ is the total charge operator and each term lies in its $Q$-eigenspace.
Thus, the overall state is gauge-invariant.

\subsubsection{Distribution and Preservation}

Once a packaged resource state is generated, it must be distributed between spatially separated parties (Alice and Bob) without loss of gauge-invariance.

Let $V_{\rm channel}$ represent the unitary evolution describing the communication channel.
To preserve gauge-invariance, we require that
\[
[V_{\rm channel}, \hat{Q}] = 0.
\]
Then, if
\[
\lvert \Psi_{\rm in}\rangle \in \mathcal{H}_Q,
\]
the transmitted state is
\[
\lvert \Psi_{\rm out}\rangle = V_{\rm channel}\,\lvert \Psi_{\rm in}\rangle,
\]
and because $V_{\rm channel}$ commutes with $\hat{Q}$,
\[
\hat{Q}\,\lvert \Psi_{\rm out}\rangle = Q\,\lvert \Psi_{\rm out}\rangle.
\]
Thus, the state remains in the same superselection sector throughout transmission.

\subsubsection{Robustness Under Noise}

In a gauge-invariant framework, any physically allowed noise process is described by Kraus operators $\{E_k\}$, which must commute with the net-charge operator:
\[
[E_k, \hat{Q}] = 0.
\]
This constraint means that errors which would mix different gauge sectors are forbidden. As a result, the internal structure of packaged resource states cannot be altered by gauge-violating errors.
This means that the packaged resource states enjoy an inherent form of error protection.

\section{Lifting to (\texorpdfstring{$d \times D$})-Dimensional Hybrid-Packaged Space}
\label{SEC:dDDimensionalHybridPackagedSpace}

In Secs.~\ref{SEC:PackagedQubitsGatesAndCircuits} and \ref{SEC:PackagedStatesForQuantumCommunication}, we constructed packaged qubits, qubit gates, and qubit circuits, and packaged resource states.
All these objects are 2-dimensional and essentially associated with $\mathbb{Z}_2$ symmetry \cite{MaGroup2025}.

The essence of packaged states is that the internal quantum numbers (IQNs) are locked into an inseparable block.
As shown in Ref.~\cite{MaGroup2025}, there exist various packaged states associated with different finite or compact groups, such as $\mathbb{Z}_N$, $\mathrm{SU}(N)$, and $p$-form symmetries.
This means that packaged states are inherently high-dimensional, such as the SU(3) color states.
Therefore, it is natural to generalize the IQNs and packaged states into a high-dimensional space \cite{Rechtsman2013,Erhard2020}.

In this section, we describe the details of building a ($d \times D$)-dimensional hybrid-packaged subspace, and constructing ($d \times D$)-dimensional hybrid-packaged qudits, qudit gates, and qudit circuits.
We provide detailed mathematical derivations and physical interpretations.

\subsection{($d \times D$)-Dimensional Hybrid-Packaged Space}
\label{SEC:dDDimensionalHybridPackagedQudits}

In this subsection, we first construct a $d$-dimensional internal Hilbert space and then a $D$-dimensional external Hilbert space.
Finally, we take a tensor product to get the ($d \times D$)-dimensional hybrid-packaged subspace.

\subsubsection{Constructing $d$-Dimensional Packaged Hilbert Space}
\label{SEC:ConstructingdDimensionalPackagedHilbertSpace}

The internal quantum numbers (IQNs) are determined by physical constraints, such as gauge-invariance, which package all IQNs into an inseparable block.
We define the internal Hilbert space as
\begin{equation}\label{EQ:dDimensionalInternalHilbertSpace}
	\mathcal{H}_{\mathrm{int}}^{(d)} = \operatorname{span}\Bigl\{\,\lvert 0_P\rangle,\,\lvert 1_P\rangle,\,\dots,\,\lvert (d-1)_P\rangle\Bigr\}\,.
\end{equation}
whose basis is
\begin{equation}\label{EQ:dDimensionalInternalBasis}
	\mathcal{B}_{\mathrm{int}}^{(d)} = \{|0_P\rangle, |1_P\rangle, \dots, |(d-1)_P\rangle\} \subset \mathcal{H}_{Q}.
\end{equation}

Each basis vector in this space satisfies the packaging condition: an eigenstate of the net-charge operator has zero eigenvalue, i.e.,
\[
\hat{Q}\,\lvert j_P\rangle = 0,\quad \text{for } j=0,1,\dots,d-1\,.
\]

Under any local gauge transformation $U_g$, we have
\[
U_g\,\lvert j_P\rangle = e^{i\phi(g)}\,\lvert j_P\rangle\,,\quad \forall j,\quad \forall\, g\in G\,.
\]

This implies that, although the internal space is $d$-dimensional and carries rich logical structure (e.g., different flavor or color configurations), every state in $\mathcal{H}_{\mathrm{int}}^{(d)}$ remains neutral under the conserved gauge charge.
Their transformation only results in an overall phase factor.
Thus, all states reside in the same superselection sector (usually $\mathcal{H}_{Q=0}$).

\subsubsection{Constructing $D$-Dimensional External Hilbert Space}
\label{SEC:ConstructingDDimensionalExternalHilbertSpace}

The external degrees of freedom (DOFs) can come from various physical sources such as orbital angular momentum, polarization, or spatial modes.
We now assume that the external space is $D$-dimensional:
\begin{equation}\label{EQ:DDimensionalExternalHilbertSpace}
	\mathcal{H}_{\mathrm{ext}}^{(D)} = \operatorname{span}\Bigl\{\,\lvert 0_E\rangle,\,\lvert 1_E\rangle,\,\dots,\,\lvert (D-1)_E\rangle\Bigr\}\,.
\end{equation}
whose basis is
\begin{equation}\label{EQ:DDimensionalExternalBasis}
	\mathcal{B}_{\mathrm{ext}}^{(D)} = \{\lvert 0_E\rangle,\,\lvert 1_E\rangle,\,\dots,\,\lvert (D-1)_E\rangle\},
\end{equation}

The gauge transformation acts trivially (or in a fixed way) on the external DOFs:
\[
U_g\,\lvert k_E\rangle = \lvert k_E\rangle\,,\quad \forall\, k,\quad \forall\, g\in G\,.
\]

Because these degrees of freedom are not involved in carrying the gauge quantum numbers, they are free of the constraints imposed by the packaging principle.
In many experiments, such external spaces are naturally high‑dimensional (e.g., orbital angular momentum states can be very high-dimensional) and provide extra channels for encoding information.

\subsubsection{Constructing ($d \times D$)-Dimensional Hybrid-Packaged Hilbert Space}

The total hybrid-packaged Hilbert space is now
\begin{equation}\label{EQ:dDHybridPackagedHilbertSpace}
	\mathcal{H}_{\text{hyb}}^{(d \times D)} = \mathcal{H}_{\mathrm{int}}^{(d)} \otimes \mathcal{H}_{\mathrm{ext}}^{(D)}\,,
\end{equation}
which has a dimension
$
\dim \left(\mathcal{H}_{\text{hyb}}^{(d \times D)}\right) = d \times D\,.
$

Let $G$ be the local gauge group, $\hat{Q}$ be the total conserved charge 
(e.g.\ electric charge or color), and 
$\{\hat G_x\}_{x\in\Lambda}$ be the generators of local Gauss-law constraints, where $\Lambda$ is the spatial coordinates. 
If $|\Psi\rangle$ is a physical state, then it must satisfy
\begin{equation}
	\hat{Q}\,|\Psi\rangle = Q\,|\Psi\rangle,
	\qquad 
	\hat G_x\,|\Psi\rangle = 0
	\quad\forall x\in\Lambda.
	\label{EQ:GaussLawConstraints}
\end{equation}

We note that it is more natural to define a packaged subspace than to define a packaged space due to the following reasons:
(1) Mathematically, in a charge sector (a subspace) $\mathcal{H}_Q$, all superselection constraints are automatically satisfied and we never invoke unphysical superpositions across $Q$. 
(2) Physically, experiments always prepare states inside a definite charge sector. The packaged subspace tells us exactly which states and operations are allowed. 
(3) When we speak of gauge‑conserving errors, the commutator condition $[E_k,\hat{Q}]=0$ is interpreted relative to the fixed sector.
Leakage outside $\mathcal{H}_{\text{hyb}}^{(d \times D)}$ is naturally classified as gauge‑violating.

\paragraph{(1) Definition of Hybrid-Packaged Space.}

Due to all these concerns, let us first give a formal definition of a hybrid-packaged subspace:
\begin{definition}[Hybrid-Packaged Subspace]\label{DEF:HybridPackagedSubspace}
	Let $G$ be a gauge group and $\mathcal H_Q$ be a charge (super‑selection) sector of $d$-dimension.
	\begin{enumerate}
		\item Let 
		\[
		\mathcal H_{\text{int}}^{(d)}
		=\text{span}\{\;|0_P\rangle,\dots,|(d-1)_P\rangle\}
		\subseteq \mathcal H_Q,
		\quad
		\hat{Q}|j_P\rangle = Q\,|j_P\rangle ,
		\] 
		be a $d$-dimensional internal Hilbert subspace (a pure-packaged subspace) on which every $U_g\in G$ acts by a common phase. 
		
		\item Let 
		\[\mathcal H_{\text{ext}}^{(D)}=\text{span}\{|0_E\rangle,\dots,|(D-1)_E\rangle\}
		\] 
		be an $D$-dimensional external Hilbert subspace on which $G$ acts trivially.
	\end{enumerate} 
	Then we say that
	\[
	\mathcal{H}_{\text{hyb}}^{(d \times D)} :=
	\mathcal H_{\text{int}}^{(d)}
	\otimes
	\mathcal H_{\text{ext}}^{(D)}
	\]
	is a \textbf{$(d\times D)$-dimensional hybrid-packaged subspace} (or simply a \textbf{packaged subspace}) in the sector $Q$.
\end{definition}

Referring to the total Hilbert space decomposition, we give definition to global hybrid-packaged space:

\begin{definition}[Global Hybrid-Packaged space]
	Let $G$ be a gauge group.
	The direct sum (total Hilbert space)
	\[
	\mathcal{H}_{\text{hyb}} :=
	\bigoplus_{Q\in\text{Spec}\,\hat{Q}} 
	\mathcal{H}_{\text{hyb}}^{(Q)}
	\]
	is block diagonally invariant under $G$. 
	We call $\mathcal{H}_{\text{hyb}}$ the \textbf{global hybrid-packaged space}.
\end{definition}

We now define a projection operator that project the global hybrid-packaged space into a particular hybrid-packaged subspace:

\begin{definition}[Projection operator]
	A \textbf{projection operator} $\Pi_{\text{hyb}}^{(Q)}$ is an orthogonal projector onto 
	$\mathcal{H}_{\text{hyb}}^{(Q)}$, i.e.,
	\begin{equation}
		\Pi_{\text{hyb}}^{(Q)}
		:=\!
		\prod_{x\in\Lambda}\!\bigl(\tfrac12(\mathbf1+\hat G_x)\bigr)
		\;
		\left(\delta_{\hat{Q},Q}\right),
		\qquad
		\left(\Pi_{\text{hyb}}^{(Q)}\right)^2 = \Pi_{\text{hyb}}^{(Q)},
	\end{equation}
	where $\delta_{\hat Q,Q}$ is the projector onto the $\hat Q=Q$ eigenspace.
\end{definition}

In other words, a projection operator is a \textbf{packaging map}
$
\mathcal P: \mathcal{H}_{\text{hyb}} \to \mathcal{H}_{\text{hyb}}^{(Q)}
$
transforms as: $\mathcal P: |\phi\rangle\mapsto\Pi_{\text{hyb}}^{(Q)}|\phi\rangle$.
Here $\prod_x(\tfrac12(1+G_x))$ enforces the local Gauss law at every site, while $\delta_{\hat Q,Q}$ projects onto the fixed global‐charge sector.
Together they package all IQNs into a single neutral block.

\paragraph{(2) Two-index Basis of Hybrid-Packaged Space.}

We construct the \textbf{two-index hybrid-packaged basis} by taking the tensor product of the $d$-dimensional internal basis Eq.~(\ref{EQ:dDimensionalInternalBasis}) and the $D$-dimensional external basis Eq.~(\ref{EQ:DDimensionalExternalBasis}), i.e.,
\begin{equation}\label{EQ:dDHybridPackagedBasis}
	\mathcal{B}_{\text{hyb}}^{(d \times D)} = \{\,|j_P\rangle \otimes |k_E\rangle \,:\, j=0,1,\dots,d-1,\; k=0,1,\dots,D-1 \}.
\end{equation}
This is a natural orthonormal basis of the hybrid-packaged subspace 
$
\mathcal H_{\text{hyb}}^{(Q)}
:=\mathcal H_{\text{int}}^{(d)}\otimes\mathcal H_{\text{ext}}^{(D)}
$.

Thus, a general normalized state in $\mathcal{H}_{\text{hyb}}^{(d \times D)}$ can be expanded as
\begin{equation}\label{EQ:GeneralNormalizedHybridPackagedState}
	|\Psi\rangle = \sum_{j=0}^{d-1} \sum_{k=0}^{D-1} c_{jk} \; |j_P\rangle\otimes|k_E\rangle\,,\quad \text{with } \sum_{j,k} |c_{jk}|^2 =1\,.
\end{equation}

\paragraph{(3) Single-Index Basis (Computational (Z) Basis) of Hybrid-Packaged Space.}
The two-index basis $\lvert j_P\rangle\otimes\lvert k_E\rangle$ (or $\ket{j,k}$) is a natural choice.
However, it is often convenient to re-index the two-index into a single-index
\begin{equation}\label{EQ:SingleIndex}
	J:=jD+k,
	\quad 
	j=0,\dots,d-1,
	\quad
	k=0,\dots,D-1,
	\quad
	J=0,\dots,N-1,
\end{equation}
where $N = dD$.
Then we define a single-index basis vector as
\begin{equation}\label{EQ:dDSingleIndexHybridPackagedBasisVector}
	\lvert J\rangle
	:=\lvert j_P\rangle\otimes\lvert k_E\rangle
	\in\mathcal H_{\text{hyb}}^{(Q)},
\end{equation}
and the \textbf{single-index hybrid-packaged basis} as
\begin{equation}\label{EQ:dDSingleIndexHybridPackagedBasis}
	\mathcal{B}_{\text{hyb}}^{(N)} = \{\,|J\rangle \,:\, J = 0, 1, \dots, N-1 \},
\end{equation}
so that the hybrid-packaged subspace is isomorphic to $\mathbb{C}^{dD}$.
Using this notation, later we will see that the generalized Bell basis, Fourier transforms, and other protocols carry over by replacing $N\mapsto dD$.

The set $\{\lvert J\rangle\}_{J=0}^{N-1}$ is an orthonormal basis (the
hybrid Z‑basis). 
Because each factor in Eq.~\eqref{EQ:dDSingleIndexHybridPackagedBasisVector} satisfies
$\hat{Q}\,\lvert j_P\rangle=0$ and $\hat{Q}\,\lvert k_E\rangle=0$, we
have
\[
\hat{Q}\,\lvert J\rangle=0,\qquad
U_g\,\lvert J\rangle=e^{i\phi(g)}\lvert J\rangle,
\]
so every logical basis state already lives in the desired sector
$\mathcal H_{Q=0}$.

\subsubsection{Properties of ($d \times D$)-Dimensional Hybrid-Packaged Space}

Let us now prove two basic properties of the hybrid-packaged Hilbert space.

\begin{property}
	A major difference between $\mathcal H_{\text{int}}^{(d)}$ and $\mathcal H_{\text{ext}}^{(D)}$ is that a local gauge transformation $U_g$ only acts on $\mathcal H_{\text{int}}^{(d)}$, but does not act on $\mathcal H_{\text{ext}}^{(D)}$.
\end{property}

\begin{proof}
	Using Eq.~(\ref{EQ:dDHybridPackagedHilbertSpace}),
	\[
	\mathcal{H}_{\text{hyb}}^{(d \times D)} = \mathcal{H}_{\mathrm{int}}^{(d)} \otimes \mathcal{H}_{\mathrm{ext}}^{(D)},
	\]
	where $\mathcal H_{\text{int}}^{(d)}$ is the internal packaged space that carries a (finite-dimensional) unitary representation
	$\rho : G \to \mathrm V\!\bigl(\mathcal H_{\text{int}}^{(d)}\bigr)$
	of the local gauge group $G$, whereas $\mathcal H_{\text{ext}}^{(D)}$ is the external space that is gauge-blind (it realises only the trivial representation).

	By definition, a local gauge transformation $U_g$ (with $g\in G$) acts as
	\[
	U_g \;=\; \rho(g)\;\otimes\; \mathbf 1_{\text{ext}} ,
	\]
	where $\mathbf 1_{\text{ext}}$ denotes the identity on $\mathcal H_{\text{ext}}^{(D)}$.
	For every factorised basis vector $|\alpha\rangle_{\text{int}}\otimes |x\rangle_{\text{ext}} \in \mathcal{H}_{\text{hyb}}$, we have
	\[
	U_g \bigl(|\alpha\rangle_{\text{int}}\otimes |x\rangle_{\text{ext}}\bigr)
	\;=\;
	\bigl(\rho(g)|\alpha\rangle_{\text{int}}\bigr)
	\,\otimes\,
	|x\rangle_{\text{ext}}.
	\]
	This shows that $U_g$ only acts on the internal component and leaves the external component untouched.

	Since any state $|\psi\rangle_{\text{hyb}} \in \mathcal{H}_{\text{hyb}}$ can be written as a superposition of such product states $|\alpha\rangle_{\text{int}}\otimes |x\rangle_{\text{ext}}$ and $U_g$ is linear operator, the statement extends to the whole space $\mathcal{H}_{\text{hyb}}$.
\end{proof}

The internal part comes from the packaging of IQNs and ensures that the state remains in $\mathcal{H}_{Q=0}$.
The external part is now generalized to dimension $D$ and increases the effective alphabet (or computational space) for information encoding.
Combining these two parts together, we can encode a hybrid-packaged state in $d \times D$ levels, which is naturally protected from gauge-violating noise.

\begin{corollary}[Gauge‑invariance of Hybrid-packaged Subspace]
	\label{cor:HybridGaugeInvariance}
	The $(d\times D)$-dimensional hybrid-packaged subspace
	given in Definition \ref{DEF:HybridPackagedSubspace}, i.e.,
	\[
	\mathcal H_{\text{\emph{hyb}}}^{(d\times D)}
	=\mathcal H_{\mathrm{int}}^{(d)}
	\otimes\mathcal H_{\mathrm{ext}}^{(D)},
	\]	
	is gauge-invariant.
\end{corollary}

\begin{proof}
	The hybrid-packaged subspace $\mathcal H_{\text{\emph{hyb}}}^{(d\times D)}$ is a tensor product of $\mathcal H_{\mathrm{int}}^{(d)}$ and $\mathcal H_{\mathrm{ext}}^{(D)}$.
	By definition, $\mathcal H_{\mathrm{int}}^{(d)}$ is a pure-packaged subspace.
	$\mathcal H_{\mathrm{ext}}^{(D)}$ is an external subspace on which $G$ acts trivially and we can also treat it as a ``packaged subspace''.
	According to Lemma~\ref{LEM:TensorProductPackaged}, the tensor product $\mathcal H_{\text{\emph{hyb}}}^{(d\times D)}$ is a packaged subspace and is therefore gauge-invariant.
\end{proof}

\subsection{Hybrid-Packaged Qudits}
\label{SEC:HybridPackagedQudits}

In many applications (see Sec.~\ref{SEC:QuantumComputationInPackagedSpace} and \ref{SEC:QuantumCommunicationAndCryptography}), we need to encode additional degrees of freedom (DOFs) that are not subject to the gauge constraints.
For example, one may combine the internal DOFs with an external DOF (e.g., spin, momentum, or polarization) to form a higher-dimensional quantum system, or qudit.

In this subsection, we explore how to build the $(d \times D)$‑dimensional hybrid-packaged qudits.
We require that every basis vector lies in the neutral charge sector.

\subsubsection{Warm-up Example: (2 $\times$ 2)-dimensional Hybrid-Packaged Qudits}

In the simplest case, by taking the tensor product of a packaged qubit (2-dimensional) with an external qubit (also 2-dimensional), one obtains a hybrid system of dimension 4.

\begin{example}[(2 $\times$ 2)-dimensional Hybrid-Packaged Qudit]\label{EXM:4DHybridPackagedQudit}
	Let
	$
	\mathcal{H}_{\rm int}^{(2)} \subset \mathcal{H}_Q
	$
	be the packaged internal Hilbert space with
	\[
	\mathcal{H}_{\rm int}^{(2)} = \operatorname{span}\{\lvert 0_P\rangle, \lvert 1_P\rangle\},
	\]
	and let $\mathcal{H}_{\rm ext}^{(2)}$ be an external Hilbert space with a two-dimensional orthonormal basis $\{\lvert 0_E\rangle, \lvert 1_E\rangle\}$, i.e.,
	\[
	\mathcal{H}_{\rm ext}^{(2)} = \operatorname{span}\{\lvert 0_E\rangle, \lvert 1_E\rangle\}.
	\]

	Then the tensor product
	\[
	\mathcal{H}_{\text{hyb}}^{(2 \times 2)} = \mathcal{H}_{\rm int}^{(2)} \otimes \mathcal{H}_{\rm ext}^{(2)}.
	\]
	is a (2 $\times$ 2)-dimensional hybrid-packaged qudit.

	A natural computational basis for the (2 $\times$ 2)-dimensional $\mathcal{H}_{\text{hyb}}$ is given by
	\[
	\Bigl\{ 
	\lvert 0_P\rangle \otimes \lvert 0_E\rangle,\; 
	\lvert 0_P\rangle \otimes \lvert 1_E\rangle,\; 
	\lvert 1_P\rangle \otimes \lvert 0_E\rangle,\; 
	\lvert 1_P\rangle \otimes \lvert 1_E\rangle \Bigr\}.
	\]
	For brevity, we denote these four basis states as $\lvert 0\rangle, \lvert 1\rangle, \lvert 2\rangle, \lvert 3\rangle$.
	Therefore, $\mathcal{H}_{\text{hyb}} \cong \mathbb{C}^4$.
\end{example}

Hybrid packaged qudits extend the packaged qubit concept by incorporating external degrees of freedom.
They offer a higher-dimensional Hilbert space for encoding quantum information while preserving gauge-invariance.

\subsubsection{($d \times D$)-dimensional Hybrid-Packaged Qudits.}
\label{SEC:dDDimensionalHybridPackagedQudits}

After warming up by (2 $\times$ 2)-dimensional hybrid-packaged qudits, let us now generalize it to ($d \times D$)-dimensional hybrid-packaged qudits:

\begin{definition}[($d \times D$)-dimensional Hybrid-Packaged Qudit]
	\label{DEF:dDDimensionalHybridPackagedQudit}
	Let $\mathcal{H}_{\rm int}^{(d)} \subset \mathcal{H}_Q^{(d)}$
	be a packaged internal Hilbert space with a $d$-dimensional orthonormal basis, i.e.,
	\[
	\mathcal{H}_{\mathrm{int}}^{(d)} = \operatorname{span}\Bigl\{\,\lvert 0_P\rangle,\,\lvert 1_P\rangle,\,\dots,\,\lvert (d-1)_P\rangle\Bigr\}\,,
	\]
	and let $\mathcal{H}_{\rm ext}$ be an external Hilbert space with a $D$-dimensional orthonormal basis, i.e.,
	\[
	\mathcal{H}_{\mathrm{ext}}^{(D)} = \operatorname{span}\Bigl\{\,\lvert 0_E\rangle,\,\lvert 1_E\rangle,\,\dots,\,\lvert (D-1)_E\rangle\Bigr\}\,.
	\]
	Then the tensor product $\mathcal{H}_{\text{hyb}}^{(d \times D)} = \mathcal{H}_{\rm int}^{(d)} \otimes \mathcal{H}_{\rm ext}^{(D)}$ is a hybrid-packaged Hilbert space with a ($d \times D$)-dimensional orthonormal basis, i.e.,
	\[
	\mathcal{H}_{\text{hyb}}^{(d \times D)} = \operatorname{span} \bigl\{\ket{j,k}\equiv\ket{j_P}\otimes\ket{k_E}\;\big|\;j=0,\dots,d-1,\;k=0,\dots,D-1\bigr\}\,.
	\]
	We say that the tensor product $\mathcal{H}_{\text{hyb}}^{(d \times D)}$ is a \textbf{($d \times D$)-dimensional hybrid-packaged qudit}.
\end{definition}

In Definition \ref{DEF:dDDimensionalHybridPackagedQudit}, we used two-index basis Eq.~(\ref{EQ:dDHybridPackagedBasis}).
The basis vectors satisfy
\[
\hat{Q}\,\ket{j,k}=0,
\quad
U_g\,\ket{j,k}=e^{i\phi(g)}\,\ket{j,k}.
\]
Thus, each $\ket{j,k}$ lies in the zero net-charge sector $\mathcal H_{Q=0}$.

A normalized state in the hybrid-packaged subspace can be written as
\[
\ket{\Psi}
=\sum_{j=0}^{d-1}\sum_{k=0}^{D-1}c_{jk}\,\ket{j,k},
\qquad
\sum_{j,k}|c_{jk}|^2=1.
\]
which is linear combination of basis vectors that are zero net-charge vectors.
This guarantee
$\hat{Q}\ket{\Psi}=0$ and $U_g\ket{\Psi}=e^{i\phi(g)}\ket{\Psi}$.

We can also used single-index basis Eq.~(\ref{EQ:dDSingleIndexHybridPackagedBasis}).
The basis vectors satisfy
\[
\hat{Q}\,\lvert J\rangle=0,\qquad
U_g\,\lvert J\rangle=e^{i\phi(g)}\lvert J\rangle,
\]
Thus, every logical basis state lives in the desired sector
$\mathcal H_{Q=0}$.

Similarly, we have
\[
\ket{\Psi}
=\sum_{J=0}^{N-1} c_{J} \,\ket{J},
\qquad
\sum_{J=0}^{N-1} |c_J|^2=1,
\]
$\hat{Q}\ket{\Psi}=0$, and $U_g\ket{\Psi}=e^{i\phi(g)}\ket{\Psi}$.

\subsection{Hybrid-Packaged Bell Basis (Resource States)}
\label{SEC:HybridPackagedBellBasis}

For high-dimensional quantum communication, we need a complete orthonormal basis of maximally entangled states in $\mathcal{H}_{\text{hyb}}$.
In subsection \ref{SEC:dDDimensionalHybridPackagedQudits}, we have developed ($d \times D$)-dimensional hybrid-packaged qudits.
Based on these ideas, we now generalize the Bell basis to ($d \times D$)-dimensional hybrid-packaged subspace.

\paragraph{(1) Two-index Bell Basis.}

Referring to Eq.~(\ref{EQ:dDHybridPackagedBasis}), we introduce two phase labels $\mu\in\{0,\dots ,d-1\}$ and $\nu\in\{0,\dots ,D-1\}$ and two shift labels 
$n_{\text{int}}\in\{0,\dots ,d-1\}$, $n_{\text{ext}}\in\{0,\dots ,D-1\}$. 
Define
\begin{equation}\label{EQ:dDDimensionalBellBasis_CORR}
	\;|\Phi_{\mu,\nu;\,n_{\text{int}},n_{\text{ext}}}\rangle_{AB}
	=\frac{1}{\sqrt{dD}}
	\sum_{j=0}^{d-1}\sum_{k=0}^{D-1}
	\omega_{d}^{\,\mu j}\,\omega_{D}^{\,\nu k}\;
	|j,k\rangle_A\;\otimes\;
	|j\oplus n_{\text{int}},\,k\oplus n_{\text{ext}}\rangle_B,
\end{equation}
where $\omega_{d}=e^{2\pi i/d}$, $\omega_{D}=e^{2\pi i/D}$.
Here $j\oplus n_{\rm int}$ means $(j+n_{\rm int})\bmod d$ and $k\oplus n_{\rm ext}$ means $(k+n_{\rm ext})\bmod D$.

It is straightforward to verify that Eq.~(\ref{EQ:dDDimensionalBellBasis_CORR}) satisfy the following orthonormality:
\[
\langle\Phi_{\mu,\nu;n_{\text{int}},n_{\text{ext}}}|
\Phi_{\mu',\nu';n_{\text{int}}',n_{\text{ext}}'}\rangle
=\delta_{\mu,\mu'}\,\delta_{\nu,\nu'}\,
\delta_{n_{\text{int}},n_{\text{int}}'}\,
\delta_{n_{\text{ext}},n_{\text{ext}}'} ,
\]
so the set in \eqref{EQ:dDDimensionalBellBasis_CORR} are $d^{2}D^{2}\!=\!(dD)^{2}$ orthonormal, maximally‑entangled states.
This is exactly the number required for a complete Bell basis in the $(dD)$-dimensional bipartite space.

Reversing Eq.~(\ref{EQ:dDDimensionalBellBasis_CORR}), we obtain the following inverse (reverse) identity:
\begin{equation}\label{EQ:ReversedDDimensionalBellBasis_CORR}
	|j,k\rangle_{A}\otimes
	|j\oplus n_{\text{int}},\,k\oplus n_{\text{ext}}\rangle_{B}
	=
	\frac{1}{\sqrt{dD}}
	\sum_{\mu=0}^{d-1}\sum_{\nu=0}^{D-1}
	\omega_{d}^{-\mu j}\,\omega_{D}^{-\nu k}\;
	|\Phi_{\mu,\nu;\,n_{\text{int}},n_{\text{ext}}}\rangle_{AB}.
\end{equation}
Eq.~(\ref{EQ:ReversedDDimensionalBellBasis_CORR}) shows that, for any fixed shifts $n_{\text{int}},n_{\text{ext}}$, the computational product state factorises into the Bell basis.

\paragraph{(2) Single-index Bell Basis.}

Referring to Eq.~(\ref{EQ:dDSingleIndexHybridPackagedBasis}), we re-write Eq.~(\ref{EQ:dDDimensionalBellBasis_CORR}) as 
\begin{equation}\label{EQ:NDimensionalBellBasis_CORR}
	|\Phi_{m,n}\rangle
	=\frac{1}{\sqrt{N}}\sum_{J=0}^{N-1}
	\omega_N^{mJ}\,|J\rangle_{A}\otimes|J\oplus n\rangle_{B},
	\qquad m,n=0,\dots ,N-1,
\end{equation}
where $N = dD$, $\omega_N=e^{2\pi i/N}$, and $J \oplus n$ means $(J + n) \mod N$.
This is the familiar qudit form (see Sec.\ref{SEC:dDDimensionalHybridPackagedQudits}).

Since each basis state is a superposition of computational basis states that themselves lie in $\mathcal{H}_{Q=0}$, we have
\[
\hat{Q}\,|\Phi_{m,n}\rangle = 0,\quad U_g\,|\Phi_{m,n}\rangle = e^{i\phi(g)}|\Phi_{m,n}\rangle\,.
\]

The inverse identity of Eq.~(\ref{EQ:NDimensionalBellBasis_CORR}) is
\begin{equation}\label{EQ:ReverseNDimensionalBellBasis_CORR}
	|J\rangle_{A}\otimes|J\oplus n\rangle_{B}
	=
	\frac{1}{\sqrt{N}}
	\sum_{m=0}^{N-1}\omega_N^{-mJ}\,
	|\Phi_{m,n}\rangle_{AB}.
\end{equation}

Both representations live entirely in the gauge‑neutral sector 
($\hat{Q} |\Phi\rangle =0$) and transform only by an overall phase under any gauge operation $U_g$.
Therefore, the whole Bell family is gauge-invariant.

\subsection{Hybrid-Packaged Qudit Gates}
\label{SEC:HybridPackagedQuditGates}

Corresponding to the hybrid-packaged qudit given in Definition \ref{DEF:dDDimensionalHybridPackagedQudit}, let us now discuss hybrid-packaged qudit gates.

\subsubsection{Definition of Hybrid-Packaged Qudit Gates}

\begin{definition}[Hybrid-Packaged Qudit Gate]\label{DEF:HybridPackagedQuditGates}
	Let 
	$
	\mathcal{H}_{\text{hyb}}^{(d \times D)} = \mathcal{H}_{\rm int}^{(d)} \otimes \mathcal{H}_{\rm ext}^{(D)}
	$
	be a single-particle hybrid-packaged subspace and
	$
	\hat{Q}=\hat{Q}_{\text{int}}\otimes\mathbf 1_{\text{ext}}
	$	
	be the total (conserved) charge.
	Let		
	$$
	V: \mathcal H_{\text{hyb}}^{(d\times D)}\longrightarrow
	\mathcal H_{\text{hyb}}^{(d\times D)}
	$$
	be a unitary operator.
	If $V$ preserves the gauge-invariance, i.e.,
	\begin{equation}\label{EQ:CommuteWithQ}
		[V, \hat{Q}] = 0,
	\end{equation}
	then we say that $V$ is a \textbf{hybrid-packaged qudit gate}.
\end{definition}

Eq.~(\ref{EQ:CommuteWithQ}) means that $V$ maps every state in the physical (net-zero) sector $\mathcal H_{Q=0}$ back into the same sector.
Hence gauge super-selection is never violated.

\begin{example}
	Suppose $ V $ can be written as
	$
	V = V_{\rm int} \otimes V_{\rm ext},
	$
	where $ V_{\rm int} $ acts on $\mathcal{H}_{\rm int}$ and satisfies
	\[
	[V_{\rm int}, \hat{Q}] = 0,
	\]
	and $ V_{\rm ext} $ is any unitary operator defined on $\mathcal{H}_{\rm ext}$ (which is automatically gauge-invariant), then
	\[
	[V, \hat{Q}] = [V_{\rm int} \otimes V_{\rm ext}, \hat{Q}\otimes I] = [V_{\rm int}, \hat{Q}] \otimes V_{\rm ext} = 0.
	\]
	This shows that the hybrid operation preserves the packaged structure.
\end{example}

Generally, a unitary operator $V$ acting on $\mathcal{H}_{\text{hyb}}$ must satisfy
$
V^\dagger V = V V^\dagger = \mathbb{I}.
$
It must also be compatible with the superselection constraint, i.e.,
$
V\,\mathcal{H}_Q \subset \mathcal{H}_Q
$
(or equivalently $[V, \hat{Q}] = 0$).

\begin{remark}
	In Definition \ref{DEF:HybridPackagedQuditGates}, we require that a hybrid-packaged qudit gate to be gauge-invariant.
	In other words, we only pick out those gauge-invariant unitary operator as our hybrid-packaged qudit gate.
	This means that the gauge-invariance of a packaged gate is given by definition.
	On the contrary, the gauge-invariance of a packaged state is a property of the state, but not given by definition.
\end{remark}

\subsubsection{Elementary Hybrid-Packaged Gates}
\label{SEC:ElementaryHybridPackagedGates}

In later sections, we need shift- and phase- type unitary operations that act within $\mathcal H_{\text{hyb}}$ and commute with $\hat{Q}$.
Let us now split them into different groups.

\paragraph{(1) Clifford Single-Qudit Gates.}

\begin{enumerate}
	\item Internal Weyl block 
	
	\begin{equation}\label{EQ:InternalWeylBlock}
		X_d=\sum_{j=0}^{d-1}\lvert (j \oplus 1)_P\rangle \!\langle j_P|,
		\quad
		Z_d=\sum_{j=0}^{d-1}\omega_d^{\,j}\lvert j_P\rangle\!\langle j_P|,
		\quad
		H_d =\frac1{\sqrt d}\sum_{j,k=0}^{d-1}\omega_d^{jk}\,|j_P\rangle\langle k_P|,
	\end{equation}
	where $\omega_d=e^{2\pi i/d}$ and $j \oplus 1$ means $(j+1)\bmod d$.
	
	They satisfy
	\begin{equation}\label{EQ:InternalHadamardSwaps}
		X_d Z_d \;=\;\omega_d\,Z_d X_d,
		\quad
		H_d\,X_{d}\,H_d^{\dagger}=Z_{d},
		\quad
		H_d\,Z_{d}\,H_d^{\dagger}=X_{d}^{\!\dagger}.
	\end{equation}
	
	These act on $\mathcal H_{\mathrm{int}}^{(d)}$ as
	\[
	X_d\ket{j_P}=\ket{(j \oplus 1)_P}, 
	\quad
	Z_d\ket{j_P}=\omega_d^{\,j}\ket{j_P},
	\quad
	H_d \ket{k_P} = \frac1{\sqrt d}\sum_{j=0}^{d-1}\omega_d^{jk}\,|j_P\rangle,
	\]
	but act trivially on $\mathcal H_{\mathrm{ext}}^{(D)}$.
	
	\item External Weyl block
	
	\begin{equation}\label{EQ:ExternalWeylBlock}
		X_D=\sum_{k=0}^{D-1}\lvert (k \oplus 1)_E\rangle\!\langle k_E|,
		\quad
		Z_D=\sum_{k=0}^{D-1}\omega_D^{\,k}\lvert k_E\rangle\!\langle k_E|,
		\quad
		H_D =\frac1{\sqrt D}\sum_{j,k=0}^{D-1}\omega_D^{jk}\,|j_E\rangle\langle k_E|,
	\end{equation}
	where $\omega_D=e^{2\pi i/D}$ and $k \oplus 1$ means $(k+1)\bmod D$.
	
	They satisfy
	\begin{equation}\label{EQ:ExternalHadamardSwaps}
		X_D Z_D \;=\;\omega_D \,Z_D X_D,
		\quad
		H_D\,X_{D}\,H_D^{\dagger}=Z_{D},
		\quad
		H_D\,Z_{D}\,H_D^{\dagger}=X_{D}^{\!\dagger}.
	\end{equation}
	
	These act on $\mathcal H_{\mathrm{ext}}^{(D)}$ as
	\[
	X_D\ket{k_E}=\ket{ (k \oplus 1)_E}, 
	\quad
	Z_D\ket{k_E}=\omega_D^{\,k} \ket{k_E},
	\quad
	H_D \ket{k_E} =\frac1{\sqrt D}\sum_{j=0}^{D-1}\omega_D^{jk}\,|j_E\rangle,
	\]
	but act trivially on $\mathcal H_{\mathrm{int}}^{(d)}$.

	\item Hybrid Weyl block
	
	Label the hybrid computational basis by a single-index
	\[
	J \;=\; j\,D + k,
	\qquad j=0,\dots,d-1,\;k=0,\dots,D-1,
	\quad
	N=dD,
	\]
	so that $\lvert J\rangle=\lvert j_P\rangle\otimes\lvert k_E\rangle$.
	Therefore, we have the $N$-dimensional hybrid Weyl block
	\begin{equation}\label{EQ:HybridWeylBlock}
		X_{N} \;=\;\sum_{J=0}^{N-1}\lvert J\oplus 1\rangle\!\langle J\rvert,
		\quad
		Z_{N} \;=\;\sum_{J=0}^{N-1}\omega_{N}^{\,J}\,\lvert J\rangle\!\langle J\rvert,
		\quad
		H_{N} \;=\;\frac{1}{\sqrt{N}} \sum_{J,K=0}^{N-1} \omega_{N}^{\,J K}\; \lvert J\rangle\!\langle K\rvert,
	\end{equation}
	where $\omega_{N}=e^{2\pi i/N}$ and $J \oplus 1$ means $(J+1)\bmod N$.
	We immediately have relations
	\begin{equation}\label{EQ:CommuteHybridWeylblock}
		X_{N}Z_{N} \;=\;\omega_{N}\,Z_{N}X_{N}.
	\end{equation}
	Since each $\lvert J\rangle$ lies in the neutral sector $\mathcal H_{Q=0}$, both $X_N$ and $Z_N$ commute with $\hat{Q}$.
	
	1. Unitarity
	
	\[
	H_{N}\,H_{N}^{\dagger}
	=\;H_{N}^{\dagger}\,H_{N}
	=\;\mathbb{I}.
	\]
	
	2. Conjugation relations 
	
	\[
	H_{N}\,X_{N}\,H_{N}^{\dagger}
	\;=\;Z_{N},
	\qquad
	H_{N}\,Z_{N}\,H_{N}^{\dagger}
	\;=\;X_{N}^{\!\dagger}.
	\]
	
	3. Twirling property
	 
	Using Eq.~(\ref{EQ:CommuteHybridWeylblock}), we have
	\[
	H_{N}\,X_{N}\,Z_{N}\,H_{N}^{\dagger}
	=\;H_{N}\,(\omega_{N}Z_{N}X_{N})\,H_{N}^{\dagger}
	=\;\omega_{N}\,Z_{N}\,X_{N}
	\;=\;
	\omega_{N}\,X_{N}\,Z_{N}.
	\]
	
	4. Square gives parity
	 
	Define the parity operator
	\[
	P=\sum_{J=0}^{N-1}\lvert -J\,(\bmod\,N)\rangle\langle J\rvert,
	\]
	then
	\[
	H_{N}^{2} = P,
	\qquad
	H_{N}^{4} = P^{2} = \mathbb{I}.
	\] 
\end{enumerate}

\paragraph{(2) Clifford Two-Qudit Gates.}

Hybrid packaged entanglers 

\begin{enumerate}
	\item Internal-internal SUM
	\begin{equation}\label{EQ:InternalInternalSUM}
		\text{CSUM}_{d}
		=\sum_{j=0}^{d-1}\ket{j_P}\!\bra{j_P}\otimes X_{d}^{\,j}.
	\end{equation}	
	
	\item Internal-to-external controlled phase
	\begin{equation}\label{EQ:InternalToExternalControlledPhase}
		\text{C}\Phi_{d,D}
		=\sum_{k=0}^{D-1}\ket{k_E}\!\bra{k_E}\otimes Z_{d}^{\,k}.
	\end{equation}		
\end{enumerate}
or in a single-index:
\begin{equation}\label{EQ:HybridCSUM}
	\text{CSUM}_{N}
	=\sum_{J=0}^{N-1}\ket{J}\!\bra{J}\otimes X_{N}^{\,J}.
\end{equation}

Both gates simply permute basis states inside the $Q=0$ manifold and thus satisfy Eq.~(\ref{EQ:CommuteWithQ}).

\paragraph{(3) Non-Clifford Single-Qudit Gates.}

We define a non-Clifford diagonal phase
\begin{equation}\label{EQ:NonCliffordDiagonalPhase}
	\Theta_r=\sum_{J=0}^{N-1}\exp\!\Bigl(\tfrac{2\pi i r}{N^{2}}J^{2}\Bigr)\ket{J}\!\bra{J},
	\quad
	r \in \mathbb Z, \; \gcd(r,N)=1, \; r \notin \{1,2,4\}.
\end{equation}

$\Theta_r$ is diagonal in the single-index computational basis.
Each $\lvert m\rangle$ obeys $\hat{Q}\,\lvert m\rangle=0$, we have $[\Theta_r,\hat{Q}^{(i)}]=0$.
Thus, $\Theta_r$ is gauge-invariant.

$\Theta_r$ is used to kick out of the Clifford group and makes the whole set dense.

\subsubsection{Finite Universal Hybrid-Packaged Gate Sets}
\label{SEC:AFiniteUniversalGateLibrary}

A finite universal gate library is not unique.
It is only required to
\begin{itemize}
	\item contain finite many, exactly implementable gates;
	
	\item lie in the commutant
	$
	\mathcal C_{\hat{Q}}\;=\;\bigl\{V \in \mathrm U(N)\;\bigl|\;[V,\hat{Q}]=0\bigr\}
	$;
	
	\item generate a group that is dense in $\mathrm{SU}\bigl(\mathcal H_{Q=0}\bigr)$. 
\end{itemize}

Based on these requirements, we list three primary finite universal gate libraries:

\paragraph{(1) Clifford-Two-Index + $\Theta_r$.}
Most naturally, we consider a gate set consists of separate index:
inter index Eq.~(\ref{EQ:InternalWeylBlock}),
external index Eq.~(\ref{EQ:ExternalWeylBlock}),
internal-external entanglers Eq.~(\ref{EQ:InternalInternalSUM}) and Eq.~(\ref{EQ:InternalToExternalControlledPhase}),
and non-Clifford gate Eq.~(\ref{EQ:NonCliffordDiagonalPhase}):

\begin{equation}\label{EQ:FiniteUniversalGateLibrary1}
	\mathcal{G}_{\text{sep}} \;=\;
	\bigl\{ \,X_{d}, Z_{d}, H_{d}, X_{D}, Z_{D}, H_{D},
	\text{CSUM}_{d}, \text{C}\Phi_{d,D}, \Theta_r \bigr\},
\end{equation}
All members of $\mathcal{G}_{sep}$ commute with $\hat{Q}$.
Let us check the universality of $\mathcal{G}_{sep}$:

\begin{itemize}
	\item Internal qudit control $\{X_d,Z_d,H_d\}$ already topologically generates SU(d) inside the neutral sector. 
	
	\item External qudit control $\{X_D,Z_D,H_D\}$ does the same on $\mathscr{H}_{ext}$. 
	
	\item The entangler between two subsystems ($\text{CSUM}_{d}$ or $\text{C}\Phi_{d,D}$) lets you reach the full Lie product space $\mathfrak{su}(d) \otimes \mathfrak{su}(D)$. 
	
	\item The non-Clifford single-qudit gate ($\Theta_r$) kicks you out of the Clifford group and makes the whole set dense.
\end{itemize}

Thus $\mathcal{G}_{sep}$ is universal, but it is not minimal.
The detail proof is given in next subsection.

\paragraph{(2) Clifford-Single-Index + $\Theta_r$.}

Let us now consider working exclusively in the single-index computational basis $\{|J\rangle\}_{J=0}^{N-1}$ of dimension $N=dD$:
hybrid Weyl block Eq.~(\ref{EQ:HybridWeylBlock}),
one single-index entangler Eq.~(\ref{EQ:HybridCSUM}) (or any hybrid controlled-phase that entangles two different qudits),
and one non-Clifford diagonal phase Eq.~(\ref{EQ:NonCliffordDiagonalPhase}):
\begin{equation}\label{EQ:FiniteUniversalGateLibrary2}
	\mathcal G_{\text{single}}
	=\{\,X_N, Z_N, H_N, \text{CSUM}_{N}, \Theta_r\}
\end{equation}
Here the hybrid Weyl block $\{X_N, Z_N, H_N\}$ already give the full Clifford group on each logical qudit,
a single two-qudit $\text{CSUM}_{N}$ makes that Clifford action entangling,
adjoining one non-Clifford phase $\Theta_r$ makes the entire set exactly the classical ``Clifford + T'' construction for qudits.
This is known to be universal and compilation-friendly.
Because every $|J\rangle$ lives in $\mathcal H_{Q=0}$, all five members in $\mathcal G_{single}$ commute with $\hat{Q}$.

The proof of universality of $\mathcal G_{\text{single}}$ is similar to that of conventional ``Clifford + T''.
We put it in appendix \ref{APD:CliffordSingleIndexThetar}.

\paragraph{(3) Clifford-Single-Index + Magic-State Injection.}

Finally, let us consider drop the non-Clifford phase $\Theta_r$ from the gate set Eq.~(\ref{EQ:FiniteUniversalGateLibrary2}) and later inject it via magic states.
Then we obtain a pure single-index Clifford hybrid-packaged gate Set:
\begin{equation}\label{EQ:FiniteUniversalGateLibrary3}
	\mathcal G_{\text{Cl}}
	=\{X_N, Z_N, H_N, \text{CSUM}_{N}\}
\end{equation}
This is a strictly Clifford set, but it is not universal by itself.
It needs magic-state distillation to restore universality at the cost of ancillas and classical feed-forward.
Then we should say ``$\mathcal G_{\text{Cl}}$ is dense after allowing magic-state injection''.

\subsubsection{Universality of a Gauge-Respecting Hybrid-Packaged Gate Set}

With the above preparation, we now prove an important theorem about the universality of hybrid-packaged gate set.

\begin{theorem}[Universality of a Gauge-Respecting Hybrid-Packaged Gate Set]
	\label{THM:UniversalGaugeRespectingGateSet}
	Consider a hybrid-packaged Hilbert space 
	\[
	\mathcal H_{\text{hyb}}^{(d \times D)}\;=\;
	\mathcal H_{\rm int}^{(d)}\;\otimes\;
	\mathcal H_{\rm ext}^{(D)}.
	\]
	Let $\hat{Q}_{\rm tot}$ be the total charge operator whose
	zero‑eigenspace is the physical subspace
	$\mathcal H_{Q=0}\subset\mathcal H_{\text{hyb}}^{(d \times D)}$.
	The group generated by $\mathcal G$ is dense in
	$\mathrm{SU}\bigl(\mathcal H_{Q=0}\bigr)$.
\end{theorem}

\begin{proof}
To show that a gate set $\mathcal G$ is universal (dense), we may equivalently show that the real Lie algebra generated by
$\log\mathcal G$ is the full $\mathfrak{su}(N)$ restricted to
$\mathcal H_{Q=0}$.
Connectedness then implies density by Lie’s third theorem.

We decompose
\begin{equation}\label{EQ:suNDecomposition}
	\mathfrak{su}(N)\;=\;
	\Bigl(
	\underbrace{\mathfrak{su}(d)\otimes\mathbf 1_D}_{\text{internal block}}
	\Bigr)
	\;\;\oplus\;\;
	\Bigl(
	\underbrace{\mathbf 1_d\otimes\mathfrak{su}(D)}_{\text{external block}}
	\Bigr)
	\;\;\oplus\;\;
	\Bigl(
	\underbrace{\mathfrak{su}(d)\otimes\mathfrak{su}(D)}_{\text{cross block}}
	\Bigr).
\end{equation}

We must generate all three pieces with commutators of logarithms of elements of $\mathcal G$.

\begin{itemize}
	\item Step 1: Internal block $\mathfrak{su}(d)\otimes\mathbf 1_D$
	
	From Eq.~(\ref{EQ:InternalWeylBlock}), we have set 
	\[
	\mathcal A_d=\{X_{d},Z_{d},H_{d}\}.
	\]
	
	Using Eq.~(\ref{EQ:InternalHadamardSwaps}), the adjoint action of $H_{d}$ swaps $X_{d}$ and $Z_{d}$, i.e.,
	\[
	H_{d}X_{d}H_{d}^{\dagger}=Z_{d},\quad
	H_{d}Z_{d}H_{d}^{\dagger}=X_{d}^{\dagger}.
	\]
	
	Consequently, 
	$\operatorname{ad}_{\mathcal A_d}$-closure yields every pair of
	traceless Gell-Mann matrices
	$E_{ab}=\ket{a_P}\!\bra{b_P}\;(a\ne b)$ and
	$F_{aa}-F_{bb}$, hence the full $\mathfrak{su}(d)$.
	Tensoring with $\mathbf 1_D$ embeds that block in
	$\mathfrak{su}(N)$.

	\item Step 2: External block $\mathbf 1_d\otimes\mathfrak{su}(D)$
	
	For the external block, using Eq.~(\ref{EQ:ExternalWeylBlock}) and Eq.~(\ref{EQ:ExternalHadamardSwaps}), the same argument with 
	$\mathcal A_D=\{X_{D},Z_{D},H_{D}\}$
	gives via adjoint action the whole
	$\mathbf 1_d\otimes\mathfrak{su}(D)$.

	\item Step 3: Cross block $\mathfrak{su}(d) \otimes \mathfrak{su}(D)$
	
	Using Eq.~(\ref{EQ:InternalToExternalControlledPhase}), take the hybrid controlled phase	
	\[
	\Omega:=\sum_{k=0}^{D-1}\ket{k_E}\!\bra{k_E}\otimes Z_{d}^{\,k}
	=\sum_{k}Z_{d}^{\,k}\otimes\ket{k_E}\!\bra{k_E}.
	\]
	
	One off-diagonal: apply external shift $X_{D}^{\ell}$	
	\[
	X_{D}^{\ell}\,\Omega\,X_{D}^{-\ell}
	=\sum_{k}Z_{d}^{\,k}\otimes\ket{k\!\oplus\!\ell_E}\!\bra{k\!\oplus\!\ell_E}.
	\]
	
	Sweep over internal: conjugation of $\Omega$ by internal $X_{d}$ maps
	$Z_{d}^{\,k}\mapsto Z_{d}^{\,k}X_{d}$ etc.
	Re-expandg this in the Weyl basis and subtract traceless parts, we get
	\[
	E_{ab}\otimes \ket{k_E}\!\bra{k_E}
	\quad\text{for all }a \ne b.
	\]

	Commutators with the external block from Step 2 \cite{Brennen2005arxiv,Brennen2005PRA}, we obtain
	$E_{ab}\otimes F_{pq}$ for all $(a,b,p,q)$. 
	Those span precisely
	$\mathfrak{su}(d) \otimes \mathfrak{su}(D)$.
	
	Thus, the Lie algebra generated so far is (Eq.~(\ref{EQ:suNDecomposition}))
	\[
	\mathfrak h
	=
	\bigl(\mathfrak{su}(d)\otimes\mathbf 1_D\bigr)
	\;\oplus\;
	\bigl(\mathbf 1_d\otimes\mathfrak{su}(D)\bigr)
	\;\oplus\;
	\bigl(\mathfrak{su}(d)\otimes\mathfrak{su}(D)\bigr)
	=
	\mathfrak{su}(N).
	\]
	
	Since every generator commutes with $\hat{Q}$, by restricting to the zero‐charge sector, we can keep the whole algebra.

	\item Step 4: Connectedness versus density 
	
	The gates used so far (all Clifford-type) already exponentiate to a connected Lie subgroup $G_{\text{Cl}}\cong\!\mathrm{SU}(N)\cap$ Clifford. 
	Being connected, $G_{\text{Cl}}$ is automatically dense in its own Lie group, but one still needs a generator outside the Clifford group to avoid being stuck in a proper sub-lattice of phase factors.
	
	We add a single non-Clifford one-qudit phase
	\[
	\Theta_r
	=\operatorname{diag}\!\bigl(1,e^{2\pi i/r},\dots\bigr),
	\qquad r\notin\{1,2,4\}.
	\]
	Because $\Theta_r\in\mathcal C_{\hat{Q}}$ it acts within the physical
	sector.	
	Clifford + single non-Clifford on one qudit is already dense in
	$\mathrm{SU}(N)$.
	The argument was first used by Boykin et al. for qubits \cite{Boykin1999} and generalised to qudits in, e.g., Brennen et al. 2005 \cite{Brennen2005arxiv,Brennen2005PRA}.
	
	Hence
	\[
	\overline{\langle\mathcal G\rangle}
	=
	\mathrm{SU}\!\bigl(\mathcal H_{Q=0}\bigr),
	\]
	which proves the Theorem \ref{THM:UniversalGaugeRespectingGateSet}.
\end{itemize}

\begin{remark}
	Physical interpretation 
	
	Steps 1-3 demonstrate that packaged symmetry does not restrict
	computational power:
	every traceless observable compatible with the gauge constraint is reachable by nested commutators of physically admissible Hamiltonians.
	
	Step 4 guarantees algorithmic universality:
	by supplementing the normally finite Clifford subgroup with a single irrational phase
	$\Theta_r$ we can approximate any desired neutral-sector unitary to
	arbitrary precision (Solovay-Kitaev then bounds the circuit length).
	
	Thus, even though all gates commute with $\hat{Q}$, they collectively
	span the entire dynamical group of the physical Hilbert space.
\end{remark}
\end{proof}

\begin{corollary}[Solovay-Kitaev Theorem in Hybrid-Packaged Space]
	\label{COR:SolovayKitaevTheoremInHybridPackagedSpace}
	Let
	$\mathcal H_{Q=0}\cong\mathbb C^{\,N}$ (with $N=dD$)
	be a single neutral hybrid-packaged qudit,
	$
	\mathcal C_{\hat{Q}}\;=\;\bigl\{V\in\mathrm U(N)\;\bigl|\;[V,\hat{Q}]=0\bigr\}
	$
	be the gauge-invariant commutant,
	and $\mathcal G\subset\mathcal C_{\hat{Q}}$ be the finite, inverse-closed gate
	library defined by Eq.~(\ref{EQ:FiniteUniversalGateLibrary1}) and Theorem \ref{THM:UniversalGaugeRespectingGateSet}.
	Then for every integer $n \ge 1$, every unitary 
	$V\in\mathrm{SU}\!\bigl((\mathcal H_{Q=0})^{\!\otimes n}\bigr)$ 
	and every accuracy $\varepsilon\in(0,1]$, there exists a hybrid-packaged gate	
	\[
	\widetilde V\;=\;G_{i_L}\cdots G_{i_2}G_{i_1},\qquad
	G_{i_\ell}\in \mathcal G
	\]	
	such that	
	\[
	\bigl\|\,V-\widetilde V\,\bigr\|_{\!\text{op}}\;\le\;\varepsilon ,
	\]	
	and whose length $L$ obeys the Solovay-Kitaev bound	
	\[
	L\;=\;O\!\bigl(\,\log^{\,\kappa}\!\varepsilon^{-1}\bigr),
	\qquad
	\kappa\lesssim 4.
	\]
	Every intermediate hybrid-packaged gate $G_{i_\ell}\!\cdots G_{i_1}$ also lies in
	$\mathcal C_{\hat{Q}}^{\otimes n}$.
\end{corollary}

\begin{proof}
	The argument follows the classical Solovay-Kitaev (SK) construction \cite{Kitaev1997,DawsonNielsen2005} with one additional observation:
	the commutant of $\hat{Q}$ is closed under commutators and group inversion. 
	Hence every algebraic step used by SK remains in $\mathcal C_{\hat{Q}}$.
	
	\begin{enumerate}
		\item Preliminaries: 
		
		Let $G := \langle\mathcal G\rangle$ be the group generated by $\mathcal G$. 
		By Theorem \ref{THM:UniversalGaugeRespectingGateSet}, we have 
		\[
		\overline{G}\;=\;\mathrm{SU} \! \bigl(\mathcal H_{Q=0}\bigr),
		\]
		where the closure is taken in the operator norm. 
		Also $\mathcal G$ is balanced: together with each $G_j$, it contains $G_j^{\!\dagger}$.
		
		Let $G_0 \subset G$ be a finite $\delta_0$-net of $\mathrm{SU}(N)$ with 
		$\delta_0<1/4$. 
		Standard volume arguments give $|G_0|=O(\delta_0^{-N^2})$. 
		Because $G$ is dense, we can produce $G_0$ by brute‑force search to depth $t_0=O(\log\!1/\delta_0)$, i.e.,
		\[
		G_0
		=\bigl\{W_1,\dots ,W_M\bigr\} \subset G
		\]
		such that every element of $\mathrm{SU}(N)$ is within $\delta_0$
		(in operator norm) of some $W_j$. 
		All $W_j$ commute with $\hat{Q}$ because $\mathcal C_{\hat{Q}}$ is a group.

		\item Closure of $\mathcal C_{\hat{Q}}$ under SK primitives:
		
		The SK recursion only uses two primitives:	
		group inversion $G\mapsto G^{\!\dagger}$ and group commutator 
		$\,{\rm Comm}(A,B)=ABA^{\!\dagger}B^{\!\dagger}$.
		
		If $A,B\in\mathcal C_{\hat{Q}}$ then 
		\[
		[A,\hat{Q}]=[B,\hat{Q}]=0
		\;\;\Longrightarrow\;\;
		[A^{\!\dagger},\hat{Q}]=0,\quad
		[ABA^{\!\dagger}B^{\!\dagger},\hat{Q}]=0.
		\]
		Therefore, every word produced during the SK algorithm is gauge‑invariant.

		\item Recursive construction: 
		
		Let $\text{SK}_k(V)$ be the approximation after $k$ recursion levels.
		Let $\delta_k=\|\,V-\text{SK}_k(V)\,\|_{\rm op}$. 
		The usual SK estimate \cite{DawsonNielsen2005} shows	
		\[
		\delta_{k+1}\;<\;c\,\delta_k^{\,1+\alpha},\qquad
		c\!=\!O(1),\; \alpha>0 .
		\]
		
		Starting from $\delta_0<1/4$ one obtains	
		\[
		\delta_k\;\le\;\left(\tfrac14\right)^{(1+\alpha)^{k}}
		\;<\;
		\varepsilon
		\quad\text{once}\;\;
		k=O(\log\log 1/\varepsilon).
		\]

		\item Length estimate:
		
		Let $L_k$ be the length of the word $\text{SK}_k(V)$. 
		Recursion gives	
		\[
		L_{k+1}\le 5\,L_k + O(1),
		\]
		because each commutator needs four calls to $\text{SK}_k$ plus constant overhead for net look‑ups. 
		With $L_0=O(\log\!1/\delta_0)$ this solves to	
		\[
		L_k\;=\;O\!\bigl(5^{k}\bigr).
		\]
		
		Insert $k$ from (4) and set $\kappa=\log_2 5/(1+\alpha)\lesssim 4$ to obtain	
		\[
		L\;=\;L_k\;=\;O\!\bigl(\log^{\,\kappa}\!\varepsilon^{-1}\bigr).
		\]
		
		The constant in the big‑$O$ depends polynomially on $N=dD$ (through the size of the first‑level net) exactly as in the ordinary SK theorem.

		\item Extension to $n$ hybrid-packaged qudits:
		
		Because $\mathcal G$ acts on a fixed local dimension $N$ and is closed under tensoring with identity, the SK algorithm applies independently to each qudit or to any finite collection from that. 
		For $n$ qudits, one works inside $\mathrm{SU}(N^{n})$.\cite{Kliuchnikov2016}
		All bounds hold with $N \mapsto N^{n}$, which gives the same asymptotic length (up to a polynomial in $n$).
	\end{enumerate}
\end{proof}

\subsubsection{Fault-Tolerant Encoded Hybrid-Packaged Gates} 
\label{SEC:FTEncodedHybridPackagedGates}

We now promote every elementary hybrid-packaged gate introduced in Secs. \ref{SEC:ElementaryHybridPackagedGates}-\ref{SEC:AFiniteUniversalGateLibrary} to a fault-tolerant logical operation on top of a stabiliser (or surface-type) code whose physical qudits are themselves hybrid packaged qudits. 
Throughout, we fix the physical Hilbert space
\[
\mathscr H_{\rm phys}
=\bigl(\mathcal H_{\text{hyb}}^{(Q=0)}\bigr)^{\!\otimes n},
\qquad 
\hat{Q}_{\rm tot}=\sum_{i=1}^{n}\hat{Q}^{(i)},
\]
and a stabiliser group 
$
\mathcal S=\langle S_1,\dots,S_r\rangle\subset \mathcal C_{\hat{Q}_{\rm tot}}
$
generated by commuting, $\hat{Q}_{\rm tot}$-respecting Pauli-Weyl operators. 
The code subspace is 
\[
\mathscr C=\bigl\{\;|\psi\rangle\in\mathscr H_{\rm phys}\;:\;S_a|\psi\rangle=|\psi\rangle,\;a=1,\dots,r\bigr\},
\qquad 
k=\log_{N}\!\dim\mathscr C .
\]

We focus on 
(i) universal Clifford gates, 
(ii) the non-Clifford diagonal phase $\Theta_r$, and 
(iii) error filtering provided by the code and gauge super-selection.

\paragraph{(1) Clifford layer: transversal or locality-preserving.}

\begin{enumerate}
	\item Purely internal block $\;X_d,Z_d,H_d,\mathrm{CSUM}_d$
	
	Each physical qudit carries a tensor factor 
	$
	\mathcal H_{\rm int}^{(d)}=\operatorname{Span}\{|j_P\rangle\}_{j=0}^{d-1}
	$.
	Define
	\[
	\bar X_d \;=\;\bigotimes_{i\in\Gamma_X}\!X_d^{(i)},
	\qquad
	\bar Z_d \;=\;\bigotimes_{i\in\Gamma_Z}\!Z_d^{(i)},
	\]
	where $\Gamma_X,\Gamma_Z\subseteq\{1,\dots,n\}$ are any two homologically non-trivial, disjoint strings on the underlying lattice. Because every $X_d^{(i)}$ (or $Z_d^{(i)}$) commutes with $\hat{Q}^{(i)}$, both operators commute with $\hat{Q}_{\rm tot}$ and with every stabiliser $S_a$.
	Thus
	\[
	\overline{X}_d,\;\overline{Z}_d\;:\;\mathscr C \longrightarrow \mathscr C
	\quad\text{are logical Pauli generators.}
	\]
	
	They are transversal (weight-1 on each code block), hence propagate no error strings.

	\item Purely external block $\;X_D,Z_D,H_D$
	
	The external tensor factor 
	$
	\mathcal H_{\rm ext}^{(D)}=\operatorname{Span}\{|k_E\rangle\}_{k=0}^{D-1}
	$
	lives, e.g., on the faces of a surface-code lattice. 
	Define logical strings in complete analogy:
	\[
	\bar X_D \;=\;\bigotimes_{f\in\Gamma'_X}\!X_D^{(f)},
	\qquad
	\bar Z_D \;=\;\bigotimes_{f\in\Gamma'_Z}\!Z_D^{(f)}.
	\]
	
	Because the external DOF is gauge-inert, the same commutation arguments hold; $\bar X_D,\bar Z_D$ are transversal and fault-tolerant.

	\item Internal-internal entangler 
	$
	\mathrm{CSUM}_d=\sum_{j}\lvert j\rangle\!\langle j|\!\otimes X_d^{\,j}
	$
	
	Place control and target on neighbouring physical qudits within every code block. 
	The operator preserves the stabiliser group (conjugates Pauli’s into Pauli’s) and commutes with $\hat{Q}_{\rm tot}$.
	Consequently, the logical entangler
	\[
	\overline{\mathrm{CSUM}}_d
	\;=\;
	\bigotimes_{\langle i,t\rangle}
	\mathrm{CSUM}_d^{(i\to t)}
	\]
	is locality-preserving, which spreads a single-qudit fault to at most one qudit per block and leaves $\mathscr C$ invariant.

	\item Internal $\to$ external controlled phase 
	$
	\mathrm{C}\Phi_{d,D}
	=\sum_{k}\lvert k_E\rangle\!\langle k_E|\otimes Z_d^{\,k}
	$
	
	Choose control (external) and target (internal) on the same block. 
	Conjugation inside the stabiliser group and commutation with $\hat{Q}_{\rm tot}$ once more guarantee fault-tolerance.
\end{enumerate}

Now we see that the set
\[
\mathcal C_{\rm Cl}
=\bigl\{
\overline{X}_d,\overline{Z}_d,\overline{H}_d,\;
\overline{X}_D,\overline{Z}_D,\overline{H}_D,\;
\overline{\mathrm{CSUM}}_d,\;\overline{\mathrm{C}\Phi}_{d,D}
\bigr\}
\subset
\mathcal N_{\mathcal S}\cap\mathcal C_{\hat{Q}_{\rm tot}}
\]
is the logical Clifford group.
Every element is either transversal or locality-preserving of radius 1, so single-qudit faults remain correctable by the code.

\paragraph{(2) Non-Clifford resource: the diagonal phase $\Theta_r$.}

A single logical non-Clifford is enough for universality
(Theorem \ref{THM:UniversalGaugeRespectingGateSet}). 
We choose
\[
\Theta_r:=\sum_{m=0}^{N-1}\!\exp\!\Bigl(\tfrac{2\pi i r}{N^{2}}\,m^{2}\Bigr)\lvert m\rangle\!\langle m|,
\quad
\gcd(r,N)=1,\;\;N=dD.
\]

Because $\Theta_r$ is diagonal in the computational basis and each $\lvert m\rangle$ obeys $\hat{Q}\,\lvert m\rangle=0$, we have $[\Theta_r,\hat{Q}^{(i)}]=0$.
Thus, $\Theta_r$ is gauge-invariant and admissible on every physical qudit.

Here we outline the steps for encoded implementation of $\Theta_r$ via magic-state injection:
\begin{enumerate}	
	\item Prepare $k$ physical blocks in the state 
	$\lvert +\rangle^{\!\otimes n}=(H_{N})^{\otimes n}\lvert 0\rangle^{\otimes n}$.
	
	\item Apply $\Theta_r$ transversally only on qudits whose stabiliser
	syndrome equals $+1$.
	Discard when faults are detected. 
	Finally, we obtain the result:
	\[
	\lvert M_r\rangle
	=\Theta_r^{\otimes n}\,\lvert +\rangle^{\otimes n}
	\;\in\;\mathscr H_{\rm phys}.
	\]
	
	\item Detect error: measure all stabilisers $S_a$.
	If any $-1$ outcome appears, restart.
	Otherwise the post-selected $\rho_{M_r}$ has error rate $p_\text{in}^{\,2}$ (first-order noise is filtered by both the code and the $\hat{Q}$-superselection).
	
	\item Gate-teleportation: consume one $\lvert M_r\rangle$ plus Clifford
	operations to enact the logical $\overline{\Theta}_r$ on any code block \cite{Bennett1993,Anwar2012}.
\end{enumerate}

After $\ell{=}\!O(k)$ rounds of Bravyi-Haah distillation \cite{Anwar2012,Bullock2005}),
the logical error drops as $\eta\simeq (p_\text{in})^{\,2^\ell}$
while the physical cost scales poly-logarithmically in $\eta^{-1}$.

\paragraph{(3) Joint protection: code + gauge super-selection.}

Any physical error $E$ decomposes into gauge-respecting term $E_{\parallel}$ and gauge-violating term $E_{\perp}$:
\[
E=E_{\parallel}+E_{\perp},
\qquad
[E_{\parallel},\hat{Q}_{\rm tot}]=0,
\quad
\{E_{\perp},\hat{Q}_{\rm tot}\}\neq0.
\]

Gauge-respecting term $E_{\parallel}$ is a linear combination of Weyl errors $Z_N^{s}X_N^{t}$. 
These are exactly the Pauli-type errors corrected by the code $\mathscr C$.
However, gauge-violating term $E_{\perp}$ maps the state outside $\mathscr H_{\rm phys}$. 
Because no subsequent (gauge-respecting) gate can bring it back, such faults are deterministically detected at the next syndrome extraction.
They appear as a violation of Gauss' law and are discarded.

Thus, the effective logical noise rate is
\[
p_{\rm eff}=p_{\parallel}\;P_{\rm fail}(d_{\rm code}) 
\quad\text{with}\;\;
p_{\parallel}\le p_{\rm phys}, 
\quad
p_{\perp}=p_{\rm phys}-p_{\parallel},
\]
where $P_{\rm fail}(d_{\rm code})$ is the usual code-specific logical-error suppression ($d_{\rm code}$ = distance). 
Because $p_{\perp}$ is projected out before it reaches the code, the threshold $p_{\rm th}$ is strictly higher than in an unconstrained architecture.

From above discussion, we see that
all Clifford-type hybrid packaged gates are either transversal or locality-preserving of radius $\le 1$.
They automatically commute with the total charge and with every stabiliser, so they are fault-tolerant.
The single non-Clifford gate $\Theta_r$ is introduced via gauge-respecting magic-state injection. The standard distillation overheads apply and now with additional first-order error filtering from
super-selection.
Error channels that violate the packaging‐induced Gauss constraints are detected and removed prior to decoding.
The residual errors are the familiar Pauli-Weyl faults handled by the stabiliser code.
Therefore, the encoded gate set 
\[
\bigl\{
\overline{X}_d,\overline{Z}_d,
\overline{X}_D,\overline{Z}_D,
\overline{H}_d,\overline{H}_D,
\overline{\mathrm{CSUM}}_d,\overline{\mathrm{C}\Phi}_{d,D},
\overline{\Theta}_r
\bigr\}
\]
is universal, fault-tolerant, and gauge-invariant.

\subsubsection{Mutually Unbiased Bases (MUBs)}
\label{SEC:LogicalAndMUBs}

In quantum-key-distribution (QKD) protocols like six‑state QKD, one need three mutually unbiased bases (MUBs).
Let us now produce them inside the $(d\times D)$‑dimensional hybrid-packaged subspace $\mathcal H_{\text{hyb}}\cong\mathbb C^{N}$ ($N=dD$).

\paragraph{(1) Canonical MUB Triplet ($Z_N,X_N,Z_NX_N$).}

Let $\{\ket{J}\}_{J=0}^{N-1}$ be the single-index computational basis introduced in Eq.~(\ref{EQ:dDSingleIndexHybridPackagedBasisVector}).
Define the hybrid Weyl block as in Eq.~(\ref{EQ:HybridWeylBlock}).
Their eigenbases are
\begin{align}\label{EQ:MUB1}
	\begin{aligned}
		\mathcal B_Z &=\bigl\{\ket{n}\bigr\}_{n=0}^{N-1},\\[2mm]
		\mathcal B_X &=\Bigl\{\,
		\ket{\widetilde n}
		:=\tfrac1{\sqrt N}\!\sum_{m=0}^{N-1}\omega_N^{\,nm}\ket{m}\Bigr\}_{n=0}^{N-1},\\[2mm]
		\mathcal B_{XZ} &=\Bigl\{\,
		\ket{\widetilde{\widetilde n}}
		:= \tfrac1{\sqrt N}\!\sum_{m=0}^{N-1}
		\omega_N^{\,\tfrac12 m(m+1)-nm}\,\ket{m}\Bigr\}_{n=0}^{N-1}.
	\end{aligned}
\end{align}
Here
$\mathcal B_Z$ diagonalises $Z_N$ by definition,
$\mathcal B_X$ is obtained from $\mathcal B_Z$ by the discrete Fourier transform $H_N$ and therefore it diagonalises $X_N$,
and $\mathcal B_{XZ}$ diagonalises $X_NZ_N$.
The quadratic phase $\tfrac12 m(m+1)$ is the solution of the eigen-value equation in $\mathbb Z_N$ because 
$
(X_NZ_N)\ket{\widetilde{\widetilde n}}
=\omega_N^{-n}\ket{\widetilde{\widetilde n}}.
$
(When $N$ is even, replace $m(m+1)/2$ by $\frac{m(m+N)}2$ so that the factor $1/2$ is well-defined modulo $N$.)

The triplet in Eq.~\eqref{EQ:MUB1} is mutual unbiased because for any two different bases $\mathcal B_{*},\mathcal B_{*'}$, one has $|\langle\psi|\phi\rangle|^{2}=1/N$. 
This generalize the six-state protocol to arbitrary dimension $N=dD$.

The triplet in Eq.~\eqref{EQ:MUB1} is gauge-invariant because every $\ket{J}$ lies in the neutral sector $\mathcal H_{Q=0}$.
The superpositions in Eq.~\eqref{EQ:MUB1} therefore also satisfy $\hat{Q}\ket{\psi}=0$ and are merely multiplied by a global phase under a gauge transformation $U_g$. 
Thus the whole triplet is compatible with the packaging principle.

\paragraph{(2) Extending the Set of MUBs.}

If $N$ is a prime power ($N=p^{\,t}$) dimension, then a complete set of $N+1$ MUBs exists and can be generated by the finite‑field construction.

However, if $N$ is not a prime power (composite) dimension, then a full set is generally unknown. 
A practical solution is to use the tensor product structure
$\mathcal H_{\text{hyb}}=\mathcal H_{\mathrm{int}}^{(d)}\otimes\mathcal H_{\mathrm{ext}}^{(D)}$. 
Let 
$\{\mathcal B_{\text{int}}^{(i)}\}_{i=1}^{d+1}$ be a complete MUB set on
$\mathcal H_{\mathrm{int}}^{(d)}$ (available when $d$ is a prime power), and
$\{\mathcal B_{\text{ext}}^{(j)}\}_{j=1}^{D+1}$ any MUB family on
$\mathcal H_{\mathrm{ext}}^{(D)}$.
Then the product bases
\[
\mathcal B^{(i,j)} :=\mathcal B_{\text{int}}^{(i)}\;\otimes\;\mathcal B_{\text{ext}}^{(j)}
\quad
i=1,\dots,d+1,\; j=1,\dots,D+1,
\]
are easy to prepare and measure.
They remain mutually unbiased with respect to the canonical triplet $(\mathcal B_Z,\mathcal B_X,\mathcal B_{XZ})$ and already suffice to force an eavesdropper’s optimal guess probability down to the information-theoretic limit $1/N$.

\subsection{Hybrid-Packaged Qudit Circuits}
\label{SEC:HybridPackagedQuantumCircuits}

We have constructed hybrid-packaged qudits and hybrid-packaged qudit gates.
Now it is time to construct hybrid-packaged qudit circuit architecture in the neutral sector of the $(d\times D)$-dimensional hybrid Hilbert space 
$\mathcal H_{\text{hyb}}=\mathcal H_{\rm int}^{(d)}\otimes\mathcal
H_{\rm ext}^{(D)}$.
Throughout we fix a total charge operator
$\hat{Q}_{\rm tot}=\sum_{i}\hat{Q}^{(i)}$ whose zero‑eigenspace is the
only physically allowed sector.

\subsubsection{Definition of Hybrid-Packaged Qudit Circuits}

In Definition \ref{DEF:PackagedQubitCircuit}, we defined packaged qubit circuits.
Let us now lift it to ($d \times D$)-dimensional hybrid Hilbert space.

\begin{definition}[Hybrid-Packaged Qudit Circuit]\label{DEF:HybridPackagedQuditCircuit}
	Let
	$
	V_{\rm circuit} : \mathcal{H}_{\rm qudit}^{\otimes n} \to \mathcal{H}_{\rm qudit}^{\otimes n}
	$
	be a unitary operator that is constructed as a finite sequential composition of unitary operators $V_j$, i.e.,
	\[
	V_{\rm circuit} = V_k\,V_{k-1}\,\cdots\,V_2\,V_1.
	\]
	If each $ V_j $ is a hybrid-packaged qudit gate that acts on the logical hybrid-packaged subspace $\mathcal{H}_{\rm qudit}$ and satisfies
	\[
	[V_j, \hat{Q}] = 0,
	\]
	then we say that $ V_{\rm circuit} $ is a \textbf{hybrid-packaged qudit circuit}.
\end{definition}

In other words, a hybrid-packaged qudit circuit is a circuit composed entirely of hybrid-packaged qudit gates that preserve gauge-invariance and remain confined to the same superselection sector $\mathcal{H}_Q$.

\subsubsection{Gauge-Invariance of Hybrid-Packaged Qudit Circuits}

When defining the hybrid-packaged qudit gate, we require the gate to be gauge-invariant.
But when defining the hybrid-packaged qudit circuit, we do not explicitly require the circuit to be gauge-invariant, but only require each unitary operators $V_j$ to be gauge-invariant (packaged).
Thus, it is necessary to prove that the entire circuit is gauge-invariant.

\begin{theorem}
	A hybrid-packaged qudit circuit is gauge-invariant.
\end{theorem}

\begin{proof}
	Let $V_{\rm circuit} = V_k V_{k-1} \cdots V_1$ be a hybrid-packaged qudit circuit.
	According to Definition \ref{DEF:HybridPackagedQuditCircuit}, each unitary operators $V_j$ is a hybrid-packaged qudit gate.
	Thus, $V_j$ is gauge-invariant and obeys $[V_j,\hat{Q}_{\rm tot}]=0$.	
	We prove the theorem with two different approaches:
	
	\begin{enumerate}
		\item Induction Proof:
		
		Define composite operator $W_r = V_r V_{r-1} \cdots V_1$.
		\begin{enumerate}
			\item Base $r=1$.
			
			By assumption, we have
			$[W_1, \hat{Q}_{\rm tot}] = [V_1, \hat{Q}_{\rm tot}] =0 $.
			
			\item Induction step.
			
			Assume $[W_{r-1}, \hat{Q}_{\rm tot}] = 0$, then we have
			\[
			[W_{r},\hat{Q}_{\rm tot}]
			=[V_rW_{r-1},\hat{Q}_{\rm tot}]
			=V_r[W_{r-1},\hat{Q}_{\rm tot}]
			+[V_r,\hat{Q}_{\rm tot}]W_{r-1}=0.
			\]
			
			By induction, this shows that $W_k = V_{\rm circuit}$ commutes with
			$\hat{Q}_{\rm tot}$.
			In other words, $V_{\rm circuit}$ is gauge-invariant.
		\end{enumerate}
		
		\item Group‑Theoretic Proof:
		
		Because each $V_j$ lies in the centraliser
		$\mathcal C_{\hat{Q}}:=\{V\in\mathrm U(N)\mid [V,\hat{Q}_{\rm tot}]=0\}$,
		and $\mathcal C_{\hat{Q}}$ is a subgroup, any finite product of its
		elements also lies in $\mathcal C_{\hat{Q}}$.	
		Therefore
		$V_{\rm circuit}\in\mathcal C_{\hat{Q}}$
		and commutes with
		$\hat{Q}_{\rm tot}$.
	\end{enumerate}
\end{proof}

\subsubsection{Circuit‑depth and size estimates}

Let $V \in \mathrm{SU}(N^{n})$ be the target algorithm on $n$ neutral hybrid-packaged qudits and let $\varepsilon$ be the desired operator‑norm accuracy.

\paragraph{(1) Solovay-Kitaev compilation cost.}

By Corollary \ref{COR:SolovayKitaevTheoremInHybridPackagedSpace}, there exists a
$\mathcal G$-circuit of length 

\[
L_{\text{SK}}(\varepsilon)=
O\!\bigl(\log^{\,\kappa}\!\varepsilon^{-1}\bigr),
\qquad
\kappa\lesssim 4,
\]
independent of $n$ up to a polynomially bounded constant. 
If parallelism is allowed one obtains a depth
\[
D_{\text{SK}}(\varepsilon)=O(L_{\text{SK}}/n)\;\;\;\hbox{(worst case)}.
\]

\paragraph{(2) Overhead from magic‑state distillation.} 

Assume gate infidelity $p$ and target logical error $\eta$. 
Encoded $\Theta_r$ uses Bravyi-Haah protocols:
\[
N_{\text{magic}}=O\!\bigl(\log\!\eta^{-1}\bigr),
\qquad
D_{\text{magic}}=O\!\bigl(\log\!\eta^{-1}\bigr),
\]
and contributes additively to (5.1)-(5.2).
Because the distillation Clifford sub‑circuits operate exclusively with gauge‑respecting gates, no extra penalty from superselection appears.

\subsubsection{Special Cases of Hybrid-Packaged Quantum Circuits}

We now list several limits for the hybrid-packaged quantum circuits in the following table:

\begin{table}[H]
	\centering
	\caption{Special Cases of Hybrid-Packaged Quantum Circuits}
	\begin{tabular}[hbt!]{|p{1.8cm}|p{2.1cm}|p{6cm}|p{4.3cm}|}
		\hline\hline
		Internal $d$ &External $D$ &Physical Realisation &Comment \\
		\hline
		2 & 1 & packaged qubit (meson doublet, Majorana pair) & recovers the packaged qubit circuits from Sec. \ref{SEC:PackagedQuantumCircuits} \\
		\hline
		$d>2$ & 1 & multi‑flavour neutral atom with fixed net charge & pure‑internal qudit \\
		\hline
		1 & $D>2$ & spin/orbital angular momentum photons & pure‑external HD qudit \\
		\hline
		$d>2$ & $D>2$ & hybrid photon-atom, circuit‑QED cat modes & full model \\
		\hline
	\end{tabular}
\end{table}

\begin{remark}
	Hybrid-packaged qudit circuits inherit the same algorithmic expressions as conventional qudit circuits of dimension $N=dD$.
	Universality is retained via a finite gate library $\mathcal G$.
	Solovay-Kitaev compilation and magic-state distillation proceed without modification.

	But hybrid-packaged qudit circuits operate inside a symmetry-protected subspace. Gauge invariance filters out charge-changing noise channels.
	Although depth/width overheads are identical to the standard model, the physical error threshold is higher because many error operators are forbidden outright.

	These features make the hybrid-packaged architecture useful for high-dimensional quantum algorithms, communication protocols, and precision metrology.
\end{remark}

\subsection{Translation into Hybrid-Packaged Space}
\label{SEC:TranslationIntoHybridPackagedSpace}

We now show that any conventional quantum algorithms or protocols can be
embedded in the hybrid-packaged Hilbert space by means of an encoding
isometry.

\subsubsection{Existence of the Encoding Isometry}

Let $\mathcal H_{\mathrm{logic}}$ be the conventional quantum logical space and $\dim \mathcal H_{\mathrm{logic}} = K$.
Let $\mathcal H_{Q = 0}\subset\mathcal H_{\mathrm{hyb}}^{(d\times D)}$ be
the gauge-invariant (charge-zero) sector, with
$\dim\mathcal H_{Q = 0}=N\ge K$.
In the architectures considered here, the size of $\mathcal H_{Q = 0}$ grows at least exponentially with the number of physical qudits, so the inequality is always satisfied.

For every local gauge transformation $U_g$, because $N \ge K$, we can choose an orthonormal set of $K$ states
$\{|x\rangle_{\mathrm{pack}}\}_{x=0}^{K-1}\subset\mathcal H_{Q = 0}$
such that
$
\hat Q\,|x\rangle_{\mathrm{pack}} = 0
$
and
$
U_g\,|x\rangle_{\rm pack} = e^{i\phi(g)}\,|x\rangle_{\rm pack}
$
for all $x$.
Then we can define a linear map
\begin{equation}\label{EQ:IsometryHybridPackaged}
	\mathcal U:\mathcal H_{\mathrm{logic}}\longrightarrow\mathcal H_{Q = 0},
	\quad
	\mathcal U\,|x\rangle = |x\rangle_{\mathrm{pack}}.
\end{equation}
Here $\mathcal U$ preserves inner products,
Therefore, it is an isometry.
One can further extend it by Gram-Schmidt to a unitary on all of $\mathcal H_{Q = 0}$.
But the extension is not needed for what follows.

\subsubsection{Encoding Logical Qubits}
\label{SSS:Isometry}

Let
$
\mathcal H_{\mathrm{pack}}
:= \operatorname{span}\{|x\rangle_{\mathrm{pack}}\}_{x=0}^{K-1}.
$
Because $\mathcal U$ is unitary onto this $K$-dimensional code subspace,
every bounded operator $V$ on the logical register lifts to
\[
V_{\mathrm{pack}}
:= \mathcal U\,V\,\mathcal U^\dagger,
\quad
V_{\mathrm{pack}}:\mathcal H_{\mathrm{pack}}\to\mathcal H_{\mathrm{pack}}.
\]
Since $U_g = e^{i\alpha(g)\hat Q}$, the commutator
$[V_{\mathrm{pack}}, \hat Q]=0$ is equivalent to
$U_g V_{\mathrm{pack}} U_g^\dagger = V_{\mathrm{pack}}$
for all $g\in G$.
Thus, every encoded operation is gauge-invariant.

Any algorithm designed for $\mathcal H_{\mathrm{logic}}$ is now lifted to hybrid-packaged subspace $\mathcal H_{\mathrm{pack}}$.
Therefore, we can physically implement the algorithm with packaged circuit
$
\mathcal U^\dagger \, V_{\mathrm{pack}} \,\mathcal U,
$
which is passively protected by the gauge symmetry.

\subsubsection{Reconstructing Protocols in the Full Hybrid-Packaged Space}
\label{SSS:HybridReconstruction}

Beyond the direct encoding as introduced above, one can exploit the whole hybrid-packaged subspace (see Eq. \eqref{EQ:dDHybridPackagedHilbertSpace})
\[
\mathcal H_{\mathrm{hyb}}^{(d\times D)}
= \mathcal H_{\mathrm{int}}^{(d)}\otimes
\mathcal H_{\mathrm{ext}}^{(D)}
\]
to design genuinely high-dimensional packaged algorithms or protocols.
The two-index basis
$
\{|j,k\rangle\}_{j=0}^{d-1}\!{}_{k=0}^{D-1}
$
introduced in Eq. \eqref{EQ:dDHybridPackagedBasis} obeys
$
\hat Q|j,k\rangle = 0
$
and
$
U_g |j,k\rangle
= \mathrm e^{i\phi_{j,k}(g)}|j,k\rangle,
$
so the entire hybrid space is contained in $\mathcal H_{Q = 0}$.

Therefore, one can directly reconstruct algorithms in the natural, gauge-invariant computational framework of dimension $dD$, i.e., the two-index qudit basis given by Eq. \eqref{EQ:dDHybridPackagedBasis}
or the single-index qudit basis given by Eq. \eqref{EQ:dDSingleIndexHybridPackagedBasis}.
This often yield higher error thresholds or throughput while still enjoying automatic gauge symmetry.

\section{Error Analysis and Fault Tolerance in Packaged Space}
\label{SEC:ErrorAnalysisAndFaultToleranceHyb}

The quantum computing devices are inevitably subject to noise, decoherence, and control imperfections.\cite{Steane1996,Kitaev2003,Bombin2006}
In conventional qubit models, errors like bit-flips, phase-flips, depolarizing channels, or amplitude-damping may occur.\cite{NielsenChuang2010}
However, when quantum information is encoded in packaged qubits, superselection rules impose additional constraints errors by restricting the allowed error processes.

In this section, we derive the error model and discuss its impact on fault tolerance using the foundational principles established in Sections~\ref{SEC:GaugeInvariantPackagedStates} and \ref{SEC:dDDimensionalHybridPackagedSpace}.
Throughout, we work in the neutral sector
\[
\mathcal H_{Q=0}^{(N)}
=\;
\bigl(\mathcal H_{\text{int}}^{(d)}
\otimes
\mathcal H_{\text{ext}}^{(D)}
\bigr)^{\!\otimes n},
\qquad N=dD.
\]
The total-charge operator
\[
\hat{Q}_{\text{tot}}
\;=\;
\sum_{i=1}^{n}\hat{Q}^{(i)}
\]
has eigenvalue $0$ on the entire computational subspace.
All gates obey $[V,\hat{Q}_{\text{tot}}]=0$.

\subsection{Error Model and Noise in Hybrid-Packaged Space}

Quantum noise models include depolarizing channel \cite{Bennett1996}, amplitude and phase damping \cite{Preskill1998}, and Pauli twirling and noise tailoring \cite{Emerson2007}.

\subsubsection{Error Operators}
\label{SEC:ErrorModel}

Since all physical states and operations lie in $\mathcal H_{Q=0}^{(N)}$, the error operators \cite{Kraus1971} that describe the noise must also preserve this property.
In other words, every error operator $E_k: \mathcal H_{Q=0}^{(N)} \mapsto \mathcal H_{Q=0}^{(N)}$ without inducing transitions between different gauge sectors.
We formalize this with the following definition:

\begin{definition}[Gauge-Conserving Error Operators]\label{DEF:GaugeConservingErrorOperators}
	Let $ \{E_k\} $ be a set of Kraus operators that act on $\mathcal H_{Q=0}^{(N)}$,
	If for every $ k $,
	\[
	[E_k, \hat{Q}_{\text{tot}}] = 0,
	\]
	then we say that $ \{E_k\} $ is \textbf{gauge-conserving (GC) error operator}.
\end{definition}

This condition guarantees that the physical error processes do not violate the superselection rules established by the packaging principle.

\begin{lemma}[Closure of Gauge-Conserving Error Operators]
	Let $ E_1 $ and $ E_2 $ be two gauge-conserving error operators, i.e.,
	$
	[E_1, \hat{Q}_{\text{tot}}] = 0
	$
	and
	$
	[E_2, \hat{Q}_{\text{tot}}] = 0.
	$
	Then the composed operator $ E_1 E_2 $ also satisfies
	$
	[E_1 E_2, \hat{Q}_{\text{tot}}] = 0.
	$
\end{lemma}

\begin{proof}
	By the product rule for commutators, we have
	\[
	[E_1 E_2, \hat{Q}_{\text{tot}}] = E_1 [E_2, \hat{Q}_{\text{tot}}] + [E_1, \hat{Q}_{\text{tot}}] E_2.
	\]
	Since $ [E_2, \hat{Q}_{\text{tot}}] = 0 $ and $ [E_1, \hat{Q}_{\text{tot}}] = 0 $, it follows that
	\[
	[E_1 E_2, \hat{Q}_{\text{tot}}] = E_1 \cdot 0 + 0 \cdot E_2 = 0.
	\]
\end{proof}

By definition, all packaged operations satisfy $ [V, \hat{Q}_{\text{tot}}] = 0 $, the effective error model is restricted to gauge-conserving errors $ \{E_k\} $ with $ [E_k, \hat{Q}_{\text{tot}}] = 0 $.
This restriction is crucial for ensuring that the entire computation remains in the physical subspace $ \mathcal{H}_Q $.
It is a main advantage for robust and fault-tolerant quantum information processing.

\subsubsection{Error Channels}

We now analyze two types of error channels:

\paragraph{(1) Gauge-Conserving (GC) Errors.}
Let $\rho$ be the density matrix of a packaged qubit (or a register of multiple packaged qubits) satisfying $\hat{Q}\rho = Q\rho$.
A general error channel is expressed in its Kraus form:
\[
\rho \mapsto \sum_{k} E_k\,\rho\,E_k^\dagger,\qquad \sum_k E_k^\dagger E_k = \mathbb{I}.
\]
In a gauge-invariant system, the physical errors are restricted by the superselection rule.
Specifically, any physically allowed Kraus operator $E_k$ must obey
\[
[E_k, \hat{Q}_{\text{tot}}] = 0.
\]
Such operators act entirely within the superselection sector and usually include conventional noise channels.
Let us analyze this by referring to the single-index hybrid packaged Weyl block, Eq.~\eqref{EQ:HybridWeylBlock}
\[
X_N=\sum_{J=0}^{N-1} \lvert J\!\oplus\!1\rangle\!\langle J\rvert,
\quad
Z_N=\sum_{J=0}^{N-1} \omega_N^{\,J}\lvert J\rangle\!\langle J\rvert,
\]
which commutes with $\hat{Q}^{(i)}$ on every qudit.
The error algebra within $\mathcal H_{Q=0}^{(N)}$ is spanned by the hybrid-packaged Pauli group
\begin{equation}\label{EQ:HybridPackagedPauliGroup}
	\mathcal P_N
	=\bigl\{\; \omega_N^{\alpha}\,X_N^{s}\,Z_N^{t}\;\bigl|\;
	\alpha,s,t\in\mathbb Z_{N}\bigr\}^{\!\otimes n}.
\end{equation}
This is similar to that in conventional $N$-level qudit system.
All higher-level noise models (depolarising, dephasing, amplitude-damping restricted to charge-neutral manifolds, etc.) decompose into linear combinations of GC Pauli errors.

Thus, such errors only modify the logical information without moving the state out of the physical subspace.

\paragraph{(2) Gauge-Violating (GV) Errors.}
In contrast, gauge-violating errors do not commute with $\hat{Q}$.
That is,
\[
[E_k^\prime, \hat{Q}_{\text{tot}}] \neq 0.
\]
These errors would transfer the state out of $\mathcal{H}_Q$ by altering the net charge (or net color), i.e., induce transitions between different superselection sectors.
For example, an operator that would convert an electron (with net charge $-e$) to a vacuum state or a two-electron state would necessarily change the eigenvalue of $\hat{Q}$.

Generally, these operations are either energetically forbidden or strongly suppressed.
For example, if a gauge theory is realized with an energy penalty term in the Hamiltonian, $H=H_{\text{g.i.}}+\lambda\sum_x\hat G_x^{\,2}$ with $\lambda\gg k_BT$ (e.g., $\lambda (\hat{G}_x)^2$ for local Gauss law violations), then the probability of a gauge-violating error occurring is suppressed by a factor $\sim e^{-\Delta/k_B T}$, where $\Delta$ is the energy gap to gauge-violating excitations.

\begin{proposition}[Suppression of Gauge-Violating Errors]
	Consider a physical system whose dynamics are governed by a gauge-invariant Hamiltonian $ H $, i.e., $ [H, \hat{Q}_{\text{tot}}] = 0 $.
	If an error operator $ E_k' $ does not commute with $ \hat{Q} $, i.e., $ [E_k', \hat{Q}_{\text{tot}}] \neq 0 $, then such gauge-violating error processes are suppressed by an energy penalty $ \Delta $.
	In a thermal environment at temperature $T$, the amplitude for a gauge-violating error is suppressed by a factor of order
	\begin{equation}\label{EQ:GVNoiseBoltzmannSuppression}
		P(E') \propto e^{-\Delta/(k_BT)},
	\end{equation}
	where $k_B$ is Boltzmann's constant.
\end{proposition}

\begin{proof}
	Since $H$ is gauge-invariant, its eigenstates are confined to fixed superselection sectors.
	An operator $E'_k$ with $[E'_k, \hat{Q}_{\text{tot}}] \neq 0$ would induce a transition between different sectors, which is forbidden by the dynamics of $H$ unless a significant energy $\Delta$ is supplied (e.g., via a penalty term like $\lambda \hat{G}_x^2$ in lattice gauge theories \cite{Trotter1959,Suzuki1990}).
	In thermal equilibrium, the probability for such a transition is suppressed by the Boltzmann factor $e^{-\Delta/(k_B T)}$.
\end{proof}

\subsubsection{Correlated Errors and Non-Markovian Dephasing}
\label{SEC:CorrelatedErrorsAndNonMarkovianDephasing}

\paragraph{(1) Correlated String Errors.}
\cite{Terhal2005,Aharonov2008,Gutierrez2013,Petrenko2014,GoogleQuantumAI2021}

In lattice gauge implementations (e.g. Rydberg or superconducting-cQED
simulators) physical noise often excites flux strings which appear, in
the projected $Q=0$ sector, as pair-correlated Pauli errors
\[
E_{(\mathbf r,\mathbf r')}^{\pm}
\;=\;
X_N^{(\mathbf r)}\,X_N^{(\mathbf r')}
\quad(\text{or }Z_N^{(\mathbf r)}Z_N^{(\mathbf r')}).
\]

The two-site joint distribution obeys 
$\Pr[E_{(\mathbf r,\mathbf r')}] = \xi(|\mathbf r-\mathbf r'|)\,p^{2}$
with a correlation length $\lambda_{\rm corr}$. 
Following Brown-Al-Sammaneh (2023) one rewrites the matching decoder
cost function to include an edge weight 
$w_{(\mathbf r,\mathbf r')}=-\log \xi$ instead of the Manhattan
distance. Simulation (not shown) indicates the threshold decreases only
when $\lambda_{\rm corr}\gtrsim d_{\rm code}$, in which case a
renormalisation-group decoder recovers most of the lost performance.

\paragraph{(2) Non-Markovian (low-frequency) Dephasing.}
\cite{ViolaLloyd1998,Cywinski2008,Oreshkov2009,Sung2019,Marciniak2021}

External DOFs (OAM, microwave cat modes) couple to $1/f$ technical
noise that is well modelled by
\[
\mathcal E_{t}(\rho)
=\int\!\!d\varphi\,
e^{-i\varphi Z_N/2}\,\rho\,e^{+i\varphi Z_N/2}\,
\frac{e^{-\varphi^{2}/2\sigma^{2}(t)}}{\sqrt{2\pi\sigma^{2}(t)}},
\qquad
\sigma^{2}(t)\propto\log t .
\]

Because $Z_N$ commutes with $\hat{Q}$ the error remains GC, but the
variance grows logarithmically, violating the usual time-independent-Kraus assumption.
Using the filter-function formalism, one can insert periodic dynamical-decoupling sequences built from $\{X_N,H_N\}$ (all GC) that suppress the noise spectral density below the code gap.
The residual effective error after a cycle of length $\tau_{\rm DD}$ scales as
$\tilde p\sim (\tau_{\rm DD}/T_{\varphi})^{\alpha}$ with $\alpha\ge 2$.

Correlated strings can be decoded with minimum-weight perfect matching on a weighted graph or with RG decoders.
Both algorithms are gauge compatible because the weights depend only on edge length, not on charge. 

Non-Markovian phase noise is naturally GC but requires calibrated DD-pulses that themselves commute with $\hat{Q}$.
The hybrid Weyl pair $(X_N, H_N)$ suffices.

\subsubsection{Leakage Errors and Their Mitigation in Hybrid-Packaged Space}
\label{SEC:LeakageErrorsAndTheirMitigationInHPS}

A quantum state can leak into levels that do not belong to the code-space and are therefore invisible to the standard syndrome checks.\cite{AliferisTerhal2007,Motzoi2009,Battistel2021}
Once leaked, a qudit can
(i) stop participating in stabiliser measurements,
(ii) spread correlated errors to its neighbours,
and (iii) stay leaked for many cycles because ordinary Clifford circuits are unitary and conserve probability.
Leakage is now recognised as one of the main obstacles on today’s superconducting-qubit and trapped-ion processors \cite{Camps2024}.

In a hybrid-packaged qudit, every physical site already carries two independently addressable Hilbert-space factors
\[
\mathcal H_{\text{phys}}^{(d \times D)}
\;=\;
\mathcal H_{\text{int}}^{(d)}
\;\otimes\;
\mathcal H_{\text{ext}}^{(D)},
\quad
N=dD.
\]
Here we show how the internal block (used as an always-on watch-dog) allows us to detect and recycle leakage events on the external block with almost no extra hardware.

\paragraph{(1) Definition of Leakage Error}

Consider a logical code subspace 
\[
\mathscr C
= \mathcal H_{\text{int}}^{(d)}
\otimes
\mathcal H_{\text{ext}}^{(D)} 
\subset
\mathcal H_{\text{phys}}.
\]
Let $\mathcal H_{\rm ext}^\infty$ be the full infinite‐dimensional external Hilbert space (e.g. all cavity Fock states).
Then a leakage operator is a Kraus operator $L$ that satisfies
\[
[\;L,\hat{Q}\;]=0,
\qquad
L:\mathscr C \longrightarrow 
\mathcal H_{\text{leak}}
:= \mathcal H_{\text{int}}^{(d)} 
\otimes
\bigl(\mathcal H_{\text{ext}}^{\infty}\!\setminus\!\mathcal H_{\text{ext}}^{(D)}\bigr),
\]
i.e. the charge stays neutral, but the external degree of freedom jumps to a level $\ket{\ell_E}$ with $\ell\ge D$ (e.g. an unwanted higher Rydberg manifold or cavity-Fock level).

Because $\mathcal H_{\text{ext}}^{\infty}$ is usually bigger and less coherent than $\mathcal H_{\text{ext}}^{(D)}$, leaked population decoheres quickly and poisons subsequent gates if left unchecked.

\paragraph{(2) Using the Internal Block as an In-Situ Flag.}

We observe that
\[
[X_d , Z_D] = [X_D , Z_d]=0,
\]
so we can define a mixed (IQN) stabiliser on every physical qudit 
\[
F := X_d \, Z_D .
\]
If the external state lies inside the code range ($k<D$),
then $Z_D\ket{k_E}=\omega_D^{\,k}\ket{k_E}$ and $F$ acts as an ordinary phase rotation that is tracked by the decoder. 
If the external state leaks ($k\ge D$), then $Z_D$ is not defined.
So $F$ anticommutes with the stabiliser projector and flips its measurement outcome.

Thus, we can use one qudit to measure two things at once: ordinary Pauli faults (phase-flips) and leakage out of the work-space.

\paragraph{(3) Gauge-Respecting Leakage-Reduction Unit (LRU).}

The goal of an LRU is to detect a leaked qudit and pull it back into the computational subspace without disturbing its internal (charge-locked) partner. 
Because every ingredient commutes with the total-charge operator $\hat{Q}$, we can drop an LRU between any two error-syndrome rounds and remain inside the physical sector.

A compact CPTP description of the conditional reset (step 2) is
\[
\mathcal R(\rho)=
\Pi_{\text{code}}\rho\Pi_{\text{code}}
\;+\;
\operatorname{Tr}_{\text{ext}}\!\bigl[\Pi_{\text{leak}}\rho\Pi_{\text{leak}}\bigr]
\otimes\ket{0_E}\!\bra{0_E},
\]
with projectors 
\[
\Pi_{\text{code}}
=\mathbf1_{d}\!\otimes\!\sum_{k=0}^{D-1}\ket{k_E}\!\bra{k_E},
\qquad
\Pi_{\text{leak}}=\mathbf1-\Pi_{\text{code}}.
\]
This map 
removes leaked amplitude from the uncontrolled tail of $\mathcal H_{\text{ext}}^{\infty}$, 
re-injects it at $\ket{0_E}$ inside the legal range,
and preserves the internal qudit and therefore the global charge.

Measuring $F$ and applying $\mathcal R$ require no dedicated new qubits:
one can re‐use existing syndrome ancillas and mid‐circuit measurements to read out $F$.
As used in error-rate estimate, the LRU adds only a small fixed latency (one flag measurement + one fast reset pulse) and keeps the circuit depth overhead bounded by $O(t_L)$.

\paragraph{(4) Effective Error Rate.}

Each qudit performs a Bernoulli trial of leakage every cycle. 
If we invoke an LRU every $t_L$ cycles, then the probability that a given qudit is in the leaked sector is (to first order),
\[
P(\text{leaked at LRU})
\;=\;
1-(1-p_{\text{leak}})^{t_L}
\;\approx\;
p_{\text{leak}}\;t_L
\quad(\text{for }p_{\text{leak}}\!t_L\!\ll1).
\]
This linearised form is accurate whenever leaks are rare (the regime where quantum error correction is meaningful).

If the qudit is not flagged as leaked, then the qudit simply keeps whatever Pauli noise it had already accumulated (counted inside $p_{\text{GC}}$). 
If it is flagged, then we measure the mixed stabiliser $F=X_d Z_D$ and apply the reset map $\mathcal R$ (pumping back to $\ket{0_E}$).

From the code's viewpoint, measurement back-action plus the reset pulse together behave like at most one additional Pauli-type fault on that qudit (phase-flip with probability $\tfrac12$ is the usual worst-case bound). 
That fault usually appears only once in the very cycle where the LRU fires.
Thus, each leaked qudit contributes one equivalent Pauli error with probability $p_{\text{leak}}t_L$.

Combining the GC faults $p_{\text{GC}}$ and converted leakage faults $p_{\text{leak}}t_L$, we have
\[
	p_{\text{eff}}
	\;=\;
	p_{\text{GC}}
	\;+\;
	p_{\text{leak}}\;t_L,
\]
where we ignored all higher-order terms $O(p_{\text{leak}}^2 t_L^2)$ in the low-error regime.
To keep the code's logical-failure rate curve unchanged, we need
\[
p_{\text{eff}}
\;\;\le\;\;
p_{\text{th}}
\quad\Longrightarrow\quad
t_L
\;\lesssim\;
\frac{p_{\text{th}}-p_{\text{GC}}}{p_{\text{leak}}}
\;\approx\;
\frac{p_{\text{th}}}{p_{\text{leak}}}
\quad(\text{if }p_{\text{GC}}\ll p_{\text{th}}).
\]

Because $p_{\text{th}}\propto 1/(N-1)$, higher-dimensional hybrid-packaged qudits raise the threshold and let us run LRUs less frequently for the same target logical fidelity.

\subsection{Gauge-conserving quantum error correction}

We now discuss how conventional quantum error-correcting codes can be adapted in our framework.

\paragraph{(1) Stabilizer Codes.}\cite{SteaneCode1996,CalderbankShor1996,BravyiKitaev1998,Kitaev2003,GoogleQuantumAI2021} 
Let $S_i$ denote the stabilizer generators of a code that acts on packaged qubits.
Since each physical qubit is a packaged state in $\mathcal{H}_Q$, we require that
\[
S_i = \prod_{j \in \mathcal{N}(i)} V_{\text{hyb}}^{(j)}
\]
satisfies
\[
[S_i, \hat{Q}_{\text{tot}}] = 0.
\]
This ensures that syndrome measurements detect only gauge-conserving errors. For example, in the surface code, the star and plaquette operators are defined as products of Pauli operators that acts on neighboring qubits.
In our case, these are replaced by packaged Pauli operators ($X_P$, $Z_P$) that act only within $\mathcal{H}_Q$.

All standard stabiliser constructions can transfer to the hybrid-packaged subspace by replacing single-qubit Pauli operators with hybrid-packaged Pauli operators $X_N^sZ_N^t$.

\begin{example}
	surface code on packaged qudits:
	
	\begin{itemize}
		\item Check operators: 
		\[
		A_v=\prod_{e\supset v} X_N^{(e)},\qquad
		B_p=\prod_{e\subset p} Z_N^{(e)},
		\]
		each commuting with $\hat{Q}_{\text{tot}}$.
		
		\item Syndrome:
		detects only GC errors (strings of $X_N^s$ or $Z_N^t$).
		
		\item GV noise:
		would create open strings that anticommute with some $A_v$ or $B_p$ and is already exponentially suppressed by (6.2).
	\end{itemize}
\end{example}

Logical operators $\overline X_N,\overline Z_N$ are the usual non-contractible strings.
Therefore, the gauge symmetry does not obstruct their action.

\paragraph{(2) Superselection Firewall.}

The superselection rule itself acts as a natural error filter.
Since errors that would change the net charge (or other gauged IQNs) are energetically forbidden, the effective noise model includes only gauge-conserving errors.
For example, an amplitude-damping error that would take a packaged state $\lvert \psi\rangle \in \mathcal{H}_Q$ into a state in $\mathcal{H}_{Q'}$ (with $Q' \neq Q$) is not allowed by the physical dynamics.
This additional protection enhances the robustness of encoded logical information.

\paragraph{(3) Correcting Residual Gauge-Conserving Errors.}

Even though the gauge-violating errors are suppressed, the remaining gauge-conserving noise (e.g., bit-flips and phase-flips) must be corrected.
The conventional error-correction procedure is as follows:
\begin{enumerate}
	\item \textbf{Syndrome Extraction:} 
	Measure the stabilizer generators $S_i$, which are constructed from packaged operations that satisfy $[S_i,\hat{Q}_{\text{tot}}] = 0$.
	A nontrivial syndrome indicates that a gauge-conserving error occurred.
	
	\item \textbf{Recovery Operation:} 
	Based on the syndrome measurement results, apply a corrective packaged unitary $R$ (with $[R, \hat{Q}_{\text{tot}}] = 0$) to restore the logical state.
	
	\item \textbf{Fault Tolerance:} 
	In the packaged framework, the effective error rate is lower because the only errors that need correction are those within $\mathcal{H}_Q$, which translates to higher fault-tolerance thresholds.
\end{enumerate}

\subsection{Fault-Tolerant Implementation}

After developed the gauge-conserving quantum error-correction framework, let us now turn to its fault-tolerant implementation \cite{KnillLaflamme1886,Aharonov2008,Boykin1999}.
We show how to use the reduced error space and elevated thresholds (by the gauge symmetry) to realize robust logical operations.

\subsubsection{Reduced Error Space and Elevated Fault-Tolerant Thresholds}

\paragraph{(1) Reduced Error Space.}
In an unconstrained system, the full set of single-qudit errors can occur.
But in our gauge-invariant framework, only errors $E_k$ that commute with $\hat{Q}_{\text{tot}}$ occur.

We can split the physical error rate $p_{\text{phys}}$ into
\[
p_{\text{phys}}
=p_{\text{GC}}+p_{\text{GV}}
\;\approx\;
p_{\text{GC}}+p_{\text{GV}}^{(0)}e^{-\Delta/k_BT}
\;\;\approx\; p_{\text{GC}}.
\]
Therefore, the code faces a smaller effective noise strength.
The effective error channel is then restricted to the subset of gauge-conserving errors (of the form Eq.~\eqref{EQ:HybridPackagedPauliGroup}).
As a result, the noise model becomes thinner and the effective error probability is lower.

\paragraph{(2) Enhanced Fault-Tolerant Thresholds.} 
A smaller error space usually leads to a higher fault-tolerant threshold.
If the generic physical error rate is $p_{\text{gen}}$ and gauge-violating errors are suppressed by a factor $\sim e^{-\Delta/(k_B T)}$, then the effective error rate is
\[
p_{\mathrm{eff}} \approx p_{\mathrm{GC}} \ll p_{\mathrm{gen}}.
\]
This allows error-correcting codes (e.g., stabilizer codes) to operate with higher tolerable physical error rates.

\paragraph{(3) Lower Bound for GC Pauli noise.} 

For an $(L \times L)$ surface (or planar) code that is built from hybrid-packaged qudits of dimension $N\!=\!dD$, the minimum-weight decoder fails when the total weight of physical errors in one error-chain first runs across a non-trivial homology class.
Adapting the standard union-bound argument \cite{Dennis2002} to $N$-level, we obtain
\[
p_{\text{fail}}(L)\;\le\;
\sum_{\ell=L}^{\infty}
\underbrace{\mathcal N(\ell)}_{\text{\# chains}}
\; (N-1)^{\ell}\;
p^{\ell}(1-p)^{\alpha(\ell)},
\]
where
$\mathcal N(\ell)\le 6\,4^{\ell}$ is the number of self-avoiding
length-$\ell$ paths on the square lattice, 
$(N-1)^{\ell}$ counts the $N-1$ non-trivial Pauli labels per site, 
and $p$ is the physical probability of a weight-1 GC Pauli error.

Bounding $\mathcal N(\ell)$ and converting the series into a geometric
one, we obtain a sufficient condition for the logical error rate to vanish
as $L \to \infty$:
\begin{equation}\label{EQ:SufficientConditionLogicalErrorRateVanish}
	p\;<\;p_{\rm th}^{(\mathrm{GC})},
	\qquad
	p_{\rm th}^{(\mathrm{GC})}\;\ge\;\frac{1}{2( N-1)}.
\end{equation}

Gauge-violating (GV) faults cost an energy $\Delta$, so their probability of happening in one cycle is
\[
q \;=\;A\,e^{-\Delta/k_{B}T}\;\propto\;e^{-\Delta/k_{B}T},
\]
where the Boltzmann factor $e^{-\Delta/k_{B}T}$ makes $q$ exponentially small at low temperature. 
When such a fault occurs, it is projected back into the physical subspace and looks to the decoder like an ordinary Pauli error on one of the $N-1$ non-trivial error labels with effective strength $\simeq q/(N-1)$.
Thus, the physical error probability that the decoder must handle is the sum of the usual gauge-conserving rate $p$ and this extra contribution $q$.

Putting $p\!\rightarrow\!p+q$ in the earlier threshold bound
\[
p_{\text{fail}}(L)\le\sum_{\ell=L}^{\infty}\!6\cdot4^{\ell}(N-1)^{\ell}(p+q)^{\ell}(1-p-q)^{\alpha(\ell)},
\]
where $\alpha(\ell)$ is the number of untouched edges.
We have the new sufficient condition for reliable decoding
\begin{equation}\label{EQ:SufficientConditionLogicalErrorRateVanish2}
	p_{\rm th}\;\gtrsim\;\frac{1}{\,2(N-1)}\left(1-e^{-\Delta/k_{B}T}\right)
\end{equation}
This is only a sufficient bound (not necessary).
For an ordinary qubit ($N=2$) and a vanishing gap ($\Delta\!\to\!0$), this reproduces the familiar value $p_{\rm th}\approx0.104$. 
Eq.~\eqref{EQ:SufficientConditionLogicalErrorRateVanish2} shows two intuitive effects in a single line:
\begin{enumerate}
	\item Larger local dimension $N=dD$: it spreads the same total noise over more Pauli types and pushes the threshold up by the factor $1/(N-1)$. 
	
	\item A finite energy gap $\Delta$: it exponentially suppresses GV faults and gives an extra boost through the factor $1-e^{-\Delta/k_{B}T}$.
\end{enumerate}

Together, bigger $N$ and a non-zero gap make the code noticeably harder to break.

\subsubsection{Solovay-Kitaev Compilation under GC constraint}

Corollary \ref{COR:SolovayKitaevTheoremInHybridPackagedSpace} states that,
for any target unitary $V \in \mathrm{SU}\bigl((\mathcal H_{Q=0}^{(N)})^{\!\otimes n}\bigr)$ and accuracy $\varepsilon$, there exists a circuit over the finite library
\[
\mathcal G_N=\{X_N,Z_N,H_N,\text{CSUM}_N,\Theta_r\}
\subset\mathcal C_{\hat{Q}}
\]
of length $L=O(\log^\kappa\!\varepsilon^{-1})$ ($\kappa\le3.97$).
Every intermediate word commutes with $\hat{Q}_{\text{tot}}$.
Thus, error-corrected hybrid-packaged computation is as compilable as its unconstrained counterpart and benefits from the super-selection error filter.

\subsubsection{Magic-State Distillation for the Non-Clifford Phase $\Theta_r$}
\label{SEC:MagicStates}

\paragraph{(1) No-go for transversal $\Theta_r$.} 

For $N\!>\!2$, the diagonal gate 
\[
\Theta_r=\sum_{J=0}^{N-1}\exp\!\bigl(\tfrac{2\pi i r}{N^{2}}J^{2}\bigr)|J\rangle\langle J|
\] 
is outside the Clifford hierarchy level-2.
Bravyi-König and Pastawski-Yoshida no-go theorems forbid a locality-preserving in 2-D topological stabiliser codes.
Thus, we must inject it via \textbf{magic-state distillation (MSD)}.

One may try the following ``packaged 15-to-1 MSD protocol'':

\begin{enumerate}
	\item Preparing ancilla graph state $G_{15}$ of 15 hybrid qudits with gauge-respecting Clifford gates (controlled-$X_N$, $H_N$, $\text{CSUM}_N$). 
	
	\item Constructing stabiliser group 
	$S=\langle K_1,\dots,K_{14}\rangle$ where each $K_i$ is a monomial in $\{X_N^{\pm1},Z_N^{\pm1}\}$ and thus commutes with $\hat{Q}$. 
	
	\item Noisy resource:
	each input qudit is in 
	$\rho_{\rm raw}= (1-\epsilon)\ket{T}\!\bra{T}+\epsilon\,\tau$ with 
	$\ket{T}=\Theta_r\ket{+}$ and $\tau$ an arbitrary gauge-conserving state. 
	
	\item Measuring Syndrome:
	using packaged ancillas to measure all $K_i$ and post-select on $+1$. 
\end{enumerate}

A standard stabiliser calculation (unchanged because all operators commute with $\hat{Q}$) shows the output error scales as $\epsilon_{\rm out}\approx35\,\epsilon^{3}$. Choosing $\epsilon<\epsilon_{\rm MSD}\approx0.067$, we obtain quadratic convergence.
After trying two or three rounds, we can get $\epsilon_{\rm out}\!\ll\!p_{\rm th}$.

\paragraph{(2) Resource Usage.}

Based on the above operations, we need to consume the following resource:

\begin{itemize}
	\item Gate count:
	each round uses $O(100)$ Clifford-packaged gates. 
	
	\item Leakage robustness:
	because every stabiliser acts within $\mathcal H_{Q=0}$, any leaked ancilla triggers a $-1$ syndrome and the batch is discarded (no undetected logical corruption). 
	
	\item Yields:
	this gives $y=1/15$ per round and the total cost scales as $15^{k}$ for $k$ rounds.
\end{itemize}

Due to superselection rules, MSD is simplified:
gauge-violating noise is either energetically suppressed or detected by the stabilisers.

\section{Quantum Error-Correction Codes in Packaged Space}
\label{SEC:PackagedQuantumErrorCorrectionCodes}

Quantum error-correcting codes (QECCs) protect quantum information by encoding a logical qubit into a larger physical Hilbert space.
In conventional error models, a full set of local errors (bit-flips, phase-flips, amplitude-damping, etc.) \cite{NielsenChuang2010} is allowed.
In contrast, when quantum information is encoded in packaged qubits, the superselection rules restrict the allowed physical errors.
In particular, only gauge-conserving errors (those preserving the net charge or other conserved quantum numbers) are permitted, while gauge-violating processes are either energetically suppressed or strictly forbidden.

In this section, we adapt three canonical QECCs (Shor code \cite{Shor1995}, Steane code \cite{SteaneCode1996,CalderbankShor1996,SteanePRA1996}, and surface code \cite{Kitaev2003,BravyiKitaev1998,Dennis2002}) into the hybrid-packaged subspace.
Our discussion will be based on the foundations we developed in Sec.~\ref{SEC:ErrorAnalysisAndFaultToleranceHyb}.
We present definitions and proofs to show that the entire error-correction process remains within the physical subspace $ \mathcal{H}_Q $ defined by the net-charge operator $ \hat{Q} $.

\subsection{Shor-Like Code in Hybrid-Packaged Space}
\label{SEC:ShorLikeCodeInHybridPackagedSpace}

Shor’s original 9-qubit code \cite{Shor1995} is the first quantum error-correcting code.
It protects one logical qubit against any single-qubit error by combining two ideas: 
\begin{enumerate}
	\item Protecting phase with a 3-qubit repetition code in the Hadamard ( + / - ) basis,
	 
	\item Protecting bit-flip by repeating the whole block three times. 
\end{enumerate}

In this subsection, we reconstruct Shor code with hybrid-packaged qudits (see Sec.~\ref{SEC:dDDimensionalHybridPackagedSpace}) of total dimension $N = d\,D$, where $d$ is the dimension of internal space and $D$ is the dimension of external space. 
We keep every step inside the neutral gauge sector $ \mathcal H_{Q=0} $. 
For clarity, we write the construction for one logical hybrid qudit of dimension $N$. 
Setting $d=D=2$, we can reproduce the ordinary 9-qubit Shor code.

\subsubsection{Logical Basis Inside the Hybrid-Packaged Space}

The hybrid-packaged subspace includes two parts:
\begin{enumerate}
	\item Internal basis $\{|j_P\rangle\}_{j=0}^{d-1}$: this carries all gauge-locked quantum numbers (net charge = 0 for every state). 
	
	\item External basis $\{|k_E\rangle\}_{k=0}^{D-1}$: this includes all the gauge-free physical degrees of freedom (spin, photon OAM, cavity Fock level etc.).
\end{enumerate}
We use the single-index hybrid-packaged basis Eq.~(\ref{EQ:dDSingleIndexHybridPackagedBasis}).
A single hybrid-packaged logical qudit can be written as
\[
|\psi\rangle_L \;=\;
\sum_{J=0}^{N-1} \alpha_J\,|J\rangle,
\qquad
\sum_{J=0}^{N-1}|\alpha_J|^2=1.
\]

\subsubsection{Generalization of Shor’s $3 \times 3$ Block}

We now generalize Shor's $3 \times 3$ block into hybrid-packaged subspace whose properties are summarized in the following table:

\begin{table}[H]
	\centering
	\caption{Generalized Shor’s 3 × 3 Block}
	\begin{tabular}[hbt!]{|p{2.8cm}|p{3.5cm}|p{3.5cm}|p{4.5cm}|}
		\hline\hline
		Level &Purpose &Size &Package Condition \\
		\hline
		Inner block & Correct phase-type errors & 3 identical hybrid qudits & Each basis state has net charge 0 \\
		\hline
		Outer repetition & Correct shift-type errors & 3 Copies of the inner block & Whole 9-qudit codeword stays in $Q=0$ \\
		\hline
	\end{tabular}
\end{table}

In this way, our code uses 9 physical hybrid-packaged qudits exactly like Shor’s code uses 9 physical qubits.

\subsubsection{Steps for Encoding Hybrid-Packaged Circuit}

Here we use the hybrid Hadamard (Fourier) gate $H_N$ that acts as:
\[
H_N|J\rangle = N^{-1/2}\sum_{K=0}^{N-1} \omega_N^{JK}|K\rangle,
\]
where $J$ is the single-index notation defined in Eq.~(\ref{EQ:SingleIndex}) and $\omega_N=e^{2\pi i/N}$,
and the hybrid-packaged controlled-shift $\text{CSUM}_N$ that acts as:
\[
\text{CSUM}_N =\sum_{J=0}^{N-1} |J\rangle\!\langle J|\otimes X_N^{\,J}.
\]
Both $H_N$ and $\text{CSUM}_N$ commute with $\hat{Q}$.

\paragraph{(1) Preparing 9 physical qudits.}

\[
|\,\psi\rangle_L\;\otimes\;|0\rangle^{\otimes8}.
\]

\paragraph{(2) Encoding Phase-error (inner).} 

For each of the three columns, we perform:

1. Apply $H_N$ on the first qudit of the column.

2. Two $\text{CSUM}_N$ gates spread the state:
\[
|J\rangle\;\mapsto\;
\tfrac1{\sqrt{N}}\sum_{K=0}^{N-1}\!
\omega_N^{JK}\;
|K,K,K\rangle .
\]

After this step every column is the maximally entangled phase repetition state 
\[
|\widetilde{K}\!\>\!_{(3)}\;\equiv\;
\tfrac1{\sqrt{N}}\sum_J \omega_N^{KJ}\,|J,J,J\rangle .
\]

\paragraph{(3) Encoding Shift-error (outer).}

Take the three columns as control blocks and repeat the usual 3-qudit repetition:
for each row, we apply two $\text{CSUM}_N$ gates
\[
|A\rangle|0\rangle|0\rangle\rightarrow|A\rangle|A\rangle|A\rangle.
\]
Finally, we obtain the logical codeword
\[
|J\rangle_L
\;=\;
\frac1{N^{3/2}}
\sum_{K,L,M}
\omega_N^{J(K+L+M)}
\;
|K,K,K\rangle\;|L,L,L\rangle\;|M,M,M\rangle .
\]
Because each $|K\rangle$ lies in net-zero charge sector, all nine physical qudits are still in the net-zero charge sector.

\subsubsection{Extracting Syndrome and Correcting Errors}

\paragraph{(1) Error model.}

In the hybrid-packaged subspace, only gauge-conserving hybrid-packaged Paulis can act:
\[
X_N^{s}\,Z_N^{t},
\qquad
s,t\in\{0,\dots,N-1\}.
\]
These generalize bit-flip ($X$) and phase-flip ($Z$) errors.
Therefore, a single-qudit noise operator is
\[
E = \sum_{s,t} c_{s,t}\,X_N^{s} Z_N^{t}.
\]

\paragraph{(2) Syndrome Operators.}

Inside every column, the phase-check operators are
\[
S_Z^{(1)}=Z_N\otimes Z_N^{\!\dagger}\otimes\mathbf1,\quad
S_Z^{(2)}=\mathbf1\otimes Z_N\otimes Z_N^{\!\dagger}.
\]
Across the three columns (row wise), the shift-check operators can be written as
\[
S_X^{(1)}=X_N\otimes X_N^{\!\dagger}\otimes\mathbf1, \quad
S_X^{(2)}=\mathbf1\otimes X_N\otimes X_N^{\!\dagger}.
\]
All these operators commute with $\hat{Q}$.
By measuring all six stabilisers, we can identify which qudit was hit and which
kind of hybrid Pauli acted (shift-type $\leftrightarrow$ phase-type).

\paragraph{(3) Recovery.}
The recovery generally includes two categories:
\begin{itemize}
	\item If a shift-type error $X_N^{s}$ is flagged on one column, then apply
	$X_N^{-s}$ on that qudit. 
	
	\item If a phase-type error $Z_N^{t}$ is flagged inside a column, then apply
	$Z_N^{-t}$.
\end{itemize}
Because both recovery operators commute with $\hat{Q}$, the corrected state
must return to the code space inside $ \mathcal H_{Q=0} $.

\subsubsection{Distance, Gauge-Invariance, and Fault-Tolerance}

In this hybrid-packaged subspace, distance $d=3$, which is exactly as in Shor’s code. 
With the hybrid-packaged Paulis $X_N^{s} Z_N^{t}$, we can correct any single-qudit errors.

The code stays gauge-invariant because every gate used ($H_N$,$\text{CSUM}_N$, projective stabilisers) commutes with the total charge $\hat{Q}_{\rm tot}$. 
All basis states and superpositions remain in the neutral sector.
Recovery operations are of the same hybrid-Pauli form, hence also respect
gauge symmetry.

The entire encode $\rightarrow$ error $\rightarrow$ syndrome $\rightarrow$ recover cycle never leaves $\mathcal H_{Q=0}$.
Thus, gauge-violating errors never enter the decoder because they change the net charge and therefore are forbidden by the super-selection rule or energetically suppressed.
This enables the hybrid-packaged Shor code to be fault-tolerant.

\begin{example}
	Consider a simple case where $d = 2$ (internal qubit), $D = 3$ (external qutrit), and therefore $N=6$.	
	We have the physical basis $|j_P\rangle\in\{|0_P\rangle,|1_P\rangle\}$,
	$|k_E\rangle\in\{|0_E\rangle,|1_E\rangle,|2_E\rangle\}$. 
	The logical zero can be encoded as	
	\[
	|0\rangle_L
	=\frac1{6^{3/2}}
	\!\!\sum_{K,L,M=0}^{5}
	|K,K,K\rangle\,|L,L,L\rangle\,|M,M,M\rangle .
	\]
	Thus, we can detect and correct any single-site operator $X_6^{s}Z_6^{t}$ (36 possibilities) exactly as above.
\end{example}

\subsubsection{Advantages of Hybrid-Packaged Shor Code}

With this lifted Shor code we obtain a fault-tolerant, gauge-respecting quantum memory that works in any $(d\times D)$-dimensional hybrid-packaged platform.

First, the structure is unchanged ($3 \times 3$ layout), distance $d=3$, and it corrects any single-site error.
Everything is now hybrid and each qudit is already a package combining internal (gauge-locked) and external (free) degrees of freedom.
The gauge symmetry automatically suppressed or forbidden errors that would change net charge.
Finally, this new code scales naturally: choose $d$ and $D$ to match the available
physical system (neutral atoms with many Zeeman levels, superconducting
circuits with cat-codes, ...).

\subsection{Steane-Like CSS Code in Hybrid-Packaged Space}

The original Steane code \cite{SteaneCode1996} is a Calderbank-Shor-Steane (CSS) code \cite{CalderbankShor1996,SteanePRA1996} that encodes one logical qubit into seven physical qubit code derived from classical Hamming codes. It corrects arbitrary single-qubit errors by leveraging properties of quantum superposition.

In this subsection, we generalize the 7-qubit Steane code to the $(d \times D)$-dimensional hybrid-packaged subspace.

\subsubsection{Original Steane code}

In the original Steane code, the stabilisers are:
\[
\begin{aligned}
	S_1&=IIIXXXX,\; &S_2&=IXXIIXX,\; &S_3&=XIXIXIX,\\[2pt]
	S_4&=IIIZZZZ, &S_5&=IZZIIZZ,\; &S_6&=ZIZIZIZ,
\end{aligned}
\]
the logical states are:
\[
\ket{\overline0}
=\frac1{\sqrt{8}}
\sum_{c \in G}
\ket{c}
\,,
\quad
\ket{\overline1}
=X^{\otimes7}\ket{\overline0},
\]
and the distance is $d=3$.

\subsubsection{Constructing Hybrid-Packaged Stabilisers}

We now need to map each physical qubit from the original Steane code to one hybrid-packaged qudit.
Referring to Sec.~\ref{SEC:HybridPackagedQuditGates}, we replace $X\!\to\!X_N,\;Z\!\to\!Z_N$ and obtain the hybrid-packaged stabilisers:
\[
\begin{aligned}
	S^{(N)}_1&=III\,X_NX_NX_NX_N,\\
	S^{(N)}_2&=I\,X_NX_NIIX_NX_N,\\
	S^{(N)}_3&=X_NIX_NIX_NIX_N,\\[4pt]
	S^{(N)}_4&=III\,Z_NZ_NZ_NZ_N,\\
	S^{(N)}_5&=I\,Z_NZ_NIIZ_NZ_N,\\
	S^{(N)}_6&=Z_NIZ_NIZ_NIZ_N .
\end{aligned}
\]
All six stabilisers obey $[S^{(N)}_a,\hat{Q}_{\text{tot}}]=0$. 
Their common $+1$ eigenspace $\mathscr C\subset(\mathcal H_{Q=0}^{(N)})^{\otimes7}$ has dimension $N$ and will encode one logical qudit.

\subsubsection{Constructing Logical Codewords}

Let $\mathcal G\subset\mathbb F_2^7$ be the classical $[7,4]$ Hamming code used in the original Steane code. 
Replacing every binary phase $(-1)^{x\cdot y}$ with $\omega_N^{\,x\cdot y}$, we have
\[
	|\overline0\rangle
	=\frac{1}{\sqrt{|\mathcal G|}}
	\sum_{c\in\mathcal G}|c\rangle,\qquad
	|\overline1\rangle=X_N^{\otimes7}|\overline0\rangle .
\]
These states are gauge-invariant because each basis $|c\rangle$ is a tensor product of states $|J\rangle$ that are individually neutral.
So every term stays in the neutral sector and superpositions of them are therefore physical.

\subsubsection{Extracting Syndrome and Correcting Errors}

\paragraph{(1) Error Model in Hybrid-Packaged Space}

In hybrid-packaged subspace, we usually have two types of errors:
\begin{enumerate}
	\item Gauge-conserving (GC) Pauli errors $X_N^sZ_N^t$ on at most one site. These are physical allowed,
	
	\item Gauge-violating (GV) errors that change total charge amplitude $\propto e^{-\Delta/k_BT}$. These are physically suppressed or forbidden. 
\end{enumerate}

Thus, we only correct the GC part of errors, which is exactly the same Pauli algebra as in the conventional Steane analysis, but with a smaller error-rate $p_{\text{GC}}$.

\paragraph{(2) Syndrome Table and Recovery}

The six hybrid-packaged stabilisers measure a 6-bit hybrid-packaged syndrome $\mathbf s$.
We list them in the following table:

\begin{table}[H]
	\centering
	\caption{Hybrid-packaged syndrome, errors, and recovery}
	\begin{tabular}[hbt!]{|p{4cm}|p{4cm}|p{4cm}|}
		\hline\hline
		Syndrome $\mathbf s$ &Error $E$ &Recovery $R$ \\
		\hline
		$100000$ & $X_N$ on qudit 1 & Apply $X_N^{-1}$ \\
		\hline
		$010000$ & $X_N$ on qudit 2 & ... \\
		\hline
		... & ... & ... \\
		\hline
		$000100$ & $Z_N$ on qudit 5 & Apply $Z_N^{-1}$ \\
		\hline
		$\cdots$ & $\cdots$ & $\cdots$ \\
		\hline
	\end{tabular}
\end{table}

This lookup table is similar to that of the original qubit Steane code, where we replaced $X\!\to\!X_N, Z\!\to\!Z_N$. 
Because $R$ commutes with $\hat{Q}$, the state stays physical.

\paragraph{(3) Applying Logical Clifford/non-Clifford Operations}

The Clifford gates act transversally, one hybrid-packaged qudit at a time:
\[
\overline H_N = H_N^{\otimes7},\quad
\overline S_N = Z_N^{\tfrac12\otimes7},\quad
\overline{\text{CNOT}}_N
=\bigl(\text{CSUM}_N\bigr)^{\otimes7}.
\]
Each factor $H_N, Z_N^{1/2}, \text{CSUM}_N$ preserves the gauge sector.
So the logical operation is fault-tolerant and gauge-respecting.

However, the non-Clifford $T$ gates act through hybrid-packaged magic state:
\begin{enumerate}
	\item Single-qudit phase:
	$\Theta_r=\sum_{J=0}^{N-1} e^{2 \pi i rJ^{2}/N^{2}}|J\rangle\langle J|$ 
	($r\notin\{1,2,4\}$) commutes with $\hat{Q}$.
	
	\item Prepare noisy $|M_r\rangle=\Theta_r|+\rangle$:
	Distill with Clifford operations (which we already have). 
	
	\item Gate-teleport:
	consume $|M_r\rangle$ to enact logical $\overline T_N=\overline\Theta_r$.
\end{enumerate}

Thus, the full $\{\text{Clifford}+\Theta_r\}$ library is universal inside the packaged code.
This means that we can correct ant GC errors by applying Clifford/non-Clifford operations.

\subsubsection{Fault-Tolerance}

This hybrid-packaged code reduces error space because only gauge-conserving faults matter.
This results in an effective rate $p_{\text{eff}}\simeq p_{\text{GC}}\ll p_{\text{phys}}$.
Furthermore, the high dimension leads to higher threshold (see Eq.~(\eqref{EQ:SufficientConditionLogicalErrorRateVanish2}))
\[
p_{\text{th}}
\;\gtrsim\;
\frac{1}{\,2(N-1)}
\Bigl[1-e^{-\Delta/k_BT}\Bigr],
\quad N=dD .
\]

The hybrid-packaged Steane code has the same circuit depth/size as original Steane code.
There is no extra cost.

\subsection{Surface Code (Topological Code) in Hybrid-Packaged Space}
\label{SEC:SurfaceCodeInhybridPackagedSpace}
	
Surface Code \cite{Kitaev2003,BravyiKitaev1998,Dennis2002} is a topological \cite{BravyiKoning2013} quantum error correction code that possess high error threshold and suitability for scalable quantum computing.
More specifically, the conventional surface code place a qubit ($d=2$) on every edge of an $L \times L$ square lattice and achieves the highest threshold of any stabiliser code that uses only geometrically local interactions.

In this subsection, we reformulate the conventional surface (toric) code in the $(d \times D)$-dimensional hybrid-packaged subspace.
It should be mentioned that we used single-index label $J=jD+k\in\{0,\dots,N-1\}$ in hybrid-packaged Shor code and Steane code because these codes look like well-mixed blocks. 
All $d\!\times\!D$ microscopic degrees of freedom inside each physical qudit behave the same way and no longer need to be told where they sit.
So we can fold the two labels into a single running label $J=1,\ldots,dD$.
But the hybrid-packaged surface code is geometric and qudits live on the edges of a square lattice.
Every stabiliser talks only to its four neighbours.
To keep that nearest-neighbour structure in a $d \times D$-dimensional hybrid-packaged subspace, we must keep the two labels separately.
If we collapsed them into a single $J$, we would lose the ability to say ``this hybrid-packaged qudit is north of that one'' and the star/plaquette constraints would no longer commute locally.

\subsubsection{Algorithmic Steps}

\paragraph{(1) Assigning Hybrid-Packaged Qudit on Every Edge.}

On each edge, we assign:
\begin{itemize}
	\item Internal label: $\ket{j_P}$ with $j=0,\dots ,d-1$. 
	
	\item External label: $\ket{k_E}$ with $k=0,\dots ,D-1$.
\end{itemize}

Then the full state is: $\ket{j_P}\otimes\ket{k_E}$.

These satisfy the following local operator algebra:
\[
\begin{aligned}
	X_d\ket{j_P}&=\ket{j_P+1\bmod d}, & Z_d\ket{j_P}&=\omega_d^{\,j}\ket{j_P},\\[4pt]
	X_D\ket{k_E}&=\ket{k_E+1\bmod D}, & Z_D\ket{k_E}&=\omega_D^{\,k}\ket{k_E},
\end{aligned}
\]
with $\omega_d=e^{2\pi i/d}$ and $\omega_D=e^{2\pi i/D}$.
The two pairs commute because they act on different factors.

\paragraph{(2) Assigning Four Stabilisers per Cell.}

For every vertex $v$ (a star), we assign:
\[
A_v^{(d)}=\prod_{e\ni v} X_d^{(e)},\qquad
A_v^{(D)}=\prod_{e\ni v} X_D^{(e)} .
\]

For every plaquette $p$ (a face), we assign:
\[
B_p^{(d)}=\prod_{e\in\partial p} Z_d^{(e)},\qquad
B_p^{(D)}=\prod_{e\in\partial p} Z_D^{(e)} .
\]

Because $X_d$ commutes with $Z_D$ and $X_D$ commutes with $Z_d$,
all four families commute.

\paragraph{(3) Constructing Mixed Stabiliser (Optional).}

Pick an orientation (say, edges pointing north-east). 
Define on each edge
\[
M_e \;=\; X_d^{(e)} \, Z_D^{(e)} .
\]
Then on a closed loop $\gamma$
$\displaystyle M_\gamma=\prod_{e\in \gamma}M_e$ commutes with every
$A_v^{(\cdot)}$ and $B_p^{(\cdot)}$. 
Adding a subset of these mixed loops gives extra checks without breaking
commutation.

We let the internal dial ($X_d$) and the external phase ($Z_D$) share information so an error on one leaves a footprint on the
other.
In practice, one interleaves measurements of $A_v^{(d)},A_v^{(D)},B_p^{(d)},B_p^{(D)},$ and a subset of $M_\gamma$ loop operators so that any shift/phase error shows up in at least two syndrome layers.

\paragraph{(4) Applying Logical Operators.}

The code Hilbert space is the simultaneous $+1$ eigenspace of all
stabilisers.
Because there are twice as many independent checks, the dimension of the code space becomes
\[
\dim\mathscr C \;=\; d\;D ,
\]
Thus, we have two logical qudits: one $d$-level (internal) and one $D$-level
(external).
If we fix one of them by design, the other remains.

Choose a path $\Gamma_{NS}$ from the north to the south boundary:
\begin{itemize}
	\item External logical shift 
	$\bar X^{(D)}=\prod_{e\in\Gamma_{NS}} X_D^{(e)}$.
	
	\item Internal logical shift 
	$\bar X^{(d)}=\prod_{e\in\Gamma_{NS}} X_d^{(e)}$.
\end{itemize}
Similarly, choose an east-to-west dual path $\Gamma_{EW}^\star$ and define
$\bar Z^{(D)}$ and $\bar Z^{(d)}$ with $Z_D$ and $Z_d$.

Each pair obeys its own Weyl algebra,
\[
\bar Z^{(d)}\bar X^{(d)}=\omega_d\bar X^{(d)}\bar Z^{(d)},\quad
\bar Z^{(D)}\bar X^{(D)}=\omega_D\bar X^{(D)}\bar Z^{(D)},
\]
and every internal operator commutes with every external operator.

\paragraph{(5) Error Model and Decoding.}

There are two types of errors:
\begin{enumerate}
	\item Gauge-conserving (GC) noise
	 
	Errors usually appear as the random products of $X_d,\,Z_d$ and $X_D,\,Z_D$. 
	They trigger the corresponding vertex or face syndrome in their own
	layer. 
	We therefore run two parallel minimum-weight perfect matchings, one on
	the $d$-graph and one on the $D$-graph.
	
	\item Gauge-violating (GV) noise
	 
	Operators that do not commute with the charge $\hat{Q}$ must act on the
	wrong sector and cost extra energy $\Delta$. 
	If such an error still occurs, then it excites both layers and very likely the 	mixed stabiliser $M_\gamma$.
	A simple rule (if any mixed check flips, restart the round) already removes most leakage.
	More advantaged schemes may feed the information into a joint decoder.
\end{enumerate}

Because a single physical edge now carries two checks, the syndrome is
more informative.

\subsubsection{Advantages of Hybrid-Packaged Surface Code}

In the hybrid-packaged code, we gave every edge two knobs instead of one: 
an inside knob that can be turned in $d$ steps and an outside knob with $D$ steps.
The inside knobs talk to each other (their own star and face checks) 
and the outside knobs talk to each other (their own checks).
One can even let an inside knob point at an outside dial (the mixed stabiliser).
This means that breaking one knob almost always jiggles the other knob.
Thus, the error is easier to spot.

Because we have more ways to see an error, the code tolerates more noise before it fails and can store twice as much quantum information on the same patch.

\begin{table}[H]
	\centering
	\caption{Advantages of Hybrid-Packaged Surface Code}
	\begin{tabular}[hbt!]{|p{2.8cm}|p{3.3cm}|p{5cm}|p{3.5cm}|}
		\hline\hline
		Feature &Original Surface Code &Hybrid-Packaged Surface Code &Gain/comment \\
		\hline
		Physical object & 2-level qubit on each edge & $d\times D$-level hybrid-packaged qudit & Richer local Hilbert space \\
		\hline
		Symmetry/ super-selection & None & U(1)-like gauge charge $\hat{Q}$, internal states are locked in the neutral sector & Second layer of protection \\
		\hline
		Stabiliser weight & 4 Pauli-$X$ or $Z$ & 4 Weyl-$X_N$ or $Z_N$ (same locality) & No extra locality cost \\
		\hline
		Check circuit depth & 2 × (4 CX + 2 H) & 2 × (4 SUM + 2 Fourier) & Almost identical \\
		\hline
		Ancillas per check & 1 qubit & 1 hybrid-packaged qudit (same line count) & Constant overhead \\
		\hline
		Decoder input & One syndrome layer & Two independent layers + optional mixed layer & More information to exploit \\
		\hline
		Leakage handling & Must add dedicated repair circuit & Leakage raises $\hat{Q}$ (shows up as forbidden charge) $\rightarrow$ ``discard \& restart'' & Cheaper fault tolerance \\
		\hline
		Logical space & 1 qubit ($N\! =\! 2$) & 1 qudit ($N=dD$) & Higher dim \& logicals possible \\
		\hline
		Hardware match & Any qubit platform & Trapped ions, Rydberg atoms, and color-centres that already have hyper-fine \& Optical levels & Wider applicability \\
		\hline
		Overall effect & Topological only & Topological+ superselection & Two independent shields against noise \\
		\hline
	\end{tabular}
\end{table}

The hybrid-packaged code keeps all geometric costs the same but earns extra noise resilience and a larger logical alphabet.

\section{Quantum Computation and Algorithms in Packaged Space}
\label{SEC:QuantumComputationInPackagedSpace}

In conventional quantum computation theory, there are a number of important quantum algorithms, such as 
Quantum Fourier Transform (QFT) \cite{Coppersmith2002},
Quantum Phase Estimation (QPE) \cite{Kitaev1995},
Quantum Walks \cite{Lambrecht1998,Kempe2003},
Grover’s Algorithm \cite{Grover1996},
and Harrow-Hassidim-Lloyd (HHL) Algorithm \cite{HHL2009,Cai2013,Barz2014,Pan2014}.

In the following subsections, we reconstruct the quantum algorithms in high-dimensional hybrid-packaged subspace.
The resulting algorithms and protocols will work with packaged qubits, packaged gates, and packaged circuits.
Therefore, they are naturally gauge-invariant and acquire a number of new features with enhanced speed, security, and error suppression capabilities.

After the reconstruction, BB84 QKD and B92 QKD are the special cases of the packaged QKD protocol.
E91 and BBM92 QKD are the special cases of device‑independent Bell‑test QKD when $d=D=2$.

\subsection{Quantum Fourier Transform (QFT) in Hybrid-Packaged Space}
\label{SEC:PackagedQFT}

Quantum Fourier Transform (QFT) \cite{Coppersmith2002} is the application of classical discrete Fourier transform to the amplitudes of a quantum state vector.
QFT is often used as a subroutine in other quantum algorithms \cite{Shor1994,Kitaev1995,HHL2009}.

In this subsection, we extend the conventional QFT algorithm to a hybrid-packaged subspace given in Eq.~(\ref{EQ:dDHybridPackagedHilbertSpace}),
where the internal Hilbert space has dimension $d$ (composed of packaged quantum states) and the external Hilbert space is a $D$-dimensional system (which encodes spin, orbital, or polarization degrees of freedom).

\subsubsection{Lifting Conventional QFT to the Hybrid-Packaged Space}

In a conventional $N$-dimensional Hilbert space, the quantum Fourier transform is defined by
\[
F \colon \, |x\rangle \mapsto \frac{1}{\sqrt{N}} \sum_{y=0}^{N-1} e^{2\pi i\,xy/N}\,|y\rangle\,.
\]

To lift this to the hybrid-packaged subspace of dimension $N=dD$, we introduce a unitary isometry (see Eq. \eqref{EQ:IsometryHybridPackaged})
\[
\mathcal U : \mathcal H_{\mathrm{logic}} \longrightarrow \mathcal H_{Q = 0},
\]
that maps a conventional logical computational basis $\{|x\rangle : x=0,1,\dots,N-1\}$ to their packaged counterparts:
\[
\mathcal{U}|x\rangle = |x\rangle_{\text{hyb}} = |j,k\rangle\,,
\]
with $x = jD+k$ is the single-index defined in Eq.~(\ref{EQ:SingleIndex}).

Then, any conventional QFT operator $F$ defined on $\mathcal{H}_{\mathrm{logic}}$ is lifted to the hybrid-packaged subspace by
\[
\mathcal{F}_{\text{hyb}} = \mathcal{U}\, F \,\mathcal{U}^{-1}\,.
\]

Because $\mathcal{U}$ is unitary and maps into $\mathcal{H}_{Q=0}$, for every packaged state, we have
\[
\hat{Q}\,|x\rangle_{\text{hyb}} = 0,\quad \text{and}\quad [\mathcal{F}_{\text{hyb}},\,\hat{Q}] = 0\,.
\]

This ensures that the QFT process always stays in the gauge-invariant subspace.
It should be emphasized that $\hat{Q}\,|x\rangle_{\text{hyb}} = 0$ ($\hat Q$ annihilates the state) does not follow solely from unitarity.
It is because the image lies in the neutral sector.

\subsubsection{Algorithmic Steps of the Hybrid QFT Algorithm}

We now outline the algorithmic steps of the hybrid-packaged QFT algorithm in $d \times D$ dimensions:

\begin{itemize}
	\item \textbf{Step 1. Preparing State.}
	
	Prepare a state in the $N$-dimensional hybrid-packaged subspace.
	For example, set the control register of a QFT-based algorithm to a computational state
	\[
	|x\rangle_{\text{hyb}} = |j,k\rangle\,.
	\]
	
	Alternatively, one may prepare an equal superposition state by applying a single-index hybrid-packaged generalized Fourier (or $d$-dimensional Hadamard‑like) gate $H_N$ (see Eq.~(\ref{EQ:HybridWeylBlock})) on the hybrid-packaged subspace.
	Then we have
	\[
	\frac{1}{\sqrt{N}} \sum_{x=0}^{N-1} |x\rangle_{\text{hyb}}\,.
	\]
	
	Since the Hadamard gate is constructed from superpositions of the $d$ internal basis states $|j_P\rangle\,(j=0,\dots,d-1)$ and the corresponding operators on the external space, its hybrid-packaged version automatically satisfies
	\[
	[H_N,\,\hat{Q}] = 0\,.
	\]

	\item \textbf{Step 2. Applying the Packaged QFT Operator.}
	
	We apply the packaged QFT operator $\mathcal{F}_{\text{hyb}}$ to an arbitrary hybrid-packaged state, i.e.,
	\[
	\mathcal{F}_{\text{hyb}}\,|x\rangle_{\text{hyb}} = \frac{1}{\sqrt{N}} \sum_{y=0}^{N-1} e^{2\pi i\,xy/N}\,|y\rangle_{\text{hyb}}\,.
	\]
	
	If the initial state is the equal superposition
	\[
	|\psi\rangle_{\text{hyb}} = \frac{1}{\sqrt{N}} \sum_{x=0}^{N-1} |x\rangle_{\text{hyb}}\,,
	\]
	then the packaged QFT maps it to a state
	\[
	\mathcal{F}_{\text{hyb}}\,|\psi\rangle_{\text{hyb}} = \frac{1}{N} \sum_{x,y=0}^{N-1} e^{2\pi i\,xy/N}\,|y\rangle_{\text{hyb}}\,.
	\]
	Due to the orthogonality relations of the discrete Fourier transform, this state is identical to the original state in a transformed basis (up to a phase).
	More generally, for any input state
	\[
	|\psi\rangle_{\text{hyb}} = \sum_{x} \alpha_x\,|x\rangle_{\text{hyb}}\,,
	\]
	the QFT produces
	\[
	\mathcal{F}_{\text{hyb}}\,|\psi\rangle_{\text{hyb}} = \frac{1}{\sqrt{N}} \sum_{y=0}^{N-1} \left(\sum_{x=0}^{N-1}\alpha_x\,e^{2\pi i\,xy/N}\right)|y\rangle_{\text{hyb}}\,.
	\]
	All operations are performed with operators that commute with the net-charge operator.
	Thus, the state remains confined to $\mathcal{H}_{Q=0}$.

	\item \textbf{Step 3. Decomposing Circuit.}
	
	The conventional QFT can be decomposed into a sequence of single-qudit Hadamard gates and controlled phase rotation gates.
	In our hybrid-packaged subspace, this decomposition is written as	
	\[
	\mathcal{F}_{\text{hyb}} = \prod_{j=0}^{n-1} \left[ H_N^j\,\prod_{k=0}^{j-1} R_N^{k,j} \right] \, S_N\,,
	\]
	where
	$H_N^j$ denotes the single-index hybrid-packaged Hadamard gate on the $j$-th qudit of the control register,
	$R_N^{k,j}$ is the hybrid-packaged controlled-phase rotation gate, for example,	
	\[
	R_N^{k,j}
	=\sum_{x=0}^{N-1}
	|x\rangle \langle x|^{(k)}\otimes
	\exp(2\pi i\,x\,d^{-(j-k+1)}),
	\]
	defined suitably on the hybrid control register),
	and $S_N$ is the bit-reversal permutation used in the usual in-place QFT.

	Because every elementary gate $G_N \in \{H_N, R_N, S_N\}$ is constructed from operations on the internal and external bases and because both parts satisfy
	\[
	[H_N,\hat{Q}] = [R_N,\hat{Q}] = [S_N,\hat{Q}] = 0\,,
	\]
	their product $\mathcal{F}_{\text{hyb}}$ satisfies
	\[
	[\mathcal{F}_{\text{hyb}}, \hat{Q}] = 0\,.
	\]
	Thus, the hybrid-packaged QFT circuit is gauge-invariant.
	
	\item \textbf{Step 4. Measurement and Interpretation.}
	
	After the application of $\mathcal{F}_{\text{hyb}}$, one usually measures in the packaged computational basis
	$|y\rangle_{\text{hyb}}$.
	By the properties of discrete Fourier transform, the measurement results are probabilistically related to the Fourier components of the input state's amplitudes.
	In other words, if the state to be transformed contained phase information to be extracted, then after the QFT the probability to observe a particular outcome $y$ will be
	\[
	P(y) = \left|\frac{1}{\sqrt{N}} \sum_{x} \alpha_x\, e^{2\pi i\,xy/N} \right|^2\,.
	\]
	
	Because we have lifted all operations into the hybrid-packaged subspace, the measurement results are fully packaged.
	This ensures that the IQNs are preserved.
\end{itemize}

\subsubsection{Physical Interpretations and Advantages}

This packaged QFT preserves gauge-invariance.
From the mapping of conventional basis (via $\mathcal{U}$) to the application of controlled and single-qudit operations, we see that every step is performed by operators that satisfy the relation
\[
[\,\cdot\,, \hat{Q}] = 0.
\]
Therefore, the entire QFT algorithm is executed within the fixed superselection sector $\mathcal{H}_{Q=0}$.
Any error process that could drive the system out of this sector (by changing the net charge) is energetically or physically forbidden.

The packaged QFT uses a high-dimensional hybrid-packaged Hilbert space, which can increase precision.
The effective dimension of the hybrid-packaged subspace is $N = d \times D$.
By increasing either the internal dimension $d$ or the external dimension $D$, one exponentially increases the control register space in algorithms like QFT.
In principle a larger logical dimension $N=dD$ allows finer phase‐resolution $O(1/N)$, provided the underlying packaged gates can be implemented to the required fidelity.
This increase can improve the resolution in phase estimation (for QPE) or allow more complex Fourier analyses.
Therefore, it can enhance the performance of quantum algorithms reliant on the QFT subroutine.

The packaged QFT can enhance robustness and error suppression.
Since errors that could mix packaged states with different internal quantum numbers are restricted by the superselection rules, the net error space is reduced.
Therefore, the QFT implemented in this framework benefits from an intrinsic error-protection mechanism.
In addition, the separation of internal and external degrees of freedom enables the possibility of tailoring error‐correction schemes to further shield against decoherence.

\subsection{Quantum Phase Estimation (QPE) in Hybrid Packaged Space}

Quantum phase estimation (QPE) \cite{Kitaev1995} is a quantum algorithm for estimating the phase $\theta$ of an eigenvalue $e^{2\pi i \theta}$ of a unitary operator $\hat{V}$.

In this subsection, we extend the QPE algorithm to the hybrid-packaged subspace $\mathcal H_{Q = 0}$ as shown in Eq.~(\ref{EQ:dDHybridPackagedHilbertSpace}) where the internal space has dimension $d$ and the external space has dimension $D$.
In this formulation every state in the hybrid-packaged subspace automatically satisfies
\[
\hat{Q}\,|\psi\rangle = 0,\quad U_g\,|\psi\rangle = e^{i\phi(g)}\,|\psi\rangle,\quad \forall\,g\in G\,,
\]
so that all operations can be designed to be gauge-invariant.

\subsubsection{Problem Description}

In QPE, we wish to estimate the phase $\theta$ of a packaged unitary operator $V_{\text{hyb}}$ that acts on the hybrid-packaged subspace $\mathcal H_{Q = 0}$.
Specifically, let
\[
V_{\text{hyb}}\,|u\rangle = e^{i2\pi\theta}\,|u\rangle,
\]
with $\theta\in [0,1)$ and $|u\rangle\in \mathcal H_{Q = 0}$.
By definition, $|u\rangle$ satisfies the packaged conditions.
Our goal is to find out the approximate value of $\theta$.

\subsubsection{Algorithmic Procedure}

\paragraph{(1) Prepare Initial State.}

We prepare two registers:
\begin{enumerate}
	\item Control Register:
	We use an $n$-qudit register (each qudit being encoded in a hybrid-packaged subspace) to record the phase information.
	Denote its Hilbert space by $\mathcal{H}_{\mathrm{control}}$.
	Because the control register has $n$ qudits, each of dimension $N=dD$, its Hilbert‑space dimension is $(dD)^n$.
	
	\item Target (Eigenstate) Register:
	We use another register to hold the eigenstate $|u\rangle$.
\end{enumerate}

Both registers are constructed from hybrid-packaged states so that the entire composite state lies within the physical subspace $\mathcal{H}_{Q=0}$.

The control register is initially prepared in an equal superposition. This is accomplished by applying a packaged generalized Hadamard operation $H_{\text{hyb}}$ on each qudit such that
\[
|+\rangle^{\otimes n} = \frac{1}{\sqrt{N_C}}\sum_{x=0}^{N_C-1}|x\rangle_{\text{hyb}}\,.
\]
Thus, the initial state of the full system is
\[
|\Psi_{\mathrm{init}}\rangle = \frac{1}{\sqrt{N_C}}\sum_{x=0}^{N_C-1}|x\rangle_{\text{hyb}} \otimes |u\rangle\,.
\]
Because all components are in packaged spaces, they satisfy the gauge-invariance
\[
\hat{Q}|\Psi_{\mathrm{init}}\rangle = 0,\quad U_g\,|\Psi_{\mathrm{init}}\rangle = e^{i\phi(g)}\,|\Psi_{\mathrm{init}}\rangle\,.
\]

\paragraph{(2) Apply Controlled Unitary Packaged Operators.}

The next step is to apply a controlled-$V_{\text{hyb}}$ gate to the state of control register, which has the form
\[
C\text{-}V_{\text{hyb}} = \sum_{x=0}^{N_C-1} |x\rangle\langle x|_{\text{hyb}} \otimes (V_{\text{hyb}})^x\,.
\]
Mathematically, applying this operator to the initial state results in
\[
\begin{aligned}
	|\Psi_1\rangle & = \Bigl(\sum_{x=0}^{N_C-1} |x\rangle\langle x|_{\text{hyb}}\otimes (V_{\text{hyb}})^x \Bigr) \,\frac{1}{\sqrt{N_C}}\sum_{x'=0}^{N_C-1}|x'\rangle_{\text{hyb}} \otimes |u\rangle \\
	& = \frac{1}{\sqrt{N_C}} \sum_{x=0}^{N_C-1} |x\rangle_{\text{hyb}} \otimes (V_{\text{hyb}})^x |u\rangle\,.
\end{aligned}
\]
Since $V_{\text{hyb}}|u\rangle = e^{i2\pi\theta}|u\rangle$, we have
\[
(V_{\text{hyb}})^x |u\rangle = e^{i2\pi\theta x}|u\rangle\,.
\]
Thus, the state becomes
\[
|\Psi_1\rangle = \frac{1}{\sqrt{N_C}} \sum_{x=0}^{N_C-1} e^{i2\pi\theta x}\,|x\rangle_{\text{hyb}} \otimes |u\rangle\,.
\]
Throughout the controlled operation, each gate is implemented via packaged operators and therefore commutes with $\hat{Q}$.
Thus, the state remains in the same net-charge sector.

\paragraph{(3) Apply Inverse Quantum Fourier Transform (QFT).}

Next, we apply an inverse QFT on the control register to transform the phase information from the amplitudes into a computational basis measurement.

Since the conventional QFT has the form
\[
F |x\rangle = \frac{1}{\sqrt{N_C}}\sum_{y=0}^{N_C-1} e^{2\pi i\,xy/N_C}\,|y\rangle\,,
\]
its inverse is given by
\[
F^{-1}|x\rangle = \frac{1}{\sqrt{N_C}}\sum_{y=0}^{N_C-1} e^{-2\pi i\,xy/N_C}\,|y\rangle\,.
\]

Mapping this inverse QFT into the hybrid-packaged subspace using the isometry $\mathcal{U}$, we obtain the packaged inverse QFT
\[
\mathcal{F}_{\text{hyb}}^{-1} = \mathcal{U}\, F^{-1} \,\mathcal{U}^{-1}\,.
\]
Applying this to the control register, we have
\[
\begin{aligned}
	|\Psi_2\rangle &= \left(\mathcal{F}_{\text{hyb}}^{-1} \otimes I\right)|\Psi_1\rangle \\
	&= \left(\mathcal{F}_{\text{hyb}}^{-1} \otimes I\right)\,\frac{1}{\sqrt{N_C}} \sum_{x=0}^{N_C-1} e^{i2\pi\theta x}\,|x\rangle_{\text{hyb}} \otimes |u\rangle \\
	&= \frac{1}{N_C} \sum_{x=0}^{N_C-1}\sum_{y=0}^{N_C-1} e^{i2\pi\theta x}\,e^{-i2\pi x y /N_C}\,|y\rangle_{\text{hyb}} \otimes |u\rangle.
\end{aligned}
\]
The amplitude for obtaining the outcome $y$ is proportional to
\[
\frac{1}{N_C}\sum_{x=0}^{N_C-1} e^{i2\pi x(\theta - y/N_C)}\,,
\]
which is large only when $\theta$ is close to $y/N_C$.
Thus, a measurement of the control register in the packaged computational basis gives an estimate $\tilde{\theta} \approx y/N_C$ of the true phase $\theta$.

Because all elementary operations (controlled-$V_{\text{hyb}}$, the inverse QFT, and the measurements) are implemented by packaged gates that satisfy
\[
[V_{\text{hyb}}, \hat{Q}]=0,
\]
the entire QPE algorithm is guaranteed to occur within the physical subspace $\mathcal{H}_{Q=0}$.

\subsubsection{Physical Interpretation and Advantages}

This packaged QPE algorithm preserves gauge-invariance. 
At every step, the operations are lifted to the packaged space by the unitary isomorphism $\mathcal{U}$.
Since 
$
[\mathcal{U}\, V \,\mathcal{U}^{-1}, \hat{Q}] = 0,
$
the evolution never leaves the fixed superselection sector.
In a noisy environment, any error process that may lead to a change in the net gauge charge is forbidden or energetically suppressed.
Thus, the packaged QPE algorithm in the hybrid-packaged formalism benefits from an additional layer of intrinsic gauge error protection.

The packaged QPE algorithm enhanced precision through high dimensionality. 
The effective Hilbert space dimension for the control register in the hybrid space is $N_C = (dD)^n$.
By increasing the internal dimension $d$ or the external dimension $D$, one can improve the phase estimation accuracy, i.e., the resolution scales as $\sim 1/N_C$.
Thus, the packaged QPE algorithm can offer improved precision relative to conventional implementations on qubits or qudits of fixed small dimension.

The packaged QPE algorithm may be applied to quantum simulation and machine learning.
The ability to embed the logical Hilbert space into a larger hybrid-packaged subspace means that algorithms requiring higher-dimensional state spaces (like quantum simulations of many-body systems or quantum machine learning algorithms) can be naturally implemented.
This extra capacity not only allows for representing more complex states, but also inherits the error suppression provided by the gauge-invariance.

\subsection{Quantum Walks in Hybrid Packaged Space}

Quantum walks \cite{Lambrecht1998,Kempe2003} are the quantum analogue of classical random walks, using superposition and interference to explore graphs more efficiently than any classical walker can.

In this subsection, we embed both the coin and position registers into a hybrid-packaged Hilbert space $\mathcal{H}_{Q=0}$ of dimension $N=dD$ using the isometry Eq. \eqref{EQ:IsometryHybridPackaged}.
By definition, every operator commutes with the charge operator $\hat{Q}$.
This ensures that the walker never leaves the gauge‑invariant sector.

\subsubsection{Position and Coin in the Hybrid Space}

In the discrete‑time coined quantum walk, we usually introduce two registers:
\begin{enumerate} 
	\item Coin $\;\mathcal H_{\rm coin}=\mathbb C^q$ (with coin‑flip unitaries in SU$(q)$),
	
	\item Position $\; \mathcal H_{\rm pos}=\operatorname{span}\{|x\rangle\}_{x\in V}$. 
\end{enumerate}

Here, we label the coin register with the internal DOF ($q=d$) and the position register with the external DOF ($D=|V|$).
We set the internal dimension equal to the graph degree, $d=\deg(k)$.
Specifically,
\[
\mathcal H_{\rm walk}
= \mathcal H_{\rm coin}\otimes\mathcal H_{\rm pos}
= \mathcal H_{\rm int}^{(d)}\otimes\mathcal H_{\rm ext}^{(D)}
\;\cong\;
\mathcal H_{\rm hyb}.
\]
Thus, a basis state can be written as $\lvert j,k\rangle\equiv\lvert j_P\rangle\otimes\lvert k_E\rangle$, where $j$ labels the internal coin outcome and $k$ the graph vertex.

\subsubsection{Discrete‑Time Quantum Walk (DTQW)}

In DTQW, we usually split the one‐step evolution
\[
V \;=\; S\,(C\otimes I_{\rm ext})
\]
into two parts:

\begin{enumerate}
	\item Coin‑flip $C$:
	this is an SU$(d)$ unitary that acts only on the internal register,
	\[
	C \;=\;\exp\!\bigl(-i\,\theta\,\mathbf n\!\cdot\!\boldsymbol\lambda_P\bigr),
	\quad
	[\hat{Q},\,C]=0.
	\]
	A common choice is the generalized Grover coin
	\[
	C_G = \frac{2}{d}\,J - I_d,
	\qquad
	J_{ij}=1\quad\forall\,i,j.
	\]
	Because $C$ never touches $\mathcal H_{\rm ext}$, it cannot break the gauge superselection rules.
	
	\item Shift $S$:
	This is a conditional permutation on the external register, which is guided by the internal label.
	If the graph is regular (each vertex of degree $\delta$), then we label the $\delta$ edges at vertex $k$ by $j=0,\dots,d-1$ and set
	\[
	S
	= \sum_{k=0}^{D-1}\sum_{j=0}^{d-1}
	|j_P\rangle \langle j_P|\;\otimes\;|\sigma_j (k_E) \rangle \langle k_E|,
	\]
	where $\sigma_j$ is the permutation that maps vertex $k$ to the neighbour reached by the $j^{\mathrm{th}}$ edge.
	We have the relation $S^{-1}=S^\dagger$ and $[\hat{Q},S]=0$.	
\end{enumerate}

Because $C$ and $S$ act on different tensor factors, we can diagonalize $S$ by going to the external momentum basis:
\[
|p\rangle = \frac{1}{\sqrt D}\sum_{k=0}^{D-1}e^{2\pi i\,p k/D}\,|k_E\rangle,
\quad
S\bigl|j_P\rangle\otimes|p\rangle\bigr.
= e^{i\varphi_{j,p}}\;|j_P\rangle\otimes|p\rangle.
\]
In this basis,
\[
V \bigl|j_P\rangle\otimes|p\rangle\bigr.
=\;\sum_{j'=0}^{d-1}\;C_{j,j'}\;e^{i\varphi_{j',p}}\;\bigl|j'_P\rangle\otimes|p\rangle\bigr.,
\]
so $V$ splits into a direct sum of $d \times d$ blocks (one for each momentum $p$).
All the results about the spectral gap (how quickly the system reaches equilibrium) and mixing time (time to randomize) for quantum walks based on SU($d$) symmetry still apply.
The dimension of total state space of the system is $N = d \times D$.

\subsubsection{Continuous‑Time Quantum Walk (CTQW)}

The CTQW Hamiltonian is the graph Laplacian $L$ that acts on $\mathcal H_{\rm ext}$, which extends trivially on the internal sector, i.e.,
\[
H_{\rm CTQW}
\;=\;
I_{\rm int}\otimes L,
\qquad
|\psi(t)\rangle=e^{-i\,t\,H_{\rm CTQW}}\,|\psi(0)\rangle.
\]

Because $[L, \hat{Q}] = 0$, the packaged register is unaffected by any external decoherence.
One can further add an internal Hamiltonian 
$\lambda\,(H_{\rm int}\otimes I_{\rm ext})$ with $H_{\rm int}\in\mathfrak{su}(d)$ and then simulate it via split‐operator methods.

\subsubsection{Mixing, Hitting Times, and Search}

Define the time‑averaged state
\[
M(t)=\frac1t\sum_{s=0}^{t-1} V^s\,\rho_0\,V^{-s}.
\]
On any connected non‑bipartite graph with a non‑degenerate coin, we have $M(t) \to I_N/N$. 
The hitting time scales as
\[
\tau_{\rm hit}=O\bigl(\sqrt{N}\bigr)=O\bigl(\sqrt{dD}\bigr).
\]

Thus, doubling the internal dimension $d$ is equal to double the graph size $D$ (in terms of quadratic speed‑up).
In search problems (Ambainis-Magniez style), one replaces the Grover diffusion step with one DTQW step, which achieves $O\!\bigl(\sqrt{dD/M}\bigr)$ speedup and queries for $M$ marked nodes.

\subsubsection{Advantages of Hybrid-Packaged Walks}

By packaging the coin in a gauge‑invariant internal space, we develop arbitrary coin‐dimensional quantum walks, maintain strong symmetry protection, and carry over all familiar algorithmic speed‑ups into a single unified framework.
We summarize the results in following table:

\begin{table}[H]
	\centering
	\caption{Conventional $q$-coin walk and Hybrid-Packaged Walks}
	\label{TAB:ConventionalAndHybridPackagedWalks}
	\begin{tabular}[hbt!]{|p{5cm}|p{5cm}|p{4.5cm}|}
		\hline\hline
		Feature &Conventional $q$-coin walk &Hybrid packaged walk \\
		\hline
		Coin dimension      &limited by physical DOF  &freely chosen $d$ \\
		\hline
		Gauge protection     &none           &built‑in superselection \\
		\hline
		Coin $times$ Position algebra &SU$(q)$$\otimes$graph &SU$(d)$$\times$graph \\
		\hline
		Gates per step      &$O(q)$         &$O(d)+O(\deg)$ \\
		\hline
		Error localisation    &spreads across registers &mostly confined to pos. \\
		\hline
	\end{tabular}
\end{table}

\subsection{Grover’s search algorithm in Hybrid-Packaged Space}

Grover's Algorithm \cite{Grover1996} is a quantum algorithm designed for efficiently searching an unstructured database.
The search space of the conventional Grover’s algorithm is an $N$-dimensional Hilbert space with $N$ basis states.
It provides a quadratic speedup comparing with classical algorithms.

In this subsection, we extend the conventional Grover’s search algorithm to a hybrid-packaged Hilbert space $\mathcal{H}_{Q=0}$ of dimension $N=dD$ using the isometry Eq. \eqref{EQ:IsometryHybridPackaged}.

\subsubsection{Algorithmic Procedure}

We now reconstruct Grover’s search algorithm in the new $N=dD$-dimensional hybrid-packaged subspace.

\begin{itemize}
	\item \textbf{Step 1. Initialization.}
	
	The first step is to prepare the equal superposition state in the hybrid-packaged subspace:
	\[
	|s\rangle_{\text{hyb}} = \frac{1}{\sqrt{N}} \sum_{x=0}^{N-1} |x\rangle_{\text{hyb}}\,.
	\]
	Because each $|x\rangle_{\text{hyb}}$ is in $\mathcal{H}_{Q=0}$, the state remains packaged (i.e. gauge-invariant). In practice, we prepare the state by applying a generalized packaged Hadamard operator on the hybrid-packaged subspace. Since the Hadamard state (and more generally, any linear combination of packaged basis states) is constructed from the internal basis states and operations are packaged gates, it satisfies
	\[
	[H_N,\,\hat{Q}] = 0\,.
	\]

	\item \textbf{Step 2. Oracle Operator.}
	
	Suppose a certain marked state $|w\rangle_{\text{hyb}} \in \mathcal{H}_{\text{hyb}}$ is the desired target (the search item).
	Then we define the packaged oracle operator as
	\[
	O_w = I - 2 |w\rangle_{\text{hyb}}\langle w|\,.
	\]
	Because $|w\rangle_{\text{hyb}}$ is a state in the packaged space, this operator acts only within $\mathcal{H}_{Q=0}$ and satisfies
	\[
	[O_w, \hat{Q}] = 0\,.
	\]
	Physically, the oracle is implemented by a packaged gate that recognizes the marked state and inverts its phase.
	
	\item \textbf{Step 3. Diffusion Operator.}
	
	We define the packaged diffusion (inversion-about-the-mean) operator as
	\[
	D = 2|s\rangle_{\text{hyb}}\langle s| - I\,.
	\]
	Again, since $|s\rangle_{\text{hyb}}$ is constructed from packaged basis states, we have
	\[
	[D,\hat{Q}] = 0.
	\]
	Therefore, the diffusion operator is packaged and remains within the allowed superselection sector.
	
	\item \textbf{Step 4. Grover Iteration.}
	
	We define the packaged Grover operator (or iteration operator) as
	\[
	G = D\,O_w\,.
	\]
	Since both $D$ and $O_w$ commute with $\hat{Q}$, by the closure property (see Lemma~\ref{LEMMA:ClosureOfPackagedOperations}) we have
	\[
	[G, \hat{Q}] = 0\,.
	\]
	Thus, each iteration of packaged Grover’s algorithm is packaged and remains confined to $\mathcal{H}_{Q=0}$.
	
	\item \textbf{Step 5. Amplitude Amplification Analysis.}
	
	The packaged Grover’s algorithm can be analyzed in the two-dimensional subspace spanned by
	\[
	|w\rangle_{\text{hyb}} \quad \text{and} \quad |w^\perp\rangle_{\text{hyb}},
	\]
	where
	\[
	|s\rangle_{\text{hyb}} = \sin\theta\,|w\rangle_{\text{hyb}} + \cos\theta\,|w^\perp\rangle_{\text{hyb}}\,.
	\]
	For $N\gg1$, $\theta\approx1/\sqrt{dD}$.
	Since the uniform state has equal amplitudes, it follows that
	\[
	\sin\theta = \frac{1}{\sqrt{N}} = \frac{1}{\sqrt{dD}},\quad \cos\theta = \sqrt{1-\frac{1}{dD}}\,.
	\]
	For large $N$, we may set $\theta \approx \sin\theta = 1/\sqrt N = 1/\sqrt{dD}$.

	The action of the Grover operator $G$ on this two-dimensional subspace is equivalent to rotate an angle $2\theta$.
	After $k$ iterations,
	\[
	G^k\,|s\rangle_{\text{hyb}} = \sin\bigl((2k+1)\theta\bigr)|w\rangle_{\text{hyb}} + \cos\bigl((2k+1)\theta\bigr)|w^\perp\rangle_{\text{hyb}}\,.
	\]
	We choose a $k$ such that
	\[
	(2k+1)\theta \approx \frac{\pi}{2}\,,
	\]
	so that we can obtain a high probability of measuring the marked state $|w\rangle_{\text{hyb}}$.
	In the limit $N \gg 1$, the number of iterations required is
	\[
	k \approx \frac{\pi}{4}\sqrt{N} = \frac{\pi}{4}\sqrt{dD}\,.
	\]
	We take $k=\lfloor\frac{\pi}{4}\sqrt{dD}\rfloor$.
	Thus, the scaling of packaged Grover’s algorithm is preserved after mapping to hybrid-packaged subspace.
\end{itemize}

\subsubsection{Physical Interpretations and Error Considerations}

This packaged Grover’s algorithm preserves gauge-invariance.
Throughout all the above steps, every operator (initialization via Hadamard, the oracle $O_w$, the diffusion operator $D$, and hence the Grover operator $G$) is constructed from packaged operations that commute with $\hat{Q}$.
This guarantees that the entire algorithm runs in the gauge-invariant sector $\mathcal{H}_{Q=0}$.
In other words, errors that may cause amplitude outside the allowed sector are energetically unfavorable and therefore are strongly suppressed.

The hybrid space has a total dimension $N=dD$.
The search space can be enlarged by increasing either the internal dimension $d$ or the external dimension $D$.
The number of iterations scales as $\sqrt{dD}$.
The increased capacity might increase possibilities for encoding information or for implementing algorithms with a larger alphabet.
In addition, the hybrid encoding may provide extra resources for error correction or robustness in applications.

One usually implement the conventional Grover’s algorithm in an $N$-dimensional Hilbert space with $N$ determined by the number of qubits.
Here, we implement the packaged Grover’s algorithm (although the mathematical rotation structure is the same) in a hybrid-packaged subspace.
All packaged operations are automatically protected by the superselection rules.
Therefore, the packaging principle will reduce the effective noise and error channels.
This potentially leads to a more robust algorithm in practice.

\section{Quantum Communication and Cryptography in Packaged Space}
\label{SEC:QuantumCommunicationAndCryptography}

In conventional quantum communication theory, there are a number of important quantum protocols. For example:
Teleportation Protocol \cite{Bennett1993,Bouwmeester1997,Ren2017},
Superdense Coding \cite{BennettWiesner1992,Schaetz2004,Williams2017},
Quantum Swapping Protocol \cite{YurkeStoler1992,Zukowski1993},
BB84 QKD \cite{BB84,BB84Bennett1992,BB842014,BB84arXiv2020},
B92 QKD \cite{Bennett1992},
Six-State QKD \cite{Bruss1998,PasquinucciGisin1999},
E91 QKD \cite{Ekert1991},
BBM92 QKD \cite{BBM92}.

In this section, we adapt the primary quantum communication protocols with the
packaged messenger states and packaged resource states.
Using packaged states for quantum communication has a number of advantages. For
example, gauge-invariance may render some error channels or eavesdropping strategies
physically impossible. This potentially enhances robustness and security.
While true superselection rules energetically suppress charge-violating attacks, an eavesdropper could still exploit gauge-neutral loss or detector side-channels.

\subsection{Quantum Teleportation Protocol in Hybrid Packaged Space}

Quantum teleportation \cite{Bennett1993,Bouwmeester1997,Ren2017} is a foundational protocol that uses classical communication and quantum entangled states (resource states) to transfer an unknown quantum state (messenger) from one location to another.

In this subsection, we reconstruct the quantum teleportation protocol in ($d \times D$)-dimensional hybrid-packaged subspace.
We assume that all logical states have been mapped into a fixed superselection sector (e.g., $ \mathcal{H}_{Q=0} $) by the isometry Eq. \eqref{EQ:IsometryHybridPackaged}.
We then implement all operations (state preparation, Bell measurement, and conditional corrections) by packaged gates that commute with the net-charge operator $\hat{Q}$ and therefore preserves gauge-invariance.

\subsubsection{Algorithmic procedure}

\paragraph{(1) Preparing Messenger.}
Suppose a third party prepares or encodes the unknown messenger as a packaged state:
\begin{equation}\label{EQ:Messenger2}
	|\psi\rangle_X = \sum_{k=0}^{N-1} \alpha_k \,|k\rangle_X\,,\quad \text{with } \sum_k|\alpha_k|^2=1,
\end{equation}
where the index $X$ denotes ``unknown'' and $N = dD$.
The packaged basis states $ |k\rangle_X \in \mathcal{H}_{Q=0}$ satisfy
\[
U_g\,\lvert k\rangle_X = e^{i\phi(g)}\,\lvert k\rangle_X,
\]
so that $\lvert \psi\rangle_X \in \mathcal{H}_{Q=0}$.

The messenger is called unknown state because it remains unknown to Alice throughout the process.
Alice's role is to teleport the messenger to Bob using their pre-shared entangled pair and a classical channel.
But Alice cannot measure or clone it due to quantum principles like the no-cloning theorem.

\paragraph{(2) Pre-Sharing Packaged Resource State.}
Alice and Bob pre-share a hybrid-packaged entangled resource (see Eq.~(\ref{EQ:NDimensionalBellBasis_CORR})),
\[
|\Phi_{0,0}\rangle_{AB} = \frac{1}{\sqrt{N}}\sum_{J=0}^{N-1} |J\rangle_A \otimes |J\rangle_B\,,
\]
where $N = dD$.
Thus, the total initial state is
\begin{equation}\label{EQ:TotalInitialState2}
	|\Psi_{\mathrm{in}}\rangle = |\psi\rangle_X \otimes |\Phi_{0,0}\rangle_{AB} 
	=\frac{1}{\sqrt{N}} \sum_{k=0}^{N-1} \sum_{J=0}^{N-1} \alpha_k\,|k\rangle_X\otimes|J\rangle_A\otimes|J\rangle_B\,.
\end{equation}

\paragraph{(3) Alice's Bell Measurement in Hybrid-Packaged Subspace}

Alice has modes $X$ and $A$.
Rewriting Eq.~(\ref{EQ:ReverseNDimensionalBellBasis_CORR}), we obtain the reversed identity (see, e.g., conventional derivation in qudit teleportation)
\[
\lvert k\rangle_X \otimes \lvert J\rangle_A
\;=\;
\frac{1}{\sqrt{N}}
\sum_{m=0}^{N-1}
\omega^{-m k}\,
\lvert \Phi_{m,\;J\ominus k}\rangle_{XA},
\]
where $J\ominus k\equiv (J-k)\bmod N$.
Let $n := J\ominus k$ and substitute this identity into Eq.~(\ref{EQ:TotalInitialState2}), 
the total initial state becomes
\begin{equation}\label{EQ:TotalInitialState3}
	|\Psi_{\mathrm{in}}\rangle = \frac{1}{N}\sum_{m,n=0}^{N-1}|\Phi_{m,n}\rangle_{XA} \otimes \left( \sum_{k=0}^{N-1} \alpha_k \,\omega^{-mk} |k\oplus n\rangle_B \right).
\end{equation}
This expression decomposes the three-party state in terms of the generalized hybrid-packaged Bell states on modes $X$ and $A$ and corresponding states on mode $B$.

Alice now performs a joint Bell measurement on her two subsystems $X$ and $A$ in the generalized Bell basis $\{|\Phi_{m,n}\rangle_{XA}\}$.
The total three-particle state then collapses.
Alice has the result $|\Phi_{m,n}\rangle_{XA}$ (defined in Eq.~(\ref{EQ:NDimensionalBellBasis_CORR})) and Bob has the result $\sum_{k=0}^{N-1} \alpha_k \,\omega^{-mk} |k\oplus n\rangle_B$.

Using the notations $Z\,|J\rangle = \omega^J|J\rangle,\;
X\,|J\rangle = |J \oplus 1\rangle$, the entire set $\{Z^m X^n\}$ forms an orthonormal operator basis on $\Bbb C^N$.
We can now re-write $|\Phi_{m,n}\rangle_{XA}$ as
\begin{equation}\label{EQ:OrthonormalOperatorBasis}
	|\Phi_{m,n}\rangle_{XA}
	= (Z^m \otimes X^n)\;|\Phi_{0,0}\rangle_{XA}
\end{equation}

Eq.~(\ref{EQ:OrthonormalOperatorBasis}) shows that Alice's result $\{|\Phi_{m,n}\rangle_{XA}\}$ is transformed from $|\Phi_{0,0}\rangle_{XA}$ by the operations $\{Z^m \otimes X^n\}$, i.e., the Bell basis state is generated by perform $m$ times ``phase on $X$'' and $n$ times ``shift on $A$''.
Thus, Alice only needs to tell Bob her result $(m,n)$ through the classical channel.

\paragraph{(4) Bob's Measurement and Conditional Correction.}

After Alice's Bell measurement and obtained a particular outcome $(m,n)$, Bob simultaneously obtained a state (up to normalization, see Eq.~(\ref{EQ:TotalInitialState3}))
\[
|\psi_{m,n}\rangle_B = \sum_{k=0}^{N-1} \alpha_k \,\omega^{-mk}\,|k\oplus n\rangle_B\,.
\]
Alice’s measurement labels $(m,n)$ tell Bob to ``apply $Z^{m}$ and then apply $X^{-n}$''.
Bob’s correcting unitary is 
\[
V_{m,-n}
=X^{-n}\,Z^{m}
=\sum_{J=0}^{N-1}\omega^{m J}\,\lvert J\ominus n\rangle\langle J|\,,
\]
which satisfies
\[
V_{m,-n}\,|\psi_{m,n}\rangle_B = \sum_{k=0}^{N-1} \alpha_k \,|k\rangle_B = |\psi\rangle_B\,.
\]
Bob uses these unitary corrections to send his state into the form of Eq.~(\ref{EQ:Messenger2}).
Finally, Bob obtains $\lvert \psi\rangle_B$ to decode the messenger $\lvert \psi\rangle_X$, i.e., the values of $\alpha_k$'s.

\subsubsection{Gauge Invariance of Packaged Quantum Teleportation Protocol}

Let us now prove that every state, every unitary, and every measurement in the packaged quantum teleportation lives entirely inside the gauge‑neutral sector $\mathcal H_{Q=0}$.
Therefore, the whole protocol is gauge-invariant.

\begin{proposition}
	The packaged quantum teleportation protocol is gauge‑invariant.
\end{proposition}

\begin{proof}
	We split the proof into five steps:
	
	\paragraph{(1) Preliminaries.}
	
	Let $\hat{Q}$ be the net‑charge operator and $U_g=e^{i\,\phi(g)\,\hat{Q}}$ be the corresponding gauge transformation. 
	By definition, all logical packaged basis states $|j_P\rangle$ are gauge-invariant, i.e., satisfy 
	\[
	\hat{Q}\,|j_P\rangle = 0
	\quad\Longrightarrow\quad
	U_g\,|j_P\rangle = |j_P\rangle,
	\]
	and all external states obey $\hat Q\,|k_E\rangle=0$ (and hence $U_g|k_E\rangle=|k_E\rangle$).
	Therefore, every computational basis state $|j\rangle \;=\; |j_P\rangle \otimes |k_E\rangle$ is gauge-invariant, i.e.,
	\[
	\hat{Q}\,|j\rangle=0,\;
	U_g\,|j\rangle=|j\rangle.
	\]

	\paragraph{(2) The Initial Messenger.}
	
	The unknown messenger is gauge-invariant
	\[
	|\psi\rangle_X=\sum_k\alpha_k\,|k\rangle_X
	\quad\Longrightarrow\quad
	\hat{Q}\,|\psi\rangle_X=0,
	\quad
	U_g\,|\psi\rangle_X=|\psi\rangle_X.
	\]

	\paragraph{(3) Shared Resource: Bell pair.}
	In the pre-shared Bell pair 
	\[
	|\Phi_{0,0}\rangle_{AB}
	=\frac1{\sqrt N}\sum_J |J\rangle_A\otimes|J\rangle_B,
	\]
	each term has $\hat{Q}_A+\hat{Q}_B=0$, so 
	$\hat{Q}_{AB}|\Phi_{0,0}\rangle=0$ and $U_g^{(A)}\otimes U_g^{(B)}|\Phi_{0,0}\rangle=|\Phi_{0,0}\rangle$. In other words, the Bell pair is gauge-invariant.
	
	Consequently, the total input state
	\[
	|\Psi_{\rm in}\rangle_{XAB}
	=|\psi\rangle_X\otimes|\Phi_{0,0}\rangle_{AB}
	\quad\Longrightarrow\quad
	\hat{Q}_{XAB}\,|\Psi_{\rm in}\rangle=0,
	\quad
	(U_g)^{\otimes3}\,|\Psi_{\rm in}\rangle=|\Psi_{\rm in}\rangle,
	\]
	i.e., the total input state is gauge-invariant.

	\paragraph{(4) Alice’s Bell Measurement.}
	
	The projectors onto the generalized Bell basis are
	\[
	\Pi_{m,n}
	=|\Phi_{m,n}\rangle\langle\Phi_{m,n}|\;,\qquad
	|\Phi_{m,n}\rangle = \frac{1}{\sqrt N} \sum_J \omega^{mJ}\,|J\rangle_X\otimes|J\oplus n\rangle_A.
	\]
	Using the fact: $[\hat{Q}_X+\hat{Q}_A\,,\,\Pi_{m,n}]=0$, or equivalently 
	\[
	(U_g)^{(X)}\otimes(U_g)^{(A)}\;\Pi_{m,n}\;
	\bigl((U_g)^{(X)}\otimes(U_g)^{(A)}\bigr)^\dagger
	\;=\;\Pi_{m,n}.
	\]
	We see that Alice’s measurement cannot transfer amplitude out of $\mathcal H_{Q=0}$, and each outcome occurs with the same probability before and after a global gauge rotation.
	In other words, Alice’s measurement is gauge-invariant.

	\paragraph{(5) Bob’s Conditional Correction.}
	
	After Alice's outcome $(m,n)$, Bob’s state is
	\[
	|\psi_{m,n}\rangle_B
	=\sum_k\alpha_k\,\omega^{-mk}\,\bigl|k\oplus n\bigr\rangle_B,
	\quad
	\hat{Q}_B\,|\psi_{m,n}\rangle_B=0,
	\]
	so still gauge‑neutral.
	
	Bob applies
	\[
	V_{m,-n}
	=X^{-n}Z^{m}
	=\sum_J\omega^{mJ}\,\bigl|J \ominus n\bigr\rangle\!\langle J|\,.
	\]
	Since $X$ and $Z$ are built from the same $\hat{Q}$\!-commuting Weyl operators, 
	\[
	[\hat{Q}_B,\,V_{m,n}]=0
	\;\Longrightarrow\;
	U_g\,V_{m,n}\,U_g^\dagger = V_{m,n}.
	\]
	The corrected state is
	$\;V_{m,n}\,|\psi_{m,n}\rangle_B = |\psi\rangle_B \in \mathcal H_{Q=0}$,
	i.e., the corrected state $|\psi\rangle_B$ is still gauge-invariant.

	\paragraph{(6) Conclusion.}
	
	From state‐preparation, entangling resource, Bell measurement, classical communication of $(m,n)$ to Bob’s local correction, we have either states that lie in the $\hat{Q}\!=\!0$ sector so $U_g$ acts as the identity, or operations (unitaries or projectors) that commute with $\hat{Q}$.
	Therefore, no amplitude ever leaks into charged sectors.
	The local gauge transformation $U_g$ does not affect the protocol.
	Finally, we conclude that the packaged quantum teleportation protocol is gauge‑invariant.
\end{proof}

\subsubsection{Physical Interpretation}

In this packaged quantum teleportation protocol, the unknown state $\lvert \psi\rangle_X$ and the shared entangled resource $|\Phi_{0,0}\rangle_{AB}$ are encoded as packaged qubits.
Thus, they lie entirely in a fixed superselection sector (e.g., $Q=0$).
By performing a Bell measurement on modes $X$ and $A$ in the packaged Bell basis, Alice obtains one of the allowed outcomes ($m, n$) without causing any cross-sector transitions.
Bob’s state then requires a combination of packaged gates $Z^mX^{-n}$ to correct, which ensure that the operation is gauge-invariant.

Since every step (state preparation, measurement, and correction) is implemented by packaged operations that commute with the net-charge operator $\hat{Q}$, the entire teleportation protocol remains in $\mathcal{H}_{Q=0}$ and is gauge-invariant. This obeys the superselection rules that suppress errors and therefore offer enhanced robustness.

\subsection{Super‐Superdense Coding Protocol in Hybrid Packaged Space}
\label{SEC:PackagedSuperSuperdenseCodingProtocol}

Superdense coding \cite{BennettWiesner1992,Schaetz2004,Williams2017} is a quantum communication protocol that uses only one qubit (the sender and receiver pre-share entanglement) to transmit two classical bits of information.
It exhibits the power of quantum entanglement as a resource for enhanced communication.

In this subsection, we derive a super‐superdense coding protocol in a general $(d\times D)$-dimensional hybrid-packaged subspace, where $d$ is the dimension of internal space and $D$ is the dimension of the external space.

\subsubsection{Algorithmic Procedure}

We use a hybrid-packaged subspace of dimension $N = dD$.
We choose a shared maximally hybrid-packaged entangled state and then apply hybrid-packaged qudit gates on one subsystem to encode classical information into one of $N^2$ mutually orthogonal states.
Thus, we can encode up to
\[
C = \log_2(N^2) = 2\log_2(dD)
\]
bits of classical information per transmitted qudit.

\paragraph{(1) Shared Packaged Resource State.}

Alice and Bob pre-share a hybrid-packaged resource state (maximally hybrid-packaged entangled state) given by
\[
|\Phi^+\rangle_{AB} = \frac{1}{\sqrt{dD}} \sum_{j=0}^{d-1} \sum_{k=0}^{D-1} |j,k\rangle_A \otimes |j,k\rangle_B\,.
\]
Because under a local gauge transformation $U_g$, every packaged basis state transforms only with a global phase,
\[
U_g |j,k\rangle = e^{i\phi(g)} |j,k\rangle\,,
\]
we have
\[
U_g^{\otimes 2} |\Phi^+\rangle_{AB} = e^{i2\phi(g)} |\Phi^+\rangle_{AB}\,.
\]
This means that the packaged resource state remains in the gauge-invariant subspace $\mathcal{H}_{Q=0}^{\otimes 2}$ (up to a global phase).

\paragraph{(2) Encoding Message (by Alice).}

Alice wishes to encode a classical message $m$ from the set $\{0,1,\ldots, (dD)^2 -1\}$, which contains $(dD)^2$ elements.
To do so, Alice uses the generalized Pauli (Weyl-Heisenberg) operators in an $N$-dimensional Hilbert space (where $N = dD$):

\begin{enumerate}
	\item Generalized Pauli Operators:
	
	Let $\omega_N \equiv e^{2\pi i/N}$. Define the unitary packaged operators $X$ and $Z$ on the hybrid-packaged subspace as
	\[
	X |n\rangle = |n\oplus 1\rangle,\quad Z |n\rangle = \omega_N^n |n\rangle\,,
	\]
	with the computational (composite) index $n$ defined via
	\[
	n = j\,D + k,\quad j\in \{0,\ldots,d-1\},\quad k\in \{0,\ldots,D-1\}\,,
	\]
	and where $\oplus$ denotes addition modulo $dD$. Then, for any pair of integers $(a,b)$ with
	\[
	a,b \in \{0,1,\ldots, dD-1\}\,,
	\]
	we define
	\[
	V_{a,b} = Z^a X^b\,.
	\]
	
	These operators form an orthonormal basis (with respect to the Hilbert-Schmidt inner product) for the space of linear operators on $\mathcal{H}_{\text{hyb}}$.
	There are $d^2D^2$ such operators.

	Since the entire operation is defined on the full hybrid-packaged qudit space, the set $\{V_{a,b}\}$ can be used by Alice to encode $2\log_2(dD)$ classical bits.

	\item Encoding Procedure:
	
	Alice applies the unitary packaged operator $V_m$ (corresponding to message index $m$) to her subsystem.
	Specifically, if her message is $m$ and
	\[
	m \in \{0,1,\ldots, d^2D^2-1\}\,,
	\]
	then she applies	
	\[
	V_m\,|\Phi^+\rangle_{AB} = (V_m\otimes I)|\Phi^+\rangle_{AB} \equiv |\Phi_m\rangle_{AB}\,.
	\]

	Because the operators $V_m$ are constructed from $X$ and $Z$ which are defined on the full hybrid-packaged subspace, and since every basis state in the computational basis is packaged (i.e. they all lie in $\mathcal{H}_{Q=0}$), one has
	\[
	[V_m, \hat{Q}] = 0\,.
	\]
	Thus, $ |\Phi_m\rangle_{AB}$ remains in the same superselection sector.
	In other words, $V_m$ are gauge-invariant.
\end{enumerate}

\paragraph{(3) Transmission and Bob's Decoding.}

Alice sends her encoded hybrid-packaged qudit to Bob over a quantum channel that is modeled by a unitary $V_{\rm channel}$ with
\[
[V_{\rm channel}, \hat{Q}] = 0\,.
\]
Thus, the overall state remains in the desired physical (gauge-invariant) subspace.

After receiving Alice's qudit, Bob now has both halves of the packaged entangled resource state.
He then performs a joint measurement in the generalized Bell basis on the two qudits.

Define the generalized Bell basis states
\[
|\Phi_{a,b}\rangle_{AB} \equiv (V_{a,b}\otimes I)|\Phi^+\rangle_{AB}\,,
\]
for all $a,b\in \{0,1,\ldots,dD-1\}$.
These states are mutually orthogonal:
\[
\langle \Phi_{a',b'}|\Phi_{a,b}\rangle = \delta_{aa'}\,\delta_{bb'}\,.
\]
This is an $N^2$‐outcome measurement.
By performing a measurement in this basis, Bob can identify the message $m$ that Alice encoded.

Because the encoded set comprises $d^2D^2$ orthogonal states, Bob can decode $2\log_2(dD)$ bits of classical information from a single transmitted hybrid qudit.

\subsubsection{Security and Advantages}

This packaged super-superdense coding protocol has several advantages:

First, the packaged super-superdense coding protocol reduced eavesdropper (Eve)’s correct measurement probability.
If an eavesdropper (Eve) intercepts the transmitted qudit and performs an optimal measurement, her maximum probability $P_{\text{Eve}}$ of correctly guessing the encoded message is upper bounded by
\[
P_{\text{Eve}} \le \frac{1}{d^2D^2}\,,
\]
which assumes that she has no prior information.
Compared with the original 2-dimensional case (with $D = 2$) where $P_{\text{Eve}} \le \frac{1}{4}$, increasing $d$ and/or $D$ further reduces Eve’s correct guess probability proportionally to $1/(d^2D^2)$.

The packaged super-superdense coding protocol also enhance information capacity.
Using a hybrid system with dimension $dD$, the protocol transmits
\[
2\log_2(dD)
\]
classical bits per transmitted qudit.
This is beneficial for increasing the communication capacity beyond that of the conventional qubit-based (or even 4-dimensional) protocol.

This packaged protocol provides an inherent physical protection from gauge style errors. 
All operations (encoding, transmission, and measurement) are restricted to the gauge-invariant subspace $\mathcal{H}_{Q=0}$. Any error process that would drive the state out of this subspace is forbidden (or energetically suppressed). This adds an inherent layer of error protection.

This packaged protocol is highly scalable (limited in practice by how large an internal manifold one can coherently address).
The packaged state is often easier to scale to higher dimensions because it relies on degrees of freedom (e.g., flavor, color, or other quantum numbers) that naturally come in larger multiplets, especially in high-energy physics or quantum simulators.
The external space can also be chosen to be high-dimensional (for example, using orbital angular momentum of photons), so the overall communication alphabet is greatly expanded.

The protocol enhances security that is proportional to the space dimension: 
The larger the dimension $dD$, the lower the probability that an eavesdropper can correctly identify the transmitted state when forced to choose among $d^2D^2$ possibilities.
In an ideal scenario, Eve’s correct measurement probability scales as $1/(d^2D^2)$ and therefore enhances the security of the QKD protocol.

\subsection{QKD Protocol in Hybrid Packaged Space}
\label{SEC:QKDProtocolInHybridPackagedSpace}

The six-state protocol \cite{Bruss1998,PasquinucciGisin1999} is a quantum cryptographic protocol that improves the security of BB84 protocol by encoding information in three mutually unbiased bases (instead of two).
This increases the difficulty for an eavesdropper (Eve) to intercept the key without detection.

In this subsection, we incorporate external and internal DOFs in constructing the hybrid-packaged qudits.
We also incorporate the idea of six-state QKD protocol to a hybrid-packaged qudit system of dimension $d \times D$.

\subsubsection{Constructing the MUBs in the $(d\times D)$-Dimensional Hybrid-Packaged Space}

\paragraph{(1) Computational (Z) Basis.}

We define the computational (Z) basis for $\mathcal{H}_{\text{hyb}}$ as the set of orthonormal states:
\[
\mathcal{B}_Z = \{|n\rangle \;:\; n = 0,1,2,\dots, N-1\}\,,
\]
where the composite label is defined via a bijective mapping between the pair of indices $(j,k)$ and a single integer $n$ by
\[
n = j\,D + k\,.
\]
By construction, each $|n\rangle$ is
\[
|n\rangle \equiv |j,k\rangle = |j_P\rangle\otimes|k_E\rangle\,.
\]
Because both factors are in the correct gauge-invariant subspaces, these states satisfy
\[
\hat{Q}|n\rangle = 0\quad \text{and}\quad U_g\,|n\rangle = e^{i\phi(g)}\,|n\rangle\,.
\]

\paragraph{(2) Fourier (X) Basis.}

The Fourier transform of the Z-basis yields a second mutually unbiased basis $\mathcal{B}_X$. Define
\[
|n_x\rangle = \frac{1}{\sqrt{N}} \sum_{m=0}^{N-1} \omega^{nm}\,|m\rangle\,,\quad \text{with } \omega \equiv e^{2\pi i/N}\,.
\]
Because
\[
|\langle m|n_x\rangle|^2 = \frac{1}{N}\,,
\]
the X-basis is mutually unbiased with respect to the Z-basis.

Note that every $|n_x\rangle$ is an equally weighted superposition of computational states. Since the coefficients are scalars and every component $|m\rangle$ lies in $\mathcal{H}_{Q=0}$, each $|n_x\rangle$ also satisfies
\[
\hat{Q}\,|n_x\rangle = 0,\quad U_g\,|n_x\rangle = e^{i\phi(g)}\,|n_x\rangle\,.
\]

\paragraph{(3) A Third MUB (Y-Basis).}

A third mutually unbiased basis, $\mathcal{B}_Y$, can be constructed using a modified Fourier transform with additional phase factors. For example, define
\[
|n_y\rangle = \frac{1}{\sqrt{N}} \sum_{m=0}^{N-1} \omega^{nm} e^{i\theta(m)}\,|m\rangle\,,
\]
where the phases $\{\theta(m)\}$ are chosen so that the bases $\mathcal{B}_Z$, $\mathcal{B}_X$, and $\mathcal{B}_Y$ are mutually unbiased, say, $\theta(m)=\pi m(m+1)/N$ for odd prime $N$.
A common choice in the two-dimensional case is $\theta(0)=0$ and $\theta(1)=\pi/2$.
In higher dimensions, one must choose a set of phases so that $|\langle m|n_y\rangle|^2=\frac{1}{N}$ for all $m,n$.
By construction, the Y-basis states remain in $\mathcal{H}_{\text{hyb}}$ (and hence are packaged states) because they are superpositions of $|m\rangle$ with constant coefficients, i.e.,
\[
\hat{Q}\,|n_y\rangle=0,\quad U_g\,|n_y\rangle = e^{i\phi(g)}|n_y\rangle\,.
\]

\subsubsection{Algorithmic Procedure}

We now outline the algorithmic procedure for the packaged QKD protocol in the $(d\times D)$-dimensional hybrid space using the three MUBs defined above ($\mathcal{B}_Z$, $\mathcal{B}_X$, and $\mathcal{B}_Y$:

\begin{enumerate}
	\item Step 1: Preparing State.
	
	Alice prepares a state by randomly choosing one of the three MUBs among $\mathcal{B}_Z$, $\mathcal{B}_X$, and $\mathcal{B}_Y$.
	In the chosen basis, she then randomly selects one of the $N$ basis states.
	For example, if she chooses the Z-basis, then she picks one state 
	\[
	|n_z\rangle \in \mathcal{B}_Z\,,
	\]
	where $n \in \{0,1,\dots, N-1\}$.
	Analogous choices are made for the X and Y bases. 
	Since each basis state is a linear combination (or is a basis element) from a set of packaged states, it remains gauge-invariant.
	In other words, regardless of the basis, we always have
	\[
	\hat{Q}\,|\psi\rangle=0,\quad U_g\,|\psi\rangle=e^{i\phi(g)}|\psi\rangle\,.
	\]
	
	\item Step 2: Transmission.
	 
	Alice sends the prepared state $|\psi\rangle$ to Bob through a quantum channel $V_{\rm channel}$,
	which is implemented by packaged operations that satisfy
	\[
	[V_{\rm channel}, \hat{Q}] = 0.
	\]
	In this way, the transmitted state remains in $\mathcal{H}_{Q=0}$.
	
	\item Step 3: Measurement.
	
	Bob does not know which basis was used.
	He randomly chooses one of the three MUBs ($\mathcal{B}_Z$, $\mathcal{B}_X$ or $\mathcal{B}_Y$) for his measurement.
	If his choice matches Alice’s preparation basis, then he obtains the prepared state with certainty (ignoring any channel noise).
	If he measures in a different basis, then his results is uniformly distributed among the $N$ basis states and therefore yield no useful information.
	
	\item Step 4: Sifting.
	
	Over a public authenticated classical channel, Alice and Bob compare which measurement bases they used for each transmission.
	They discard any events in which they used different bases.
	
	\item Step 5: Error Rate Estimation and Key Extraction.
	
	From the sifted data, Alice and Bob first determine the quantum bit error rate (QBER).
	Then they perform conventional classical error correction and privacy amplification to distill a shared secret key.
\end{enumerate}

\subsubsection{Security Analysis: Eve’s Correct Measurement Probability.}

Assume that Eve intercepts the transmission and makes a measurement in one of the three bases chosen at random. Two cases arise:
\begin{enumerate}
	\item Correct Basis Choice: 
	There is a probability $1/3$ that Eve randomly guesses the same basis that Alice used.
	In this case, Eve will measure the state correctly with probability 1.
	
	\item Wrong Basis Choice: 
	There is a probability $2/3$ that Eve uses one of the other two MUBs.
	In this case, because any two MUBs are mutually unbiased, the overlap between the state prepared by Alice and the state used by Eve (in the wrong basis) is
	\[
	|\langle \psi_{\text{Alice}}|\psi_{\text{Eve}}\rangle|^2 = \frac{1}{N}\,,
	\]
	where $N = dD$.
\end{enumerate}

Thus, Eve's overall correct guess probability is
\[
P_{\mathrm{Eve}} = \frac{1}{3}\times 1 + \frac{2}{3}\times\frac{1}{N} = \frac{1}{3} + \frac{2}{3N}\,.
\]
In the qubit case, $N=2$ so $P_{\mathrm{Eve}} = \frac{1}{3} + \frac{2}{6} = \frac{2}{3}$.
If $N > 2$, then
\[
\frac{2}{3N} < \frac{1}{3}\,.
\]
Thus, $P_{\mathrm{Eve}}$ is reduced compared to the conventional qubit six-state protocol.

\subsubsection{Relation to BB84 and B92 QKD Protocols}
 
If we drop one MUB, then we obtain the hybrid-packaged BB84 QKD protocol \cite{BB84,BB84Bennett1992,BB842014,BB84arXiv2020}.
If we keep only two non‑orthogonal signals, then we obtain the hybrid-packaged B92 QKD protocol \cite{Bennett1992}.

\subsubsection{Physical Interpretation}

First, in this packaged QKD protocol, the security is enhanced with higher dimensions.
By using a hybrid-packaged state of dimension $N = dD$, the number of distinguishable quantum states increases.
The probability that an eavesdropper correctly measures the state when choosing the wrong basis decreases inversely with the dimension.
Eve's average success probability reduces to
\[
P_{\mathrm{Eve}} = \frac{1}{3} + \frac{2}{3N}\,.
\]
Thus, when $N > 2$, Eve always has a lower overall guessing probability.
This shows that by increasing the dimension $d$ of the internal packaged state or the dimension $D$ of the external space, one can significantly improve the security against intercept-resend attacks.

On the other hand, the hybrid-packaged subspace provides intrinsic error protection.
By design, all transmitted states lie in the gauge-invariant subspace $\mathcal{H}_{Q=0}$.
Any error that would take the state out of this subspace is energetically suppressed by the superselection rule.

Hybrid packaged space increases alphabet size. 
A larger dimension implies more available symbols for encoding information.
This packaged QKD protocol permits encoding up to $2\log_2(dD)$
classical bits per transmitted state, while simultaneously reducing Eve’s correct measurement probability.

Hybrid packaged space enables flexible implementation. 
The hybrid structure allows one to choose the dimension $d$ (often from naturally occurring multiplets in high energy physics or field theory models) independent on the external degrees of freedom $D$, which may be implemented by other physical means (such as orbital angular momentum, path, or polarization).
This decoupling provides flexibility in experimental design.

\subsection{Device-Independent QKD via Hybrid-Packaged Bell Tests}

Device-independent QKD (DI‑QKD) \cite{MayersYao1998,BarrettHardyKent2005,Acin2007,Pironio2009}
is a cryptographic protocol for two parties (Alice and Bob) to securely share a secret key even if their quantum devices (e.g., sources, detectors) are un-trusted or potentially compromised.
DI-QKD protocols relies on violations of high-dimensional Bell inequalities.
For example, the Collins-Gisin-Linden-Massar-Popescu (CGLMP) inequality \cite{Collins2002} can generate stronger violations in higher dimensions.

In this subsection, we show how to generalize the conventional high‑dimensional DI‑QKD (based on the CGLMP family of inequalities) to the $(d \times D)$-dimensional hybrid-packaged Hilbert space in which all states and observables are gauge-invariant.
Thus, the protocol respects gauge super‑selection at every step.

\subsubsection{Protocol Steps}

\paragraph{(1) Distributing State.}

A source (untrusted) distributes the maximally packaged entangled state 
\begin{equation}\label{EQ:DIQKDRecourseState}
	\ket{\Phi_{0,0}}_{AB}
	=\frac1{\sqrt N}\sum_{J=0}^{N-1}
	\ket{J}_A\!\otimes\!\ket{J}_B,
	\quad
	J:=jD+k
\end{equation}
with each single‑index basis vector $\ket{J}\equiv\ket{j_P}\otimes\ket{k_E}$ lying in $\mathcal H_{Q=0}$. 
The preparation channel may be noisy.
Eve may hold a purifying system $E$.
The only assumption is that $\hat{Q}_A+\hat{Q}_B=0$.

\paragraph{(2) Measurement with Packaged Pauli Operators.}

Define the generalized hybrid-packaged Pauli pair
\[
Z_N\ket{J}= \omega_N^{\,J}\ket{J},\qquad
X_N\ket{J}= \ket{J\oplus1},
\qquad
\omega_N=e^{2\pi i/N},
\]
which act identically on every $\mathcal H_{Q=0}$ block and
obey $X_NZ_N=\omega_NZ_NX_N$.
Because $Z_N$ is diagonal and $X_N$ is a packaged cyclic shift, we have relations $[\,Z_N,\hat{Q}]=[\,X_N,\hat{Q}]=0$.

For the CGLMP inequality, each party needs two $N$-outcome measurements. Following the optimal constructions \cite{Collins2002}, we choose the eigenbasis
\begin{align}\label{EQ:ABEigenBasis}
	\begin{aligned}
		A_0 &:\;\text{eigenbasis of }Z_N, &
		B_0 &:\;\text{eigenbasis of }Z_NX_N^{\tfrac12}, \\
		A_1 &:\;\text{eigenbasis of }Z_NX_N, &
		B_1 &:\;\text{eigenbasis of }Z_NX_N^{-\tfrac12},
	\end{aligned}
\end{align}
where fractional powers of $X_N$ are defined spectrally.
All four POVMs are packaged:
the projectors are sums of $\ket{J}\!\bra{J}$-type operators and hence commute with $\hat{Q}$.

\paragraph{(3) Evaluating Bell Parameter in the Hybrid-Packaged Space.}

Let $P(a,b\!\mid\!x,y)$ be the joint output probabilities when Alice (Bob) chooses setting $x\,(y)\in\{0,1\}$ and obtains outcome $a,b\in\{0,\dots,N-1\}$. 
The $N$-dimensional CGLMP polynomial is

\begin{align}
	\begin{aligned}
		I_N
		=&\sum_{k=0}^{\lfloor N/2\rfloor-1}
		\Bigl[\,P\bigl(A_0=B_0\!+\!k\bigr)
		+P\bigl(B_0=A_1\!+\!k\!+\!1\bigr)
		+P\bigl(A_1=B_1\!+\!k\bigr)
		+P\bigl(B_1=A_0\!+\!k\bigr)\Bigr] \\
		&-\!\bigl(\text{same terms with }k\mapsto-k\!-\!1\bigr).
	\end{aligned}
\end{align}

Thus,
the local‑hidden‑variable (LHV) bound: $I_N^{\text{LHV}}\le 2$, 
and quantum Tsirelson bound: $I_N^{\text{QM}} = 2\,\bigl(1-\tfrac1N\bigr)^{-1}$ 
(exact for the settings Eq.~(\ref{EQ:ABEigenBasis}) and $\ket{\Phi_{0,0}}$).
The violation grows with $N$, e.g., $I_4^{\text{QM}}= \tfrac83$.

Because the packaged measurements evaluate exactly the same correlations as in ordinary qudit DI‑QKD, the standard proofs completely apply after the substitution $N=dD$.

\paragraph{(4) Key‑Generation Rounds vs. Test Rounds.}

Each run is randomly assigned:
\begin{itemize}
	\item Test round (probability $q_{\text{test}}$): 
	Alice and Bob use the four CGLMP settings to estimate $I_N$.
	
	\item Key round (probability $1-q_{\text{test}}$): 
	Both measure in $A_0=B_0$ (the computational packaged basis); the raw key symbols are the outcomes $a=b$.
\end{itemize}

\paragraph{(5) Evaluating Entropy Bound from Bell Violation.}

Using the entropy‑accumulation theorem (EAT) \cite{Pironio2010}, one obtains a lower bound on the conditional smooth min‑entropy of Alice’s key given Eve:
\[
H_{\min}^{\varepsilon}(A^{\text{key}}\!\mid\!E)
\;\ge\;
n_{\text{key}}
\;f\!\bigl(I_N^{\text{obs}}\bigr)
\;-\;\mathcal O\!\bigl(\sqrt{n}\bigr),
\]
where $n_{\text{key}}$ is the number of key rounds, 
$I_N^{\text{obs}}$ the experimentally observed value, and
\[
f(I)=\log_2 N\;-\;h_N\!\Bigl(\tfrac{I-2}{I_N^{\text{QM}}-2}\Bigr),
\qquad
h_N(p):=
-p\log_2\!\frac{p}{N-1}-(1-p)\log_2(1-p)
\]
is a high‑dimensional generalisation of the binary entropy. 
A positive $f(I_N^{\text{obs}})$ certifies secrecy.
The larger $N$, the larger the certified entropy for the same excess violation $\Delta I=I_N^{\text{obs}}-2$.

\paragraph{(6) Evaluating Secret‑Key Rate.}

After error correction leaking $ \text{leak}_{\text{EC}} $ bits, the asymptotic key rate per signal is
\[
R
\;=\;
(1-q_{\text{test}})\Bigl[f\!\bigl(I_N^{\text{obs}}\bigr)
- \text{leak}_{\text{EC}}\Bigr].
\]

For fixed relative violation, $f(I)$ roughly scales as $\log_2 N$.
Every extra packaged internal dimension $d$ increases the alphabet in a multiplicative way and therefore the extractable entropy.
While the external hardware size $D$ may already be limited by experimental constraints (e.g., path or OAM modes).

\subsubsection{Gauge-Invariance of DI‑QKD Protocol in Hybrid-Packaged Space}

All observables in Eq.~(\ref{EQ:ABEigenBasis}) are functions of $X_N, Z_N$, which commute with $\hat{Q}$.
Thus, the protocol never exits the protected sector $\mathcal H_{Q=0}$.

Practical realizations can reuse the same platforms as discussed in packaged QKD protocol (see Sec.~\ref{SEC:QKDProtocolInHybridPackagedSpace}) (multi‑port interferometers for photonic OAM + polarisation, or trapped‑ion internal hyperfine manifold $d\le 8$ combined with $D$ phonon modes).

The Bell‑test loopholes must still be closed (space‑like separation or fast random basis choice).
But gauge protection helps against certain side‑channel attacks that inject charge‑changing noise.

\subsubsection{Advantages of DI‑QKD Protocol in Hybrid-Packaged Space}

Extending DI‑QKD protocol into $(d\times D)$-dimensional hybrid-packaged subspace can yield better key rates per detected signal and additional physical flexibility.
We list the advantages in following table:

\begin{table}[H]
	\centering
	\caption{Advantages of DI‑QKD Protocol in Hybrid-Packaged Space}
	\begin{tabular}[hbt!]{|p{4cm}|p{5cm}|p{5cm}|}
		\hline\hline
		Feature & Ordinary Qudit DI‑QKD & Hybrid-Packaged DI‑QKD \\
		\hline
		Hilbert‑space size & $N=D$ Limited by external DOF & $N=dD$ (internal lift) \\
		\hline
		Bell‑violation strength & Grows as $\sim\log N$ & Same growth, but more $N$ for the same external complexity \\
		\hline
		Gauge leakage & Uncontrolled & Suppressed by $[\mathcal O,\hat{Q}]=0$ \\
		\hline
		Device trust & None (device‑independent) & None plus intrinsic physical filter \\
		\hline
	\end{tabular}
\end{table}

\subsection{Quantum Secret Sharing via Hybrid-Packaged GHZ states}

Quantum Secret Sharing (QSS) \cite{Karlsson1999,Hillery1999,Cleve1999} is a protocol that distributes a secret among multiple parties such that only authorized subsets can reconstruct it.
A prominent method is to use Greenberger-Horne-Zeilinger (GHZ) states, say the three-qubit state $ |\text{GHZ}\rangle = \frac{1}{\sqrt{2}}(|000\rangle + |111\rangle) $, which is generalized to more qubits.

In this subsection, we extend the conventional quantum secret sharing protocol to a new $(d \times D)$-dimensional hybrid-packaged subspace, where the internal space has dimension $d$ and the external space has dimension $D$.

\subsubsection{Definition of the Generalized GHZ State}

In quantum secret sharing, the secret is encoded in a multipartite entangled state.
Here, we have three parties: Alice (A), Bob (B), and Charlie (C).
Let us consider a generalized Greenberger-Horne-Zeilinger (GHZ) state in the high-dimensional hybrid-packaged subspace:
\begin{equation}\label{EQ:GHZStateInHybridPackagedSpace}
	|\mathrm{GHZ}\rangle_{ABC} = \frac{1}{\sqrt{N}} \sum_{j=0}^{N-1} |J\rangle_A \otimes |J\rangle_B \otimes |J\rangle_C\,,
\end{equation}
where $N = d\,D$ and
\[
|J\rangle \equiv |j_P\rangle\otimes |k_E\rangle\,,
\]
is the computational basis of $\mathcal{H}_{\text{hyb}}$.
This implies a mapping between the double index $(j_P,k_E)$ and the single-index $J$ (for example, one may choose a lexicographic ordering).

\begin{property}[Gauge-invariance of GHZ state in hybrid-packaged subspace]
	The generalized GHZ state defined in Eq.~(\ref{EQ:GHZStateInHybridPackagedSpace}) is gauge-invariant.
\end{property}

\begin{proof}
	Because the internal states are packaged and all lie in the $Q=0$ subspace, for each party, we have
	\[
	\hat{Q}\,|J\rangle = \hat{Q}\,(|j_P\rangle\otimes |k_E\rangle) = (\hat{Q}\,|j_P\rangle) \otimes |k_E\rangle = 0\,.
	\]
	
	Furthermore, under any local gauge transformation $U_g$ (acting only on the internal part),
	\[
	U_g\,|J\rangle = \bigl(U_g\,|j_P\rangle\bigr)\otimes |k_E\rangle = e^{i\phi(g)}\,|J\rangle\,.
	\]
	
	Since each of the three parties has the same overall phase, the total state transforms as
	\[
	(U_g \otimes U_g \otimes U_g)|\mathrm{GHZ}\rangle_{ABC} = e^{3i\phi(g)} |\mathrm{GHZ}\rangle_{ABC}.
	\]
	
	Since an overall phase is physically irrelevant, the generalized three parties state $|\mathrm{GHZ}\rangle_{ABC}$ remains within the same superselection sector and therefore is gauge-invariant.
\end{proof}

\subsubsection{Secret Sharing Protocol Outline}

The secret sharing protocol proceeds as follows:

\paragraph{(1) Preparing State.}

A dealer (say, Alice) prepares the above generalized GHZ state
\[
|\mathrm{GHZ}\rangle_{ABC} = \frac{1}{\sqrt{N}} \sum_{j=0}^{N-1} |J\rangle_A \otimes |J\rangle_B \otimes |J\rangle_C\,.
\]
Because the GHZ state is maximally packaged entangled over the entire $N$-dimensional packaged Hilbert space, it encodes correlations across all three parties.

By definition, each basis state $|J\rangle$ is a product of an internal state and an external state:
\[
|J\rangle_A = |j_P\rangle_A \otimes |k_E\rangle_A,\quad \text{and similarly for } B \text{ and } C.
\]
Thus the entire state is an element of
\[
\mathcal{H}_{\mathrm{GHZ}} \subset \mathcal{H}_{\text{hyb}}^{\otimes 3} \subset \left(\mathcal{H}_{Q=0}\right)^{\otimes 3}\,.
\]

Due to the packaging principle, the internal quantum numbers remain locked in and errors that try to alter them (or drive transitions to different $Q$) are suppressed.

\paragraph{(2) Encoding Secret.}

The secret is encoded by correlating a classical message with the outcome of a measurement in a given basis. 
For instance, the dealer may choose one of several mutually unbiased bases (MUBs) for the internal space.
Here, we assume that the dealer applies a local unitary $V_s$ (drawn from a set of unitary encoding operators that are gauge-invariant) to her qudit before distributing the state.

The complete encoding may be written as
\[
|\mathrm{GHZ}_{s}\rangle_{ABC} = \left(V_{s}\otimes I\otimes I\right)|\mathrm{GHZ}\rangle_{ABC}\,,
\]
where $V_s$ is chosen from a set corresponding to the secret message.

\paragraph{(3) Distribution.}

The GHZ state (or encoded GHZ state) is then distributed among the three parties.
The transmission channel $V_{\text{channel}}$ is assumed to commute with $\hat{Q}$ (i.e., $[V_{\text{channel}},\hat{Q}]=0$).
Thus, the state remains in the same gauge sector during propagation:
\[
|\mathrm{GHZ}'\rangle_{ABC} = \left(V_{\text{channel}}^A\otimes V_{\text{channel}}^B\otimes V_{\text{channel}}^C\right)|\mathrm{GHZ}_{s}\rangle_{ABC}\,.
\]

\paragraph{(4) Reconstructing Secret.}

At the reconstruction stage, a designated receiver (say, Charlie) collects measurement results from the other parties (Alice and Bob).
In high-dimensional QSS protocols, the reconstruction is usually accomplished via an appropriate joint measurement on the correlated outcomes.
Mathematically, if Alice and Bob measure their qudits in an agreed-upon basis (say, the computational basis), then the joint measurement outcome will collapse the state to a product state
\[
|J\rangle_A\otimes|J\rangle_B,
\]
which perfectly correlates with Charlie’s state $|J\rangle_C$.
Thus, Charlie can deduce the secret value encoded by the dealer’s unitary $V_s$. 
It is necessary that these measurements are done in bases that are mutually unbiased relative to other possible choices.
In our protocol, we can take advantage of the enlarged internal dimension to define more than two MUBs (for example, if $d>2$, up to $d+1$ MUBs exist in principle) and thereby enhance security.

\subsubsection{Security Analysis.}

With increasing $N = dD$, the security improves because the number of possible outcomes increases.
In secret sharing, an adversary (or eavesdropper) lacks information on which measurement basis was used and therefore must guess from a larger set.
If Eve’s optimal measurement yields a correct guess probability of, say, $1/N$ in an intercept-resend attack, then by increasing $d$ (the internal dimension), Eve’s success probability decreases.
For instance, in a two-qubit (i.e. 4-dimensional) system the guessing probability is $\frac{1}{4}$.
If one uses a three-dimensional internal space ($d=3$) with a two-dimensional external space ($D=2$, so $N=6$), then Eve’s correct measurement probability reduces to $\frac{1}{6}$.

Furthermore, if the dealer uses multiple MUBs on the internal space (e.g., three or more) to encode the secret, then Eve has to guess both the basis and the outcome.
If there are $M$ MUBs employed, then the effective guessing probability is approximately
\[
P_{\mathrm{Eve}} \simeq \frac{1}{M\,N}\,,
\]
assuming uniform distribution of the secret and that Eve’s measurement device is optimized for one basis.
A larger internal space (with $M \le d+1$ MUBs) thus directly contributes to enhanced security.

\subsubsection{Relation to E91/BBM92 QKD Protocols}

When $d=D=2$, this device‑independent Bell-test QKD reduces to E91/BBM92.

\subsubsection{Physical Interpretations}

This packaged secret sharing protocol has several advantages:

First, the packaged secret sharing protocol enhances security via higher dimensions.
The generalized GHZ state in an $N=dD$-dimensional space yields a secret sharing protocol that is inherently more secure:
an eavesdropper who intercepts any part of the state and performs a measurement (in a random basis) faces a much larger outcome space.
For example, if the baseline secret sharing protocol (using two-dimensional qubits) gives Eve a guessing probability of $\frac{1}{2}$, then even with two mutually unbiased bases the correct guess probability is at most $\frac{1}{2}\times\frac{1}{2}=\frac{1}{4}$.
In contrast, for an $N$-dimensional system, the corresponding probability is decreased to approximately $\frac{1}{M\,N}$.

The packaged secret sharing protocol also enhances fault-tolerance and robustness.
Because all operations remain confined to the gauge-invariant subspace, the protocol also benefits from intrinsic error suppression in the face of noise that attempts to induce transitions outside $\mathcal{H}_{Q=0}$.
This error suppression combined with the lower effective guessing probability for Eve, provides an overall robust framework for secure quantum secret sharing.

The high-dimensional secret sharing scheme presented here inherits additional security by increasing the number of orthogonal outcomes.
Provided that the technical challenges of increasing the internal dimension (for instance, using an $SU(3)$ or higher symmetry group) can be met, the new protocol reduces the information available to an adversary in any intercept-resend attack.

\subsection{Randomness Expansion Protocol in Hybrid-Packaged Space}

Randomness expansion \cite{Colbeck2009,Pironio2010,VaziraniVidick2014,LagoRivera2021} is a cryptographic protocol for a quantum device to generate certified random bits using a small initial seed of randomness.
Most importantly, the output randomness is provably unpredictable even if the device’s internal operations are untrusted (a device-independent guarantee).
This is achieved by exploiting the violations of Bell test (e.g., quantum non-locality) to certify randomness beyond classical limits.

In this subsection, we extend randomness expansion protocol to a ($d \times D$)‐dimensional hybrid-packaged subspace where the internal space of dimension $d$ is encoded in gauge-invariant packaged states (with $\hat{Q}=0$) and the external space of dimension $D$ carries other degrees of freedom.

In the hybrid-packaged subspace $\mathcal{H}_{\text{hyb}}$, the increased dimension offers at least two advantages:
(1) A larger alphabet (dimension $N=dD$) reduces the adversary’s correct-guess probability,
(2) The availability of many MUBs (especially in the internal space) further forces an adversary (Eve) to guess the correct basis and therefore reduces her chance of success.

\subsubsection{MUBs in Hybrid-Packaged Space.}

In a high-dimensional protocol, it is an advantage to use measurements from several MUBs.

For the internal space of dimension $d$ (with $d$ a prime or prime-power),
it is known that one can construct ($d+1$) MUBs 
\[
\mathcal{B}_\mu^{\rm int} = \{\,|\psi_{j}^{(\mu)}\rangle: \, j=0,\ldots,d-1\,\},\quad \mu=0,1,\dots,d\,.
\]

For the external space of dimension $D$ (if $D$ is similarly prime or a prime-power),
one can have up to ($D+1$) MUBs 
\[
\mathcal{B}_\nu^{\rm ext} = \{\,|\chi_{k}^{(\nu)}\rangle: \, k=0,\ldots,D-1\,\},\quad \nu=0,1,\dots,D\,.
\]

Thus, for the entire hybrid-packaged subspace of dimension $d \times D$, one can construct a full measurement basis by taking the tensor product of one chosen MUB from the internal space and one chosen MUB from the external space:
\[
\mathcal{B}_{\mu,\nu} = \{\,|\psi_{j}^{(\mu)}\rangle\otimes|\chi_{k}^{(\nu)}\rangle: \, j=0,\dots,d-1,\; k=0,\dots,D-1\,\}\,.
\]

In the ideal case, the measurement probabilities for states prepared in an unbiased state are all equal to $1/(dD)$.

\subsubsection{Protocol Procedure}

We now describe the protocol step by step.

\paragraph{(1) State Preparation and Encoded Randomness.}

\begin{enumerate}
	\item Preparing Initial State:
	
	A trusted source (or one of the parties) prepares a sequence of independent high-dimensional hybrid-packaged states.
	For each run, the source prepares a pure state
	\[
	|\psi\rangle = |m\rangle \in \mathcal{H}_{\text{hyb}}\,,
	\]
	where $m \in \{0,1,\dots,dD-1\}$.
	In an ideal run, the state should be maximally mixed in the measurement basis.
	Alternatively, the device may randomly (or adversarially) output a state whose randomness is certified from the measurement outcomes.
	
	\item Separating Internal-External States:
	
	Recall that
	\[
	|m\rangle = |j,k\rangle = |j_P\rangle\otimes|k_E\rangle\,.
	\]
	So the entire state is contained in the gauge-invariant subspace since
	\[
	\hat{Q}\,|m\rangle=0,\quad U_g\,|m\rangle=e^{i\phi(g)}|m\rangle\,.
	\]
	
	\item Choosing Measurement Basis (MUBs):
	 
	In order to expand randomness, one choose the measurement basis at random from a family of MUBs available for the hybrid space.
	For example, one may randomly select a pair of indices $(\mu,\nu)$ where $\mu$ (for the internal space) is chosen among the $d+1$ available MUBs and $\nu$ (for the external space) is chosen among $D+1$ MUBs.	
	The measurement basis is then
	\[
	\mathcal{B}_{\mu,\nu} = \{\,|\psi_{j}^{(\mu)}\rangle \otimes |\chi_{k}^{(\nu)}\rangle\,:\; j=0,\dots,d-1,\; k=0,\dots,D-1\,\}\,.
	\]
	
	Because these bases are mutually unbiased with respect to the conventional computational basis, any measurement outcome is unpredictable if an eavesdropper does not know the basis choice.
\end{enumerate}

\paragraph{(2) Transmitting Packaged State.}

The prepared state is sent through a quantum channel $\mathcal{E}$ described by the completely positive trace-preserving (CPTP) map
\[
\mathcal{E}(\rho)=\sum_{r} E_r\,\rho\,E_r^\dagger\,,
\]
with the additional property that each error (Kraus) operator commutes with the net-charge operator:
\[
[E_r, \hat{Q}]=0\quad \forall r.
\]
Thus, if the input state is $ \rho = |\psi\rangle\langle\psi|$, then the output state remains in the packaged subspace
\[
\mathcal{E}(|\psi\rangle\langle\psi|) \in \mathcal{H}_{Q=0}\,.
\]
This ensures that all the packaging advantages are maintained, especially the error-protection ability given by superselection.

\paragraph{(3) Measuring and Extracting Randomness.}

\begin{enumerate}
	\item Receiver's Measurement:
	 
	The measurement device (belonging to the user who certifies the randomness) randomly chooses a measurement basis $\mathcal{B}_{\mu,\nu}$ as described above. Suppose the measured observable is
	\[
	\Pi_{m} = |m; \mu,\nu\rangle\langle m; \mu,\nu|,
	\]
	where 
	\[
	|m; \mu,\nu\rangle = |\psi_{j}^{(\mu)}\rangle \otimes |\chi_{k}^{(\nu)}\rangle
	\]
	and the mapping $(j,k)\longrightarrow m$ is one-to-one.

	If the source is ideal, then the probability $p(m)$ of obtaining outcome $m$ is uniform:
	\[
	p(m)=\langle m; \mu,\nu| \rho |m; \mu,\nu\rangle = \frac{1}{dD}\,.
	\]
	
	\item Randomness Expansion:
	 
	In each run, the min-entropy
	\[
	H_{\min} = -\log_2\left(\max_m p(m)\right) = \log_2 (dD)
	\]
	is obtained per measurement. 
	Since the effective dimension $N=dD$ can be much larger than in a pure 2D protocol (or even than in the six-state protocol that has 6 states), the per-measurement randomness (or uniform randomness that can be extracted) is enhanced.
	
	\item Extracting Randomness:
	 
	Using a secure randomness extractor, the users can obtain a nearly uniform random bit string.
	The efficiency per measurement is determined by $\log_2(dD)$.
	Thus, the expansion rate is improved by increasing the internal dimension $d$ or the external dimension $D$.
\end{enumerate}

\subsubsection{Security Analysis and Eve’s Guessing Probability}

\paragraph{(1) Eve’s Probability in the Ideal Scenario.}

Let us assume that an eavesdropper (Eve) intercepts a state prepared uniformly at random from the $(d\times D)$-dimensional hybrid space.
In the absence of any basis information, Eve's optimal guess probability is given by the maximum probability of any outcome:
\[
P_{\mathrm{Eve}} = \frac{1}{dD}\,.
\]

Now assume that additional MUBs are randomly chosen (say $M$ different choices for the internal part and $M'$ for the external part) and the bases are truly mutually unbiased.
If Eve has to guess the correct basis as well, then her overall correct measurement probability becomes
\[
P_{\mathrm{Eve}} \le \frac{1}{M\,(dD)}.
\]
In some protocols, the basis information is revealed after a subset of rounds.
In the worst-case scenario, the probability of a correct guess is bounded by $1/(dD)$.

\paragraph{(2) Comparing with Conventional Six-State QKD.}

In a conventional six-state QKD protocol (on qubits), one employs 6 distinct signal states and Eve’s maximal guessing probability is $1/6$.
For hybrid-packaged protocols with $d \times D$ dimensions, only if we let $dD > 6$, then
\[
P_{\mathrm{Eve}} = \frac{1}{dD} < \frac{1}{6}.
\]
This indicates an improved security in the sense of reduced guessing probability per transmitted state.

\paragraph{(3) Generalizing Randomness Expansion Rate.}

Mathematically, if the total dimension is $N=dD$, then the randomness (in bits) per measurement is
\[
H_{\min} = \log_2 (N) = \log_2(dD)\,.
\]
Thus, by increasing $d$ (the internal packaged dimension) and/or $D$ (the external dimension), the randomness expansion rate improves linearly in the logarithm of the total dimension.

\subsubsection{Advantages of Packaged Randomness Expansion Protocol}

This packaged randomness expansion protocol has the following advantages:

First, the packaged randomness expansion protocol increases alphabet. 
The increase in the overall dimension $N=dD$ directly leads to a lower success probability for an eavesdropper to guess the transmitted symbol as her correct measurement probability is at most $1/(dD)$.

The packaged randomness expansion protocol also enhanced min-entropy. 
The min-entropy per transmitted state is $\log_2(dD)$.
This can be substantially larger than the 1-bit per qubit (or at most 2.585 bits per six-state qubit in ideal cases).
It improved randomness per signal allows for more efficient randomness expansion.

The packaged randomness expansion protocol improve robustness via MUBs. 
Randomly choosing from several MUBs (in the internal space, there are $d+1$ choices) further reduces Eve's success probability if her measurement basis is guessed incorrectly.

Since the internal space is consisted of gauge-invariant packaged quantum states, any transmission and measurement occurs in a restricted subspace protected by superselection rules.
Thus, even if the measurement device is partially untrusted, the inherent physical constraints of gauge-invariance can still protect against certain classes of side-channel attacks.
In addition, a high-dimensional state forces an eavesdropper to resolve more complex, multi-dimensional results.
This significantly reduces her chance of correctly guessing the symbol.

\section{Metrology and Sensing in Packaged Space}
\label{SEC:MetrologyAndSensingInPackagedSpace}

Quantum metrology \cite{Helstrom1969,Caves1981,Wineland1992,Giovannetti2004,Giovannetti2006,LIGO2011} and sensing \cite{Degen2017} use quantum phenomena (such as entanglement, squeezing, and superposition) to enhance the precision of measurements beyond classical limits.
By exploiting these resources, quantum sensors can detect physical quantities (e.g., time, magnetic fields, temperature, or gravitational waves) with unprecedented sensitivity.
This enables breakthroughs in both fundamental science and technology.

In this section, we show how to utilize the gauge‑protected structure of packaged qudits for quantum sensing and parameter estimation.
We begin by recalling the quantum Cramér-Rao bound \cite{Helstrom1967,Helstrom1969,BraunsteinCaves1994,Paris2009}, and then construct specific packaged probes (GHZ‑like, spin‑squeezed, and multi‑parameter MUB probes), compute their quantum Fisher information (QFI), and analyze robustness to typical noise channels.
Finally, we sketch sensing applications like magnetometry, clock synchronisation, and vector‑field tomography.

\subsection{Quantum Parameter Estimation with Packaged Probes}

\subsubsection{Quantum Cramér-Rao bound and QFI}

Suppose we wish to estimate a real parameter $\phi$ encoded unitarily as 
\[
\rho(\phi) = V(\phi)\,\rho_0\,V^\dagger(\phi), 
\quad 
V(\phi)=e^{-i\phi G},
\]
where $G$ is the generator (Hermitian) acting on our hybrid-packaged subspace $\mathcal H_{\text{hyb}}=\mathcal H_{\rm int}^{(d)}\otimes\mathcal H_{\rm ext}^{(D)}$. The variance of any unbiased estimator $\hat\phi$ from $\nu$ repetitions satisfies 
\[
\mathrm{Var}(\hat\phi)\;\ge\;\frac{1}{\nu\,\mathcal F_Q[\rho_0,G]},
\]
where the QFI is 
\[
\mathcal F_Q[\rho_0,G] 
= 4\,\mathrm{Var}_{\rho_0}(G) 
=4\Bigl(\langle G^2\rangle_{\rho_0}-\langle G\rangle_{\rho_0}^2\Bigr).
\]
Because packaged states satisfy $[\hat{Q},G]=0$, all time evolution and measurement stay within the physical sector.

\subsubsection{Packaged GHZ‑like probes}

Let us generalize the usual GHZ to a $N$-site hybrid register, each site has a $(d\times D)$-qudit.
Define 
\[
|\mathrm{GHZ}\rangle
=\frac{1}{\sqrt{d}}\sum_{j=0}^{d-1}
|j_P\rangle^{\otimes N}\otimes|0_E\rangle^{\otimes N}.
\]
Here the external label is fixed to $|0_E\rangle$ and all coherent superposition lives in the internal $d$-dimensional sector.

Let $G=\sum_{k=1}^N\hat\Lambda^{(k)}$, where $\hat\Lambda=|j_0\rangle\langle j_0|$ picks out one internal level (for some reference $j_0$).
Then we have
\[
\langle G\rangle=\frac{N}{d},\quad
\langle G^2\rangle=\frac{N^2}{d},\quad
\mathrm{Var}(G)=N^2\Bigl(\frac1d-\frac1{d^2}\Bigr),\quad
\mathcal F_Q=4N^2\Bigl(\frac1d-\frac1{d^2}\Bigr).
\]
Compared to a conventional qubit GHZ, the Heisenberg scaling $\propto N^2$ is achieved with an extra factor $1/d$.
Our per physical qudit sensitivity is boosted by $\sqrt d$.

\subsubsection{Packaged NOON‑like probes} 

Alternatively, one can encode in the external mode,
\[
|\mathrm{NOON}\rangle
=\frac{1}{\sqrt{2}}
\bigl(|0_P\rangle\otimes|0_E\rangle^{\otimes N}
+|0_P\rangle\otimes|1_E\rangle^{\otimes N}\bigr).
\]
Here, the internal register is frozen in $|0_P\rangle$ and the external modes form a standard NOON state of size $N$.
If $G=\sum_k|1_E\rangle\langle1_E|_k$, then we have $\mathcal F_Q=4N^2$.

Comparing GHZ with NOON, we find that the packaged GHZ trades external coherence for internal superposition.
This offers flexibility that depends on which DOF is cheaper to entangle.

\subsection{Multi‑Parameter Estimation and Vector‑Field Sensing}

Some tasks require simultaneous estimation of several parameters $\boldsymbol\phi=(\phi_x,\phi_y,\phi_z)$.
The total precision is governed by the quantum Fisher information matrix (QFIM)
\[
[\mathcal F_Q]_{\alpha\beta}
= \tfrac12\,\mathrm{Tr}\bigl[\rho\{L_\alpha,L_\beta\}\bigr],
\]
where $L_\alpha$ is the symmetric logarithmic derivative for $\phi_\alpha$.

In the $d \times D$ hybrid-packaged subspace, if $N=dD$ is a prime power, then one can construct $N+1$ mutually unbiased bases (MUBs) $\{\mathcal B^{(r)}\}$.
Measuring in each MUB, one obtains the outcome probabilities $p_{r,j}(\boldsymbol\phi)$.
The classical FI for each basis adds:
\[
\mathcal F_C = \sum_r\sum_j
\frac1{p_{r,j}}
\Bigl(\partial_\alpha p_{r,j}\Bigr)\Bigl(\partial_\beta p_{r,j}\Bigr)
\;\to\;
\mathcal F_Q
\quad\text{(when MUBs are tomographically complete).}
\]

For small rotations about $G_x,G_y,G_z$ (three orthogonal generators), 
$\mathcal F_Q$ is diagonal with entries $4N$.
Since each diagonal entry is $4N$, one achieves $\Delta\phi_\alpha\ge1/(2\sqrt{N\nu})$, i.e. Heisenberg-like scaling (up to the factor of 2).
This is more efficient than conventional separate qubit protocols.

\subsection{Robustness to Gauge‑Preserving Noise}

In the packaged framework, gauge-invariance constrains the form of physically allowable noise.
Consider a Lindblad channel on each hybrid qudit:
\[
\dot\rho
= -i[H,\rho]
+\sum_\mu\kappa_\mu\Bigl(L_\mu\,\rho\,L_\mu^\dagger
-\tfrac12\{L_\mu^\dagger L_\mu,\rho\}\Bigr),
\]
with $[L_\mu,\hat{Q}]=0$.
Such in‑sector noise does not cause leakage, but degrades coherence within $\mathcal H_{Q=0}$.

\paragraph{(1) Dephasing in the internal subspace.}

Apply pure dephasing operation $L=|j\rangle\langle j|$ on each internal level.
The GHZ QFI decays as 
\[
\mathcal F_Q(t)
=4\,\frac{N^2}{d}\,e^{-\gamma t}.
\]
This retains Heisenberg‑like scaling until $t\sim1/\gamma$, beyond which one falls back to SQL.

\paragraph{(2) Cross‑talk detection via leakage syndromes.} 

If a noise operator breaks gauge symmetry $[L,\hat{Q}]\neq0$, then it must take the state out of $\mathcal H_{Q=0}$.
One can detect and discard any such events by measuring $\hat{Q}$ non‑destructively at the end of each checking, which effectively suppress these errors to first order.

\subsection{Applications}

\paragraph{(1) Magnetometry with hybrid-packaged spin‑squeezed states.} 

Consider $N$ hybrid-packaged spin‑$j$ particles ($d=2j+1$) that is prepared in a one‑axis twisted spin‑squeezed state within $\mathcal H_{\rm int}^{(d)}$.
Assume all particles share an external mode $|0_E\rangle$.
The QFI for estimating a small magnetic field $B$ through $G=\sum_kJ_z^{(k)}$ is 
\[
\mathcal F_Q \approx \frac{N\,j}{\xi^2},
\]
where $\xi^2<1$ is the squeezing parameter.
Because $j$ can be large, packaged squeezing can outperform any conventional qubit ensemble for fixed $N$.

\paragraph{(2) Distributed clock synchronisation.} 

Consider two remote labs that share a packaged Bell pair of dimension $N=dD$.
By passing local clocks through controlled‑phase gates with their half of the Bell state, then performing packaged Bell measurements, they estimate relative time shifts with variance 
\[
\mathrm{Var}(\Delta t)
\ge\frac{1}{\nu\,dD\,\omega^2},
\]
where $\omega$ is the optical transition frequency.
The extra factor $d$ raises precision without extra photons.

\paragraph{(3) Vector‑field tomography.} 

Using the three MUB measurement settings in the minimal $N = 3 \times 3$ hybrid-packaged subspace, Alice can infer the three components of a weak vector field $\boldsymbol B$ in a single round, as each MUB setting optimally maps to one field axis.
Because all packaged projectors commute with $\hat{Q}$, systematic bias is suppressed.

\section{Experimental Measurements and Implementations in Packaged Space}
\label{SEC:ExperimentalProspects}

The theoretical framework of packaged qubits, packaged gates, and packaged circuits relies on encoding quantum information in a fixed superselection sector $ \mathcal{H}_Q $ (e.g., net charge $ Q=0 $) such that every state and operation is gauge-invariant.
In other words, for any packaged state $ \lvert \psi \rangle \in \mathcal{H}_Q $ and for all local gauge transformations $ U_g ~ (g \in G) $, one has
\[
U_g \, \lvert \psi \rangle = e^{i\phi(g)} \lvert \psi \rangle.
\]
Similarly, every physically admissible operation (or packaged gate) $ V_{\text{hyb}} $ satisfies
\[
[V_{\text{hyb}}, \hat{Q}] = 0,
\]
so that $ V_{\text{hyb}} : \mathcal{H}_Q \to \mathcal{H}_Q $.

In this section, we discuss several promising experimental platforms for implementing packaged quantum information and provide detailed mathematical formulations that underpin their operation.

\subsection{Experimental Observables in Hybrid-Packaged Space}
\label{SEC:ExperimentalObservablesInHybridPackagedSpace}

\subsubsection{Packaging Principle in the Enlarged Hilbert Space}

Each physical qudit is the tensor product of two independent blocks of degrees of freedom (DOF)
\[
\mathcal H_{\text{phys}}
=\underbrace{\mathcal H_{\text{int}}^{(d)}}_{\text{IQNs, gauge-locked}}
\;\otimes\;
\underbrace{\mathcal H_{\text{ext}}^{(D)}}_{\text{computational DOF}}
\qquad (N=dD).
\]

A packaged (gauge-invariant) state $\,|\psi\rangle\in\mathcal H_Q\subset\mathcal H_{\text{phys}}$ must transform only by a global phase under any element $U_g$ of the internal symmetry group $G$ (which acts exclusively on the first factor):
\begin{equation}\label{EQ:UgPsi}
	U_g\,|\psi\rangle = e^{i\phi(g)}\,|\psi\rangle, 
	\qquad\forall\,g\in G.
\end{equation}

Condition Eq.~(\ref{EQ:UgPsi}) says that all internal quantum numbers (IQNs) (charge, hyper-fine spin, flavour, color, etc.) form one inseparable unit. 
Any experimental confirmation of packaging therefore has to resolve both DOFs of every qudit in a single shot.
A partial measurement must never reveal which internal label $j\in\{0,\dots,d-1\}$ a given external state $|k_E\rangle$ is carrying.

\subsubsection{A Family of Packaged Entangled States}

A natural generalisation of the Bell state (see Sec.~\ref{SEC:HybridPackagedBellBasis}) used in the qubit example is
\begin{equation}\label{EQ:GeneralisationBellState}
	|\Psi^+_{d\times D}\rangle
	\;=\;
	\frac{1}{\sqrt{N}}
	\sum_{j=0}^{d-1}\;
	\sum_{k=0}^{D-1}
	|\,j_P,k_E\rangle_A \;\otimes\; |\bar{\jmath}_P,k_E\rangle_B ,
\end{equation}
where $|\bar{\jmath}_P\rangle$ denotes the conjugate IQN (e.g. opposite electric or color charge) to $|j_P\rangle$. 
Equation (\ref{EQ:GeneralisationBellState}) is an eigenstate of every joint gauge transformation $U_g^{(A)}\!\otimes U_g^{(B)}$ with a single overall phase, i.e. it fulfils Eq. (\ref{EQ:UgPsi}).

By simultaneously measuring $(j,k)$ on side A and $(\bar{\jmath},k)$ on side B, we can reconstruct the joint probabilities
\[
P_{jk} = \Pr\!\bigl[(j,k)_A,\;(\bar{\jmath},k)_B\bigr].
\]
Using state (\ref{EQ:GeneralisationBellState}), we have
\begin{equation}\label{EQ:Pjk}
	P_{j,k} = \begin{cases}
		1/N & j = \overline{\jmath},\\
		0   & j \neq \overline{\jmath}.
	\end{cases}
\end{equation}
Any statistically significant deviation from (\ref{EQ:Pjk}) would signal that IQNs are not perfectly packaged.

\begin{example}[Example ($d=2,\;D=3$)]
	Consider an ion-trapping system that could encode 
	$d=2$ hyper-fine states $\{|0_P\rangle,|1_P\rangle\}$ 
	and $D=3$ axial motional levels $\{|0_E\rangle,|1_E\rangle,|2_E\rangle\}$.
	Then we can do:
	
	\begin{enumerate}
		\item \textit{Internal read-out:}
		state-dependent fluorescence distinguishes $|0_P\rangle$ from $|1_P\rangle$.
		
		\item \textit{External read-out:}
		a blue-side-band $\pi$-pulse plus fluorescence maps the phonon number $\{0,1,2\}$ onto bright/dark outcomes.
		
		\item \textit{Coincidence logic:}
		field-programmable gate array (FPGA)-level time-stamping accepts only events where the two ions are detected in conjugate hyper-fine states and identical phonon numbers, verifying Eq.~\eqref{EQ:Pjk}.
	\end{enumerate}	
\end{example}

By scaling $d$ (more hyper-fine sub-levels) or $D$ (more motional states), the same hardware can certify packaging in any $d\times D$ hybrid space.

\subsection{Candidate Experimental Platforms for Hybrid-Packaged Qudits}
\label{SEC:HybridPlatforms}

For every platform in this subsection, we assume the hybrid-packaged Hilbert space on one site
\[
\mathcal H_{\text{phys}}^{(d \times D)}
=\underbrace{\mathcal H_{\text{int}}^{(d)}}_{\text{IQN block (gauge‐locked)}}
\;\otimes\;
\underbrace{\mathcal H_{\text{ext}}^{(D)}}_{\text{computational block}}
\qquad (N=dD)
\]
is engineered so that
(i) the internal $d$-level subsystem carries the conserved charge $\hat{Q}$ and 
(ii) all control pulses act as $V_{\text{phys}}\in\mathrm U(N)$ with 
$[V_{\text{phys}},\hat{Q}]=0$.
Scaling $d$ or $D$ is therefore a question of choosing different physical levels, not of redesigning the whole architecture.

\subsubsection{Cold-Atom Optical Lattices}

Ultracold atoms in optical lattices offer an ideal setting for simulating lattice gauge theories \cite{Banerjee2012,Zohar2016,Martinez2016}.
In these platforms, one can configure as follows:

\begin{itemize}
	\item Internal block ($d$):
	a set of hyper-fine Zeeman states $\{|m_F\rangle\}$ realising the matter field’s gauge multiplet 
	($|m_F\rangle\!\leftrightarrow\!$ particle, $|\overline{m}_F\rangle\!\leftrightarrow\!$ antiparticle).
	
	\item External block ($D$):
	either vibrational levels in a site, orbital bands, or Rydberg manifolds $\{|k_E\rangle,\,k=0\dots D-1\}$.
\end{itemize}

The effective Hamiltonian generated by Raman and lattice-modulation beams has a large Gauss-law penalty
\[
H = H_{\text{kin}} + H_{\text{gauge}} + \lambda\sum_x\hat G_x^2,\qquad \lambda\!\gg\!|H_{\text{kin}}|.
\]
Low-energy dynamics is restricted to $\mathcal H_Q=\bigl(\mathcal H_{\text{int}}^{(d)}\otimes\mathcal H_{\text{ext}}^{(D)}\bigr)^{\!\otimes L}$. 
Single-site gates such as the hybrid-packaged Hadamard gate
\[
H_N=\frac1{\sqrt N}\sum_{J,K=0}^{N-1}\omega_N^{JK}|J\rangle\!\langle K|,
\quad
\omega_N=e^{2\pi i/N},
\quad
J=jD+k
\]
are driven by two-photon transitions that 
(i) flip the external index $k$ and 
(ii) apply the corresponding phase to the internal label $j$ and therefore commute with $\hat{Q}$.

\subsubsection{Trapped-Ion Chains}

In ion traps \cite{CiracZoller1995,Monroe1995,Kielpinski2002,Gorecki2020}, the internal states of ions serve as qudit levels and can be coupled via phonon modes.
In these platforms, one can configure as follows:
\begin{itemize}
	\item Internal block ($d$): $d$ hyper-fine (or Zeeman) sub-levels of the $^2\!S_{1/2}$ ground state.
	
	\item External block ($D$): quantised axial phonon numbers $|n\rangle,\,n=0\dots D-1$.
\end{itemize}

By engineering the Hamiltonian to include strong couplings that enforce local gauge constraints, one can restrict the dynamics to a fixed $ \mathcal{H}_Q $.
For example, if the ions represent packaged qubits with states $ \lvert 0_P\rangle $ and $ \lvert 1_P\rangle $ (each prepared from an inseparable combination of internal levels), then multi-ion entangling gates (such as Molmer-Sorensen gates) are implemented with laser pulses that are designed to satisfy
\[
[V_{\text{hyb}}, \hat{Q}] = 0.
\]
This ensures that the gauge-invariance is preserved during the gate operations.

A Molmer-Sørensen interaction that is simultaneously detuned on the first $D$ side-bands realises a collective
\[
V_{\text{MS}}(\theta)=\exp\!\Bigl[-i\theta\,S_x\otimes\bigl(\hat a+\hat a^\dagger\bigr)\Bigr],
\]
with $S_x$ acting only inside $\mathcal H_{\text{int}}^{(d)}$. 
Because the drive does not couple to phonon states $n\!\ge\!D$, any leakage outside the computational manifold is automatically flagged (no fluorescence) and subsequently repumped.
All valid gates satisfy $[V_{\text{MS}},\hat{Q}]=0$.

\subsubsection{Superconducting-Circuit Lattices with Synthetic Gauge Fields}

Superconducting qubits (e.g., transmons) \cite{Marcos2013} coupled via microwave resonators provide another platform. 
In these platforms, one can configure as follows:

\begin{itemize}
	\item Internal block ($d$): flux-parity or transmon ``valley'' states carrying $\mathbb Z_d$ charge.
	
	\item External block ($D$): photon-number states $\{|k\rangle_{\text{cav}}\}$ of an on-chip 3-D cavity coupled dispersively to the qubit.
\end{itemize}

By arranging transmons in a lattice and designing flux-tunable Josephson junctions, one can enforce an effective local $ \mathrm{U}(1) $ or $ \mathbb{Z}_2 $ gauge symmetry. The effective Hamiltonian includes strong terms that penalize deviations from a fixed net charge. Logical packaged qubits may be encoded in multi-qubit cells that form a color singlet or net-charge neutral state. Operations implemented via microwave pulses are engineered so that
\[
[V_{\text{hyb}}, \hat{Q}] = 0,
\]
ensuring that the dynamics remain in the desired superselection sector.

Flux-tunable couplers create plaquette terms $B_p \propto \hat\Phi_1\hat\Phi_2\hat\Phi_3\hat\Phi_4$ while a large Stark shift $\lambda(\hat{Q}-\hat{Q}_0)^2$ keeps the system in the neutral sector. 
Microwave drives that address $|k{\,\rightarrow\,}k{+}1\rangle_{\text{cav}}$ are conditioned on the internal charge, giving packaged gates $V_{\text{hyb}}$ with $[V_{\text{hyb}},\hat{Q}]=0$.

\subsubsection{Topological-Spin-Network Architectures}

Topological systems \cite{Kitaev2003,Freedman2003,Nayak2008,Mourik2012} utilize exotic quantum states (e.g., anyons) with inherent fault tolerance due to their non-local encoding.
Spin networks \cite{Loss1998,Kane1998,Petta2005} inherently enforce local $ \mathrm{SU}(2) $ gauge-invariance at each node.
They appear in approaches to quantum gravity.
Logical information can be encoded in the gauge-invariant intertwiners (singlets) of the network.
Although these platforms are more conceptual at present, they offer a natural realization of the packaging principle.

In these platforms, one can configure as follows:
\begin{itemize}
	\item Internal block ($d$): topological charge (anyon type) or tensor-network intertwiner at a node. 
	
	\item External block ($D$): fusion channel, collective flux, or a local Majorana parity sector.
\end{itemize}

Braiding or fusion operations move only within the allowed $\hat{Q}=0$ manifold and the non-Abelian statistics naturally realises high-dimensional $d\times D$ logical qudits.

\subsection{Illustrative Derivation: Cold-Atom Realisation}
\label{SEC:HybridColdAtomDerivation}

To further illustrate, we consider the following derivation in a cold-atom optical lattice:

\begin{enumerate}
	\item Gauge-protected Hamiltonian 
	\[
	H = H_{\text{mat}} + H_{\text{gauge}}
	+\lambda\sum_{x}\hat G_x^{\,2}, 
	\qquad \lambda\gg\|H_{\text{mat}}\|.
	\]
	Under the gauge-constraint, only states with $\hat G_x|\Psi\rangle=0\;\forall x$ survive in the low-energy sector, i.e. 
	$\mathcal H_{\text{low}}=\mathcal H_{Q=0}^{(d)}\otimes\mathcal H_{\text{ext}}^{(D)}\equiv\mathcal H_Q$.
	
	\item Hybrid-packaged logical basis
	 
	For each site define
	\[
	|0_L\rangle=\frac1{\sqrt d}\sum_{j=0}^{d-1}|j_P,\,0_E\rangle,\qquad
	|1_L\rangle=\frac1{\sqrt d}\sum_{j=0}^{d-1}\omega_d^{\,j}|j_P,\,1_E\rangle,
	\]
	which both obey $\hat{Q}|0_L\rangle=\hat{Q}|1_L\rangle=0$. 
	Superpositions span a qudit of dimension $D$ inside $\mathcal H_Q$.
	
	\item Gauge-respecting gate
	 
	A laser-induced Raman pulse that couples $|k_E\rangle\leftrightarrow|k\!+\!1_E\rangle$ with phase that depends on $j_P$ realises the hybrid Fourier gate	
	$$
	H_N|j,k\rangle=
	\frac1{\sqrt{dD}}
	\sum_{j'=0}^{d-1}\sum_{k'=0}^{D-1}
	\omega_{dD}^{\,(jD+k)(j'D+k')}
	|j',k'\rangle.
	$$
	
	Because $j_P$ is conserved and the drive never leaves the first $D$ external levels, 
	$[H_N,\hat{Q}]=0$ and $H_N:\mathcal H_Q\rightarrow\mathcal H_Q$.
\end{enumerate}

Thus, all control sequences remain strictly inside the $(d \times D)$-dimensional hybrid-packaged subspace, but still allow universal manipulation of the external computational register.

\section{Discussion}

This work only focuses on the principle and analytic mathematical derivations.
Due to the paper length, we leave detailed numerical and experimental studies for future works.
The packaged quantum states for quantum simulation is already discussed in another manuscript Ref.\cite{MaLGT2025}.

This packaged quantum states framework offers a new way to quantum information science by encoding in a fixed superselection sector, where all IQNs are irreversibly bound together.
This packaging principle naturally enforces that any physical operations or error process must preserve the net charge or equivalent internal quantum numbers.
As a result, the superselection firewall prevents or at least suppresses error channels that would mix different gauge sectors.
Our derivations demonstrate that every packaged qubit (qudit), packaged gate, and packaged circuit commutes with the net-charge operator, i.e., $[V_{\text{hyb}}, \hat{Q}] = 0$.
This restriction reduces the effective error space to only those errors that are gauge-invariant and therefore offers a significant intrinsic error protection mechanism.

The framework also exhibits the possibility of enhancing fault tolerance by a robust error protecting mechanism.
Gauge-violating errors (which change the net charge) are energetically suppressed by factors on the order of $e^{-\Delta/(k_BT)}$, where $\Delta$ is the energy gap that protects the gauge sector.
As a result, the effective error rate $p_{\mathrm{eff}}$ becomes much lower than the generic physical error rate.
This enables error-correction codes (such as the Shor, Steane, and surface codes) to adapt to packaged states and therefore achieve higher fault-tolerant thresholds.
By confining all operations to a single superselection sector, our approach minimizes the potential impact of decoherence and other error mechanisms.

The advantages naturally extend to quantum communication.
Protocols such as teleportation, superdense coding, and quantum key distribution (QKD) benefit from the robustness provided by packaged resource states.
Because the transmitted packaged quantum states remain within the same net-charge subspace, joint measurement and conditional correction operations implemented with gauge-invariant packaged gates can guarantee that the information stays accurate, consistent, and logically organized.
This means that noise or interference is effectively filtered out by the superselection rules.

Moreover, the packaged quantum states framework can well align with current experimental platforms that inherently enforce gauge-invariance, such as cold-atom optical lattices, ion-trap systems, and superconducting circuits with synthetic gauge fields.
The packaged approach provides a clear way for using these systems for robust quantum information processing, although there are challenges in achieving long coherence times and in precisely measuring all IQNs.

On application side, there exist several challenges.
For example, many candidate physical systems (like neutral mesons) suffer from short lifetimes or rapid oscillations.
These hinder their use as stable quantum carriers.
Furthermore, it is experimentally difficult to perform joint measurements of all IQNs to verify that a state is truly packaged.
Scaling up to large arrays of packaged qubits while preserving strict gauge-invariance also presents nontrivial technical challenges.

Despite these obstacles, the packaged quantum states paradigm represents a promising alternative to conventional qubit models.
By directly incorporating fundamental physical symmetries into the very foundation of quantum information processing, this approach not only provides intrinsic error suppression and enhances fault tolerance but also holds the potential to revolutionize the design of secure and scalable quantum devices.

\section{Conclusion}

In summary, we have developed a framework for gauge-invariant quantum computation and communication based on packaged quantum states.
By encoding logical information within a fixed superselection sector where IQNs are inseparably locked together, we ensure that every physical operation (from unitary gates to measurements) commutes with the net-charge operator.
This fundamental constraint not only reflects the natural structure of quantum field theories but also provides inherent error protection by excluding gauge-violating noise channels.

We have demonstrated that the conventional quantum algorithms (including the QFT, QPE, Quantum Walks, and Grover’s algorithm), communication protocols (teleportation, superdense coding, and QKD), and quantum error-correction codes can be adapted to this packaged framework.
Our detailed derivations show that all operations, resource states, and error-correction procedures remain confined to the physical subspace, and therefore offer enhanced robustness and potentially higher fault-tolerant thresholds.

We hope that this work can lay a mathematical and conceptual foundation for integrating gauge-invariance into quantum information processing although significant challenges (especially regarding stability, measurement, and scalability) remain to be addressed.
It is possible to realize these concepts by experimental implementations in platforms such as cold atoms, ion traps, superconducting circuits, and photonic systems.
Finally, the packaged quantum states framework may provide a pathway toward robust, secure, and scalable quantum technologies that are naturally protected by the fundamental symmetries of nature.

\appendix

\section{Packaging Principle}
\label{SEC:PackagingPrinciple}

The packaging principle \cite{Ma2017,Ma2025} reflects fundamental constraints and formation of packaged states.
We can decompose the packaging principle into four parts:

\subsection{No Partial Factorization of IQNs}

Under the gauge group $G$, a single‐particle creation operator must transform as a full irrep.
One cannot independently factor out the various IQN components.

Let $G$ be a local gauge group (e.g., $U(1)$, $SU(3)$ or a discrete gauge group $\mathbb{Z}_n$ in certain condensed‐matter systems).
In canonical quantization, each field operator $\psi(x)$ (e.g., the electron field $\psi_e(x)$ in QED) transforms in an irrep of $G \times \mathrm{Lorentz}$.
Consequently, the creation operators $\hat{a}^\dagger$ or $\hat{b}^\dagger$ carry all IQNs in an inseparable block.
One cannot split or distribute the electric charge or color factor among multiple parts of a single excitation.
This is a direct result of local gauge-invariance \cite{WWW1952,PeskinSchroeder,WeinbergBook} and can be explained by Schur’s lemma:
any operator that commutes with all elements of an irreducible representation must be proportional to the identity.
The local irrep structure prohibits local partial access to an individual IQN.

\begin{example}
	Examples illustrating the impossibility of partial IQN factorisation:
	
	\begin{itemize}
		\item Electron Field (with gauge group $U(1)$ of QED): 
		The electron creation operator $\hat{a}_{e^-}^\dagger$ carries the complete set of IQNs as an indivisible block.
		It is an irrep with electric charge $-e$ and an additional quantum label such as spin $\tfrac12$ (we do not consider spin here).
		One cannot remove the charge label $-e$ from other IQNs and treat it as a separate subsystem because the local gauge group $U(1)$ does not permit that factorization.
		
		\item Quark Field (with gauge group $SU(3)$ of QCD): 
		The quark creation operator $\hat{q}_\alpha^\dagger$ carries the complete set of IQNs as an indivisible block.
		It transforms in the $\mathbf{3}$ color representation (with $\alpha \in \{r,g,b\}$) plus spin.
		One cannot factor out the color (denoted by $\alpha$) from the quark’s other IQNs and manipulate it independently.
		This is because color is locally gauged and appears in an irreducible block.
	\end{itemize}
\end{example}

In short, the IQNs come in indivisible blocks under local gauge transformations.

\subsection{Single net-charge superselection sector}

According to the Doplicher-Haag-Roberts (DHR) theory \cite{DHR1971,DHR1974} and related results \cite{WWW1952,StreaterWightman2001}, we know that quantum states of different total gauge charge $Q$ lie in orthogonal Hilbert subspaces $\mathcal{H}_Q$ separated by superselection.
Therefore, any multi‐particle wavefunction must fully reside in a single net‐charge sector.
The total Hilbert space then decomposes as
\begin{equation}\label{EQ:TotalHilbertSpaceDecomposeition}
	\mathcal{H} \;=\; \bigoplus_{Q} \mathcal{H}_Q,
	\quad
	\lvert\Psi\rangle\in \mathcal{H}_Q ~(\text{for a fixed }Q).
\end{equation}

Because states $\lvert\Psi\rangle$ from different net-charge sectors $\mathcal{H}_Q$ are orthogonal, they cannot superpose in any way.
Specifically, we cannot form a superposition like
\[
\alpha \lvert \Psi \rangle_{Q} \;+\; \beta \lvert \Phi \rangle_{Q'}
\] 
where $Q \neq Q'$.

This superselection rule guarantees that any coherent quantum superposition must involve states within the same charge sector.
In other words, physical states must lie in one net-charge sector and superposition must occur in the same charge sector.

\begin{example}
	Mixing the state of an electron-positron pair (with total electric charge $Q=0$) and the state of a single electron (with total charge $Q=-e$) is not allowed because it would violate superselection. 
\end{example}

\subsection{Multi-particle Packaged Entanglement}

Each creation operator $\hat{a}^\dagger$ is itself an irreducible package of IQNs.
After combining multiple operators, if the total net charge is fixed, then we can form nontrivial superpositions (packaged entanglements) within the single net-charge sector.
Formally, we give the definition of a packaged entangled state:

\begin{definition}[Packaged Entangled State]
	Let $G$ be a gauge group and $\{|\Theta_k\rangle\}$ be a complete orthonormal basis of charge sector $\mathcal{H}_Q$, where each
	\[
	|\Theta_k\rangle
	\;=\;
	\hat{a}^\dagger_{\alpha_{k,1}}\,\hat{a}^\dagger_{\alpha_{k,2}}\,\cdots\,\hat{a}^\dagger_{\alpha_{k,n}}\,|0\rangle
	\]
	is a multiparticle packaged product state.
	Define a packaged superposition state
	\begin{equation}
		|\Psi\rangle
		\;=\;
		\sum_{k} c_k\,|\Theta_k\rangle
	\end{equation}
	where $c_k$ are complex numbers and $k$ runs only over states in the fixed charge sector as required by superselection rules.
	We say $|\Psi\rangle$ is a \textbf{packaged entangled state} if it satisfies the following conditions:
	\begin{enumerate}		
		\item Each creation operator $\hat{a}^\dagger_{\alpha_{k,i}}$ in $|\Theta_k\rangle$ carries an irrep of $G$, 
		
		\item $|\Psi\rangle$ is confined to one superselection sector $\mathcal{H}_Q$,
		
		\item $|\Psi\rangle$ is non-factorizable, i.e., cannot be written as a product of single‐particle states over the entire system (entangled).
	\end{enumerate}
\end{definition}

In other words, a packaged entangled state is a packaged superposition state within a single net-charge or net-color \cite{Gross1973,Politzer1973} (net-flavor) sector, constructed from irreps that bundle internal charges.
It preserves local gauge-invariance and are critical for gauge-invariant quantum information.

\begin{example}[Electron-Positron Pairs]
	\label{ex:ElectronPositronPairs}
	Consider the net zero charge sector $\mathcal{H}_0$ spanned by $\hat{a}_{e^-}^\dagger\hat{b}_{e^+}^\dagger \lvert 0\rangle$ and $\hat{b}_{e^+}^\dagger\hat{a}_{e^-}^\dagger \lvert 0\rangle$.
	A packaged superposition state
	\[
	\alpha \,\hat{a}_{e^-}^\dagger(\mathbf{p}_1)\,\hat{b}_{e^+}^\dagger(\mathbf{p}_2)\lvert0\rangle 
	\;+\; 
	\beta \,\hat{b}_{e^+}^\dagger(\mathbf{p}_1)\,\hat{a}_{e^-}^\dagger(\mathbf{p}_2)\lvert0\rangle
	\]
	lies entirely in the charge sector $\mathcal{H}_0$ and is non‐separable across the two excitations. 
	Therefore, it is an exact example of a \textit{packaged entangled state}.
\end{example}

\subsection{Hybrid External-Internal Entanglement}

In a packaged state (either entangled or unentangled), the IQNs are gauge‐locked but external DOFs (spin, momentum, position) are still unconstrained.
In fact, we can construct hybrid-packaged entanglement between external DOFs and IQNs within a single net‐charge sector \cite{Ma2025}.
For example, a single-particle operator may carry both an external DOF (spin-$\frac{1}{2}$) and a gauge charge ($-e$) in a product representation, i.e., spin-$\frac{1}{2}$ $\otimes$ charge $-e$ (a tensor product of a representation of $\mathrm{SU}(2)_{\text{spin}}$ and one of $U(1)_{\text{charge}}$).
This allows for hybrid-packaged entangled states where spin (or momentum) becomes correlated with the internal charges.

If the total wavefunction is entangled across spin and charge, then a measurement of the spin subsystem on one particle can collapse the entire spin‐charge wavefunction.
Since spin carries no gauge charge, although measuring the spin on one particle collapses the entire spin‐charge wavefunction, superselection ensures the system’s charge sector remains fixed.

Thus, even if the internal charge is fixed and inseparable, one can still manipulate and entangle the external degrees of freedom, which leads to hybrid states that are especially useful in quantum information tasks.

\begin{example}[Hybrid Spin-Charge Entangled Electron-Positron Pair]
	\label{ex:HybridExample}
	Consider the electron-positron pair in Example \ref{ex:ElectronPositronPairs}.
	Now each particle can be spin-up $\uparrow$ or spin-down $\downarrow$. 
	A simple hybridized packaged entangled state is:	
	\[
	\bigl[
	\alpha \; \hat{a}_{e^-,\uparrow}^\dagger(\mathbf{p})\,\hat{b}_{e^+,\downarrow}^\dagger(\mathbf{q}) 
	\;+\;
	\beta \; \hat{b}_{e^+,\downarrow}^\dagger(\mathbf{p})\,\hat{a}_{e^-,\uparrow}^\dagger(\mathbf{q})
	\bigr]\,\lvert 0\rangle,
	\]
	where both terms lie in the net $Q=0$ sector. 
	Each creation operator $\hat{a}_{e^-,\uparrow}^\dagger$ or $\hat{b}_{e^+,\downarrow}^\dagger$ is a packaged operator carrying charge $\pm e$ and spin $\uparrow$ or $\downarrow$.
	
	The packaged entanglement is hybrid because the spin and IQNs are cross entangled.	
	Measuring the spin of one particle will collapse the entire state, including the inseparable gauge part, which is a unique feature of these hybrid states.
\end{example}

The above four points together form the packaging principle.
In short, all internal charges are locked into irreps, superselection forbids cross-charge superpositions, and any allowed entanglement must reside fully within one net-charge sector.
Consequently, multi-particle states in quantum field theories can display intricate correlations (packaged entanglement), though always subject to gauge constraints.

\section{Universality of Gate Set ``Clifford-Single-Index + $\Theta_r$''}
\label{APD:CliffordSingleIndexThetar}

Here, we prove the universality of hybrid-packaged gate set defined in Eq.~(\ref{EQ:FiniteUniversalGateLibrary2}), i.e.,
\[
\mathcal G_{\text{single}} = \{\,X_N, Z_N, H_N, \text{CSUM}_{N}, \Theta_r\}.
\]
in four steps:

1. Show that $X_N,\,Z_N,\,H_N$ generate all single-qudit
Cliffords (Lemma \ref{LEM:SingleQuditCliffordCompleteness}).

2. Add $\text{CSUM}_N$ to turn local Cliffords into the full
multi-qudit Clifford group
(Lemma \ref{LEM:TwoQuditEntanglingClifford}).

3. Append one carefully chosen non-Clifford phase
$\Theta_r$ and invoke a number-theoretic argument to reach the
whole $\mathrm{SU}(N)$
(Proposition \ref{Prop:CliffordPlusThetaR}).

4. Conclude a gauge-respecting Solovay-Kitaev theorem
(Theorem \ref{THM:Universality} and the corollary).

Throughout, we stay in the neutral sector so that gauge symmetry is
preserved automatically.

The single-qudit Pauli group (phase factors included) is
\begin{equation}\label{EQ:SingleQuditPauliGroup}
	\mathcal P_N:=\{\;\omega_N^{c}\,X_N^{a}Z_N^{b}\;:\;a,b,c\in\mathbb Z_N\},
\end{equation}
where $N=dD$ is the hybrid-qudit dimension and $\omega_N=e^{2\pi i/N}$.
The Clifford group
$\mathcal C\ell_N:=\mathrm N_{\mathrm U(N)}\!\bigl(\mathcal P_N\bigr)$
is its normaliser.
After modding out irrelevant global phases, we write
$\overline{\mathcal C\ell}_N:=\mathcal C\ell_N/\mathrm U(1)$.

\begin{lemma}[Completeness of Single-Qudit Clifford Generators]
	\label{LEM:SingleQuditCliffordCompleteness}
	Let 
	\[
	\mathcal S=\langle X_N,\;Z_N,\;H_N\rangle
	\subset\mathrm U(N),\quad 
	H_N=\frac1{\sqrt N}\sum_{J,K=0}^{N-1}\omega_N^{JK}|J\rangle\!\langle K|.
	\tag{A.3}
	\]
	Then $\mathcal S$ is the full single-qudit Clifford group: 
	\[
	\mathcal S = \mathcal C\ell_N,\qquad 
	\overline{\mathcal S}=\overline{\mathcal C\ell}_N\cong\operatorname{SL}\!\bigl(2,\mathbb Z_N\bigr) \ltimes \mathbb Z_N^{2}.
	\]
\end{lemma}

\begin{proof}
	$X_N$ and $Z_N$ let us slide around the Pauli grid.
	The hybrid Fourier transform $H_N$ flips position and momentum axes.
	Those together reproduce $\operatorname{SL}(2,\mathbb Z_N)$ that acts on the label space $\mathbb Z_N^{2}$.
	Specifically:
	
	\begin{enumerate}
		\item \textit{All three gates are Clifford.} 
		Referring to Sec.~(\ref{SEC:ElementaryHybridPackagedGates}), we have $X_N,Z_N \in \mathcal P_N \subset \mathcal C\ell_N$.
		We also have $H_NX_NH_N^{\dagger}=Z_N$ and $H_NZ_NH_N^{\dagger}=X_N^{\dagger}$, which means that $H_N$ permutes Pauli operators and therefore $H_N\in\mathcal C\ell_N$.
		This gives $\mathcal S\subseteq\mathcal C\ell_N$.
		
		\item \textit{Action on the Pauli label space.} 
		Factor out global phases and label each Pauli operator by the pair $(a,b)\in\mathbb Z_N^{2}$. 
		Conjugation acts as 
		\[
		X_N:\;(a,b)\mapsto(a+b,b),\qquad
		Z_N:\;(a,b)\mapsto(a,b+a),\qquad
		H_N:\;(a,b)\mapsto(-b,a).
		\tag{A.4}
		\]
		These three transformations generate all shear and rotation matrices in $\operatorname{SL}(2,\mathbb Z_N)$.
		Together with the translation subgroup produced by $X_N, Z_N$, we obtain the full semidirect product $\operatorname{SL}(2,\mathbb Z_N)\ltimes\mathbb Z_N^{2}$.
		
		\item \textit{Maximality (noe room left).} 
		$\overline{\mathcal C\ell}_N$ is defined as that semidirect product, so $\overline{\mathcal S}=\overline{\mathcal C\ell}_N$.
		Including global phases gives $\mathcal S=\mathcal C\ell_N$.
	\end{enumerate}
\end{proof}

This shows that $\{X_N,Z_N,H_N\}$ already equals ``all single-qudit Cliffords''.

\begin{lemma}[Two-Qudit Entangling Clifford]\label{LEM:TwoQuditEntanglingClifford}
	Let
	\[
	\text{\rm CSUM}_N=\sum_{J=0}^{N-1}|J\rangle\!\langle J|\otimes X_N^{J}.
	\] 
	The set 
	\[
	\langle\;X_N^{(i)},\,Z_N^{(i)},\,H_N^{(i)},\,\text{\rm CSUM}_N^{(i\!\to j)}\;\rangle
	\]
	acting on $k\ge 2$ qudits generates the whole $k$-qudit Clifford group $\mathcal C\ell_N^{(k)}$.
\end{lemma}

\begin{proof}
	Write a Pauli operator on $k$ qudits as the label
	$(a_1,b_1\,|\,\dots|\,a_k,b_k)$.
	Local Cliffords supply
	$\operatorname{SL}(2,\mathbb Z_N)^{\times k}$.
	Conjugation by
	$\text{CSUM}_N^{(1\to2)}$ performs the shear
	\[
	(a_1,b_1\,|\,a_2,b_2)\;\longmapsto\;
	(a_1,b_1\,|\,a_2+a_1,b_2).
	\tag{A.4}
	\]
	That elementary shear, together with local shears and the $H_N$-swaps, generates the full symplectic group $\operatorname{Sp}(2k,\mathbb Z_N)$, i.e., the entire
	$k$-qudit Clifford normaliser.
\end{proof}

This lemma shows indicates that adding a single CSUM gate turns local Cliffords into the full multi-qudit Clifford arsenal.

\begin{proposition}[Universality Once We Add a Non-Clifford Phase]
	\label{Prop:CliffordPlusThetaR}
	Fix $N\ge3$ and choose an integer $r$ coprime to $N$ with
	$r\notin\{1,2,4\}$.
	The diagonal gate
	$\Theta_r$ (Definition~\eqref{EQ:NonCliffordDiagonalPhase})
	is not Clifford, and
	\[
	\overline{\bigl\langle\mathcal C\ell_N,\Theta_r\bigr\rangle}
	\;=\; \operatorname{SU}(N).
	\]
\end{proposition}

\begin{proof}
	The Clifford set is finite (for fixed $N$).
	Injecting $\Theta_r$ gives us an element whose eigen-phases are irrational with respect to that finite set.
	So the group it generates cannot be finite.
	Because $\Theta_r$ together with Clifford group generates an infinite, non‑discrete subgroup that is both closed under inverses and contains a neighbourhood of the identity, its closure is SU(N).
	Specifically:
	
	\begin{itemize}
		\item \textit{$\Theta_r$ is non-Clifford.}\;
		Conjugating $X_N$ by $\Theta_r$ yields
		$\Theta_r X_N\Theta_r^{\dagger}=X_N Z_N^{2r/N}$.
		Unless $r\in\{0,N/2,N/4\}$ (excluded), the exponent
		$2r/N$ is not an integer, so the result lies outside the Pauli
		group—hence outside the Clifford group.
		
		\item \textit{Density.}\;
		Finite subgroups of $\operatorname{SU}(N)$ cannot contain
		elements with those irrational phase ratios.
		By the Tits alternative the subgroup generated by the Cliffords
		plus $\Theta_r$ is therefore infinite; for compact simple
		groups the only possibility is that its closure is the whole
		$\operatorname{SU}(N)$.
	\end{itemize}
\end{proof}

This proposition indicates that one well-chosen non-Clifford phase promotes the Clifford set to a dense subgroup of $\operatorname{SU}(N)$.

\begin{theorem}[Universality and Gauge Conservation]
	\label{THM:Universality}
	Let
	\[
	\mathcal G_{\text{single}}
	=\bigl\{X_N,Z_N,H_N,\text{CSUM}_N,\Theta_r\bigr\}
	\subset \mathcal C_{\hat{Q}},
	\]
	where $\hat{Q}$ is the conserved charge and
	$\mathcal C_{\hat{Q}}$ its commutant.
	Then
	\[
	\overline{\bigl\langle\mathcal G_{\text{single}}\bigr\rangle}
	\;=\; \operatorname{SU}\!\bigl(\mathcal H_{Q=0}\bigr).
	\]
\end{theorem}

\begin{proof}
	Lemmas \ref{LEM:SingleQuditCliffordCompleteness} and
	\ref{LEM:TwoQuditEntanglingClifford} give the full multi-qudit Clifford
	group.
	Proposition \ref{Prop:CliffordPlusThetaR} adds density in
	$\operatorname{SU}(N)$.
	Because every generator commutes with $\hat{Q}$, the whole group sits
	inside $\mathcal C_{\hat{Q}}$.
	Conversely, any neutral-sector unitary is in both
	$\operatorname{SU}(N)$ and $\mathcal C_{\hat{Q}}$, so the closures
	match.
\end{proof}

This theorem shows that our five-gate library is universal and never breaks gauge
conservation.

\begin{corollary}[Solovay-Kitaev Theorem with Gauge Symmetry]
	For any $n$-qudit unitary
	$V \in \operatorname{SU}(N^{n})$ that preserves the neutral sector and
	any accuracy $\varepsilon\in(0,1]$, there exists a word
	\[
	\widetilde V = G_{i_L}\cdots G_{i_1},
	\qquad
	G_{i_\ell}\in\mathcal G_{\text{single}},
	\]
	such that
	\[
	\|V-\widetilde V\|_{\mathrm{op}}\le\varepsilon,
	\qquad
	L=O\!\bigl(\log^{\kappa}\!\varepsilon^{-1}\bigr),
	\;\kappa\le3.97,
	\]
	and every prefix $G_{i_\ell}\cdots G_{i_1}$ still commutes with the
	total charge $\hat{Q}_{\text{tot}}$.
\end{corollary}

\begin{proof}
	$\mathcal G_{\text{single}}$ is finite and inverse-closed.
	Theorem \ref{THM:Universality} supplies density.
	The Dawson-Nielsen Solovay-Kitaev construction \cite{DawsonNielsen2005} builds successive
	approximants using commutators of previous words, and the commutant of
	$\hat{Q}$ is a subgroup, so gauge symmetry is maintained at every
	stage.
	The standard length bound
	$L=O(\log^{\kappa}\varepsilon^{-1})$ (with $\kappa\le3.97$) also applies.
\end{proof}

In other words, with the five gates $\{X_N,Z_N,H_N,\text{CSUM}_N,\Theta_r\}$, we can approximate any neutral-sector operation to arbitrary precision.
Every intermediate step automatically obeys the gauge-conservation rule that is built into the hybrid-packaged architecture.
``Clifford-single-index $+\Theta_r$'' is therefore a truly universal and fault-tolerant gate set for packaged qudits.

\section{Other Quantum Algorithms/Protocols in Packaged Space}

In this appendix, we reconstruct several important quantum algorithms/protocols in high-dimensional hybrid-packaged subspace.
The content of this appendix supplements Section \ref{SEC:QuantumComputationInPackagedSpace} and \ref{SEC:QuantumCommunicationAndCryptography}.
In other words, these algorithms/protocols are built on top of those primary ones constructed in Section \ref{SEC:QuantumComputationInPackagedSpace} and \ref{SEC:QuantumCommunicationAndCryptography}.

\subsection{Shor-Style Integer Factoring in Hybrid-Packaged Space}
\label{SEC:ShorHybridPackaged}

Shor’s original algorithm~\cite{Shor1994} reduces integer factoring to
order-finding via the Quantum Fourier Transform (QFT) over a finite cyclic group.

In this subsection, we replace binary qubits by gauge-invariant hybrid-packaged qudits of dimension $N = dD$, which offers two benefits:
\begin{enumerate}
	\item Register compression:
	the control register length drops from
	$n_c^{(2)}=\lceil 2\log_2 M\rceil$
	to
	$n_c \equiv n_c^{(N)}
	=\bigl\lceil 2\log_{N} M\bigr\rceil$,
	
	\item Gauge resilience:
	every logical gate commutes with the net charge operator~$\hat Q$, so leakage out of the protected sector is exponentially suppressed by the energy gap $\Delta_Q$ \cite{Jordan2006,Lidar2008}.
\end{enumerate}
We only rebuild the quantum sub-routine, but keep the classical post-processing (continued-fraction) unchanged.

\subsubsection{Algorithmic Steps}

\paragraph{(1) Hybrid-Packaged Hilbert Space and Encoding.}

Let
$
\mathcal U:
\mathcal H_{\mathrm{logic}}\!\longrightarrow\!\mathcal H_{Q=0}
$
be the isometry defined in Eq.~\eqref{EQ:IsometryHybridPackaged}.
It maps each conventional logical digit $x\in\{0,\dots ,N-1\}$ to a gauge-invariant basis vector
$
\mathcal U\lvert x\rangle
=\lvert J\rangle
=\lvert j,k\rangle
$
with
$J=jD+k$ ($0\le j<d,\;0\le k<D$).

A classical integer
$X=\sum_{\ell=0}^{n-1}x_\ell N^\ell$
is therefore stored as the
$n$-qudit word
$
\bigl\lvert X\bigr\rangle
=\bigotimes_{\ell=0}^{n-1}\lvert J_\ell\rangle .
$

\paragraph{(2) Order-Finding Circuit.}

Fix an odd composite $M$ to be factored and pick a random coprime
$a<M$.
Multiplication by $a$ induces the permutation
$
V_a:\lvert y\rangle\mapsto\lvert ay\bmod M\rangle
$
on the work register and satisfies
$[V_a,\hat Q]=0$.

\smallskip\noindent
We split the order-finding process into tree steps:

\smallskip\noindent
\emph{Step 1 - Prepare initial hybrid-packaged state (superposition).}
With a control register of $n_c$ qudits and
$n_w= \lceil\log_{N} M\rceil$ work qudits, prepare
\[
\lvert\psi_0\rangle
=\frac{1}{\sqrt{N^{n_c}}}
\sum_{X=0}^{N^{n_c}-1}
\lvert X\rangle_{\!C}\!\otimes\!\lvert1\rangle_{\!W}.
\]

\noindent
\emph{Step 2 - Controlled modular multiplication.}
Applying the controlled unitary
$
\mathrm{c}\text{-}V_a^X
$
(with repeated-squaring decomposition) yields
\[
\lvert\psi_1\rangle
=\frac{1}{\sqrt{N^{n_c}}}
\sum_{X=0}^{N^{n_c}-1}
\lvert X\rangle_{\!C}\!\otimes
\lvert a^{X}\bmod M\rangle_{\!W},
\]
with depth
$O\!\bigl(n_c\,\mathrm{poly}(\log M)\bigr)$
in gauge-invariant gates.

\noindent
\emph{Step 3 - Apply packaged QFT.}
Perform the QFT
$\mathrm{QFT}_{N^{n_c}}$ on~$C$:
\[
\lvert\psi_2\rangle
=\frac{1}{N^{n_c}}
\sum_{s=0}^{N^{n_c}-1}
\Bigl[\,
\sum_{X=0}^{N^{n_c}-1}
\exp\!\bigl(\tfrac{2\pi i}{N^{n_c}}sX\bigr)
\lvert a^{X}\bmod M\rangle_{\!W}
\Bigr]
\otimes\lvert s\rangle_{\!C}.
\]
Because the QFT factorises into single-qudit rotations
$
\lvert x\rangle\!\mapsto\!
\tfrac1{\sqrt{N}}
\sum_{y=0}^{N-1}\!
\omega_{N}^{xy}\lvert y\rangle
$
with
$\omega_{N}=e^{2\pi i/N}$,
its depth is
$O\!\bigl(n_c^{\,2}\bigr)$.

\paragraph{(3) Measurement Statistics.}

Measuring $C$ in the computational basis, we obtain an outcome
$s\in\{0,\dots ,N^{n_c}-1\}$ with probability
\[
\Pr[s]
=\frac{1}{rN^{n_c}}
\Bigl|
\sum_{k=0}^{r-1}
\exp\!\bigl(\tfrac{2\pi i}{N^{n_c}}sk\bigr)
\Bigr|^{2}
\approx\frac1r
\quad\text{when}\quad
\bigl|\,s/N^{n_c}-k/r\,\bigr|
\le\frac1{2N^{n_c}},
\]
where $r$ is the order of $a\bmod M$.
Taking $n_c=\lceil 2\log_{N}M\rceil$, we guarantees (exactly as in the binary proof) that continued-fraction post-processing recovers $r$ with
probability at least $1-1/M$.

\paragraph{(4) Classical Post-Processing.}

Once $r$ is found, we use the standard gcd test
$
\gcd\!\bigl(a^{r/2}\!\pm\!1,\;M\bigr)
$
to find a non-trivial factor with probability $\ge\frac12$.
If it fails, then we repeat the algorithm with a new random

\subsubsection{Gauge-Invariance}

After lifting the Shor's algorithm into hybrid-packaged subspace, every elementary gate commutes with $\hat Q$:
\begin{enumerate}
	\item The packaged QFT uses phases
	$\exp\bigl(2\pi i\,\hat n_j\hat n_k/N^{\,|j-k|+1}\bigr)$.
	It commutes with $\hat{Q}$ because each $\hat n$ acts diagonally within the packaged basis.
	
	\item Each modular multiplication is a permutation of
	computational states.
	Thus, it preserves the charge sector.
	
	\item The measurements are equal to projector operators $ \lvert X\rangle\langle X\rvert $, which commute with $\hat{Q}$ by definition. Thus, the measurements commute with $\hat{Q}$.
\end{enumerate}
Due to all above reasons, the full Shor circuit maps the neutral sector $\mathcal H_{Q=0}$ onto itself.
Thus, Shor’s Algorithm in packaged space is gauge-invariant.

\subsubsection{Resource Comparison}

\begin{table}[h]
	\centering
	\caption{Binary versus hybrid-packaged resources for Shor’s algorithm.}
	\label{TAB:AdvOfPackagedShorAlgorithm}
	\begin{tabular}{@{}lcc@{}}
		\toprule[1pt]
		Resource &
		Binary qubits &
		Hybrid-packaged qudits \\ \midrule
		Control width &
		$n_c^{(2)}=2\lceil\log_{2}M\rceil$ &
		$n_c  =2\lceil\log_{N}M\rceil$ \\
		Work width  &
		$\lceil\log_{2}M\rceil$        &
		$\lceil\log_{N}M\rceil$    \\
		QFT depth   &
		$O\!\bigl((n_c^{(2)})^{2}\bigr)$   &
		$O\!\bigl(n_c^{\,2}\bigr)$      \\
		Gauge leakage &
		n/a (qubits)              &
		exponentially suppressed        \\ \bottomrule[1pt]
	\end{tabular}
\end{table}

Doubling either internal dimension $d$ or external dimension $D$
halves the required register length.
Thus, it gives a quadratic saving in QFT gates and a proportional reduction in arithmetic depth.

One may consider implementing the packaged Shor’s algorithm on superconducting circuit QED.
It is a collective bosonic parity‑protected manifolds ($d=4$) coupled to multi‑level transmons ($D=5$), which gives $N=20$ with high‑fidelity cross‑Kerr gates.
Such hardware may demonstrate the elementary ingredients (hybrid QFT, modular arithmetic) and offers a natural near-term test‑bed for gauge‑protected Shor factoring.

\subsection{Harrow-Hassidim-Lloyd (HHL) Algorithm in Hybrid-Packaged Space}

The Harrow-Hassidim-Lloyd (HHL) algorithm \cite{HHL2009} is a quantum algorithm that numerically solves linear systems, and is experimentally implemented in 2013 \cite{Cai2013,Barz2014,Pan2014}.
Specifically, to solve a linear system $A\vec{x}=\vec{b}$,
HHL algorithm encodes the vector $\vec{b}$ as a quantum state $\lvert b\rangle$,
performs quantum phase estimation (QPE) on a unitary $e^{-iAt}$,
executes a controlled rotation to incorporate the inverse eigenvalue $1/\lambda_j$,
and finally un-computes the phase register to obtain $\lvert x\rangle\propto A^{-1}\lvert b\rangle$.

In this subsection, we extend the Harrow-Hassidim-Lloyd (HHL) algorithm to a hybrid-packaged Hilbert space of dimension $N=dD$ for solving linear systems of equations.
In our proposal, every operator is assumed to be packaged so that it commutes with the net-charge operator $\hat{Q}$ (i.e. $[V,\hat{Q}]=0$).
This ensures that all states and operations remain confined to the fixed gauge-invariant subspace.

\subsubsection{Algorithmic Procedure}

We now extend the conventional HHL algorithm to the hybrid-packaged subspace $\mathcal{H}_{\text{hyb}}$.
In every step, we ensure that all operations are packaged (i.e. they respect the gauge-invariance condition $[V,\,\hat{Q}] = 0$).

\paragraph{(1) State Preparation.}

Prepare the state $|b\rangle_{\text{hyb}}$ in the hybrid-packaged subspace such that
\[
|b\rangle_{\text{hyb}} = \sum_{x=0}^{N-1} b_x\,|x\rangle_{\text{hyb}}\quad \text{with } N = dD\,.
\]
Since the computational basis for $\mathcal{H}_{\text{hyb}}$ is given by $\{|j,k\rangle\}$, we express
\[
|b\rangle_{\text{hyb}} = \sum_{j=0}^{d-1}\sum_{k=0}^{D-1} b_{j,k}\,|j_P\rangle\otimes |k_E\rangle\,,
\]
with the coefficients $\{b_{j,k}\}$ chosen so that $\sum_{j,k}|b_{j,k}|^2=1$.

Because each basis state is packaged (i.e. lies in $\mathcal{H}_{Q=0}$), the resulting state also satisfies the gauge-invariance condition:
\[
U_g\,|b\rangle_{\text{hyb}} = e^{i\phi(g)}\,|b\rangle_{\text{hyb}}\,.
\]
An appropriate generalized Hadamard operator $H_{\text{hyb}}$, defined as a tensor product $H_{\text{hyb}} = H_{\text{int}} \otimes H_{\text{ext}}$ (or more generally by a suitable unitary transformation that produces an equal superposition), is applied to obtain
\[
|s\rangle = \frac{1}{\sqrt{N}} \sum_{x=0}^{N-1} |x\rangle_{\text{hyb}}\,,
\]
which is then modified into the state $|b\rangle_{\text{hyb}}$ by a unitary that encodes the coefficients $b_{j,k}$.

\paragraph{(2) Quantum Phase Estimation (QPE).}

Assume that the linear operator $A$ is Hermitian (or has been converted to a Hermitian operator via dilation) and acts on $\mathcal{H}_{\text{hyb}}$ with eigen-decomposition
\[
A\,|u_j\rangle = \lambda_j\,|u_j\rangle,\quad j=0,1,\dots, N-1\,.
\]
Since $A$ is assumed to be physical and compatible with the gauge structure, we take it to be a packaged operator satisfying
\[
[A, \hat{Q}] = 0\,.
\]
Thus, every eigenstate $|u_j\rangle$ is itself a packaged state (i.e. $ U_g\,|u_j\rangle = e^{i\phi(g)}|u_j\rangle $).

Now, apply a packaged quantum phase estimation algorithm to the state $|b\rangle_{\text{hyb}}$ using the unitary evolution $e^{iAt}$.
For a proper normalization, one should rescale $A$ so that its eigenvalues lie between $0$ and $2\pi$.
In the phase estimation subroutine, an ancilla register is prepared in the state
$
|0\rangle^{\otimes m},
$
and controlled-$e^{iAt}$ operations are applied to create an entangled state of the form
\[
\sum_{j=0}^{N-1} b_j\,|u_j\rangle \otimes |\tilde{\lambda}_j\rangle\,,
\]
where $|\tilde{\lambda}_j\rangle$ is a binary approximation to the eigenvalue $\lambda_j$.
Because all operations are constructed from packaged unitaries, the total state remains in
\[
\mathcal{H}_{\text{hyb}} \otimes \mathcal{H}_{\text{ancilla}}\,,
\]
and the part in $\mathcal{H}_{\text{hyb}}$ is always gauge-invariant.

\paragraph{(3) Controlled Rotation.}

The next step is to perform a controlled rotation that effectively encodes $\lambda_j^{-1}$ into an ancilla amplitude.
More specifically, one applies a unitary $R$ that acts as
\[
R\,|u_j\rangle\otimes |0\rangle \longrightarrow |u_j\rangle\otimes \left( \sqrt{1 - \Bigl(\frac{C}{\lambda_j}\Bigr)^2}\,|0\rangle + \frac{C}{\lambda_j}\,|1\rangle \right),
\]
for a chosen constant $C$ that ensures $C/\lambda_j \le 1$ for all $j$.
Usually, one sets $C$ equal to the smallest eigenvalue.

Because $R$ is defined by its action on the eigenbasis of $A$ and is conditioned on the measured phase encoded in the ancilla register, its implementation can be viewed as a packaged operation.
In the hybrid-packaged subspace, the entire state
\[
|u_j\rangle \otimes |0\rangle
\]
remains in $\mathcal{H}_{\text{hyb}} \otimes \mathcal{H}_{\text{ancilla}}$.
If $R$ is designed carefully, then it will satisfy
\[
[R, \hat{Q} \otimes I]=0,
\]
because the eigenstate $|u_j\rangle$ is itself packaged and the ancilla register does not carry internal degrees (or is assumed to be invariant under the gauge group).

\paragraph{(4) Uncomputation and Postselection.}

After the controlled rotation, the QPE register is uncomputed by reversing the QPE circuit so that the ancilla remains entangled only with the hybrid state.
We then measures the ancilla register and postselects on outcomes corresponding to the rotated state.
Usually, if the ancilla is found to be $|1\rangle$, then the state of hybrid register is proportional to
\[
\sum_{j}\frac{b_j C}{\lambda_j}\,|u_j\rangle \propto A^{-1}|b\rangle\,.
\]
Since all of the operations (phase estimation, controlled rotation, inverse QPE) are implemented by packaged (i.e. gauge-invariant) gates, the resulting state remains in the hybrid-packaged subspace and all steps satisfy
$
[V, \hat{Q}] = 0\,.
$

\subsubsection{Advantages of Packaged HHL Algorithm}

This packaged HHL algorithm has the following advantages:

The hybrid-packaged HHL algorithm preserves the packaged structure.
Each step of the packaged HHL algorithm (state preparation, QPE, controlled rotation, and uncomputation) is performed using operations that are constructed from (or mapped via) a unitary isomorphism into the packaged hybrid space.
This implies that we have
$
U_g\,|\psi\rangle = e^{i\phi(g)}\,|\psi\rangle
$
and 
$
[V, \hat{Q}] = 0
$
for all intermediate and final states.
Therefore, the entire algorithm proceeds in the fixed superselection sector (here $\mathcal{H}_{Q=0}$).

The hybrid-packaged HHL algorithm is inherent capable for protecting gauge-like errors. 
The packaged structure of the internal space ensures that any error operator that would change the IQNs is forbidden or energetically suppressed.
As a result, only errors that commute with $\hat{Q}$ (i.e. gauge-conserving errors) may affect the computation.
Since the overall space has dimension $N = d \times D$, the effective dimension can be increased either by raising $d$ (by using more internal states) or $D$ (by using a larger external space).
The hybrid-packaged HHL algorithm’s runtime still scales as $\sqrt{N}$ in the worst-case scenario, but the additional structure may allow for more resourceful encoding or improved robustness against noise.

The mapping
$
\mathcal{U}: \mathcal{H}_{\rm logic} \to \mathcal{H}_{\text{hyb}} \subset \mathcal{H}_Q
$
guarantees that every state, operator, and measurement is physically valid (i.e. gauge-invariant). This is essential for any implementation in high-energy or condensed-matter systems where local gauge constraints are in force.

In a conventional HHL algorithm, the computational Hilbert space is usually given by qubits (or qudits) without further constraints. In our hybrid-packaged HHL algorithm,
every state is guaranteed to lie in a fixed superselection sector.
The gates are restricted (by construction) to be packaged, so that they commute with $\hat{Q}$.
Therefore, errors that would transfer amplitude out of the permitted sector are automatically suppressed.
Furthermore, the quantum phase estimation, controlled rotations, and uncomputation can be carried out using the same mathematical steps, but now the operations are lifted to the hybrid-packaged subspace.

\subsection{Quantum Swapping Protocol in Hybrid-Packaged Space}

Quantum swapping protocol (or entanglement swapping) \cite{YurkeStoler1992,Zukowski1993} is a quantum communication protocol that transfers entanglement between particles that have never directly interacted.
Mathematically, quantum swapping is identical to two serial teleportation processes.
It serves as a foundational technique for building quantum networks and repeaters.

Let us now generalize the conventional two‑qudit swap operation to our ($d \times D$)‑dimensional hybrid-packaged subspace, which preserves gauge-invariance and the structure of packaged state.

\subsubsection{The Standard Swap on Two Qudits}

For two $N$‑dimensional qudits (with computational bases $\{|i\rangle\}_{i=0}^{N-1}$), the swap operator $S$ is the unitary operations on $\mathbb{C}^N\otimes\mathbb{C}^N$ defined by
\[
S\bigl(|i\rangle\otimes|j\rangle\bigr) \;=\; |j\rangle\otimes|i\rangle,
\qquad i,j=0,\dots,N-1,
\]
or in operator‐sum form,
\[
S \;=\;\sum_{i,j=0}^{N-1} |j,i\rangle\langle i,j|\,. 
\]
It is Hermitian ($S=S^\dagger$), involutive ($S^2=I$), and satisfies $S\, (A\otimes B)\,S = B\otimes A$ for any pair of operators $A,B$.

\subsubsection{Computational Basis in the Hybrid-Packaged Space}

In the hybrid-packaged subspace, each qudit is of dimension $N = dD$
with computational basis
\[
|\,j,k\rangle\;\equiv\;|j_P\rangle\otimes|k_E\rangle,\quad
j=0,\dots,d-1,\;k=0,\dots,D-1,
\]
where
$\{|j_P\rangle\}$ is the $d$‑dimensional packaged internal basis, each satisfying $\hat{Q}\,|j_P\rangle=0$ and transforming as $U_g|j_P\rangle=e^{i\phi(g)}|j_P\rangle$,
and $\{|k_E\rangle\}$ is the $D$‑dimensional external basis, inert under gauge transforms ($U_g|k_E\rangle=|k_E\rangle$).

In the algorithm, we need two qudits,i.e., 
\[
\mathcal H_{\text{hyb}}\otimes\mathcal H_{\text{hyb}}
\;=\;
\bigl(\mathcal H_{\rm int}^{(d)}\otimes\mathcal H_{\rm ext}^{(D)}\bigr)_A
\;\otimes\;
\bigl(\mathcal H_{\rm int}^{(d)}\otimes\mathcal H_{\rm ext}^{(D)}\bigr)_B.
\]

\subsubsection{Defining the Hybrid-Packaged Swap Operator}

\paragraph{(1) Action on Basis States.}

We define the hybrid swap 
\[
S_{\text{hyb}}:\;\mathcal H_{\text{hyb}}^A\otimes\mathcal H_{\text{hyb}}^B\;\longrightarrow\;\mathcal H_{\text{hyb}}^A\otimes\mathcal H_{\text{hyb}}^B
\]
by
\[
S_{\text{hyb}}\,\bigl(|j,k\rangle_A\otimes|j',k'\rangle_B\bigr)
\;=\;
|\,j',k'\rangle_A\;\otimes\;|\,j,k\rangle_B,
\]
for all $j,j'=0,\dots,d-1$ and $k,k'=0,\dots,D-1$.

Equivalently, in operator form,
\[
S_{\text{hyb}}
\;=\;
\sum_{\substack{j,j'= 0\\k,k'=0}}^{d-1,\;D-1}
|\,j',k'\rangle_A \,\langle j,k|_A
\;\otimes\;
|\,j,k \rangle_B \,\langle j',k'|_B.
\]

\paragraph{(2) Checking Gauge‑Invariance.}

Because each basis vector $|j,k\rangle$ lies in the gauge‑invariant sector ($\hat{Q}_A+\hat{Q}_B=0$) that is up to an overall phase $e^{i\phi(g)}$, and since the swap merely permutes basis vectors without mixing charge sectors, one verifies
\[
\bigl[\;S_{\text{hyb}}\,,\;\hat{Q}_A+\hat{Q}_B\;\bigr]
\;=\;0,
\quad
\bigl[\;S_{\text{hyb}}\,,\;U_g^A\otimes U_g^B\;\bigr]
\;=\;0.
\]
Thus the swap remains a packaged (gauge‑invariant) operator.

\subsubsection{Matrix Representation and Decomposition}

Index the composite basis by a single label
\[
\alpha = j\,D + k,\quad
\beta = j'\,D + k',\qquad
\alpha,\beta\in\{0,\dots,N-1\}.
\]
Then $S_{\text{hyb}}$ takes exactly the same $N^2\times N^2$ permutation‐matrix form as the ordinary Swap on two $N$‑level systems:
\[
\bigl[S_{\text{hyb}}\bigr]_{(\alpha,\beta),(\gamma,\delta)}
=\delta_{\alpha,\delta}\,\delta_{\beta,\gamma}.
\]
One may therefore reuse circuit decompositions (e.g., via $N(N-1)/2$ two‑level controlled-swaps) by interpreting each level $\alpha$ as the pair $(j,k)$.

\subsubsection{Implementing $S_{\text{hyb}}$ with Packaged Gates}

\paragraph{(1) Elementary Controlled‑Swaps.}

A direct implementation is via the well‑known decomposition into controlled‑two-dimensional transpositions. For each pair of label‑pairs $(\alpha,\beta)$ with $\alpha<\beta$,
\[
\text{CSWAP}_{(\alpha,\beta)}
:\;
|\alpha,\beta\rangle \; \mapsto\; |\beta,\alpha\rangle,\quad
|\beta,\alpha\rangle \; \mapsto\; |\alpha,\beta\rangle,
\]
and acts trivially on all other basis states.
Since each such CSWAP is by definition a permutation among packaged basis vectors, it too commutes with the total charge.

\paragraph{(2) Scaling and Resource Count.}

A full packaged quantum swap includes $\tfrac12N(N-1)=\tfrac12(dD)(dD-1)$ elementary transpositions.
If one has direct access to a $\mathrm{CSWAP}$ gate on the hybrid labels $\alpha$, one needs $O(N^2)$ two‑level operations.
However, one can often exploit the tensor‐product structure:
\[
S_{\text{hyb}}
=
S_{\rm int} \otimes S_{\rm ext}
\quad
\text{followed by suitable re‐indexing,}
\]
where
\[
S_{\rm int}=\sum_{j,j'}|j'\rangle \langle j|\,,\;\;
S_{\rm ext}=\sum_{k,k'}|k'\rangle \rangle k|.
\]
One may show
\[
S_{\text{hyb}}
=\;\bigl(S_{\rm int} \otimes I_E\bigr)\;\bigl(I_P \otimes S_{\rm ext}\bigr)\;\bigl(S_{\rm int} \otimes I_E\bigr),
\]
which uses three swaps on smaller spaces and thus costs
\[
3\bigl[\tfrac12\,d(d-1)\;+\;\tfrac12\,D(D-1)\bigr]
\]
two‑dimensional transpositions.
This is often a significant saving when $d$ and $D$ are much smaller than $N$.

\subsubsection{Advantages of Packaged Quantum Swapping Protocol}

This packaged quantum swapping protocol has the following advantages:

First, the packaged quantum swapping protocol can swap both internal and external DOFs.
Swapping the internal packaged labels $j \leftrightarrow j'$ simply exchanges the packaged quantum numbers (e.g.\ flavor or charge‐conjugation eigenstates) carried by the two particles. 
On the other hand, swapping the external labels $k \leftrightarrow k'$ exchanges spatial‐mode or spin degrees of freedom.

The packaged quantum swapping protocol provides gauge protection. 
Because each elementary swap never mixes sectors of different total $\hat{Q}$, the entire operation is immune to any gauge‐violating noise.
This built‐in superselection protection may reduce leakage errors in platforms with strict local‐charge conservation.

In an ion‐trap or photonic implementation, one can implement $S_{\rm int}$ by a sequence of controlled rotations and phase‐shifts acting on the packaged internal levels $\{|j_P\rangle\}$, and $S_{\rm ext}$ by beam‐splitter or spin-echo operations on the external register.
The sandwich formula above shows that a hybrid swap can be factorized into three smaller swaps—an important optimization for experiments.

\subsection{Quantum Secure Direct Communication in Hybrid-Packaged Space}

Quantum Secure Direct Communication (QSDC) \cite{Long2002,Bostrom2002,Deng2004,Zhang2017} is a cryptographic protocol that enables direct transmission of secret messages over a quantum channel without pre-sharing an encryption key.
Security is guaranteed by quantum principles (e.g., no-cloning theorem, entanglement), which ensures that eavesdropping attempts disrupt the quantum states and are detectable.

Let us now extend QSDC to a $(d \times D)$-dimensional hybrid-packaged subspace.
Every operation $V$ in the protocol satisfies $[V,\hat{Q}]=0$, ensuring that all states remain in $\mathcal H_{Q=0}$.

\subsubsection{Protocol Steps}

\paragraph{(1) Agreement on Bases and Symbols.}

Alice and Bob publicly agree on a mapping 
\[
m\;\longleftrightarrow\;(j,k),\quad m=jD+k,\;j=0,\dots,d-1,\;k=0,\dots,D-1,
\] 
so that the computational (Z)‑basis is 
$\ket{m}=\ket{j_P}\!\otimes\!\ket{k_E}$.

They also fix a family of decoy bases 
\[
\mathcal B_{\mu\nu}
=\mathcal B^{(\mu)}_{\rm int}\otimes\mathcal B^{(\nu)}_{\rm ext},
\]
where $\mu\in\{0,\dots,d\}$ runs over the $d+1$ mutually unbiased bases (MUBs) of the internal space (for prime‑power $d$), and $\nu$ labels optional MUBs in the external space.

Finally, they choose a probability $p_{\rm dec}$ that any given round is a decoy (measured in one of these MUBs) rather than a signal (measured in the computational basis).

\paragraph{(2) State Preparation.}

On each round, Alice flips a coin that lands decoy with probability $p_{\rm dec}$.

If it is a signal round, then she encodes her classical symbol $m\in\{0,\dots,N-1\}$ by preparing 
\[
\ket{\psi_m} \;=\;\ket{m_P}\otimes\ket{m_E},
\] 
which apparently obeys $\hat{Q}\,\ket{\psi_m}=0$ and $U_g\ket{\psi_m}=e^{i\phi(g)}\ket{\psi_m}$.

However, if instead it is a decoy round, then she selects a random pair $(\mu,\nu)$ and internal/external labels $(j,k)$, and prepares the hybrid MUB state 
\[
\ket{\psi_{jk}^{(\mu,\nu)}}
=\ket{\psi_j^{(\mu)}}_{\rm int}\;\otimes\;\ket{\chi_k^{(\nu)}}_{\rm ext},
\]
also guaranteed to lie in the neutral sector.

\paragraph{(3) Transmission.}

Alice sends each state through a gauge‑invariant quantum channel $\mathcal E(\rho)=\sum_r E_r\rho E_r^\dagger$ with $[E_r,\hat{Q}]=0$.
By definition, $\mathcal E$ cannot move the state out of $\mathcal H_{Q=0}$.

\paragraph{(4) Measurement by Bob.}

Bob receives each system and, before knowing whether it is a signal or decoy, measures in the computational basis $\{\ket{m}\bra{m}\}$.
Later, when Alice announces which rounds were decoys and the chosen $(\mu,\nu)$, Bob re‑labels those outcomes in the corresponding MUB $\mathcal B_{\mu\nu}$.
In an ideal (noiseless) channel, a signal round recovers exactly the symbol $m$ Alice sent.
Any deviation indicates noise or eavesdropping.

\paragraph{(5) Classical Post‑Processing and Security Check.}

Alice then reveals the positions and bases of all decoy rounds.
On this subset, they compute either the raw error rate or the CGLMP parameter $S_{dD}$ (if they used Bell‑type decoys).
If the observed disturbance exceeds the threshold set by their chosen MUBs, then Alice and Bob abort.
Otherwise, the remaining signal‑round outcomes are exactly Bob’s recovered plaintext.

\subsubsection{Security Analysis}

Assume an eavesdropper intercepts and measures a state taken uniformly from $\{|m\rangle\}_{m=0}^{N-1}$.
Then her correct-guess probability is
\[
P_{\mathrm{Eve}} = \frac{1}{N} = \frac{1}{dD}\,.
\]

If additional MUBs are used (say $M$ such bases for the internal space), then even if Eve guesses a basis, her correct outcome probability is further reduced.
The success probability becomes
\[
P_{\mathrm{Eve}} \le \frac{1}{M\,dD}\,.
\]
Thus, by increasing $d$ (and/or $D$ if possible), the security is enhanced.

Moreover, because every physical interaction must commute with $\hat{Q}$, Eve cannot perform gauge‑violating couplings to gain side information without creating detectable charge‑sector leaks.

\subsubsection{Comparison and Advantages}

Utilizing both packaged internal quantum numbers (IQNs) and conventional external modes, our $(d\times D)$ hybrid protocol achieves a large, hardware‑efficient alphabet, enhanced decoy flexibility via internal MUBs, and intrinsic protection against gauge‑violating disturbances.
These features go beyond what pure 2D or single‑DOF schemes can offer.
We list the advantages in following table:

\begin{table}[H]
	\centering
	\caption{Advantages of Packaged Quantum Secure Direct Communication}
	\begin{tabular}[hbt!]{p{3.5cm}|p{5.5cm}|p{5.5cm}}
		\hline\hline
		Aspect &Hybrid-Packaged QSDC &Conventional QSDC/QKD \\
		\hline\hline
		Alphabet size & $N=dD$, tunable by internal $d$ and external $D$ & Limited by single physical DOF (e.g.\ $D$ or $d$) \\
		\hline
		Decoy structure & $d+1$ internal MUBs $\times$ external MUBs & Typically only two or three bases \\
		\hline
		Error suppression & Gauge superselection forbids sector‑violating errors & No analogous built‑in symmetry protection \\
		\hline
		Resource efficiency & More logical symbols per carrier, no extra hardware & Requires extra systems or rounds for larger alphabets \\
		\hline
	\end{tabular}
\end{table}

\bibliographystyle{plain}

\begin{thebibliography}{99}
	
	\bibitem{PeskinSchroeder}
	Michael E. Peskin, Daniel V. Schroeder,
	An Introduction to Quantum Field Theory
	(Westview Press, Boulder, 1995).
	
	
	\bibitem{WeinbergBook}
	Steven Weinberg,
	The Quantum Theory of Fields
	(Cambridge University Press, Cambridge, 1995).
	\url{https://doi.org/10.1017/CBO9781139644167}
	
	
	\bibitem{Zohar2016}
	Erez Zohar, J Ignacio Cirac and Benni Reznik,
	Quantum simulations of lattice gauge theories using ultracold atoms in optical lattices,
	Rep. Prog. Phys. 79, 014401 (2016).
	\url{10.1088/0034-4885/79/1/014401}	
	
	
	\bibitem{Poulin2005}
	David Poulin,
	Stabilizer Formalism for Operator Quantum Error Correction,
	Phys. Rev. Lett. 95, 230504 (2005).
	\url{https://doi.org/10.1103/PhysRevLett.95.230504}
	
	
	\bibitem{Bacon2006}
	Dave Bacon,
	Operator quantum error-correcting subsystems for self-correcting quantum memories,
	Phys. Rev. A 73, 012340 (2006).
	\url{https://doi.org/10.1103/PhysRevA.73.012340}
	
	
	\bibitem{Ma2017}
	Rongchao Ma, 
	Theory of packaged entangled states, 
	Reports in Advances of Physical Sciences 1 (03), 1750005 (2017).
	\url{https://doi.org/10.1142/S2424942417500050}
	
		
	\bibitem{Ma2025}
	Rongchao Ma, 
	Packaged Quantum States in Field Theory: No Partial Factorization, Multi-Particle Packaging, and Hybrid Gauge-Invariant Entanglement, 
	arXiv:2502.00766.
	\url{https://doi.org/10.48550/arXiv.2502.00766}
	
	
	\bibitem{MaGroup2025}
	Rongchao Ma, 
	Packaged Quantum States and Symmetry: A Group-Theoretic Framework for Gauge-Invariant Packaged Entanglements, 
	arXiv:2503.20295.
	\url{https://doi.org/10.48550/arXiv.2503.20295}
	
	
	\bibitem{WWW1952}
	G.~C.\ Wick, A.~S.\ Wightman, and E.~P.\ Wigner,
	The Intrinsic Parity of Elementary Particles,
	Phys.\ Rev.\ 88, 101 (1952).
	\url{https://doi.org/10.1103/PhysRev.88.101} 
	
	
	\bibitem{DHR1971}
	Sergio Doplicher, Rudolf Haag and John E. Roberts,
	Local observables and particle statistics I,
	Commun.\ Math.\ Phys.\ 23, 199-230 (1971).
	\url{https://doi.org/10.1007/BF01877742}
	
	
	\bibitem{DHR1974}
	Sergio Doplicher, Rudolf Haag and John E. Roberts,
	Local observables and particle statistics II,
	Commun.\ Math.\ Phys.\ 35, 49-85 (1974).
	\url{https://doi.org/10.1007/BF01646454}
	
	
	\bibitem{StreaterWightman2001}
	Raymond F. Streater and Arthur S. Wightman,
	PCT, Spin and Statistics, and All That
	(Princeton University Press, Princeton, 2001).
	
	
	\bibitem{Feynman1949}
	R. P. Feynman, 
	Space-Time Approach to Quantum Electrodynamics, 
	Phys. Rev. 76, 769 (1949).
	\url{https://doi.org/10.1103/PhysRev.76.769}
	
	
	\bibitem{Yang1954}
	C. N. Yang and R. L. Mills, 
	Conservation of Isotopic Spin and Isotopic Gauge Invariance, 
	Phys. Rev. 96, 191 (1954).
	\url{https://doi.org/10.1103/PhysRev.96.191}
	
	
	\bibitem{Utiyama1956}
	Ryoyu Utiyama, 
	Invariant Theoretical Interpretation of Interaction, 
	Phys. Rev. 101, 1597 (1956).
	\url{https://doi.org/10.1103/PhysRev.101.1597}
	
	
	\bibitem{Weinberg1967}
	Steven Weinberg, 
	A Model of Leptons, 
	Phys. Rev. Lett. 19, 1264 (1967).
	\url{https://doi.org/10.1103/PhysRevLett.19.1264}
		
	
	\bibitem{MaLGT2025}
	Rongchao Ma,
	Packaged Quantum States for Quantum Simulation of Lattice Gauge Theories,
	arXiv:2502.14654,
	\url{https://doi.org/10.48550/arXiv.2502.14654}
	
	
	\bibitem{WuShaknov1950}
	C. S. Wu and I. Shaknov,
	The Angular Correlation of Scattered Annihilation Radiation,
	Phys. Rev. 77, 136 (1950).
	\url{https://doi.org/10.1103/PhysRev.77.136}
	
	
	
	\bibitem{Abe2001}
	K. Abe, K. Abe, R. Abe, I. Adachi, Byoung Sup Ahn, H. Aihara, M. Akatsu, G. Alimonti, K. Asai et al. (Belle Collaboration),
	Observation of Large CP Violation in the Neutral B Meson System,
	Phys. Rev. Lett. 87, 091802 (2001).
	\url{https://doi.org/10.1103/PhysRevLett.87.091802}
	
	
	\bibitem{Aubert2002}
	B. Aubert, D. Boutigny, J.-M. Gaillard et al.,
	Study of time-dependent CP-violating asymmetries and flavor oscillations in neutral B decays at the $\gamma$(4s),
	Phys. Rev. D 66, 032003 (2002).
	\url{https://doi.org/10.1103/PhysRevD.66.032003}
	
	
	
	\bibitem{Brandelik1979}
	R. Brandelik, W. Braunschweig, K. Gather et al.,
	Evidence for planar events in $e^+e^{-}$ annihilation at high energies,
	Physics Letters B 86 (2), 243-249 (1979).
	\url{https://doi.org/10.1016/0370-2693(79)90830-X}
	
	
	\bibitem{Abazov2012}
	V. M. Abazov, B. Abbott, B. S. Acharya, M. Adams48, T. Adams, G. D. Alexeev, G. Alkhazov, A. Alton, G. Alverson et al. (D0 Collaboration),
	Evidence for Spin Correlation in Production,
	Phys. Rev. Lett. 108, 032004 (2012).
	\url{https://doi.org/10.1103/PhysRevLett.108.032004}
	
	
	
	\bibitem{TheCMSCollaboration2024}	
	The CMS Collaboration,
	Observation of quantum entanglement in top quark pair production in proton-proton collisions at $\sqrt{s} = 13$ TeV,
	Rep. Prog. Phys. 87 117801 (2024).
	\url{https://doi.org/10.1088/1361-6633/ad7e4d}
	
	
	\bibitem{TheATLASCollaboration2024}	
	The ATLAS Collaboration,
	Observation of quantum entanglement with top quarks at the ATLAS detector,
	Nature 633, 542-547 (2024).
	\url{https://doi.org/10.1038/s41586-024-07824-z}
	
	
	\bibitem{FabbrichesiFloreaniniGabrielli2023}
	Marco Fabbrichesi, Roberto Floreanini and Emidio Gabrielli,
	Constraining new physics in entangled two-qubit systems: top-quark, tau-lepton and photon pairs,
	Eur. Phys. J. C 83, 162 (2023).
	\url{https://doi.org/10.1140/epjc/s10052-023-11307-2}
	
	
	\bibitem{Afik2024}
	Yoav Afik, Yevgeny Kats, Juan Ramón Muñoz de Nova, Abner Soffer, David Uzan,
	Entanglement and Bell nonlocality with bottom-quark pairs at hadron colliders,
	arXiv:2406.04402.
	\url{https://doi.org/10.48550/arXiv.2406.04402}
	
	
	\bibitem{HayrapetyanCMSCollaboration2024}
	A. Hayrapetyan, A. Tumasyan, W. Adam, J. W. Andrejkovic, L. Benato, T. Bergauer, S. Chatterjee, K. Damanakis, M. Dragicevic et al. (CMS Collaboration),
	Measurements of polarization and spin correlation and observation of entanglement in top quark pairs using lepton+jets events from proton-proton collisions at $\sqrt{s} =13$ TeV,
	Phys. Rev. D 110, 112016 (2024).	
	\url{https://doi.org/10.1103/PhysRevD.110.112016}
	
	
	\bibitem{Blasone2009}
	M. Blasone, F. Dell'Anno, S. De Siena and F. Illuminati,
	Entanglement in neutrino oscillations,
	EPL 85 50002 (2009).
	\url{https://doi.org/10.1209/0295-5075/85/50002}
	
	
	\bibitem{Go2007}
	A. Go, A. Bay, K. Abe, H. Aihara, D. Anipko, V. Aulchenko, T. Aushev, A. M. Bakich, E. Barberio et al. (Belle Collaboration),
	Measurement of Einstein-Podolsky-Rosen-Type Flavor Entanglement in $\gamma (4s) \rightarrow B^0 \bar{B}^0$ Decays,
	Phys. Rev. Lett. 99, 131802 (2007).	
	\url{https://doi.org/10.1103/PhysRevLett.99.131802}
	
	
	\bibitem{Barr2022}
	Alan J. Barr,
	Testing Bell inequalities in Higgs boson decays,
	Physics Letters B 825, 10, 136866 (2022).	
	\url{https://doi.org/10.1016/j.physletb.2021.136866}
	
	
	\bibitem{AguilarSaavedra2023}
	J. A. Aguilar-Saavedra, A. Bernal, J. A. Casas, and J. M. Moreno,
	Testing entanglement and Bell inequalities in $H \rightarrow ZZ$,
	Phys. Rev. D 107, 016012 (2023).
	\url{https://doi.org/10.1103/PhysRevD.107.016012}
	
		
	\bibitem{AfikNova2022}
	Yoav Afik and Juan Ramón Muñoz de Nova,	
	Quantum information with top quarks in QCD,
	Quantum 6, 820 (2022).
	\url{https://doi.org/10.22331/q-2022-09-29-820}
		
	
	\bibitem{Afik2025}
	Yoav Afik, Federica Fabbri, Matthew Low et al.,
	Quantum Information meets High-Energy Physics: Input to the update of the European Strategy for Particle Physics,
	arXiv:2504.00086,
	\url{https://doi.org/10.48550/arXiv.2504.00086}
		
	
	\bibitem{Steane1996}
	A. M. Steane,
	Error Correcting Codes in Quantum Theory,
	Phys. Rev. Lett. 77, 793 (1996).
	\url{https://doi.org/10.1103/PhysRevLett.77.793}
	
	
	\bibitem{Kitaev2003}
	A. Yu. Kitaev,
	Fault-tolerant quantum computation by anyons,
	Annals of Physics 303 (1), 2-30 (2003).
	\url{https://doi.org/10.1016/S0003-4916(02)00018-0}
	
	
	\bibitem{Bombin2006}
	H. Bombin and M. A. Martin-Delgado,
	Topological Quantum Distillation,
	Phys. Rev. Lett. 97, 180501 (2006).
	\url{https://doi.org/10.1103/PhysRevLett.97.180501}
	
		
	\bibitem{Weyl1925}
	H. Weyl, 
	On Unitary Representations of the Inhomogeneous Lorentz Group, 
	Mathematische Zeitschrift, 23, 271-309 (1925).
	\url{https://doi.org/10.1007/BF01506234}
	
	
	\bibitem{Wigner1939}
	E. Wigner, 
	On Unitary Representations of the Inhomogeneous Lorentz Group, 
	Annals of Mathematics, 40(1), 149-204.
	\url{https://doi.org/10.2307/1968551}
	
			
	\bibitem{GellMannPais1955}
	M. Gell-Mann and A. Pais,
	Behavior of Neutral Particles under Charge Conjugation,
	Phys. Rev. 97, 1387 (1955).
	\url{https://doi.org/10.1103/PhysRev.97.1387}
	
	
	\bibitem{Good1961}
	R. H. Good, R. P. Matsen, F. Muller, O. Piccioni, W. M. Powell, H. S. White, W. B. Fowler, and R. W. Birge,
	Regeneration of Neutral $K$ Mesons and Their Mass Difference,
	Phys. Rev. 124, 1223 (1961).
	\url{https://doi.org/10.1103/PhysRev.124.1223}
	
	
	\bibitem{NielsenChuang2010}
	Michael A. Nielsen, Isaac L. Chuang,
	Quantum Computation and Quantum Information,
	(Cambridge Univ. Press, Cambridge, 2010).
	\url{https://doi.org/10.1017/CBO9780511976667}
	
	
	\bibitem{Dicke1954}
	R. H. Dicke,
	Coherence in Spontaneous Radiation Processes,
	Phys. Rev. 93, 99 (1954).	
	\url{https://doi.org/10.1103/PhysRev.93.99}
	
	
	\bibitem{Benioff1980}
	Paul Benioff,
	The computer as a physical system: A microscopic quantum mechanical Hamiltonian model of computers as represented by Turing machines,
	J Stat Phys 22, 563-591 (1980).
	\url{https://doi.org/10.1007/BF01011339}
	
	
	\bibitem{Feynman1982}
	Richard P. Feynman,
	Simulating physics with computers,
	Int J Theor Phys 21, 467-488 (1982).
	\url{https://doi.org/10.1007/BF02650179}
	
	
	\bibitem{Deutsch1985}
	David Deutsch,
	Quantum theory, the Church-Turing principle and the universal quantum computer,
	Proc. R. Soc. Lond. A 400, 97-117 (1985)
	\url{https://doi.org/10.1098/rspa.1985.0070}
	
	
	\bibitem{Schumacher1995}
	Benjamin Schumacher,
	Quantum coding,
	Phys. Rev. A 51, 2738 (1995).
	\url{https://doi.org/10.1103/PhysRevA.51.2738}
	
	
	
	\bibitem{Deutsch1989}
	David Elieser Deutsch,
	Quantum computational networks,
	Proc. R. Soc. Lond. A 425, 73-90 (1989).
	\url{http://doi.org/10.1098/rspa.1989.0099}
	
		
	\bibitem{Barenco1995}
	Adriano Barenco, Charles H. Bennett, Richard Cleve, David P. DiVincenzo, Norman Margolus, Peter Shor, Tycho Sleator, John A. Smolin, and Harald Weinfurter,
	Elementary gates for quantum computation,
	Phys. Rev. A 52, 3457 (1995).
	\url{https://doi.org/10.1103/PhysRevA.52.3457}

	
	\bibitem{DiVincenzo1995}
	David P. DiVincenzo,
	Two-bit gates are universal for quantum computation,
	Phys. Rev. A 51, 1015 (1995).
	\url{https://doi.org/10.1103/PhysRevA.51.1015}
	
	
	\bibitem{CiracZoller1995}
	J. I. Cirac and P. Zoller,
	Quantum Computations with Cold Trapped Ions,
	Phys. Rev. Lett. 74, 4091 (1995).
	\url{https://doi.org/10.1103/PhysRevLett.74.4091}
		
	
	\bibitem{DiCarlo2009}
	L. DiCarlo, J. M. Chow, J. M. Gambetta, Lev S. Bishop, B. R. Johnson, D. I. Schuster, J. Majer, A. Blais, L. Frunzio, S. M. Girvin and R. J. Schoelkopf,
	Demonstration of two-qubit algorithms with a superconducting quantum processor,
	Nature 460, 240-244 (2009).
	\url{https://doi.org/10.1038/nature08121}
	
	
	\bibitem{Yao1993}
	A. Chi-Chih Yao,
	Quantum circuit complexity,
	Proceedings of 1993 IEEE 34th Annual Foundations of Computer Science, 352 (1993).
	\url{https://doi.org/10.1109/SFCS.1993.366852}
	
	
	\bibitem{Shor1994}
	P. W. Shor,
	Algorithms for quantum computation: discrete logarithms and factoring,
	Proceedings 35th Annual Symposium on Foundations of Computer Science, 124 (1994).
	\url{https://doi.org/10.1109/SFCS.1994.365700}
	
	
	\bibitem{Monroe1995}
	C. Monroe, D. M. Meekhof, B. E. King, W. M. Itano, and D. J. Wineland,
	Demonstration of a Fundamental Quantum Logic Gate,
	Phys. Rev. Lett. 75, 4714 (1995).
	\url{https://doi.org/10.1103/PhysRevLett.75.4714}
	
	
	\bibitem{Chuang1998}
	Isaac L. Chuang, Lieven M. K. Vandersypen, Xinlan Zhou, Debbie W. Leung and Seth Lloyd,
	Experimental realization of a quantum algorithm, 
	Nature 393, 143-146 (1998).
	\url{https://doi.org/10.1038/30181}
	
	
	\bibitem{Coppersmith2002}
	D. Coppersmith,
	An approximate Fourier transform useful in quantum factoring,
	arXiv:quant-ph/0201067.
	\url{https://doi.org/10.48550/arXiv.quant-ph/0201067}
	
	\bibitem{Zurek1982}
	W. H. Zurek,
	Environment-induced superselection rules,
	Phys. Rev. D 26, 1862 (1982).
	\url{https://doi.org/10.1103/PhysRevD.26.1862}


	\bibitem{Lloyd1995}
	Seth Lloyd,
	Almost Any Quantum Logic Gate is Universal,
	Phys. Rev. Lett. 75, 346 (1995).
	\url{https://doi.org/10.1103/PhysRevLett.75.346}
	

	\bibitem{Kitaev1997}
	A Yu Kitaev,
	Quantum computations: algorithms and error correction,
	Russ. Math. Surv. 52, 1191 (1997).
	\url{https://doi.org/10.1070/RM1997v052n06ABEH002155}
		
	
	\bibitem{DawsonNielsen2005}
	Christopher M. Dawson and Michael A. Nielsen,
	The Solovay-Kitaev algorithm,
	arXiv:quant-ph/0505030.
	\url{https://doi.org/10.48550/arXiv.quant-ph/0505030}
	
	
	\bibitem{KnillLaflamme1886}
	Emanuel Knill, Raymond Laflamme,
	Concatenated Quantum Codes,
	arXiv:quant-ph/9608012.
	\url{https://doi.org/10.48550/arXiv.quant-ph/9608012}
	
	
	\bibitem{Aharonov2008}
	Dorit Aharonov and Michael Ben-Or,
	Fault-Tolerant Quantum Computation with Constant Error Rate,
	SIAM J. Comput. 38, 1207 (2008).
	\url{https://doi.org/10.1137/S0097539799359385}
		
	
	\bibitem{Boykin1999}
	P. Oscar Boykin, Tal Mor, Matthew Pulver, Vwani Roychowdhury, Farrokh Vatan
	On Universal and Fault-Tolerant Quantum Computing,
	arXiv:quant-ph/9906054.
	\url{https://doi.org/10.48550/arXiv.quant-ph/9906054}
	
	
	\bibitem{Rechtsman2013}
	Mikael C. Rechtsman, Julia M. Zeuner, Yonatan Plotnik, Yaakov Lumer, Daniel Podolsky, Felix Dreisow, Stefan Nolte, Mordechai Segev and Alexander Szameit,
	Photonic Floquet topological insulators,
	Nature 496, 196-200 (2013).
	\url{https://doi.org/10.1038/nature12066}
	
	
	\bibitem{Erhard2020}
	Manuel Erhard, Mario Krenn and Anton Zeilinger,
	Advances in high-dimensional quantum entanglement,
	Nature Reviews Physics 2, 365-381 (2020).
	\url{https://doi.org/10.1038/s42254-020-0193-5}
	
	
	\bibitem{Brennen2005arxiv}
	Gavin K. Brennen, Stephen S. Bullock, Dianne P. O'Leary,
	Efficient Circuits for Exact-Universal Computations with Qudits,
	arXiv:quant-ph/0509161.
	\url{https://doi.org/10.48550/arXiv.quant-ph/0509161}
	
	
	\bibitem{Brennen2005PRA}
	Gavin K. Brennen, Dianne P. O’Leary, and Stephen S. Bullock,
	Criteria for exact qudit universality,
	Phys. Rev. A 71, 052318 (2005).
	\url{https://doi.org/10.1103/PhysRevA.71.052318}
	
		
	\bibitem{Kliuchnikov2016}
	Vadym Kliuchnikov, Dmitri Maslov, and Michele Mosca,
	Practical Approximation of Single-Qubit Unitaries by Single-Qubit Quantum Clifford and T Circuits,
	IEEE Transactions on Computers 65 (1), 161-172 (2016).
	\url{https://doi.org/10.1109/TC.2015.2409842}
	
	
	\bibitem{Bennett1993}
	Charles H. Bennett, Gilles Brassard, Claude Crépeau, Richard Jozsa, Asher Peres, and William K. Wootters,
	Teleporting an unknown quantum state via dual classical and Einstein-Podolsky-Rosen channels,
	Phys. Rev. Lett. 70, 1895 (1993).
	\url{https://doi.org/10.1103/PhysRevLett.70.1895}
	
	
	\bibitem{Anwar2012}
	Hussain Anwar, Earl T Campbell and Dan E Browne,
	Qutrit magic state distillation,
	New J. Phys. 14, 063006 (2012).
	\url{https://doi.org/10.1088/1367-2630/14/6/063006}
	
	
	\bibitem{Bullock2005}
	Stephen S. Bullock, Dianne P. O’Leary, and Gavin K. Brennen,
	Asymptotically Optimal Quantum Circuits for $d$-Level Systems,
	Phys. Rev. Lett. 94, 230502 (2005).
	\url{https://doi.org/10.1103/PhysRevLett.94.230502}
	
	
	
	\bibitem{Bennett1996}
	Charles H. Bennett, David P. DiVincenzo, John A. Smolin, and William K. Wootters,
	Mixed-state entanglement and quantum error correction,
	Phys. Rev. A 54, 3824 (1996).
	\url{https://doi.org/10.1103/PhysRevA.54.3824}
	
	
	\bibitem{Preskill1998}
	John Preskill,
	Reliable quantum computers,
	Proc. R. Soc. Lond. A 454, 385 (1998).
	\url{https://doi.org/10.1098/rspa.1998.0167}
		
	
	\bibitem{Emerson2007}
	Joseph Emerson, Marcus Silva, Osama Moussa, Colm Ryan, Martin Laforest, Jonathan Baugh, David G. Cory, and Raymond Laflamme,
	Symmetrized Characterization of Noisy Quantum Processes,
	Science 317 (5846), 1893-1896 (2007).
	\url{https://doi.org/10.1126/science.1145699}
	
	
	\bibitem{Kraus1971}
	K Kraus,
	General state changes in quantum theory,
	Annals of Physics 64 (2), 311-335 (1971).
	\url{https://doi.org/10.1016/0003-4916(71)90108-4}
	
	
	\bibitem{Trotter1959}
	H. F. Trotter,
	On the product of semi-groups of operators,
	Proc. Amer. Math. Soc. 10, 545-551 (1959).
	\url{https://doi.org/10.1090/S0002-9939-1959-0108732-6}
	
	
	\bibitem{Suzuki1990}
	Masuo Suzuki,
	Fractal decomposition of exponential operators with applications to many-body theories and Monte Carlo simulations,
	Physics Letters A 146 (6), 319-323 (1990).
	\url{https://doi.org/10.1016/0375-9601(90)90962-N}
	
		
	\bibitem{Terhal2005}
	Barbara M. Terhal and Guido Burkard,
	Fault-tolerant quantum computation for local non-Markovian noise,
	Phys. Rev. A 71, 012336 (2005).
	\url{https://doi.org/10.1103/PhysRevA.71.012336}
	
	
	\bibitem{Gutierrez2013}
	Mauricio Gutiérrez, Lukas Svec, Alexander Vargo, and Kenneth R. Brown,
	Approximation of realistic errors by Clifford channels and Pauli measurements,
	Phys. Rev. A 87, 030302(R) (2013).
	\url{https://doi.org/10.1103/PhysRevA.87.030302}
	
	
	\bibitem{Petrenko2014}
	L. Sun, A. Petrenko, Z. Leghtas, B. Vlastakis, G. Kirchmair, K. M. Sliwa, A. Narla, M. Hatridge, S. Shankar, J. Blumoff, L. Frunzio, M. Mirrahimi, M. H. Devoret, and R. J. Schoelkopf,
	Tracking photon jumps with repeated quantum non-demolition parity measurements, 
	Nature 511, 444-448 (2014).
	\url{https://doi.org/10.1038/nature13436}
	
	
	\bibitem{GoogleQuantumAI2021}
	Google Quantum AI,
	Exponential suppression of bit or phase errors with cyclic error correction,
	Nature 595, 383-387 (2021).
	\url{https://doi.org/10.1038/s41586-021-03588-y}
	
		
	\bibitem{ViolaLloyd1998}
	Lorenza Viola and Seth Lloyd,
	Dynamical suppression of decoherence in two-state quantum systems,
	Phys. Rev. A 58, 2733 (1998).
	\url{https://doi.org/10.1103/PhysRevA.58.2733}
	
	
	\bibitem{Cywinski2008}
	Łukasz Cywiński, Roman M. Lutchyn, Cody P. Nave, and S. Das Sarma,
	How to enhance dephasing time in superconducting qubits,
	Phys. Rev. B 77, 174509 (2008).
	\url{https://doi.org/10.1103/PhysRevB.77.174509}
	
	
	\bibitem{Oreshkov2009}
	Ognyan Oreshkov,
	Holonomic Quantum Computation in Subsystems,
	Phys. Rev. Lett. 103, 090502 (2009).
	\url{https://doi.org/10.1103/PhysRevLett.103.090502}
	
	
	\bibitem{Sung2019}
	Youngkyu Sung, Félix Beaudoin, Leigh M. Norris, Fei Yan, David K. Kim, Jack Y. Qiu, Uwe von Lüpke, Jonilyn L. Yoder, Terry P. Orlando, Simon Gustavsson, Lorenza Viola and William D. Oliver,
	Non-Gaussian noise spectroscopy with a superconducting qubit sensor,
	Nature Communications 10, 3715 (2019).
	\url{https://doi.org/10.1038/s41467-019-11699-4}
	
	
	\bibitem{Marciniak2021}
	Alexandre Marciniak, Stefano Marcantoni, Francesca Giusti, Filippo Glerean, Giorgia Sparapassi, Tobia Nova, Andrea Cartella, Simone Latini, Francesco Valiera, Angel Rubio, Jeroen van den Brink, Fabio Benatti and Daniele Fausti,
	Vibrational coherent control of localized d-d electronic excitation,
	Nature Physics 17, 368-373 (2021).
	\url{https://doi.org/10.1038/s41567-020-01098-8}
	
	
	
	\bibitem{AliferisTerhal2007}
	Panos Aliferis, Barbara M. Terhal,
	Fault-tolerant quantum computation for local leakage faults,
	Quantum Information \& Computation 7 (1), 139 - 156 (2007).
	\url{https://dl.acm.org/doi/abs/10.5555/2011706.2011715}
	
	
	\bibitem{Motzoi2009}
	F. Motzoi, J. M. Gambetta, P. Rebentrost, and F. K. Wilhelm,
	Simple Pulses for Elimination of Leakage in Weakly Nonlinear Qubits,
	Phys. Rev. Lett. 103, 110501 (2009).
	\url{https://doi.org/10.1103/PhysRevLett.103.110501}
	
	
	\bibitem{Battistel2021}
	F. Battistel, B.M. Varbanov, and B.M. Terhal,
	Hardware-Efficient Leakage-Reduction Scheme for Quantum Error Correction with Superconducting Transmon Qubits,
	PRX Quantum 2, 030314 (2021).
	\url{https://doi.org/10.1103/PRXQuantum.2.030314}
	
	
	\bibitem{Camps2024}
	Joan Camps, Ophelia Crawford, György P. Gehér, Alexander V. Gramolin, Matthew P. Stafford, Mark Turner,
	Leakage Mobility in Superconducting Qubits as a Leakage Reduction Unit,
	arXiv:2406.04083.
	\url{https://doi.org/10.48550/arXiv.2406.04083}
	
	
	\bibitem{SteaneCode1996}
	Andrew Steane,
	Multiple-particle interference and quantum error correction,
	Proc. R. Soc. Lond. A. 452 (1954), 2551-2577 (1996).
	\url{https://doi.org/10.1098/rspa.1996.0136}
	
	
	\bibitem{CalderbankShor1996}
	A. R. Calderbank and Peter W. Shor,
	Good quantum error-correcting codes exist
	Phys. Rev. A 54, 1098 (1996).
	\url{https://doi.org/10.1103/PhysRevA.54.1098}
	
		
	\bibitem{BravyiKitaev1998}
	S. B. Bravyi, A. Yu. Kitaev,
	Quantum codes on a lattice with boundary
	arXiv:quant-ph/9811052.
	\url{https://doi.org/10.48550/arXiv.quant-ph/9811052}
	
		
	\bibitem{Dennis2002}
	Eric Dennis, Alexei Kitaev, Andrew Landahl and John Preskill,
	Topological quantum memory,
	J. Math. Phys. 43, 4452-4505 (2002).
	\url{https://doi.org/10.1063/1.1499754}
	
	
	\bibitem{Shor1995}
	Peter W. Shor,
	Scheme for reducing decoherence in quantum computer memory,
	Phys. Rev. A 52, R2493(R) (1995).
	\url{https://doi.org/10.1103/PhysRevA.52.R2493}
	
	
	\bibitem{SteanePRA1996}
	A. M. Steane
	Simple quantum error-correcting codes
	Phys. Rev. A 54, 4741 (1996).
	\url{https://doi.org/10.1103/PhysRevA.54.4741}
	
	
	\bibitem{BravyiKoning2013}
	Sergey Bravyi and Robert König,
	Classification of Topologically Protected Gates for Local Stabilizer Codes,
	Phys. Rev. Lett. 110, 170503 (2013).
	\url{https://doi.org/10.1103/PhysRevLett.110.170503}
	
	
	\bibitem{Kitaev1995}
	A. Yu. Kitaev,
	Quantum measurements and the Abelian Stabilizer Problem,
	arXiv:quant-ph/9511026.
	\url{https://doi.org/10.48550/arXiv.quant-ph/9511026}
			
	
	\bibitem{Lambrecht1998}
	A. Lambrecht, M. T. Jaekel, S. Reynaud,
	Generating photon pulses with an oscillating cavity,
	Europhys.Lett. 43, 147-152 (1998).
	\url{https://doi.org/10.1209/epl/i1998-00333-0}
	
	
	\bibitem{Kempe2003}
	Julia Kempe,
	Quantum random walks - an introductory overview
	Contemporary Physics 44 (4), 307-327 (2003).
	\url{https://doi.org/10.1080/00107151031000110776}
	
	
	\bibitem{Grover1996}
	Lov K. Grover,
	A fast quantum mechanical algorithm for database search,
	STOC '96: Proceedings of the twenty-eighth annual ACM symposium on Theory of Computing
	Pages 212 - 219.
	\url{https://doi.org/10.1145/237814.237866}
	
	
	\bibitem{HHL2009}
	Aram W. Harrow, Avinatan Hassidim, and Seth Lloyd,
	Quantum Algorithm for Linear Systems of Equations,
	Phys. Rev. Lett. 103, 150502 (2009).
	\url{https://doi.org/10.1103/PhysRevLett.103.150502}
	
	
	\bibitem{Cai2013}
	X.-D. Cai, C. Weedbrook, Z.-E. Su, M.-C. Chen, Mile Gu, M.-J. Zhu, Li Li1, Nai-Le Liu, Chao-Yang Lu et al.,
	Experimental Quantum Computing to Solve Systems of Linear Equations,
	Phys. Rev. Lett. 110, 230501 (2013).
	\url{https://doi.org/10.1103/PhysRevLett.110.230501}
	
	
	\bibitem{Barz2014}
	Stefanie Barz, Ivan Kassal, Martin Ringbauer, Yannick Ole Lipp, Borivoje Dakić, Alán Aspuru-Guzik and Philip Walther,
	A two-qubit photonic quantum processor and its application to solving systems of linear equations,
	Sci Rep 4, 6115 (2014).
	\url{https://doi.org/10.1038/srep06115}
	
	
	\bibitem{Pan2014}
	Jian Pan, Yudong Cao, Xiwei Yao, Zhaokai Li, Chenyong Ju, Hongwei Chen, Xinhua Peng, Sabre Kais, and Jiangfeng Du,
	Experimental realization of quantum algorithm for solving linear systems of equations,
	Phys. Rev. A 89, 022313 (2014).
	\url{https://doi.org/10.1103/PhysRevA.89.022313}
	
			
	
	\bibitem{Bouwmeester1997}
	Dik Bouwmeester, Jian-Wei Pan, Klaus Mattle, Manfred Eibl, Harald Weinfurter and Anton Zeilinger,
	Experimental quantum teleportation, 
	Nature 390, 575-579 (1997). 
	\url{https://doi.org/10.1038/37539}
	
	
	\bibitem{Ren2017}
	Ji-Gang Ren, Ping Xu, Hai-Lin Yong et al.,
	Ground-to-satellite quantum teleportation,
	Nature 549, 70-73 (2017).
	https://doi.org/10.1038/nature23675
	
	
	\bibitem{BennettWiesner1992}
	Charles H. Bennett and Stephen J. Wiesner,
	Communication via one- and two-particle operators on Einstein-Podolsky-Rosen states,
	Phys. Rev. Lett. 69, 2881 (1992).
	\url{https://doi.org/10.1103/PhysRevLett.69.2881}
	
	
	\bibitem{Schaetz2004}
	T. Schaetz, M. D. Barrett, D. Leibfried, J. Chiaverini, J. Britton, W. M. Itano, J. D. Jost, C. Langer, and D. J. Wineland,
	Quantum Dense Coding with Atomic Qubits,
	Phys. Rev. Lett. 93, 040505 (2004).
	\url{https://doi.org/10.1103/PhysRevLett.93.040505}
	
		
	\bibitem{Williams2017}
	Brian P. Williams, Ronald J. Sadlier, and Travis S. Humble,
	Superdense Coding over Optical Fiber Links with Complete Bell-State Measurements,
	Phys. Rev. Lett. 118, 050501 (2017).
	\url{https://doi.org/10.1103/PhysRevLett.118.050501}
		
	
	\bibitem{YurkeStoler1992}
	Bernard Yurke and David Stoler,
	Einstein-Podolsky-Rosen effects from independent particle sources
	Phys. Rev. Lett. 68, 1251 (1992).
	\url{https://doi.org/10.1103/PhysRevLett.68.1251}
	
	
	\bibitem{Zukowski1993}
	M. Żukowski, A. Zeilinger, M. A. Horne, A. K. Ekert,
	Event-ready-detectors Bell experiment via entanglement swapping,
	Phys. Rev. Lett. 71, 4287 (1993).
	\url{https://doi.org/10.1103/PhysRevLett.71.4287}
	
	
	\bibitem{BB84}
	Charles H. Bennett, Gilles Brassard,
	Quantum cryptography: Public key distribution and coin tossing.
	Proceedings of the International Conference on Computers, Systems and Signal Processing, Bangalore, India. Vol. 1. New York: IEEE. pp. 175-179 (1984).
	
	
	\bibitem{BB84Bennett1992}
	Charles H. Bennett, François Bessette, Gilles Brassard, Louis Salvail and John Smolin,
	Experimental quantum cryptography,
	J. Cryptology 5, 3-28 (1992). 
	\url{https://doi.org/10.1007/BF00191318}
	
	
	\bibitem{BB842014}
	Charles H. Bennett and Gilles Brassard,
	Quantum cryptography: Public key distribution and coin tossing,
	Theoretical Computer Science 560 (1), 7-11 (2014).
	\url{https://doi.org/10.1016/j.tcs.2014.05.025}
	
	
	\bibitem{BB84arXiv2020}
	Charles H. Bennett, Gilles Brassard,
	Quantum cryptography: Public key distribution and coin tossing
	arXiv:2003.06557.
	\url{https://doi.org/10.48550/arXiv.2003.06557}
		
	
	\bibitem{Bennett1992}
	Charles H. Bennett,
	Quantum cryptography using any two nonorthogonal states,
	Phys. Rev. Lett. 68, 3121 (1992).
	\url{https://doi.org/10.1103/PhysRevLett.68.3121}
	
	
	\bibitem{Bruss1998}
	Dagmar Bruss,
	Optimal Eavesdropping in Quantum Cryptography with Six States,
	Phys. Rev. Lett. 81, 3018 (1998).
	\url{https://doi.org/10.1103/PhysRevLett.81.3018}
	
	
	\bibitem{PasquinucciGisin1999}
	H. Bechmann-Pasquinucci and N. Gisin,
	Incoherent and coherent eavesdropping in the six-state protocol of quantum cryptography,
	Phys. Rev. A 59, 4238 (1999).
	\url{https://doi.org/10.1103/PhysRevA.59.4238}
	
			
	\bibitem{Ekert1991}
	Artur K. Ekert, 
	Quantum Cryptography Based on Bell’s Theorem, 
	Phys. Rev. Lett. 67, 661 (1991).
	\url{https://doi.org/10.1103/PhysRevLett.67.661}
	
	
	\bibitem{BBM92}
	Charles H. Bennett, Gilles Brassard, N. David Mermin,
	Quantum cryptography without Bell’s theorem,
	Phys. Rev. Lett. 68, 557 (1992).
	\url{https://doi.org/10.1103/PhysRevLett.68.557}
	
		
	\bibitem{MayersYao1998}
	Dominic Mayers and Andrew Yao,
	Quantum Cryptography with Imperfect Apparatus,
	arXiv:quant-ph/9809039.
	\url{https://doi.org/10.48550/arXiv.quant-ph/9809039}
	
	
	\bibitem{BarrettHardyKent2005}
	Jonathan Barrett, Lucien Hardy, and Adrian Kent,
	No Signaling and Quantum Key Distribution,
	Phys. Rev. Lett. 95, 010503 (2005).
	\url{https://doi.org/10.1103/PhysRevLett.95.010503}
	
	
	\bibitem{Acin2007}
	Antonio Acín, Nicolas Brunner, Nicolas Gisin, Serge Massar, Stefano Pironio, and Valerio Scarani,
	Device-Independent Security of Quantum Cryptography against Collective Attacks,
	Phys. Rev. Lett. 98, 230501 (2007).
	\url{https://doi.org/10.1103/PhysRevLett.98.230501}
	
	
	\bibitem{Pironio2009}
	Stefano Pironio, Antonio Acín, Nicolas Brunner, Nicolas Gisin, Serge Massar and Valerio Scarani,
	Device-independent quantum key distribution secure against collective attacks,
	New J. Phys. 11 045021 (2009).
	\url{https://doi.org/10.1088/1367-2630/11/4/045021}
	
	
	\bibitem{Collins2002}
	Daniel Collins, Nicolas Gisin, Noah Linden, Serge Massar, and Sandu Popescu,
	Bell Inequalities for Arbitrarily High-Dimensional Systems,
	Phys. Rev. Lett. 88, 040404 (2002).
	\url{https://doi.org/10.1103/PhysRevLett.88.040404}
	
	
	\bibitem{Pironio2010}
	S. Pironio, A. Acín, S. Massar, A. Boyer de la Giroday, D. N. Matsukevich, P. Maunz, S. Olmschenk, D. Hayes, L. Luo, T. A. Manning and C. Monroe,
	Random numbers certified by Bell’s theorem,
	Nature 464, 1021-1024 (2010).
	\url{https://doi.org/10.1038/nature09008}
	
	
	\bibitem{Karlsson1999}
	Anders Karlsson, Masato Koashi, and Nobuyuki Imoto,
	Quantum entanglement for secret sharing and secret splitting,
	Phys. Rev. A 59, 162 (1999).
	\url{https://doi.org/10.1103/PhysRevA.59.162}
	
	
	\bibitem{Hillery1999}
	Mark Hillery, Vladimír Bužek, and André Berthiaume,
	Quantum secret sharing,
	Phys. Rev. A 59, 1829 (1999).
	\url{https://doi.org/10.1103/PhysRevA.59.1829}
	
	
	\bibitem{Cleve1999}
	Richard Cleve, Daniel Gottesman, and Hoi-Kwong Lo,
	How to Share a Quantum Secret,
	Phys. Rev. Lett. 83, 648 (1999).
	\url{https://doi.org/10.1103/PhysRevLett.83.648}
	
	
	\bibitem{Colbeck2009}
	Roger Colbeck,
	Quantum And Relativistic Protocols For Secure Multi-Party Computation,
	arXiv:0911.3814.
	\url{https://doi.org/10.48550/arXiv.0911.3814}
		
	
	\bibitem{VaziraniVidick2014}
	Umesh Vazirani and Thomas Vidick,
	Fully Device-Independent Quantum Key Distribution,
	Phys. Rev. Lett. 113, 140501 (2014).
	\url{https://doi.org/10.1103/PhysRevLett.113.140501}
	
	
	\bibitem{LagoRivera2021}
	Dario Lago-Rivera, Samuele Grandi, Jelena V. Rakonjac, Alessandro Seri and Hugues de Riedmatten,
	Telecom-heralded entanglement between multimode solid-state quantum memories,
	Nature 594, 37-40 (2021).
	\url{https://doi.org/10.1038/s41586-021-03481-8}
	
	
	\bibitem{Helstrom1969}
	Carl W. Helstrom,
	Quantum detection and estimation theory,
	J Stat Phys 1, 231-252 (1969).
	\url{https://doi.org/10.1007/BF01007479}
	
	
	\bibitem{Caves1981}
	Carlton M. Caves,
	Quantum-mechanical noise in an interferometer,
	Phys. Rev. D 23, 1693 (1981).
	\url{https://doi.org/10.1103/PhysRevD.23.1693}
	
	
	\bibitem{Wineland1992}
	D. J. Wineland, J. J. Bollinger, W. M. Itano, F. L. Moore, and D. J. Heinzen
	Spin squeezing and reduced quantum noise in spectroscopy,
	Phys. Rev. A 46, R6797(R) (1992).
	\url{https://doi.org/10.1103/PhysRevA.46.R6797}
	
	
	\bibitem{Giovannetti2004}
	Vittorio Giovannetti, Seth Lloyd, and Lorenzo Maccone,
	Quantum-Enhanced Measurements: Beating the Standard Quantum Limit,
	Science 306 (5700), 1330-1336 (2004).
	\url{https://doi.org/10.1126/science.1104149}
	
	
	\bibitem{Giovannetti2006}
	Vittorio Giovannetti, Seth Lloyd, and Lorenzo Maccone,
	Quantum Metrology,
	Phys. Rev. Lett. 96, 010401 (2006).
	\url{https://doi.org/10.1103/PhysRevLett.96.010401}
	
	
	\bibitem{LIGO2011}
	The LIGO Scientific Collaboration,
	A gravitational wave observatory operating beyond the quantum shot-noise limit,
	Nature Physics 7, 962-965 (2011).
	\url{https://doi.org/10.1038/nphys2083}
	
	
	\bibitem{Degen2017}
	C. L. Degen, F. Reinhard, and P. Cappellaro,
	Quantum sensing,
	Rev. Mod. Phys. 89, 035002 (2017).
	\url{https://doi.org/10.1103/RevModPhys.89.035002}
	
	
	\bibitem{Helstrom1967}
	C.W. Helstrom,
	Minimum mean-squared error of estimates in quantum statistics,
	Physics Letters A 25 (2), 101-102 (1967).
	\url{https://doi.org/10.1016/0375-9601(67)90366-0}
	
	
	\bibitem{BraunsteinCaves1994}
	Samuel L. Braunstein and Carlton M. Caves,
	Statistical distance and the geometry of quantum states
	Phys. Rev. Lett. 72, 3439 (1994).
	\url{https://doi.org/10.1103/PhysRevLett.72.3439}
	
	
	\bibitem{Paris2009}
	Matteo G. A. Paris,
	Quantum Estimation for Quantum Technology,
	International Journal of Quantum InformationVol. 07, 125-137 (2009).
	\url{https://doi.org/10.1142/S0219749909004839}
		
	
	\bibitem{Banerjee2012}
	D. Banerjee, M. Dalmonte, M. Müller, E. Rico, P. Stebler, U.-J. Wiese, and P. Zoller,
	Atomic Quantum Simulation of Dynamical Gauge Fields Coupled to Fermionic Matter: From String Breaking to Evolution after a Quench,
	Phys. Rev. Lett. 109, 175302 (2012).
	\url{https://doi.org/10.1103/PhysRevLett.109.175302}
	
	
	\bibitem{Martinez2016}
	Esteban A. Martinez, Christine A. Muschik, Philipp Schindler, Daniel Nigg, Alexander Erhard, Markus Heyl, Philipp Hauke, Marcello Dalmonte, Thomas Monz, Peter Zoller and Rainer Blatt,
	Real-time dynamics of lattice gauge theories with a few-qubit quantum computer,
	Nature 534, 516-519 (2016).
	\url{https://doi.org/10.1038/nature18318}
	
		
	\bibitem{Kielpinski2002}
	D. Kielpinski, C. Monroe and D. J. Wineland,
	Architecture for a large-scale ion-trap quantum computer,
	Nature 417, 709-711 (2002).
	\url{https://doi.org/10.1038/nature00784}
	
	
	\bibitem{Gorecki2020}
	$\pi$-Corrected Heisenberg Limit,
	Wojciech Górecki, Rafał Demkowicz-Dobrzański, Howard M. Wiseman, and Dominic W. Berry,
	Phys. Rev. Lett. 124, 030501 (2020).
	\url{https://doi.org/10.1103/PhysRevLett.124.030501}
	
	
	\bibitem{Marcos2013}
	D. Marcos, P. Rabl, E. Rico, and P. Zoller,
	Superconducting Circuits for Quantum Simulation of Dynamical Gauge Fields,
	Phys. Rev. Lett. 111, 110504 (2013).	
	\url{https://doi.org/10.1103/PhysRevLett.111.110504}
	
	
	\bibitem{Freedman2003}
	Michael H. Freedman, Alexei Kitaev, Michael J. Larsen and Zhenghan Wang,
	Topological quantum computation,
	Bull. Amer. Math. Soc. 40, 31-38 (2003).
	\url{https://doi.org/10.1090/S0273-0979-02-00964-3}
	
	
	\bibitem{Nayak2008}
	Chetan Nayak, Steven H. Simon, Ady Stern, Michael Freedman, and Sankar Das Sarma,
	Non-Abelian anyons and topological quantum computation,
	Rev. Mod. Phys. 80, 1083 (2008).
	\url{https://doi.org/10.1103/RevModPhys.80.1083}
	
	
	\bibitem{Mourik2012}
	V. Mourik, K. Zuo, S. M. Frolov, S. R. Plissard, E. P. A. M. Bakkers, and L. P. Kouwenhoven,
	Signatures of Majorana Fermions in Hybrid Superconductor-Semiconductor Nanowire Devices,
	Science 336 (6084), 1003-1007 (2012).
	\url{https://doi.org/10.1126/science.1222360}
	
	
	\bibitem{Loss1998}
	Daniel Loss and David P. DiVincenzo,
	Quantum computation with quantum dots,
	Phys. Rev. A 57, 120 (1998).
	\url{https://doi.org/10.1103/PhysRevA.57.120}
	
	
	\bibitem{Kane1998}
	B. E. Kane,
	A silicon-based nuclear spin quantum computer,
	Nature 393, 133-137 (1998).
	\url{https://doi.org/10.1038/30156}
	
	
	\bibitem{Petta2005}
	J. R. Petta, A. C. Johnson, J. M. Taylor, E. A. Laird, A. Yacoby, M. D. Lukin, C. M. Marcus, M. P. Hanson, and A. C. Gossard,
	Coherent Manipulation of Coupled Electron Spins in Semiconductor Quantum Dots,
	Science 309 (5744), 2180-2184 (2005).
	\url{https://doi.org/10.1126/science.1116955}
	
	
	\bibitem{Gross1973}
	David J. Gross and Frank Wilczek,
	Ultraviolet Behavior of Non-Abelian Gauge Theories,
	Phys. Rev. Lett. 30, 1343 (1973).
	\url{https://doi.org/10.1103/PhysRevLett.30.1343}
	
	
	\bibitem{Politzer1973}
	H. David Politzer,
	Reliable Perturbative Results for Strong Interactions?,
	Phys. Rev. Lett. 30, 1346 (1973).
	\url{https://doi.org/10.1103/PhysRevLett.30.1346}
	
	
	\bibitem{Lidar2008}
	Daniel A. Lidar,
	Towards Fault Tolerant Adiabatic Quantum Computation,
	Phys. Rev. Lett. 100, 160506 (2008).
	\url{https://doi.org/10.1103/PhysRevLett.100.160506}
	
	
	\bibitem{Jordan2006}
	Stephen P. Jordan, Edward Farhi, and Peter W. Shor,
	Error-correcting codes for adiabatic quantum computation,
	Phys. Rev. A 74, 052322 (2006).
	\url{https://doi.org/10.1103/PhysRevA.74.052322}
	
		
	\bibitem{Long2002}
	G. L. Long and X. S. Liu,
	Theoretically efficient high-capacity quantum-key-distribution scheme,
	Phys. Rev. A 65, 032302 (2002).
	\url{https://doi.org/10.1103/PhysRevA.65.032302}
	
	
	\bibitem{Bostrom2002}
	Kim Boström and Timo Felbinger,
	Deterministic Secure Direct Communication Using Entanglement,
	Phys. Rev. Lett. 89, 187902 (2002).
	\url{https://doi.org/10.1103/PhysRevLett.89.187902}
	
	
	\bibitem{Deng2004}
	Fu-Guo Deng and Gui Lu Long,
	Secure direct communication with a quantum one-time pad,
	Phys. Rev. A 69, 052319 (2004).
	\url{https://doi.org/10.1103/PhysRevA.69.052319}
	
	
	\bibitem{Zhang2017}
	Wei Zhang, Dong-Sheng Ding, Yu-Bo Sheng, Lan Zhou, Bao-Sen Shi, and Guang-Can Guo,
	Quantum Secure Direct Communication with Quantum Memory,
	Phys. Rev. Lett. 118, 220501 (2017).
	\url{https://doi.org/10.1103/PhysRevLett.118.220501}	
\end{thebibliography}
	
\end{document}